\documentclass[aps,amssymb,showkeys]{revtex4}

\usepackage{array}
\usepackage{scalerel}
\usepackage{tabularx}
\usepackage{multirow}
\usepackage{amsmath}
\usepackage{amsthm}
\usepackage{graphicx} 
\usepackage{amssymb}
\usepackage[pdftex, bookmarks, colorlinks=true, plainpages = false, citecolor = red, urlcolor = blue, filecolor = blue]{hyperref}
\usepackage{color}
\usepackage{mathrsfs}
\usepackage[shortlabels]{enumitem}
\setenumerate{listparindent=\parindent}
\usepackage{diagbox}
\usepackage{subeqnarray}

\newcommand{\red}{\textcolor{red}}
\newcommand{\be}{\begin{equation}}
\newcommand{\ee}{\end{equation}}
\newcommand{\bea}{\begin{eqnarray}}
\newcommand{\eea}{\end{eqnarray}}
\newcommand{\bean}{\begin{eqnarray*}}
\newcommand{\eean}{\end{eqnarray*}}

\theoremstyle{plain}
\newtheorem{theorem}{Theorem}

\newtheorem{cor}[theorem]{Corollary}
\newtheorem{lem}[theorem]{Lemma}

\theoremstyle{definition}
\newtheorem{defn}[theorem]{Definition}

\def\clap#1{\hbox to 0pt{\hss#1\hss}}
\def\mathclap{\mathpalette\mathclapinternal}

\def\mathclapinternal#1#2{%
\clap{$\mathsurround=0pt#1{#2}$}}

\allowdisplaybreaks

\begin{document}
\title{A solution space for a system of null-state partial differential equations III}

\date{\today}

\author{Steven M.\ Flores}
\email{steven.flores@helsinki.fi} 
\affiliation{
Department of Mathematics, University of Michigan, Ann Arbor, Michigan, 48109-2136, USA \\ and \\
Department of Mathematics \& Statistics, University of Helsinki, P.O. Box 68, 00014, Finland}

\author{Peter Kleban}
\email{kleban@maine.edu} 
\affiliation{LASST and Department of Physics \& Astronomy, University of Maine, Orono, Maine, 04469-5708, USA}

\begin{abstract}  

This article is the third of four that completely and rigorously characterize a solution space $\mathcal{S}_N$ for a homogeneous system of $2N+3$ linear partial differential equations (PDEs) in $2N$ variables that arises in conformal field theory (CFT) and multiple Schramm-L\"owner evolution (SLE$_\kappa$). The system comprises $2N$ null-state equations and three conformal Ward identities that govern CFT correlation functions of $2N$ one-leg boundary operators.  In the first two articles \cite{florkleb, florkleb2}, we use methods of analysis and linear algebra to prove that $\dim\mathcal{S}_N\leq C_N$, with $C_N$ the $N$th Catalan number.

Extending these results, we prove in this article that $\dim\mathcal{S}_N=C_N$ and $\mathcal{S}_N$ entirely consists of (real-valued) solutions constructed with the CFT Coulomb gas (contour integral) formalism.  In order to prove this claim, we show that a certain set of $C_N$ such solutions is linearly independent.  Because the formulas for these solutions are complicated, we prove linear independence indirectly.  We use the linear injective map of lemma \red{15} in \cite{florkleb} to send each solution of the mentioned set to a vector in $\mathbb{R}^{C_N}$, whose components we find as inner products of elements in a Temperley-Lieb algebra.  We gather these vectors together as columns of a symmetric $C_N\times C_N$ matrix, with the form of a \emph{meander matrix}.  If the determinant of this matrix does not vanish, then the set of $C_N$ Coulomb gas solutions is linearly independent.  And if this determinant does vanish, then we construct an alternative set of $C_N$ Coulomb gas solutions and follow a similar procedure to show that this set is linearly independent.  The latter situation is closely related to CFT minimal models.  We emphasize that, although the system of PDEs arises in CFT in a way that is typically non-rigorous, our treatment of this system here and in \cite{florkleb,florkleb2,florkleb4} is completely rigorous.

\end{abstract}

\keywords{conformal field theory, Schramm-L\"{o}wner evolution, Coulomb gas formalism}
\maketitle

\section{Introduction}\label{intro}

This article follows the analysis begun in \cite{florkleb,florkleb2} and concluded in \cite{florkleb4}.  In this introduction, we state the problem under consideration and summarize the results from \cite{florkleb,florkleb2}.  The introduction \red{I} and appendix \red{A} of \cite{florkleb} explain the origin of this problem in conformal field theory (CFT) \cite{bpz,fms,henkel}, its relation to multiple Schramm-L\"owner evolution (SLE$_\kappa$) \cite{bbk,dub2,graham,kl,sakai}, and its application \cite{dots,gruz,rgbw,bbk,bauber,bpz,c3,c1} to critical lattice models \cite{bax,grim,wu,fk,stan} and random walks \cite{law1,schrsheff,weintru,zcs,madraslade}.

The goal of this article and its predecessors \cite{florkleb,florkleb2} is to completely and rigorously determine a certain solution space $\mathcal{S}_N$ of the system of $2N$ null-state partial differential equations (PDEs) from CFT,
\be\label{nullstate}\Bigg[\frac{\kappa}{4}\partial_j^2+\sum_{k\neq j}^{2N}\left(\frac{\partial_k}{x_k-x_j}-\frac{(6-\kappa)/2\kappa}{(x_k-x_j)^2}\right)\Bigg]F(\boldsymbol{x})=0,\quad j\in\{1,2,\ldots,2N\},\ee
and three conformal Ward identities from CFT,
\be\label{wardid}\sum_{k=1}^{2N}\partial_kF(\boldsymbol{x})=0,\quad \sum_{k=1}^{2N}\left[x_k\partial_k+\frac{(6-\kappa)}{2\kappa}\right]F(\boldsymbol{x})=0,\quad \sum_{k=1}^{2N}\left[x_k^2\partial_k+\frac{(6-\kappa)x_k}{\kappa}\right]F(\boldsymbol{x})=0,\ee
with $\boldsymbol{x}:=(x_1,x_2,\ldots,x_{2N})$.  (In this article, but unlike its predecessors \cite{florkleb,florkleb2}, we refer to the coordinates of $\boldsymbol{x}$ as ``points.")  The main results of this article require $\kappa\in(0,8)$, but at times we need to consider the broader range $\kappa\in(0,8)\times i\mathbb{R}$.  The solution space $\mathcal{S}_N$ of interest comprises all (classical) solutions $F:\Omega_0\rightarrow\mathbb{R}$, where
\be\label{components}\Omega_0:=\{\boldsymbol{x}\in\mathbb{R}^{2N}\,|\,x_1<x_2<\ldots< x_{2N-1}< x_{2N}\},\ee
such that for each $F\in\mathcal{S}_N$, there exist some positive constants $C$ and $p$ (which we may choose to be as large as needed) such that
\be\label{powerlaw} |F(\boldsymbol{x})|\leq C\prod_{i<j}^{2N}|x_j-x_i|^{\mu_{ij}(p)},\quad\text{with}\quad\mu_{ij}(p):=\begin{cases}-p, & |x_j-x_i|<1 \\ +p, & |x_j-x_i|\geq1\end{cases}\quad\text{for all $\boldsymbol{x}\in\Omega_0.$}\ee
(We use this bound to prove many of the results in \cite{florkleb, florkleb2}.)  Restricting our attention to $\kappa\in(0,8)$, our goals are as follows:
\begin{enumerate}
\item\label{item1} Rigorously prove that $\mathcal{S}_N$ is spanned by real-valued Coulomb gas solutions. (See definition \ref{CGsolnsdef} below.)
\item\label{item2} Rigorously prove that $\dim\mathcal{S}_N=C_N$, with $C_N$ the $N$th Catalan number:
\be\label{catalan}C_N=\frac{(2N)!}{N!(N+1)!}.\ee
\item\label{item3} Argue that $\mathcal{S}_N$ has a basis consisting of $C_N$ \emph{connectivity weights} (physical quantities defined in the introduction \red{I} to \cite{florkleb}) and find formulas for all of the connectivity weights.
\end{enumerate}

To begin, we summarize some of the results in \cite{florkleb,florkleb2}.  In those articles, we use certain elements of the dual space $\mathcal{S}_N^*$ to prove that $\dim\mathcal{S}_N\leq C_N$, and in this article, we use these linear functionals again to complete goals \ref{item1}--\ref{item3}.  To construct these linear functionals, we prove in \cite{florkleb} that for all $F\in\mathcal{S}_N$ and all $i\in\{1,2,\ldots,2N-1\}$, the limit
\be\label{lim}\bar{\ell}_1F(x_1,x_2,\ldots,x_i,x_{i+2},\ldots,x_{2N})\,\,\,:=\lim_{x_{i+1}\rightarrow x_i}(x_{i+1}-x_i)^{6/\kappa-1}F(\boldsymbol{x})\ee
exists, is independent of $x_i$, and (after implicitly taking the trivial limit $x_i\rightarrow x_{i-1}$) is an element of $\mathcal{S}_{N-1}$.  (Another type of limit $\underline{\ell}_1$ fixes $x_{2N}=-x_1=R$ and sends $R\rightarrow\infty$ with the same consequences, and we denote either as $\ell_1$.)  Following $\ell_1$, we apply $N-1$ more such limits $\ell_2$, $\ell_3,\ldots,\ell_N$ sequentially to (\ref{lim}), sending $F$ to an element of $\mathcal{S}_0:=\mathbb{R}$.

\begin{figure}[b]
\centering
\includegraphics[scale=0.3]{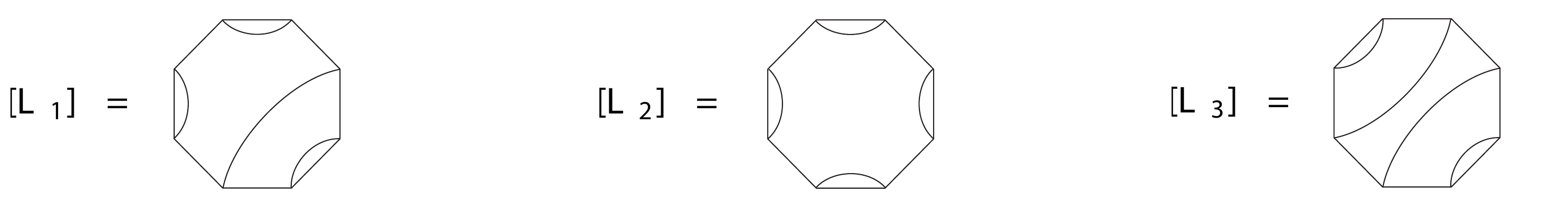}
\caption{Polygon diagrams for three different equivalence classes of allowable sequences of $N=4$ limits.  We find the other $C_4-3=11$ diagrams by rotating one of these three.}
\label{Csls}
\end{figure}

There are many ways that we may order a sequence of these limits, and in \cite{florkleb}, we list the conditions necessary to avoid various inconsistencies such as having the limit $\bar{\ell}_j$ that sends $x_{i_{2j}}\rightarrow x_{i_{2j-1}}$ precede the limit $\bar{\ell}_k$ that sends $x_{i_{2k}}\rightarrow x_{i_{2k-1}}$ if $x_{i_{2j-1}}<x_{i_{2k-1}}<x_{i_{2k}}<x_{i_{2j}}$.  We call the linear functional $\mathscr{L}:\mathcal{S}_N\rightarrow\mathbb{R}$ with $\mathscr{L}:=\ell_{j_N}\ell_{j_{N-1}}\dotsm\ell_{j_2}\ell_{j_1}$ and with the limits ordered to fulfill these conditions an \emph{allowable sequence of limits}.  Because it is linear, an allowable sequence of limits is an element of the dual space $\mathcal{S}_N^*$.

In \cite{florkleb}, we further prove that two allowable sequences $\mathscr{L}$ and $\mathscr{L}'$, which bring together the same pairs of points in different orders, have $\mathscr{L}'F=\mathscr{L}F$ for all $F\in\mathcal{S}_N$.  This fact establishes an equivalence relation among the allowable sequences of limits that partitions them into $C_N$ equivalence classes $[\mathscr{L}_1]$, $[\mathscr{L}_2],\ldots,[\mathscr{L}_{C_N}]$.  We represent the equivalence class $[\mathscr{L}_\varsigma]$ by a unique \emph{interior arc half-plane diagram}, called the \emph{half-plane diagram for $[\mathscr{L}_\varsigma]$}.  Such a diagram consists of $N$ non-intersecting curves, called \emph{interior arcs}, in the upper half-plane, with the endpoints of each interior arc brought together by a limit in every element of $[\mathscr{L}_\varsigma]$.  For convenience, we convert the half-plane diagram for $[\mathscr{L}_\varsigma]$ into an \emph{interior arc polygon diagram}, called the \emph{polygon diagram for $[\mathscr{L}_\varsigma]$} (figure \ref{Csls}), by mapping it continuously onto a regular polygon $\mathcal{P}$, with the points $x_1$, $x_2,\ldots,x_{2N}$ sent to the vertices $w_1$, $w_2,\ldots,w_{2N}$ of $\mathcal{P}$.  We call both types of diagrams \emph{interior arc connectivity diagrams}, and we refer to either of the diagrams representing $[\mathscr{L}_\varsigma]$ simply as the \emph{diagram for $[\mathscr{L}_\varsigma]$}.  We enumerate the equivalence classes $[\mathscr{L}_1]$, $[\mathscr{L}_2],\ldots,[\mathscr{L}_{C_N}]$, let $\mathscr{B}_N^*:=\{[\mathscr{L}_1],[\mathscr{L}_2,]\ldots,[\mathscr{L}_{C_N}]\}\subset\mathcal{S}_N^*$, and define for each $\varsigma\in\{1,2,\ldots,C_N\}$ the \emph{$\varsigma$th connectivity} as the arc connectivity exhibited by the diagram for $[\mathscr{L}_{\varsigma}]$.  The interior arc connectivity diagrams have a natural interpretation as multiple-SLE$_\kappa$ arc connectivities \cite{florkleb}.

We conclude our analysis in \cite{florkleb} by proving that the linear map $v:\mathcal{S}_N\rightarrow\mathbb{R}^{C_N}$ with $v(F)_\varsigma:=[\mathscr{L}_\varsigma]F$ is injective, and therefore $\dim\mathcal{S}_N \leq C_N$.  This proof assumes the following statement, whose justification spans all of \cite{florkleb2}.  If all of $(x_2,x_3)$, $(x_3,x_4),\ldots,(x_{2N-2},x_{2N-1})$, and $(x_{2N-1},x_{2N})$ are two-leg intervals of $F$, where $(x_i,x_{i+1})$ is defined to be a \emph{two-leg interval} of $F\in\mathcal{S}_N$ if the limit (\ref{lim}) vanishes, then $F$ is zero. 

In this article, we complete goals \ref{item1} and \ref{item2} listed above.  In section \ref{CGsolnsSect}, we briefly explain a method for constructing explicit (real-valued) formulas for elements of $\mathcal{S}_N$, called \emph{Coulomb gas (contour integral) solutions}, originally proposed in \cite{df1, df2}.  Then in section \ref{MeanderMatrix}, we use the map $v$ mentioned above to show that a particular set of $C_N$ Coulomb gas solutions is linearly independent, thus achieving goals \ref{item1} and \ref{item2}.  We state this result as theorem \ref{maintheorem} in section \ref{MeanderMatrix}.  In this section, our proof establishes an interesting connection between the system (\ref{nullstate}, \ref{wardid}), the Temperley-Lieb algebra, and the meander matrix \cite{fgg,fgut,difranc,franc}.  Appendix \ref{asymp} presents most of the calculations required for this proof.

In the last article \cite{florkleb4} of this series, we prove some theorems and corollaries concerning the system (\ref{nullstate}, \ref{wardid}) that follow from these results and that relate to CFT and multiple SLE$_\kappa$.  In particular, we prove that any solution equals a sum of at most two Frobenius series in powers of the distance between two neighboring points, except for certain special $\kappa$ values, where a logarithmic term is possible.  This establishes part of the operator product expansion (OPE) of two one-leg boundary operators, generally assumed to exist in CFT.  Addressing item \ref{item3} above, we also discuss connectivity weights, which are proportional to the probability that the curves of a multiple-SLE$_\kappa$ process join in a particular arc connectivity of the $C_N$ that are available.  Finally, we point out and propose a reason for the connection between certain {\it exceptional speeds} (particular $\kappa$ values, see definition \ref{exceptionalspeed} below) and the minimal models of CFT.  We mention again for emphasis that, although the system (\ref{nullstate}, \ref{wardid}) arises in CFT in a way that is typically non-rigorous, our treatment of this system here and in \cite{florkleb,florkleb2,florkleb4} is completely rigorous.

In a future article \cite{fkz}, we combine the formulas for the connectivity weights, found in \cite{florkleb4,fsk}, with a physical interpretation of the elements of the so-called ``Temperley-Lieb set" $\mathcal{B}_N\subset\mathcal{S}_N$ (definition \ref{Fkdefn}) to derive formulas for continuum-limits of cluster crossing probabilities for critical lattice models (such as percolation, Potts models, and random cluster models) in a polygon with a free/fixed side-alternating boundary condition (FFBC).  We verify our predictions with high-precision computer simulations of the $Q\in\{2,3\}$ critical random cluster model in a hexagon, finding excellent agreement.

\section{The Coulomb gas solutions}\label{CGsolnsSect}

Remarkably, one may construct many exact solutions of the system (\ref{nullstate}, \ref{wardid}) via the Coulomb gas (contour integral) formalism, introduced by V.S.\ Dotsenko and V.A.\ Fateev \cite{df1,df2}.  This approach centers on using a perturbed free boson, or Gaussian free field \cite{rgbw}, and N.\ Kang and N.\ Makarov have given a rigorous account for how one may do this  \cite{kangmak}.  To motivate the approach, we first realize each element of $\mathcal{S}_N$ as a CFT $2N$-point correlation function, 
\be\label{2Npoint}\langle\psi_1(x_1)\psi_1(x_2)\dotsm\psi_1(x_{2N})\rangle,\ee
where $\psi_1$ is a one-leg boundary operator, or a $(1,2)$ (resp.\ $(2,1)$) Kac operator in the dense, or $\kappa>4$, (resp.\ dilute, or $\kappa\leq4$) phase of SLE$_\kappa$, in a CFT with central charge 
\be\label{central}c=(6-\kappa)(3\kappa-8)/2\kappa,\quad\kappa>0,\ee
as discussed in the introduction \red{I} of the preceding article \cite{florkleb}.  In CFT, an $(r,s)$ Kac operator is a primary operator with conformal weight \cite{bpz, fms, henkel}
\be\label{Kacweight}h_{r,s}(\kappa)=\frac{1-c(\kappa)}{96}\Bigg[\Bigg(r+s+(r-s)\sqrt{\frac{25-c(\kappa)}{1-c(\kappa)}}\,\Bigg)^2-4\Bigg]=\frac{1}{16\kappa}\begin{cases}(\kappa r-4s)^2-(\kappa-4)^2&\kappa>4\\(\kappa s-4r)^2-(\kappa-4)^2&\kappa\leq4\end{cases}.\ee
(We note that this formula, and all others that we encounter below, are continuous at the phase transition $\kappa=4$.)

Next, we use the Coulomb gas formalism to write explicit formulas for this $2N$-point function (\ref{2Npoint}).  In this approach, we realize a primary operator with conformal weight $h$ as a chiral operator $V_\alpha(x)$ with the same conformal weight.  This chiral operator is the (normal ordered) exponential of $-i\alpha\sqrt{2}\varphi(x)$, with $\varphi(x)$ the holomorphic part of the free boson \cite{df1}, and with the \emph{charge} $\alpha=\alpha(h)$ given by
\be\label{alphapm}\alpha^\pm(h)=\alpha_0\pm\sqrt{\alpha_0^2+h},\quad\alpha_0:=\sqrt{\frac{1-c(\kappa)}{24}}=\frac{1}{2}\left(\frac{\sqrt{\kappa}}{2}-\frac{2}{\sqrt{\kappa}}\right)\times\begin{cases}+1,&\kappa>4 \\ -1, & \kappa\leq4\end{cases}.\ee
We say that the charge $\alpha^\mp(h)$ is \emph{conjugate} to the charge $\alpha^\pm(h)$, and we call the quantity $\alpha_0$ the \emph{background charge} because Coulomb gas calculations implicitly assume the presence of a chiral operator with charge $-2\alpha_0$ at infinity.  In this formalism, we realize an $(r,s)$ Kac operator as the chiral operator $V_{r,s}^\pm:=V_{\alpha_{r,s}^\pm}$ with the \emph{Kac charge}
\be\label{alphars}\alpha_{r,s}^\pm=\alpha^\pm(h_{r,s})=\alpha_0\pm\sqrt{\alpha_0^2+h_{r,s}}=\frac{1}{4\sqrt{\kappa}}\times\begin{cases}\kappa-4\pm|r\kappa-4s|,&\kappa>4 \\ 4-\kappa\pm|s\kappa-4r|, &\kappa\leq 4\end{cases}.\ee
In addition to these charges, the two \emph{screening charges} $\alpha^\pm$ are useful.  By definition, a screening charge is either one of the two possible charges that a chiral operator with conformal weight one may have.  According to (\ref{alphapm}), these are
\be\label{screeningcharges}\alpha^\pm:=\alpha_0\pm\sqrt{\alpha_0^2+1}=\pm\begin{cases}(\sqrt{\kappa}/2)^{\pm1}, & \kappa>4 \\ (\sqrt{\kappa}/2)^{\mp1}, & \kappa\leq4\end{cases}.\ee
One reason that screening charges are useful is that any Kac charge may be written as a sum of half-integer multiples of either or both of them:
\be\label{kaccharge} \alpha_{r,s}^\pm=\frac{(1+r)}{2}\alpha^++\frac{(1+s)}{2}\alpha^-\quad\text{or}\quad\frac{(1-r)}{2}\alpha^++\frac{(1-s)}{2}\alpha^-.\ee
For example, in the dense and dilute phases respectively, the charges $\alpha_{1,s}^\pm$ and $\alpha_{r,1}^\pm$ (\ref{alphars}), respectively corresponding to the conformal weights $h_{s,1}$ and $h_{r,1}$, may be written as half-integer multiples of the screening charges thus:
\be\label{densedilute}\kappa>4:\quad\left\{\begin{array}{l}\alpha_{1,s}^+=\dfrac{(1-s)}{2}\alpha^- \\ \alpha_{1,s}^-=\alpha^++\dfrac{(1+s)}{2}\alpha^-\end{array}\right.,\qquad\kappa\leq4:\quad\left\{\begin{array}{l}\alpha_{r,1}^+=\dfrac{(1+r)}{2}\alpha^++\alpha^- \\ \alpha_{r,1}^-=\dfrac{(1-r)}{2}\alpha^{+}\end{array}\right..\ee
(In this article we follow  the superscript sign conventions established in (\ref{alphars}, \ref{kaccharge}, \ref{densedilute}), which differ from those used in our previous articles \cite{skfz, pinchpt} and in \cite{js}.)

If we realize each one-leg boundary operator of the correlation function (\ref{2Npoint}) representing $F$ as a chiral operator, then we have
\be\label{corrfunc}F(x_1,x_2,\ldots,x_{2N})=\begin{cases}\langle V_{1,2}^\pm(x_1)V_{1,2}^\pm(x_2)\dotsm V_{1,2}^\pm(x_{2N})\rangle, & \kappa>4\\
\langle V_{2,1}^\pm(x_1)V_{2,1}^\pm(x_2)\dotsm V_{2,1}^\pm(x_{2N})\rangle, & \kappa\leq4\end{cases}.\ee
We are free to choose either the plus sign or the minus sign on each individual chiral operator in this correlation function.  After we do this, we may use the simple formula for a correlation function of chiral operators,
\be\label{vertexcorrformula}\langle V_{\alpha_1}(x_1)V_{\alpha_2}(x_2)\dotsm V_{\alpha_M}(x_M)\rangle=\delta_{\sum_j\alpha_j,2\alpha_0}\prod_{i<j}^M|x_j-x_i|^{2\alpha_i\alpha_j},\ee
and the formula (\ref{alphars}) for the charges to write explicit solutions for the system (\ref{nullstate}, \ref{wardid}).

The product on the right side of (\ref{vertexcorrformula}) satisfies the CFT conformal Ward identities \cite{bpz,fms,henkel} if and only if the sum of the charges of the chiral operators on the left side equals $2\alpha_0$.  We call this the \emph{neutrality condition}.  Thus, the $2N$-point correlation function (\ref{corrfunc}) is nontrivial if the collection of chiral operators within it satisfies the neutrality condition.  (Interestingly, there are some examples of such nontrivial correlation functions that do not satisfy the neutrality condition \cite{kype2}.)  Unfortunately, if $N>2$, then no assignment of $\pm$ signs to the chiral operators in (\ref{corrfunc}) satisfies this condition, so this approach seems to produce only the trivial solution.

However, we may circumvent this problem and glean nontrivial (potential) solutions by inserting screening operators into the correlation function (\ref{corrfunc}).  We create a \emph{screening operator} $Q_m^\pm$ by integrating the location $u_m$ of the chiral operator $V^\pm(u_m)$ with charge $\alpha^\pm$ (and thus conformal weight one) around a loop $\Gamma$ in the complex plane \cite{df1,df2}:
\be\label{screeningop}Q_m^\pm:=\oint_{\Gamma} V^\pm(u_m)\,{\rm d}u_m.\ee
This operator is primary, is non-local, and has conformal weight zero.  Therefore, it is effectively an identity operator, and its insertion into a correlation function cannot alter the pointwise information of that function.  But unlike the identity chiral operator, which has charge zero or $2\alpha_0$ because its conformal weight is zero, the screening operator $Q^\pm$ has charge $\alpha^\pm$.  Thus, we may change the total charge of the correlation function in (\ref{vertexcorrformula}) by positive integer multiples $M$ of $\alpha^\pm$ by inserting $M$ distinct screening charges $Q_1^\pm$, $Q_2^\pm,\ldots, Q_M^\pm$ into that correlation function.

After selecting some $c\in\{1,2,\ldots,2N\}$, if we choose in (\ref{corrfunc}) the plus (resp.\ minus) sign for all of the chiral operators except the one at $x_c$ and the minus (resp.\ plus) sign for the chiral operator at $x_c$ in the dense (resp.\ dilute) phase of SLE$_\kappa$, then the sum of the charges of the chiral operators is
\be\label{totalcharge}\begin{cases}(2N-1)\alpha_{1,2}^++\alpha_{1,2}^-=2\alpha_0-(N-1)\alpha^-, & \kappa>4 \\ (2N-1)\alpha_{2,1}^-+\alpha_{2,1}^+=2\alpha_0-(N-1)\alpha^+, & \kappa\leq4\end{cases}.\ee
(Here, we have used the property $\alpha^++\alpha^-=2\alpha_0$ implied by (\ref{alphapm}, \ref{screeningcharges}).)  Thus, by inserting $N-1$ screening operators of charge $\alpha^-$ (resp.\ $\alpha^+$) into the correlation function (\ref{corrfunc}), we satisfy the neutrality condition:
\begin{multline}\label{chiralrep}F(x_1,x_2,\ldots,x_{2N})\\
=\begin{cases}\langle V_{1,2}^+(x_1)V_{1,2}^+(x_2)\dotsm V_{1,2}^+(x_{c-1})V_{1,2}^-(x_c)V_{1,2}^+(x_{c+1})\dotsm V_{1,2}^+(x_{2N})Q_1^-Q_2^-\dotsm Q^-_{N-1}\rangle, & \kappa>4 \\ 
\langle V_{2,1}^-(x_1)V_{2,1}^-(x_2)\dotsm V_{2,1}^-(x_{c-1})V_{2,1}^+(x_c)V_{2,1}^-(x_{c+1})\dotsm V_{2,1}^-(x_{2N})Q_1^+Q_2^+\dotsm Q^+_{N-1}\rangle, & \kappa\leq4\end{cases}.\end{multline} 
Our choice of signs for (\ref{corrfunc}) is the choice that requires the fewest screening operators.  (See appendix \ref{proofappendix}.)  Equation (\ref{vertexcorrformula}) with (\ref{alphars}, \ref{screeningcharges}, \ref{screeningop}) leads to an explicit formula for (\ref{chiralrep}).  This is
\begin{multline}\label{CGsolns}F\Big(\kappa\,\Big|\,\Gamma_1,\Gamma_2,\ldots,\Gamma_{N-1}\,\Big|\,\boldsymbol{x}\Big):=\Bigg(\prod_{\substack{i<j \\ i,j\neq c}}^{2N}(x_j-x_i)^{2/\kappa}\Bigg)\Bigg(\prod_{k\neq c}^{2N}(x_c-x_k)^{1-6/\kappa}\Bigg)\\
\times \mathcal{J}^{(N-1,2N)}\left(\beta_l=\left\{\begin{array}{ll}-4/\kappa,&l\neq c\\ 12/\kappa-2,&l=c\end{array}\right\};\gamma=\frac{8}{\kappa}\,\,\,\Biggl|\,\Gamma_1,\Gamma_2,\ldots,\Gamma_{N-1}\,\Biggl|\,\boldsymbol{x}\right),
\end{multline}
where $c\in\{1,2,\ldots,2N\}$ (and we call $x_c$ \emph{the point bearing the conjugate charge}), $\mathcal{J}^{(M,K)}$ with $M\in\mathbb{Z}^+$ is the $M$-fold \emph{Coulomb gas} (or \emph{Dotsenko-Fateev}) \emph{integral}
\begin{multline}\label{eulerintegralch2}\mathcal{J}^{(M,K)}\Big(\{\beta_l\};\gamma\,\Big|\,\Gamma_1,\Gamma_2,\ldots,\Gamma_M\,\Big|\,x_1,x_2,\ldots,x_K\Big):=\\
\sideset{}{_{\Gamma_M}}\oint\dotsm\sideset{}{_{\Gamma_2}}\oint\sideset{}{_{\Gamma_1}}\oint\left(\prod_{l=1}^K\prod_{m=1}^{M}(x_l-u_m)^{\beta_l}\right)\left(\prod_{p<q}^{M}(u_p-u_q)^{\gamma}\right)\,{\rm d}u_1\,{\rm d}u_2\dotsm\,{\rm d}u_M,\end{multline}
and $\Gamma_1,$ $\Gamma_2,\ldots,\Gamma_M$ are any nonintersecting, closed contours in the complex plane.  (If $\kappa>0$, as it is for the main results of this article, then these contours actually may intersect because $\gamma=8/\kappa>0$.  Rather, it must be possible for us to continuously deform the contours so they do not intersect.  See appendix \ref{proofappendix}.)  According to (\ref{alphars}, \ref{screeningcharges}, \ref{vertexcorrformula}), the powers in the algebraic factors multiplying $\mathcal{J}^{(N-1,2N)}$ in (\ref{CGsolns}) are 
\be\label{powers1}\left\{\begin{array}{ll}2\alpha_{1,2}^+\alpha_{1,2}^+, & \kappa>4\\ 2\alpha_{2,1}^-\alpha_{2,1}^-, & \kappa\leq 4\end{array}\right\}=2/\kappa,\quad \left\{\begin{array}{ll}2\alpha_{1,2}^+\alpha_{1,2}^-, & \kappa>4\\ 2\alpha_{2,1}^-\alpha_{2,1}^+, & \kappa\leq 4\end{array}\right\}=1-6/\kappa,\ee
and the powers in the Coulomb gas integral $\mathcal{J}^{(N-1,2N)}$ in (\ref{CGsolns}) are 
\begin{gather}\label{powers2}\beta_l=\left\{\begin{array}{ll}2\alpha_{1,2}^+\alpha^-, & \kappa>4\\ 2\alpha_{2,1}^-\alpha^+, & \kappa\leq 4\end{array}\right\}=-4/\kappa\quad \text{if $l\neq c$,} \qquad\beta_c=\left\{\begin{array}{ll}2\alpha_{1,2}^-\alpha^-, & \kappa>4\\ 2\alpha_{2,1}^+\alpha^+, & \kappa\leq 4\end{array}\right\}=12/\kappa-2,\\
\label{powers3}\gamma=\left\{\begin{array}{ll}2\alpha^-\alpha^-, & \kappa>4\\ 2\alpha^+\alpha^+, & \kappa\leq 4\end{array}\right\}=8/\kappa.\end{gather}
We note that the formulas for these powers are the same in either phase.  (In more general scenarios, the powers $\beta_l$ and $\gamma$ in (\ref{eulerintegralch2}) may carry double indices $m,l$ and $p,q$ respectively, but we do not encounter those cases in this article.)  We also note that with these powers and $\kappa>0$, the integrand of (\ref{CGsolns}) is absolutely integrable, so we may use Fubini's theorem to change the order of integration.  At times, we do this without explicit reference to this theorem.

Throughout this article, we use the branch of the logarithm for each power function $z^p$ in the integrand of (\ref{eulerintegralch2}) such that $-\pi<\arg z\leq\pi$ for all complex $z$.  This choice determines the orientations of the branch cuts of the integrand in (\ref{eulerintegralch2}).  In spite of this choice, each power function could, in principle, use a different branch.  But instead of allowing this variance, if necessary, we either switch the order of the terms in the differences inside these power functions or explicitly show the phase factors that would otherwise accompany a different choice of logarithm branch. Furthermore, we allow any integration contour to cross these branch cuts and pass onto a different Riemann sheet of the integrand.  If this happens, then the contour must end on the same Riemann sheet as its starting point in order for it to close.

With these conventions established, we define the following:
\begin{defn}\label{CGsolnsdef} Supposing that $\text{Re}\,\kappa>0$, we define a \emph{Coulomb gas function} to be any function of the form (\ref{CGsolns}) (allowing the order of the terms in the differences in the integrand to be switched, or allowing the power functions to use different branches of the logarithm), with all of its integration contours closed and no two contours intersecting.
\begin{enumerate}
\item\label{itemdef1} Also, we define a \emph{Coulomb gas solution} to be any linear combination of Coulomb gas functions, or the following:
\item\label{itemdef2} For some $M\in\mathbb{Z}^+$ and every $j\in\{1,2,\ldots,M\}$, we let $F_j$ be a Coulomb gas function (\ref{CGsolns}) and $a_j$ be a function analytic at every speed $\kappa=\varkappa$ with $\text{Re}\,\varkappa>0$.  If for some $m\in\mathbb{Z}^+$ and a particular $\kappa$ with $\text{Re}\,\kappa>0$, the function
\be\label{CGsumm} F(\varkappa):=(\varkappa-\kappa)^{-m}\sideset{}{_{j=1}^M}\sum a_j(\varkappa)F_j(\varkappa)\ee
extends to a function analytic at $\kappa$ (which clearly happens if and only if the sum in (\ref{CGsumm}) has an order-$m$ zero at $\kappa$), then we call the limit $F(\kappa):=\lim_{\varkappa\rightarrow\kappa}F(\varkappa)$ a \emph{Coulomb gas solution} too.
\end{enumerate}
\end{defn}
\noindent
The requirement in item \ref{itemdef2} of definition \ref{CGsolnsdef} that $F_j(\varkappa)$ is a Coulomb gas function (\ref{CGsolns}) for all $\text{Re}\,\varkappa>0$  is not entirely trivial.  Indeed, if $\varkappa$ is rational, then some contours of $F_j(\varkappa)$ may close by wrapping around the branch points $x_1,$ $x_2,\ldots,x_{2N}$ of the integrand a finite number of times.  But if $\varkappa$ is irrational, then these same contours do not close, so $F_j(\varkappa)$ is not a Coulomb gas function.  Thus, this requirement mandates us to choose contours that close for all $\text{Re}\,\varkappa>0$.  With this requirement, $F_j$ is evidently an analytic function of $(\varkappa,\boldsymbol{x})\in\{\varkappa\in\mathbb{C}\,|\,\text{Re}\,\varkappa>0\}\times\Omega_0$.

The construction of the Coulomb gas functions (\ref{CGsolns}) via the Coulomb gas formalism strongly suggests, but does not rigorously prove, that these candidate solutions actually satisfy the system (\ref{nullstate}, \ref{wardid}).  In \cite{dub}, J.\ Dub\'{e}dat provides this proof for Coulomb gas solutions of item \ref{itemdef1} in definition \ref{CGsolnsdef}, and we present a slightly altered exposition of his proof in appendix \ref{proofappendix}.  The extension of this proof to item \ref{itemdef2} of definition \ref{CGsolnsdef} is straightforward, and we present it below.  Because any Coulomb gas solution obviously satisfies the bound (\ref{powerlaw}) too, we have the following theorem.
\begin{theorem}\label{vertexop} Suppose that $\kappa>0$.  Then every real-valued Coulomb gas solution is an element of $\mathcal{S}_N$.\end{theorem}
\begin{proof} Ref.\ \cite{dub} proves that any Coulomb gas solution of item \ref{itemdef1} in definition \ref{CGsolnsdef} satisfies the system (\ref{nullstate}, \ref{wardid}) for $\kappa>0$.  In appendix \ref{proofappendix}, we give a version of the same proof that more explicitly links the neutrality condition described beneath (\ref{vertexcorrformula}) to the conformal Ward identities (\ref{wardid}).

To show that any Coulomb gas solution of item \ref{itemdef2} in definition \ref{CGsolnsdef}  satisfies this system too, we insert the Taylor series for $F(\varkappa)$ centered on $\varkappa=\kappa$ into (\ref{nullstate}) with $\kappa$ replaced by $\varkappa$.  After recalling that $F$ is an analytic function of $(\varkappa,\boldsymbol{x})\in\{\varkappa\in\mathbb{C}\,|\,\text{Re}\,\varkappa>0\}\times\Omega_0$, we differentiate the series term by term with respect to $x_1$, $x_2,\ldots,x_{2N}$ to find
\be \sum_{m=0}^\infty\frac{(\varkappa-\kappa)^m}{m!}\Bigg[\frac{\varkappa}{4}\partial_j^2+\sum_{k\neq j}^{2N}\left(\frac{\partial_k}{x_k-x_j}-\frac{(6-\varkappa)/2\varkappa}{(x_k-x_j)^2}\right)\Bigg]\partial_\varkappa^mF(\kappa\,|\,\boldsymbol{x})=0.\ee
By sending $\varkappa\rightarrow\kappa$, we find that $F(\kappa)$ satisfies (\ref{nullstate}).  A similar procedure shows that $F(\kappa)$ satisfies the conformal Ward identities (\ref{wardid}) too.
\end{proof}

\begin{figure}[t]
\centering
\includegraphics[scale=0.28]{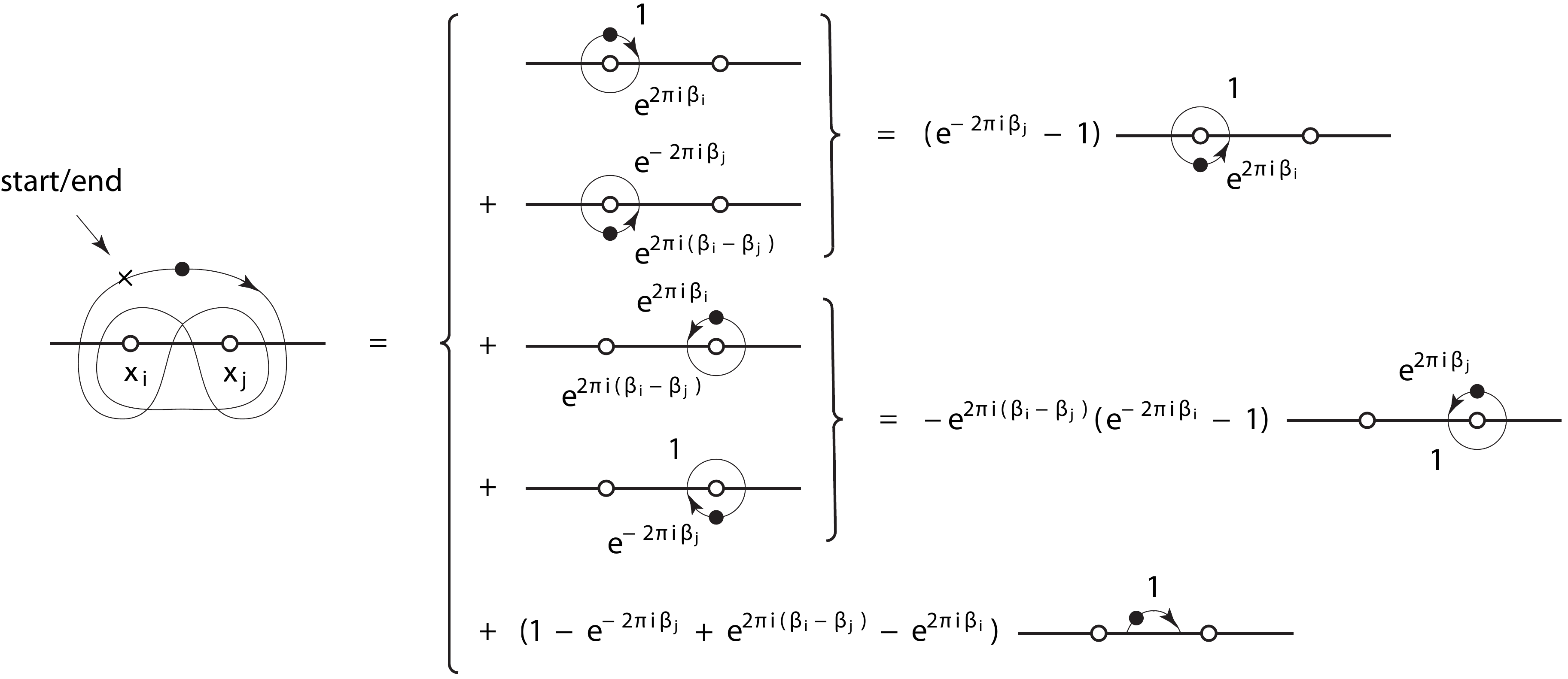}
\caption{The Pochhammer contour $\mathscr{P}(x_i,x_j)$, and the decomposition (\ref{PochDecomp}).  The phase factor of the integrand at the start point and end point (at the tip of the arrow) of each contour is shown.}
\label{BreakDown}
\end{figure}

\begin{figure}[b]
\centering
\includegraphics[scale=0.28]{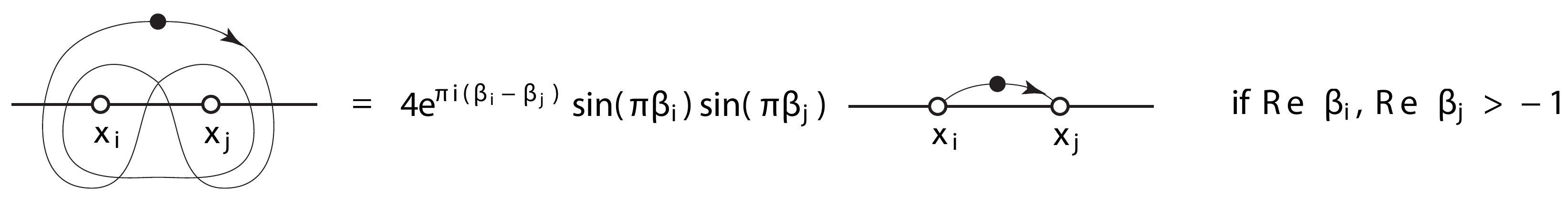}
\caption{If $e^{2\pi i\beta_i}$ and $e^{2\pi i\beta_j}$ are the monodromy factors associated with $x_i$ and $x_j$ respectively and $\text{Re}\,\beta_i,\text{Re}\,\beta_j>-1$, then we may replace the Pochhammer contour $\mathscr{P}(x_i,x_j)$ shown on the left with the simple contour shown on the right.}
\label{PochhammerContour}
\end{figure}

Now we comment on choices of integration contours for (\ref{eulerintegralch2}).  In order to guarantee that (\ref{CGsolns}) satisfies the system (\ref{nullstate}, \ref{wardid}), each integration contour in (\ref{eulerintegralch2}) must close, and no two may intersect. Moreover, the Cauchy integral theorem \cite{kod} implies that if (\ref{CGsolns}) is nontrivial, then every contour must surround at least one of the branch points $x_1$, $x_2,\ldots,x_{2N}$ of the integrand.  (A contour may surround other contours too.)  Now, if the powers $\beta_l$ and $\gamma$ of (\ref{CGsolns}, \ref{eulerintegralch2}) are irrational (as is usually the case), then the winding number of every contour around each of the points $x_1$, $x_2,\ldots,x_{2N}$ must be zero in order for it to close.  The simplest such contour is a Pochhammer contour $\mathscr{P}(x_i,x_j)$ entwining $x_i$ with $x_j$.  Figure \ref{BreakDown} illustrates this contour.  Its start point is directly above $x_i$ on the outer part of the contour, and its end point matches its start point on the same Riemann sheet.  (Also, see page 257 of \cite{witt} for another definition of this contour, noting that the authors' orientation and start point is different from ours in this article.)  Even more complicated choices of integration contours for (\ref{eulerintegralch2}) that satisfy the mentioned requirements are available.  For example, a contour may surround one or more contours and wind multiple times around many branch points at once.  Fortunately, we do not need to consider these complicated possibilities in this article.  (In fact, identity (\ref{ointext}) and figure \ref{Oint} of appendix \ref{reallemproof} indicate that many, but not all, of these possibilities give the trivial solution.)

Figure \ref{BreakDown} also shows that we may decompose a Pochhammer contour $\mathscr{P}(x_i,x_j)$ into a sum of integrations along simple contours and integrations around loops surrounding $x_i$ and $x_j$:
\bea\sideset{}{_{\mathscr{P}(x_i,x_j)}}\oint(u-x_i)^{\beta_i}(x_j-u)^{\beta_j}\dotsm\,{\rm d}u
&=&(e^{-2\pi i\beta_j}-1)\oint_{x_i}(u-x_i)^{\beta_i}(x_j-u)^{\beta_j}\dotsm\,{\rm d}u\nonumber\\
\label{PochDecomp}&-&e^{2\pi i(\beta_i-\beta_j)}(e^{-2\pi i\beta_i}-1)\oint_{x_j}(u-x_i)^{\beta_i}(x_j-u)^{\beta_j}\dotsm\,{\rm d}u\\
&+&4e^{\pi i(\beta_i-\beta_j)}\sin\pi \beta_i\sin\pi \beta_j\sideset{}{_{x_i+\epsilon}^{x_j-\epsilon}}\int (u-x_i)^{\beta_i}(x_j-u)^{\beta_j}\dotsm\,{\rm d}u.\nonumber\eea
Here, the ellipses stand for a function of $u$ analytic in the interior of a region containing $\mathscr{P}(x_i,x_j)$, and the subscript $x_i$ (resp.\ $x_j$) on the integral sign indicates that $u$ traces counterclockwise a circle centered on $x_i$ (resp.\ $x_j$) with  radius $\epsilon\ll |x_j-x_i|$, starting just above $x_i+\epsilon$ (resp.\ below $x_j-\epsilon$) where the integrand's phase is zero.  If $\text{Re}\,\beta_i,\text{Re}\,\beta_j>-1$, then sending $\epsilon\rightarrow0$ in (\ref{PochDecomp}) gives the useful identity
\be\label{Pochtostraight}\sideset{}{_{\mathscr{P}(x_i,x_j)}}\oint(u-x_i)^{\beta_i}(x_j-u)^{\beta_j}\dotsm\,{\rm d}u= 4e^{\pi i(\beta_i-\beta_j)}\sin\pi \beta_i\sin\pi \beta_j\sideset{}{_{x_i}^{x_j}}\int (u-x_i)^{\beta_i}(x_j-u)^{\beta_j}\dotsm\,{\rm d}u,\quad\text{Re}\,\beta_i,\text{Re}\,\beta_j>-1\ee
(figure \ref{PochhammerContour}).  Thus, we call $x_i$ and $x_j$ ``endpoints" of $\mathscr{P}(x_i,x_j)$ in this article.  If $x_i<x_j$ (resp.\ $x_i>x_j$), then we say that $\mathscr{P}(x_i,x_j)$ has \emph{rightward orientation} (resp.\ \emph{leftward orientation}.) Finally, we note from the decomposition (\ref{Pochtostraight}) that $\beta_i\in\mathbb{Z}^+\cup\{0\}$ or $\beta_j\in\mathbb{Z}^+\cup\{0\}$ is a simple zero of the left side of (\ref{Pochtostraight}). This fact is useful to item \ref{kappaextend} in the proof of lemma \ref{mainlem} below.

\begin{figure}[b]
\centering
\includegraphics[scale=0.27]{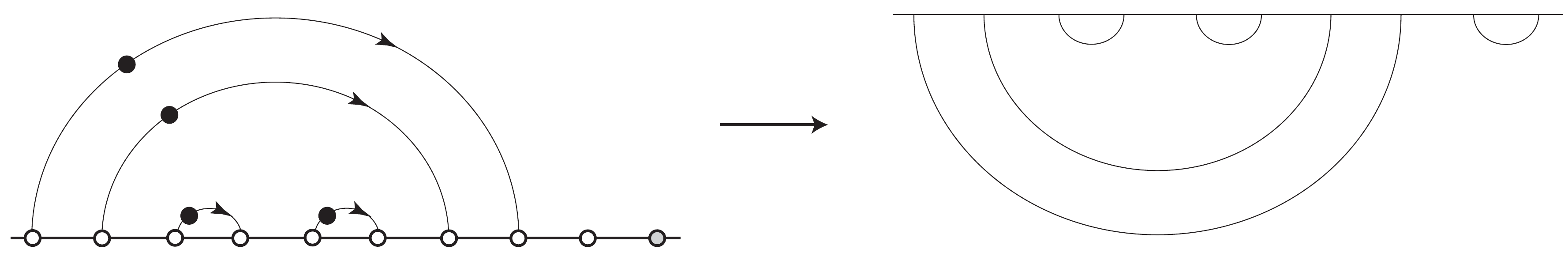}
\caption{The correspondence between the integration contours for $\mathcal{F}_{c,\vartheta}$ (left) and the half-plane diagram for $\mathcal{F}_{c,\vartheta}$ (right).  The latter is the reflection of the half-plane diagram for $[\mathscr{L}_\vartheta]$ into the lower half-plane.  The gray point bears the conjugate charge.}
\label{Halfplane}
\end{figure}

\section{A basis for $\mathcal{S}_N$ and the meander matrix}\label{MeanderMatrix}

Having proven that $\dim\mathcal{S}_N\leq C_N$ in \cite{florkleb,florkleb2}, we next prove that $\dim\mathcal{S}_N=C_N$ by showing that a certain subset of $C_N$ Coulomb gas solutions is linearly independent.  Because such a set is a basis for $\mathcal{S}_N$, we thus achieve goals \ref{item1} and \ref{item2} listed in the introduction \ref{intro}.

\begin{defn}\label{ndefn} We call the function $n:\mathbb{C}\setminus\{0\}\rightarrow\mathbb{C}$ with the following formula the \emph{O($n$)-model fugacity function:}
\be\label{fugacity}n(\kappa):=-2\cos(4\pi/\kappa).\ee
\end{defn}

\begin{figure}[p]
\centering
\includegraphics[scale=0.27]{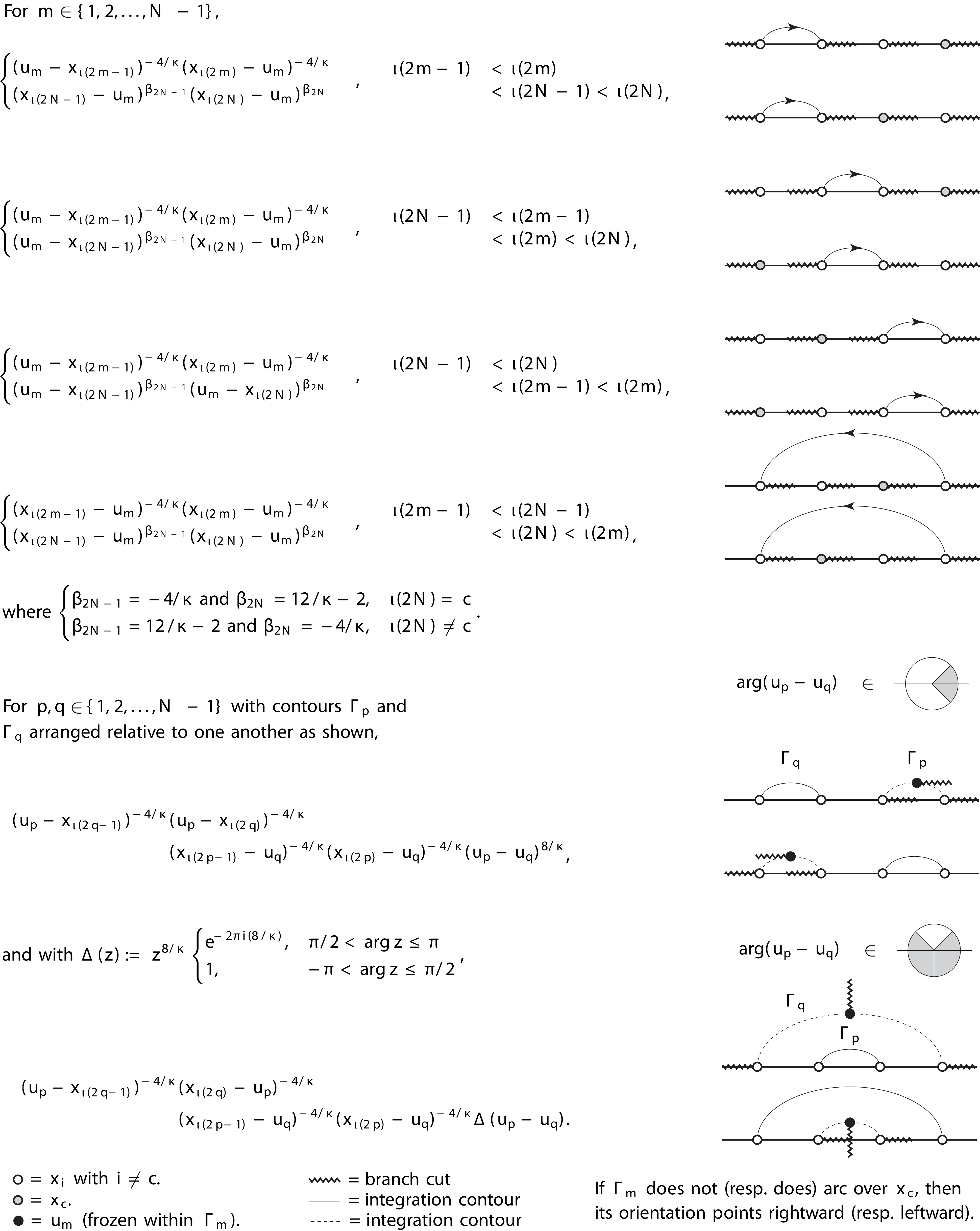}
\caption{Factors of the integrand for the Coulomb gas integral appearing in the formula (\ref{firstFexplicit1}) for $\mathcal{F}_{c,\vartheta}$.  To indicate that we are using the conventions shown in this figure, we enclose the integrand of (\ref{firstFexplicit1}) between the brackets of $\mathcal{N}[\,\,\ldots\,\,]$.}
\label{Orderings}
\end{figure}

The function $n$ inherits its name from its realization as the loop fugacity of an O$(n)$ model from statistical mechanics whose closed loops are conjectured to be (locally) statistically identical to SLE$_\kappa$ curves.  Technically, this connection between SLE$_\kappa$ and the O$(n)$ model applies only for $\kappa\in[2,8]$ \cite{gruz, rgbw, smir4,smir}.  Nonetheless, we find the notation $n(\kappa)$ useful for all $\kappa \in \mathbb{C}\setminus\{0\}$.  
\begin{defn}\label{Fkdefn}For each $c\in\{1,2,\ldots,2N\}$ and $\vartheta\in\{1,2,\ldots,C_N\}$, we let $\mathcal{F}_{c,\vartheta}:[(0,8)\times i\mathbb{R}]\times\Omega_0\rightarrow\mathbb{R}$ be the Coulomb gas function (\ref{CGsolns}) with the following details and modifications:
\begin{enumerate}
\item $\mathcal{F}_{c,\vartheta}(\kappa\,|\,\boldsymbol{x})$ is of the form (\ref{CGsolns}) multiplied by
\be\label{firstprefactor}n(\kappa)\left[\frac{n(\kappa)\Gamma(2-8/\kappa)}{4\sin^2(4\pi/\kappa)\Gamma(1-4/\kappa)^2}\right]^{N-1},\quad \text{with $n(\kappa)$ given by (\ref{fugacity}).}\ee
\item\label{third}The integration contours $\Gamma_1$, $\Gamma_2,\ldots,\Gamma_{N-1}$ are non-intersecting Pochhammer contours bent to lie completely in the upper half-plane (except for where they wrap around their endpoints) and specified as follows (figure \ref{Halfplane}):
\begin{enumerate}
\item Each contour shares its endpoints with a unique arc in the half-plane diagram for $[\mathscr{L}_\vartheta]$ that has neither of its endpoints at $x_c$.  So except for one excluded arc, those arcs correspond one-to-one with the contours.
\item\label{2cit}If an arc in the half-plane diagram for $[\mathscr{L}_\vartheta]$ does not (resp.\ does) pass over $x_c$, then its corresponding contour is oriented rightward (resp.\ leftward).  (See the comments beneath (\ref{Pochtostraight}).)
\item\label{simplifyformula} If $\text{Re}\,\kappa>4$, then identity (\ref{Pochtostraight}) allows us to simplify the formula for $\mathcal{F}_{c,\vartheta}(\kappa\,|\,\boldsymbol{x})$.  We replace each contour by a simple curve (that bends into the upper half-plane) with the same endpoints, and we replace (\ref{firstprefactor}) with
\be\label{secondprefactor}n(\kappa)\left[\frac{n(\kappa)\Gamma(2-8/\kappa)}{\Gamma(1-4/\kappa)^2}\right]^{N-1},\quad \text{with $n(\kappa)$ given by (\ref{fugacity}).}\ee
\item If $N=1$, then there are no integration contours. (See (\red{16}) and the surrounding discussion in \cite{florkleb}.)  Instead, we set $\mathcal{F}_{1,1}(\kappa\,|\,x_1,x_2)=\mathcal{F}_{2,1}(\kappa\,|\,x_1,x_2)=n(\kappa)(x_2-x_1)^{1-6/\kappa}$.
\end{enumerate}
For $m\in\{1,2,\ldots,N-1\}$, we let $\iota_{c,\vartheta}(2m-1)<\iota_{c,\vartheta}(2m)$ be indices of the endpoints of $\Gamma_m$, and for $m=N$, of the two points with no contour entwining them (one of which is $c$).  (For concision, we write $\iota$ in place of $\iota_{c,\vartheta}$.)
\item\label{4thitem}We order the differences in the integrand of (\ref{eulerintegralch2}) as figure \ref{Orderings} shows.  In this article and in \cite{florkleb4}, we indicate this ordering by enclosing the integrand for (\ref{eulerintegralch2}) between the square brackets of $\mathcal{N}[\,\,\ldots\,\,]$.
\item\label{5thitem} We order the differences in the factors multiplying $\mathcal{J}^{(N-1,2N)}$ in (\ref{CGsolns}) so each is positive.  With item \ref{4thitem}, this ordering ensures that $\mathcal{F}_{c,\vartheta}$ is real-valued if $\kappa>0$.  (The proof of lemma \ref{reallem} in appendix \ref{reallemproof} justifies this claim.)
\item\label{item2a} If $c=2N$ in (\ref{CGsolns}, \ref{eulerintegralch2}), then we let $\mathcal{F}_\vartheta:=\mathcal{F}_{2N,\vartheta}$ and $\mathcal{B}_N:=\{\mathcal{F}_1,\mathcal{F}_2,\ldots,\mathcal{F}_{C_N}\}\subset\mathcal{S}_N$.  We call $\mathcal{B}_N$ the \emph{Temperley-Lieb set} of $\mathcal{S}_N$ (due to its close relation to the Temperley-Lieb algebra, see the proof of lemma \ref{mainlem}).
\end{enumerate}
Finally, we define the \emph{exterior arc polygon (resp.\ half-plane) diagram for $\mathcal{F}_{c,\vartheta}$} (or more simply, the \emph{diagram for $\mathcal{F}_{c,\vartheta}$}) to be the diagram for $[\mathscr{L}_{\vartheta}]$, but with all interior arcs replaced by \emph{exterior arcs} drawn outside the $2N$-sided polygon (figure \ref{Fk}) (resp.\ drawn inside the lower half-plane (figure \ref{Halfplane})).  We call either diagram an \emph{exterior arc connectivity diagram}.  (We note that each exterior arc in the half-plane diagram for $\mathcal{F}_{c,\vartheta}$, except that with an endpoint at $x_c$, corresponds to an integration contour.)
\end{defn}

\begin{figure}[t]
\centering
\includegraphics[scale=0.3]{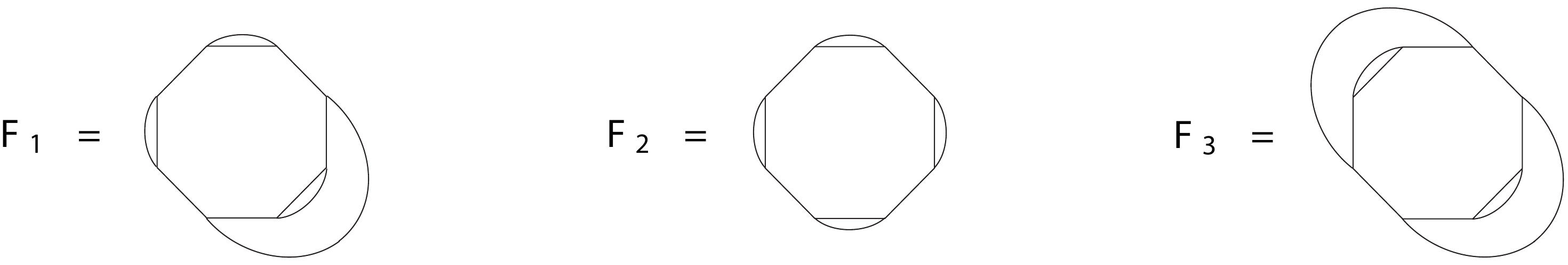}
\caption{Polygon diagrams for three different elements of the Temperley-Lieb set $\mathcal{B}_4$.  We find the other $C_4-3=11$ diagrams by rotating one of these three.}
\label{Fk}
\end{figure}

With the integration contours $\Gamma_1$, $\Gamma_2,\ldots,\Gamma_{N-1}$ and the symbol $\mathcal{N}[\,\,\ldots\,\,]$ determined by items \ref{4thitem} and figure \ref{Orderings}, the explicit formula for $\mathcal{F}_{c,\vartheta}$ is therefore\begin{multline}\label{firstFexplicit1}\mathcal{F}_{c,\vartheta}(\kappa\,|\,\boldsymbol{x})=n(\kappa)\left[\frac{n(\kappa)\Gamma(2-8/\kappa)}{4\sin^2(4\pi/\kappa)\Gamma(1-4/\kappa)^2}\right]^{N-1}\Bigg(\prod_{\substack{j<k \\ j,k\neq c}}^{2N}(x_k-x_j)^{2/\kappa}\Bigg)\Bigg(\prod_{\substack{k=1 \\ k\neq c}}^{2N}|x_c-x_k|^{1-6/\kappa}\Bigg)\oint_{\Gamma_{N-1}}{\rm d}u_{N-1}\dotsm\\ 
\oint_{\Gamma_2}{\rm d}u_2\,\,\oint_{\Gamma_1}{\rm d}u_1\,\,\mathcal{N}\Bigg[\Bigg(\prod_{\substack{l=1 \\ l\neq c}}^{2N}\prod_{m=1}^{N-1}(x_l-u_m)^{-4/\kappa}\Bigg)\Bigg(\prod_{m=1}^{N-1}(x_c-u_m)^{12/\kappa-2}\Bigg)\Bigg(\prod_{p<q}^{N-1}(u_p-u_q)^{8/\kappa}\Bigg)\Bigg].\end{multline}
If $\text{Re}\,\kappa>4$, then as per item \ref{simplifyformula} above, we may simplify this formula by replacing $\Gamma_m$ by a simple contour with the same endpoints as $\Gamma_m$ and bent into the upper half-plane, followed by dropping each factor of $4\sin^2(4\pi/\kappa)$.

In order for $\mathcal{F}_{c,\vartheta}$ to be an element of $\mathcal{S}_N$, it must be real-valued.  In fact, the ordering  required by items \ref{4thitem} (figure \ref{Orderings}) and \ref{5thitem} of definition \ref{Fkdefn} guarantees that this is true if $\kappa>0$, and we prove this claim in appendix \ref{reallemproof}.  Moreover, this ordering has one peculiarity.  With $f(z)=z^{8/\kappa}$, the integrand of (\ref{eulerintegralch2}) contains a factor of $f(u_p-u_q)$ for each $p,q\in\{1,2,\ldots,N-1\}$ such that the integration contour $\Gamma_q$ either lies to the left of $\Gamma_p$ or passes over it.  In the former case, the difference $u_p-u_q$ never touches the branch cut of $f$, which lies along the negative real axis.  But  in the latter case, it eventually does cross this branch cut from beneath.  Inserting the replacement (figure \ref{Orderings})
\be\label{Deltadefn}f(u_p-u_q)\mapsto\Delta(u_p-u_q),\quad\Delta(z):=z^{8/\kappa}\times\begin{cases}e^{-2\pi i(8/\kappa)}, & \pi/2<\arg z\leq\pi\\ 1, & -\pi<\arg z\leq\pi/2 \end{cases},\ee
shows the phase factor acquired as the difference $u_p-u_q$ crosses this branch cut and passes onto a different Riemann sheet of $f$.  This difference never touches the branch cut of $\Delta$, which lies along the positive imaginary axis.

Next, we show that $\mathcal{F}_{c,\vartheta}$ is analytic on $[(0,8)\times i\mathbb{R}]\times\Omega_0$, a property that we use in the proof of theorem \ref{maintheorem} below.  The Coulomb gas function (\ref{CGsolns}) clearly has this property, but the prefactor (\ref{firstprefactor}) that multiplies it to give $\mathcal{F}_{c,\vartheta}(\kappa\,|\,\boldsymbol{x})$ (\ref{firstFexplicit1}) is singular at $\kappa=8/r$ for some integer $r>1$.  If $r$ is odd, then the bracketed factor in (\ref{firstprefactor}) analytically extends to $\kappa=8/r$, so $\mathcal{F}_{c,\vartheta}$ is analytic there.  (In fact, $\mathcal{F}_{c,\vartheta}(\kappa=8/r)$ is zero due to the vanishing outer factor of $n(\kappa)$ in (\ref{firstprefactor}).  To avoid the trivial solution in this case, we drop this factor from (\ref{firstFexplicit1}).  See the proof of theorem \ref{maintheorem}.)  And if $r$ is even, or more simply, if $\kappa=4/r$ with $r\in\mathbb{Z}^+$, then in formula (\ref{firstFexplicit1}) for $\mathcal{F}_{c,\vartheta}(\kappa\,|\,\boldsymbol{x})$ with $\Gamma_m=\mathscr{P}(x_{\iota(2m-1)},x_{\iota(2m)})$, we find
\be\label{Freplace}\frac{n(\varkappa)\Gamma(2-8/\varkappa)}{4\sin^2(4\pi/\varkappa)\Gamma(1-4/\varkappa)^2}\oint_{\Gamma_m}\dotsm\quad\xrightarrow[\varkappa\rightarrow\kappa=4/r,\,\,r\in\mathbb{Z}^+]{}\quad\frac{(-1)^{r-1}(r-1)!^2}{2\pi i(2r-2)!}\left(\oint_{x_{\iota(2m-1)}}\dotsm\quad-\quad\oint_{x_{\iota(2m)}}\dotsm\quad\right),\ee
thanks to (\ref{PochDecomp}).  (If $\Gamma_m=\mathscr{P}(x_{\iota(2m)},x_{\iota(2m-1)})$ because $x_{\iota(2m-1)}<x_c<x_{\iota(2m)}$, then as per item \ref{2cit} in definition \ref{Fkdefn}, the terms in the difference on the right side of (\ref{Freplace}) switch.)  Here, the subscript $x_j$ on the contour integral on the right side of (\ref{Freplace}) signifies that the integration contour is a circle centered at $x_j$ with some arbitrarily small radius.  After inserting (\ref{Freplace}) into (\ref{firstFexplicit1}), we also note that the branch points $x_1$, $x_2,\ldots,x_{2N}$ of the integrand are now poles.  (Indeed, this happens only if $4/\kappa\in\mathbb{Z}^+$.)  Thus, we use the Cauchy integral formula \cite{kod} to evaluate all contour integrals on the right side of (\ref{Freplace}), finding (here, $x_{\iota(2m-1)}<x_{\iota(2m)}$ are endpoints of the $m$th arc in the $\vartheta$th connectivity)
\begin{multline}\label{explicit}\mathcal{F}_{c,\vartheta}\biggl(\kappa=\frac{4}{r}\,\bigg|\,\boldsymbol{x}\biggr)=2(-1)^{(r-1)N}\left(\frac{(r-1)!}{(2r-2)!}\right)^{N-1}\Biggl(\prod_{k\neq c}^{2N}|x_c-x_k|^{1-\frac{3r}{2}}\Biggr)\Biggl(\prod_{\substack{ i< j \\ i,j\neq c}}^{2N}(x_j-x_i)^{\frac{r}{2}}\Biggr)\\
\begin{aligned}&\times\,\sum_{\substack{\{s_1,s_2,\ldots,s_{N-1}\} \\ s_n\in\{2n-1,2n\}}}\omega\Big(\{s_n\}\Big)\,\,\partial_{u_1}^{r-1}\partial_{u_2}^{r-1}\,\dotsm\,\partial_{u_{N-1}}^{r-1}\Bigg[\Biggl(\prod_{m=1}^{N-1}|x_c-u_m|^{3r-2}\Biggr)\Biggl(\prod_{m=1}^{N-1}\prod_{l\neq\iota(s_m),\,c}^{2N}|x_l-u_m|^{-r}\Biggr)\\
&\times\,\Biggl(\prod_{p<q}^{N-1}|u_p-u_q|^{2r}\Biggr)\Bigg]_{u_m= x_{\iota(s_m)}},\quad
\omega\Big(\{s_n\}\Big):=\begin{cases}(-1)^{s_1+s_2+\dotsm+s_{N-1}+N-1}, & \text{$r$ even}\\ 1, & \text{$r$ odd}\end{cases}.\end{aligned}\end{multline}
From this formula, it is evident that $\mathcal{F}_{c,\vartheta}$ is analytic at $\kappa=4/r$ for all $r\in\mathbb{Z}^+$.  From these observations, we then conclude that $\mathcal{F}_{c,\vartheta}$ is analytic in $[(0,8)\times i\mathbb{R}]\times\Omega_0$.  This formula (\ref{explicit}) may have applications to the Gaussian free field ($r=1$) and the loop-erased random walk ($r=2$).  In particular, J.\ Dub\'edat gives a determinant formula for an element of $\mathcal{S}_N$ with $\kappa=2$ in \cite{dub}.  According to theorem \ref{maintheorem} below, this solution must equal an appropriate linear combination of the functions in (\ref{explicit}) with $\vartheta\in\{1,2,\ldots,C_N\}$, $c=2N$, and $r=2$.

\begin{table}[t]
\centering
\begin{tabular}{llllll}
SLE$_\kappa$ speed \hspace{.75cm} & Rational \hspace{.75cm} & Exceptional \hspace{.4cm} & $c(\kappa)$ (\ref{central}) a central charge \hspace{.4cm} & Indicial power of Frobenius  \hspace{.4cm} & All elements of \\
$\kappa\in(0,8)$ & speed & speed & of a CFT minimal model & series \cite{florkleb4} differ by an integer & $\mathcal{B}_N$ algebraic \\
\hline
$8/\kappa\in2\mathbb{Z}^+$ & $\quad$\checkmark & $\qquad\times$ & $\quad\qquad\qquad\times$ & $\qquad\qquad\qquad$\checkmark & $\qquad$\checkmark \\
$8/\kappa\in2\mathbb{Z}^++1$ & $\quad$\checkmark & $\qquad$\checkmark & $\quad\qquad\qquad$\checkmark &  $\qquad\qquad\qquad$\checkmark & $\qquad$\checkmark \\
$\kappa=\kappa_{q,q'}$, $q>2$ & $\quad$\checkmark & $\qquad$\checkmark & $\quad\qquad\qquad$\checkmark & $\qquad\qquad\qquad\times$ & $\qquad$? \\
$\kappa\not\in\mathbb{Q}$ & $\quad\times$ & $\qquad\times$ & $\quad\qquad\qquad\times$ & $\qquad\qquad\qquad\times$ & $\qquad?$
\end{tabular}
\caption{A table of all SLE$_\kappa$ speeds $\kappa\in(0,8)$ collected into disjoint groups with common properties.  Here, (\ref{exceptional}) defines $\kappa_{q,q'}$, and $?$ may be $\checkmark$ or $\times$, depending on the value of $\kappa$.  We prove columns four and five in sections \red{V} and \red{II} of \cite{florkleb4} respectively.}
\label{tablespeeds}
\end{table}

From here until the end of the proof of theorem \ref{maintheorem}, we work strictly with the elements $\mathcal{F}_\vartheta$ of the Temperley-Lieb set $\mathcal{B}_N$.  Their formulas are given by (\ref{firstFexplicit1}) with $c=2N$ (item \ref{item2a} in definition \ref{Fkdefn}).  (However, we later prove that $\mathcal{F}_{c,\vartheta}=\mathcal{F}_\vartheta$ for all $c\in\{1,2,\ldots,2N-1\}$.  See corollary \ref{moveconjchargecor} below.)  Now, we prove that if $\kappa\in(0,8)$, then $\mathcal{B}_N$ is linearly independent if and only if $\kappa$ is not among a certain subset of the speeds given in the following definition.
\begin{defn}\label{exceptionalspeed} We call an SLE$_\kappa$ speed $\kappa$ an \emph{exceptional speed} if it is of the form
\be\label{exceptional}\kappa_{q,q'}:=4q/q',\quad\text{with $q>1$ and $\{q,q'\}$  a pair of coprime positive integers.}\ee
\end{defn}
\noindent
More simply, the exceptional speeds are the positive rational numbers $\kappa\in\mathbb{Q}$ not equaling $4/r$ for some $r\in\mathbb{Z}^+$.  In section \red{IV} of \cite{florkleb4}, we show that the exceptional speeds $\kappa\in(0,8)$ correspond with the CFT minimal models \cite{bpz,fms,henkel} and propose a reason for this correspondence.  Table \ref{tablespeeds} groups speeds $\kappa\in(0,8)$ according to some common properties.

Lemma \red{15} of \cite{florkleb} implies that if $\kappa\in(0,8)$, then $\mathcal{B}_N$ is linearly independent if and only if the set $v(\mathcal{B}_N):=\{v(\mathcal{F}_1),v(\mathcal{F}_2),\ldots,v(\mathcal{F}_{C_N})\}$ is linearly independent, where $v$ is the linear injective map
\be v:\mathcal{S}_N\rightarrow\mathbb{R}^{C_N},\quad v(F)_\varsigma:=[\mathscr{L}_\varsigma]F.\ee  
Therefore, to determine the rank of $\mathcal{B}_N$, it suffices to determine the rank of $v(\mathcal{B}_N$).  The latter task involves calculating $[\mathscr{L}_\varsigma]\mathcal{F}_\vartheta$ for all $\mathcal{F}_\vartheta\in\mathcal{B}_N$ and all $[\mathscr{L}_\varsigma]\in\mathscr{B}_N^*$, and, as we will see, we may treat this calculation as a certain product of the interior and exterior arc diagrams for $[\mathscr{L}_\varsigma]$ and $\mathcal{F}_\vartheta$ respectively.

To motivate this approach, we start with a sample calculation.  We choose some $\mathcal{F}_\vartheta\in\mathcal{B}_N$ and $[\mathscr{L}_\varsigma]\in\mathscr{B}_N^*$ and an arc in the diagram for $[\mathscr{L}_\varsigma]$ that links a pair of adjacent points $x_i$ and $x_{i+1}$ among $x_1,x_2,\ldots,x_{2N-1}$.  For topological reasons, at least one such arc must exist, and we choose an element of $[\mathscr{L}_\varsigma]$ whose first limit 
\begin{multline}\label{thelim}\bar{\ell}_1\mathcal{F}_\vartheta(\kappa\,|\,x_1,x_2,\ldots,x_{i-1},x_{i+2},\ldots,x_{2N})\\
=\lim_{x_{i+1}\rightarrow x_i}(x_{i+1}-x_i)^{6/\kappa-1}\mathcal{F}_\vartheta(\kappa\,|\, \boldsymbol{x}),\quad i\in\{1,2,\ldots,2N-2\},\quad\kappa\in(0,8)\end{multline}
pulls its endpoints together.  As we will now discover, the value of $\bar{\ell}_1\mathcal{F}_\vartheta$ depends on whether or not $x_i$ and $x_{i+1}$ are endpoints of integration contours of $\mathcal{F}_\vartheta$.  To realize this, we consider some different cases.

First, if no contour has its endpoints at $x_i$ or $x_{i+1}$ (according to definition \ref{Fkdefn}, this is not possible because $i+1\neq2N$, but we consider this case anyway because it appears as a consequence of deforming integration contours later), then the integrand (\ref{eulerintegralch2}) of $\mathcal{F}_\vartheta$ approaches a finite value uniformly over $\Gamma_1$, $\Gamma_2,\ldots,\Gamma_{N-1}$ as $x_{i+1}\rightarrow x_i$, so the limit of the entire Coulomb gas integral in the formula (\ref{firstFexplicit1}) for $\mathcal{F}_\vartheta$ is finite.  Hence, this formula shows that $\mathcal{F}_\vartheta(\kappa\,|\,\boldsymbol{x})=O((x_{i+1}-x_i)^{2/\kappa})$, so $\bar{\ell}_1\mathcal{F}_\vartheta$ (\ref{thelim}) is zero if $\kappa\in(0,8)$.  Evidently, $(x_i,x_{i+1})$ is a two-leg interval of $\mathcal{F}_\vartheta$.  (See definition \red{13} in \cite{florkleb}.)

Next, we suppose that $x_i$ and $x_{i+1}$ are endpoints of one common contour, say $\Gamma_1=\mathscr{P}(x_i,x_{i+1})$.  By sorting all factors of $\mathcal{F}_\vartheta$ involving the endpoints and integration variable of $\Gamma_1$ into groups and finding the asymptotic behavior of each group as $x_{i+1}\rightarrow x_i$, we determine the limit $\bar{\ell}_1\mathcal{F}_\vartheta$ (\ref{thelim}).  For the first $N-2$ groups, we choose $q\in\{2,3,\ldots,N-1\}$ and combine the factors involving the endpoints and integration variables of the contours $\Gamma_1$ and $\Gamma_q$ together.  If $\Gamma_q$ passes over $\Gamma_1$, then these factors are (figure \ref{Orderings})
\begin{multline}\label{firstgroup}(x_i-x_{\iota(2q-1)})^{2/\kappa}(x_{i+1}-x_{\iota(2q-1)})^{2/\kappa}(x_{\iota(2q)}-x_i)^{2/\kappa}(x_{\iota(2q)}-x_{i+1})^{2/\kappa}\\
\times(u_1-x_{\iota(2q-1)})^{-4/\kappa}(x_{\iota(2q)}-u_1)^{-4/\kappa}(x_i-u_q)^{-4/\kappa}(x_{i+1}-u_q)^{-4/\kappa}\Delta(u_1-u_q),\end{multline}
with $\Delta$ defined in (\ref{Deltadefn}) and $x_{\iota(1)}=x_i$ and $x_{\iota(2)}=x_{i+1}$.  This product approaches one uniformly over $(u_1,u_q)\in\Gamma_1\times\Gamma_q$ as $x_{i+1}\rightarrow x_i$ because $u_1\rightarrow x_i$ too.
Alternatively, if $\Gamma_1$ lies to the right of $\Gamma_q$, then these factors are (figure \ref{Orderings})
\begin{multline}\label{thirdtolastgroup}(x_i-x_{\iota(2q-1)})^{2/\kappa}(x_{i+1}-x_{\iota(2q-1)})^{2/\kappa}(x_i-x_{\iota(2q)})^{2/\kappa}(x_{i+1}-x_{\iota(2q)})^{2/\kappa}\\
\times(u_1-x_{\iota(2q-1)})^{-4/\kappa}(u_1-x_{\iota(2q)})^{-4/\kappa}(x_i-u_q)^{-4/\kappa}(x_{i+1}-u_q)^{-4/\kappa}(u_1-u_q)^{8/\kappa},\end{multline}
and again, this product approaches one uniformly over $(u_1,u_q)\in\Gamma_1\times\Gamma_q$ as $x_{i+1}\rightarrow x_i$.  (The result is the same if $\Gamma_1$ lies to the left of $\Gamma_q$.) In the penultimate group, we consider the factors involving $x_{\iota(2N-1)}$ and $x_{\iota(2N)}$ (figure \ref{Orderings}),
\be\label{secondtolastgroup}(x_{\iota(2N-1)}-x_i)^{2/\kappa}(x_{\iota(2N-1)}-x_{i+1})^{2/\kappa}(x_{\iota(2N)}-x_i)^{1-6/\kappa}(x_{\iota(2N)}-x_{i+1})^{1-6/\kappa}(x_{\iota(2N-1)}-u_1)^{-4/\kappa}(x_{\iota(2N)}-u_1)^{12/\kappa-2}\ee
if $x_{\iota(2)}<x_{\iota(2N-1)}$, and a similar expression if otherwise.  Again, this product approaches one uniformly over $\Gamma_1$ as $x_{i+1}\rightarrow x_i$.  In the final group, we only consider factors involving $x_i$, $x_{i+1}$, and $u_1$.  With the integration along $\Gamma_1$, these are
\be\label{lastgroup}(x_{i+1}-x_i)^{2/\kappa}\oint_{\Gamma_1}(u_1-x_i)^{-4/\kappa}(x_{i+1}-u_1)^{-4/\kappa}\,{\rm d}u_1.\ee
After extracting a factor of $(x_{i+1}-x_i)^{1-8/\kappa}$ from the integrand of (\ref{lastgroup}) via the substitution $u_1(t)=(1-t)x_i+tx_{i+1}$, we find that as $x_{i+1}\rightarrow x_i$, this expression (\ref{lastgroup}) goes to the analytic continuation of the beta-function integral \cite{witt}
\be\label{betafunc}\oint_{\mathscr{P}(0,1)}t^{-4/\kappa}(1-t)^{-4/\kappa}\,{\rm d}t=4\sin^2\left(\frac{4\pi}{\kappa}\right)\frac{\Gamma(1-4/\kappa)^2}{\Gamma(2-8/\kappa)}\ee
multiplied by $(x_{i+1}-x_i)^{1-6/\kappa}$.  This latter factor cancels the factor of $(x_{i+1}-x_i)^{6/\kappa-1}$ that accompanies $\bar{\ell}_1$ (\ref{lim}, \ref{thelim}).  After inserting these results into the formula (\ref{firstFexplicit1}) (with $c=2N$) for $\mathcal{F}_\vartheta$ and then into (\ref{thelim}), we find that
\begin{multline}\label{firstlim}\bar{\ell}_1\mathcal{F}_\vartheta(\kappa\,|\,x_1,x_2,\ldots,x_{i-1},x_{i+2},\ldots,x_{2N})\,\,=\,\,n(\kappa)\,\,\times\\
\left\{\begin{aligned}&n(\kappa)\left[\frac{n(\kappa)\Gamma(2-8/\kappa)}{4\sin^2(4\pi/\kappa)\Gamma(1-4/\kappa)^2}\right]^{N-2}\Bigg(\prod_{\substack{j<k \\ j,k\neq i,i+1}}^{2N-1}(x_k-x_j)^{2/\kappa}\Bigg)\Bigg(\prod_{\substack{k=1 \\ k\neq i,i+1}}^{2N-1}(x_{2N}-x_k)^{1-6/\kappa}\Bigg)\oint_{\Gamma_{N-1}}{\rm d}u_{N-1}\dotsm\\ 
&\dotsm\,\oint_{\Gamma_3}{\rm d}u_3\,\,\oint_{\Gamma_2}{\rm d}u_2\,\,\mathcal{N}\Bigg[\Bigg(\prod_{\substack{l=1 \\ l\neq i,i+1}}^{2N-1}\prod_{m=2}^{N-1}(x_l-u_m)^{-4/\kappa}\Bigg)\Bigg(\prod_{m=2}^{N-1}(x_{2N}-u_m)^{12/\kappa-2}\Bigg)\Bigg(\prod_{2\leq p<q}^{N-1}(u_p-u_q)^{8/\kappa}\Bigg)\Bigg]\end{aligned}\right\}.\end{multline}
The bracketed function of (\ref{firstlim}) is the element of $\mathcal{B}_{N-1}$ generated by dropping from the formula (\ref{firstFexplicit1}) for $\mathcal{F}_\vartheta$ all factors involving $x_i$, $x_{i+1}$, and $u_1$, dropping the integration along $\Gamma_1$, and reducing the power in (\ref{firstprefactor}) by one.  Removing the arc connecting $x_i$ with $x_{i+1}$ in the half-plane diagram for $\mathcal{F}_\vartheta$ creates the half-plane diagram for this element of $\mathcal{B}_{N-1}$.

Incidentally, if $8/\kappa\not\in2\mathbb{Z}^++1$, then the limit (\ref{firstlim}) is not zero, so $(x_i,x_{i+1})$ is not a two-leg interval of $\mathcal{F}_\vartheta$.  If in addition $8/\kappa\not\in2\mathbb{Z}^+$, then $(x_i,x_{i+1})$ is an identity interval of $\mathcal{F}_\vartheta$ because $(x_{i+1}-x_i)^{6/\kappa-1}\mathcal{F}(\kappa\,|\,\boldsymbol{x})$ is analytic at $x_{i+1}=x_i$.  (If $8/\kappa\in2\mathbb{Z}^+$, then $(x_i,x_{i+1})$ is still an identity interval of $\mathcal{F}_\vartheta$, according to definition \red{8} of \cite{florkleb4}.)  On the other hand, if $8/\kappa\in2\mathbb{Z}^++1$, then as we previously discussed above (\ref{Freplace}), we replace $\mathcal{F}_\vartheta(\kappa)=0$ with $\mathcal{F}^{\scaleobj{0.75}{\bullet}}_\vartheta(\kappa):=\lim_{\varkappa\rightarrow\kappa}n(\varkappa)^{-1}\mathcal{F}_\vartheta(\varkappa)$.  Now, because $n(\kappa)$ (\ref{fugacity}) vanishes, the limit (\ref{thelim}) (equaling (\ref{firstlim}) with one factor of $n(\kappa)$ between the braces dropped) vanishes too, so $(x_i,x_{i+1})$ is a two-leg interval of $\mathcal{F}^{\scaleobj{0.75}{\bullet}}_\vartheta(\kappa)$.

Finally, exactly one of $x_i$ or $x_{i+1}$ may be an endpoint of one contour among $\Gamma_1$, $\Gamma_2,\ldots,\Gamma_{N-1}$, or both $x_i$ and $x_{i+1}$ may be endpoints of different contours.  In appendix \ref{asymp}, we prove that $(x_i,x_{i+1})$ is again not a two-leg interval of $\mathcal{F}_\vartheta$ in these last cases for all $\kappa\in(0,8)$.  In general, it seems that an interval sharing (resp.\ not sharing) its endpoints with integration contours of $\mathcal{F}_\vartheta$ is not (resp.\ is) a two-leg interval of $\mathcal{F}_\vartheta$.  This observation is fundamental to calculating $v(\mathcal{B}_N)$ in the proof of the following lemma and ultimately to proving theorem \ref{maintheorem}, our main result, below.

\begin{lem}\label{mainlem}Suppose that $\kappa\in(0,8)$.  Then the Temperley-Lieb set $\mathcal{B}_N$ (item \ref{item2a} of definition \ref{Fkdefn}) is linearly independent if and only if $\kappa$ is not an exceptional speed (\ref{exceptional}) with $q\leq N+1$.
\end{lem}

\begin{proof}  To prove the lemma, we show that $v(\mathcal{B}_N):=\{v(\mathcal{F}_1),v(\mathcal{F}_2),\ldots,v(\mathcal{F}_{C_N})\}$ is linearly independent if and only if $\kappa\in(0,8)$ is not an exceptional speed with $q\leq N+1$, and then we invoke lemma \red{15} of \cite{florkleb}.  In order to prove the former claim, we must calculate $[\mathscr{L}_\varsigma]\mathcal{F}_\vartheta$ for all $\varsigma,\vartheta\in\{1,2,\ldots,C_N\}$. After choosing arbitrary $\varsigma$ and $\vartheta$ and noting that the diagram of $[\mathscr{L}_\varsigma]$ has at least one arc with its endpoints at $x_i$ and $x_{i+1}$ for some $i\in\{1,2,\ldots,2N-2\}$, we choose an element $\mathscr{L}_\varsigma$ of this equivalence class whose first limit $\bar{\ell}_1$ sends $x_{i+1}\rightarrow x_i$.  As we noted earlier, the value of the limit $\bar{\ell}_1\mathcal{F}_\vartheta$ (\ref{thelim}) depends on whether or not $x_i$ or $x_{i+1}$ is an endpoint of an integration contour of $\mathcal{F}_\vartheta$.  There are four cases to consider (figure \ref{Cases}), and we explain the calculation of $\bar{\ell}_1\mathcal{F}_\vartheta$ (\ref{thelim}) in all four cases for $8/\kappa\not\in\mathbb{Z}^+$ respectively in items \ref{firstcase}--\ref{fourthcase} below, deferring details to appendix \ref{asymp}.  Afterwards, we extend our results to $8/\kappa\in\mathbb{Z}^+$ below (\ref{LkFk}).

\begin{enumerate}
\item\label{firstcase} 
\textbf{Configuration:} In case \ref{firstcase}, neither $x_i$ nor $x_{i+1}$ are endpoints of an integration contour of $\mathcal{F}_\vartheta$ (figure \ref{Cases}).  (This case does not arise at first because only one point among $x_1$, $x_2,\ldots,x_{2N-1}$ may not be an endpoint of any contour of $\mathcal{F}_\vartheta$.  However, this case does occur after we deform integration contours in cases \ref{thirdcase} and \ref{fourthcase} below.)

\noindent
\textbf{Calculation:}  Section \ref{s1} in appendix \ref{asymp} presents the calculation of the limit $\bar{\ell}_1\mathcal{F}_\vartheta$ (\ref{thelim}).  To summarize, we set $x_{i+1}=x_i$ in the formula (\ref{firstFexplicit1}) with $c=2N$.  The factor of $(x_{i+1}-x_i)^{2/\kappa}$ in this formula multiplies the factor of $(x_{i+1}-x_i)^{6/\kappa-1}$ that accompanies $\bar{\ell}_1$ to give $(x_{i+1}-x_i)^{8/\kappa-1}$, which vanishes as $x_{i+1}\rightarrow x_i$.  Because the Coulomb gas integral and all other factors of $\mathcal{F}_\vartheta$ approach finite values as $x_{i+1}\rightarrow x_i$, the limit $\bar{\ell_1}\mathcal{F}_\vartheta$ (\ref{thelim}) is zero.

\noindent
\textbf{Main result:} In case \ref{firstcase}, the limit $\bar{\ell}_1\mathcal{F}_\vartheta$ (\ref{thelim}) is zero (figure \ref{Cases}).  Thus, $(x_i,x_{i+1})$ is a two-leg interval of $\mathcal{F}_\vartheta$.

\begin{figure}[b!]
\centering
\includegraphics[scale=0.27]{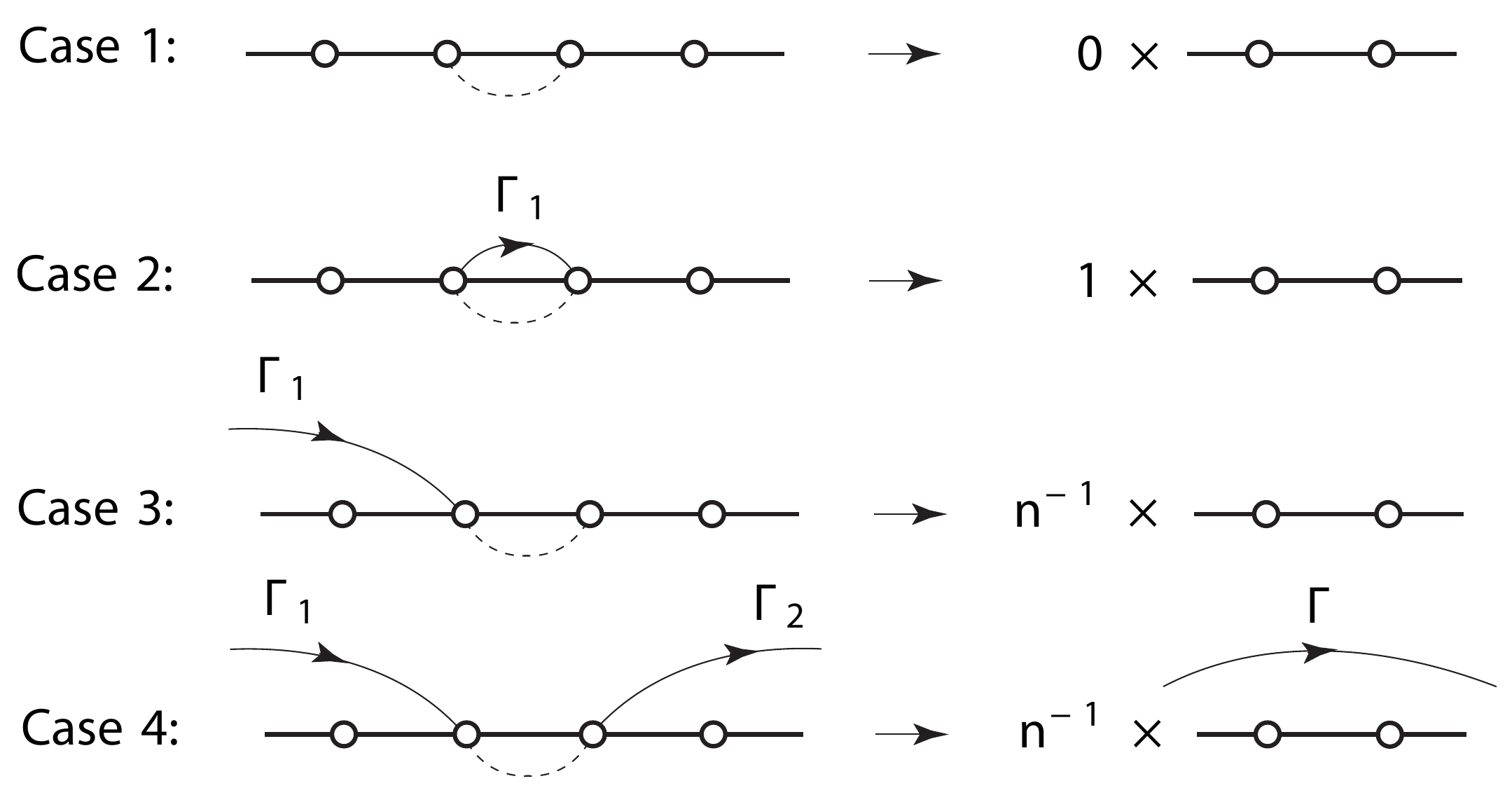}
\caption{Cases \ref{firstcase}--\ref{fourthcase} of an interval collapse.  The dashed curve connects the endpoints of the intervals to be collapsed, and the solid curves are the integration contours.  Figure \ref{Case4} shows two other contour arrangements that fall under case \ref{fourthcase}.}
\label{Cases}
\end{figure}

\item\label{secondcase} \textbf{Configuration:} In case \ref{secondcase}, both $x_i$ and $x_{i+1}$ are endpoints of a single, common integration contour $\Gamma_1$ of $\mathcal{F}_\vartheta$.  Hence, this contour is $\Gamma_1=\mathscr{P}(x_i,x_{i+1})$ or $[x_i,x_{i+1}]^+$ (figure \ref{Cases}).  (The superscript $+$ indicates that we form the contour $[x_i,x_{i+1}]^+$ by slightly bending $[x_i,x_{i+1}]$ into the upper half-plane, keeping the endpoints fixed.)

\noindent
\textbf{Calculation:} The discussion in the above paragraph with (\ref{firstgroup}--\ref{betafunc}) outlines the calculation of the limit $\bar{\ell}_1\mathcal{F}_\vartheta$ (\ref{thelim}) (figure \ref{Cases}), and the last paragraphs of section \ref{s2} in appendix \ref{asymp} present this calculation in full detail.  (See the ``main result" in that section.)

\noindent
\textbf{Main result:} In case \ref{secondcase}, the limit $\bar{\ell}_1\mathcal{F}_\vartheta$ (\ref{thelim}) equals $n$ (\ref{fugacity}) times the element of $\mathcal{B}_{N-1}$ generated from the formula (\ref{firstFexplicit1}) for $\mathcal{F}_\vartheta$ by dropping all factors involving $x_i$, $x_{i+1}$, and $u_1$, dropping the integration along $\Gamma_1$, and reducing the prefactor power in (\ref{firstprefactor}) or (\ref{secondprefactor}) by one.  Also, $(x_i,x_{i+1})$ is an identity interval of $\mathcal{F}_\vartheta$ if $8/\kappa\not\in2\mathbb{Z}^++1$.

\item\label{thirdcase} 
\textbf{Configuration:} In case \ref{thirdcase}, either $x_i$ or $x_{i+1}$ is an endpoint of a single integration contour $\Gamma_1$ of $\mathcal{F}_\vartheta$, but the other is not an endpoint of any contour (figure \ref{Cases}).  Thus, the latter point is the left endpoint of the arc terminating at $x_{2N}$ in the half-plane diagram for $\mathcal{F}_\vartheta$.

\noindent
\textbf{Initial assumptions:} We assume that $i>1$, $\kappa>4$, $x_i$ is an endpoint of $\Gamma_1$, and $\Gamma_1$ does not pass over the interval $(x_i,x_{i+1})$.  Under ``further details" below, we explain why these assumptions hold, or we extend our proof to situations in which they do not.

\noindent
\textbf{Calculation:} With $\kappa>4$, we simplify the formula (\ref{firstFexplicit1}) with $c=2N$ as per item \ref{simplifyformula} of definition \ref{Fkdefn}.  Thus, $\Gamma_1$ is a simple contour, and we decompose it into one simple contour $\Gamma_1'$ with its right endpoint at $x_{i-1}$ and another $\Gamma_1''$ with its endpoints at $x_i$ and $x_{i-1}$.  (We might have $\Gamma_1'=\emptyset$ and $\Gamma_1=\Gamma_1''$.)  Also, the limit (\ref{thelim}) breaks into
\begin{multline}\label{limdecomp}\lim_{x_{i+1}\rightarrow x_i}(x_{i+1}-x_i)^{6/\kappa-1}\mathcal{F}_\vartheta(\kappa\,|\,\boldsymbol{x})\,\,\,=\lim_{x_{i+1}\rightarrow x_i}(x_{i+1}-x_i)^{6/\kappa-1}\big(\mathcal{F}_\vartheta\big|_{\Gamma_1\mapsto\Gamma_1'}\big)(\kappa\,|\,\boldsymbol{x})\\
+\lim_{x_{i+1}\rightarrow x_i}(x_{i+1}-x_i)^{6/\kappa-1}\big(\mathcal{F}_\vartheta\big|_{\Gamma_1\mapsto\Gamma_1''}\big)(\kappa\,|\,\boldsymbol{x}).\end{multline}
The first limit on the right side of (\ref{limdecomp}) falls under case \ref{firstcase} and therefore vanishes.  Meanwhile, the second limit on the right side of (\ref{limdecomp}) still falls under case \ref{thirdcase}, and we compute it in section \ref{s3} of appendix \ref{asymp}.  (See the ``main result" in that section.)  There, we deform $\Gamma_1''$ in a way that generates terms only falling under cases \ref{firstcase} and \ref{secondcase}.  Only the latter type of term has a non-vanishing limit as $x_{i+1}\rightarrow x_i$, and a factor of $n^{-1}$ accompanies it (figure \ref{Cases}).  Thus, the second limit on the right side of (\ref{limdecomp}) equals that found in case \ref{secondcase} multiplied by this extra factor.

\noindent
\textbf{Main result:} In case \ref{thirdcase}, the limit $\bar{\ell}_1\mathcal{F}_\vartheta$ (\ref{thelim}) equals the element of $\mathcal{B}_{N-1}$ generated from the formula (\ref{firstFexplicit1}) for $\mathcal{F}_\vartheta$ by dropping all factors involving $x_i$, $x_{i+1}$, and $u_1$, dropping the integration along $\Gamma_1$, and reducing the prefactor power in (\ref{firstprefactor}) or (\ref{secondprefactor}) by one.

\noindent
\textbf{Further details:} Above, we assume that $i>1$, $\kappa>4$, $x_i$ is an endpoint of $\Gamma_1$, and $\Gamma_1$ does not pass over $(x_i,x_{i+1})$.  Here, we justify these assumptions, or we extend our proof to situations in which they do not hold.
\begin{enumerate}
\item\label{ineq1}We have $i>1$.  For if $i=1$, then $x_1$, and not $x_2$, must be the other endpoint of the arc terminating at $x_{2N}$ in the half-plane diagram for $\mathcal{F}_\vartheta$.  (Indeed, one endpoint of this arc must have an odd index while the other has the even index $2N$.)  But then, $x_1$ is an endpoint of the arc corresponding to $\Gamma_1$ too, an impossibility.
\item $\Gamma_1$ may not pass over $(x_i,x_{i+1})$.  Indeed, if it did, then its arc would cross the arc joining an endpoint of $(x_i,x_{i+1})$ with $x_{2N}$ in the half-plane diagram for $\mathcal{F}_\vartheta$, an impossibility.  (See definition \ref{Fkdefn}.)
\item If $x_{i+1}$ is an endpoint of $\Gamma_1$ but $x_i$ is not (here, we may have $i=1$), then repeating the above analysis yields the same result.  (But if $i=2N-2$ and $\kappa<4$, then a subtlety arises, which we explore next.)
\item\label{kappaextend} To extend our results to $\kappa\in(0,4]$, we revert the formula for $\mathcal{F}_\vartheta$ back to (\ref{firstFexplicit1}) with $c=2N$ on the left side of (\ref{limdecomp}) by replacing the simple contours with Pochhammer contours entwining the same endpoints and the prefactor (\ref{secondprefactor}) with (\ref{firstprefactor}).  (Here, the integration contour $\Gamma$ for $\int_\Gamma$ is simple, but the contour for $\oint_\Gamma$ is Pochhammer with the same endpoints as the former simple contour.  Although these two contours are different, we use the same symbol $\Gamma$ for both.)
\begin{multline}\label{adjustment}n(\kappa)\left[\frac{n(\kappa)\Gamma(2-8/\kappa)}{\Gamma(1-4/\kappa)^2}\right]^{N-1}\sideset{}{_{\Gamma_{N-1}}}\int\dotsm\sideset{}{_{\Gamma_2}}\int\sideset{}{_{\Gamma_1}}\int\\
\quad\longmapsto\quad n(\kappa)\left[\frac{n(\kappa)\Gamma(2-8/\kappa)}{4\sin^2(4\pi/\kappa)\Gamma(1-4/\kappa)^2}\right]^{N-1}\sideset{}{_{\Gamma_{N-1}}}\oint\dotsm\sideset{}{_{\Gamma_2}}\oint\sideset{}{_{\Gamma_1}}\oint.\end{multline}
Next, we adjust the terms on the right side of (\ref{limdecomp}) in exactly the same way as (\ref{adjustment}) shows, with $\Gamma_1$ replaced by $\Gamma_1'$ or $\Gamma_1''$, and these changes do not alter these terms if $\kappa>4$, thanks to identity (\ref{Pochtostraight}).  Because this adjusted version of (\ref{limdecomp}) is true for $\kappa>4$ and all of its terms are analytic functions of $\kappa\in(0,8)\times i\mathbb{R}$, we arrive with the analytic continuation of (\ref{limdecomp}) to this strip.  Performing the same analysis on this analytic continuation of (\ref{limdecomp}) gives the same main result that we previously found.  (See section \ref{s3} of appendix \ref{asymp}.)

There is one subtlety worth mentioning.  If $x_{i+1}=x_{2N-1}$ is an endpoint of $\Gamma_1$, then $x_{i+2}=x_{2N}$ is an endpoint of both $\Gamma_1'$ and $\Gamma_1''$.  Now, the point $x_{2N}$ is exceptional because it bears the conjugate charge.  Indeed, the integrand of $\mathcal{F}_\vartheta$ accumulates a factor of $e^{12\pi i/\kappa}$ as we cycle around it, instead of $e^{-4\pi i/\kappa}$ as we cycle around any $x_j$ with $j<2N$.  As such, some factors in the adjustment (\ref{adjustment}) to the integration along $\Gamma_1''$ and $\Gamma_1'$ in (\ref{limdecomp}) are different.  Indeed, identity (\ref{Pochtostraight}) implies that the adjustment to this former term is
\begin{multline}\label{nextadjustment}n(\kappa)\left[\frac{n(\kappa)\Gamma(2-8/\kappa)}{\Gamma(1-4/\kappa)^2}\right]^{N-1}\sideset{}{_{\Gamma_{N-1}}}\int\dotsm\sideset{}{_{\Gamma_2}}\int\sideset{}{_{\Gamma_1''}}\int(u_1-x_{2N-1})^{-4/\kappa}(x_{2N}-u_1)^{12/\kappa-2}\dotsm\\
\begin{aligned}\longmapsto\qquad&\overbrace{n(\kappa)\left[\frac{n(\kappa)\Gamma(2-8/\kappa)}{4\sin^2(4\pi/\kappa)\Gamma(1-4/\kappa)^2}\right]^{N-2}\sideset{}{_{\Gamma_{N-1}}}\oint\dotsm\sideset{}{_{\Gamma_3}}\oint\sideset{}{_{\Gamma_2}}\oint\dotsm}^{1.}\\
&\underbrace{\Bigg[\frac{e^{16\pi i/\kappa}n(\kappa)\Gamma(2-8/\kappa)}{4\sin(-4\pi/\kappa)\Gamma(1-4/\kappa)^2}\Bigg]}_{2.}\,\,\underbrace{\frac{1}{\sin(12\pi/\kappa)}\sideset{}{_{\Gamma_1''}}\oint(u_1-x_{2N-1})^{-4/\kappa}(x_{2N}-u_1)^{12/\kappa-2}\dotsm}_{3.}\,\,,\end{aligned}\end{multline}
while the adjustment to the latter term of (\ref{limdecomp}) is almost identical.  In the discussion surrounding (\ref{Freplace}), we show that the first factor in (\ref{nextadjustment}) is analytic in $\kappa\in(0,8)\times i\mathbb{R}$.  With (\ref{fugacity}), it is easy to see that the second factor in (\ref{nextadjustment}) is analytic in this strip too.  Finally, because the sine function in the third factor has a simple zero at $\kappa=12/(r+2)$ for all $r\in\mathbb{Z}$ but the contour integral it multiplies has a zero at these same locations for all $r\in\mathbb{Z}^+\cup\{0\}$ (see (\ref{PochDecomp}) with $\beta_j=r\in\mathbb{Z}^+\cup\{0\}$), the third factor is analytic in $\kappa\in(0,8)\times i\mathbb{R}$ too.
\end{enumerate}

\item\label{fourthcase} \textbf{Configuration:} In case \ref{fourthcase}, the most complicated case, $x_i$ is an endpoint of one contour $\Gamma_1$ of $\mathcal{F}_\vartheta$, and $x_{i+1}$ is an endpoint of a different contour $\Gamma_2$ (figure \ref{Cases}).

\noindent
\textbf{Initial assumptions:} We assume that $i>1$ and $\kappa>4$.  Under ``further details" below, we extend our proof to situations in which these assumptions do not hold.

\noindent
\textbf{Calculation:} With $\kappa>4$, we simplify the formula (\ref{firstFexplicit1}) with $c=2N$ as in item \ref{simplifyformula} of definition \ref{Fkdefn}.  Thus, $\Gamma_1$ and $\Gamma_2$ are simple contours.  We decompose $\Gamma_1$ (resp.\ $\Gamma_2$) into one simple contour $\Gamma_1'$ (resp.\ $\Gamma_2'$) with an endpoint at $x_{i-1}$ (resp.\ $x_{i+2}$) and another $\Gamma_1''$ (resp.\ $\Gamma_2''$) with its endpoints at $x_{i-1}$ and $x_i$ (resp.\ $x_{i+1}$ and $x_{i+2}$).   (In some cases, we might have $\Gamma_1'=\emptyset$ and $\Gamma_1''=\Gamma_1$ (resp.\ $\Gamma_2'=\emptyset$ and $\Gamma_2''=\Gamma_2$).)  This decomposition contains within it three sub-cases: neither $\Gamma_1$ nor $\Gamma_2$ passes over $(x_i,x_{i+1})$, only $\Gamma_2$ passes over $(x_i,x_{i+1})$, or only $\Gamma_1$ passes over $(x_i,x_{i+1})$ (figure \ref{Case4}).  Similarly, the limit (\ref{thelim}) decomposes into
\be\begin{aligned}\label{biggerlimdecomp}\lim_{x_{i+1}\rightarrow x_i}(x_{i+1}-x_i)^{6/\kappa-1}\mathcal{F}_\vartheta(\kappa\,|\,\boldsymbol{x})\,\,&=\lim_{x_{i+1}\rightarrow x_i}(x_{i+1}-x_i)^{6/\kappa-1}\big(\mathcal{F}_\vartheta\big|_{(\Gamma_1,\Gamma_2)\mapsto(\Gamma_1',\Gamma_2')}\big)(\kappa\,|\,\boldsymbol{x})\\
&+\lim_{x_{i+1}\rightarrow x_i}(x_{i+1}-x_i)^{6/\kappa-1}\big(\mathcal{F}_\vartheta\big|_{(\Gamma_1,\Gamma_2)\mapsto(\Gamma_1'',\Gamma_2')}\big)(\kappa\,|\,\boldsymbol{x})\\
&+\lim_{x_{i+1}\rightarrow x_i}(x_{i+1}-x_i)^{6/\kappa-1}\big(\mathcal{F}_\vartheta\big|_{(\Gamma_1,\Gamma_2)\mapsto(\Gamma_1',\Gamma_2'')}\big)(\kappa\,|\,\boldsymbol{x})\\
&+\lim_{x_{i+1}\rightarrow x_i}(x_{i+1}-x_i)^{6/\kappa-1}\big(\mathcal{F}_\vartheta\big|_{(\Gamma_1,\Gamma_2)\mapsto(\Gamma_1'',\Gamma_2'')}\big)(\kappa\,|\,\boldsymbol{x}).\end{aligned}\ee
The first limit on the right side of (\ref{biggerlimdecomp}) falls under case \ref{firstcase} and therefore vanishes.  The second (resp.\ third) limit on the right side of (\ref{biggerlimdecomp}) falls under case \ref{thirdcase} and therefore equals the element of $\mathcal{B}_{N-1}$ with contours $\Gamma_2',$ (resp.\ $\Gamma_1'$) $\Gamma_3,\ldots,\Gamma_{N-1}$.  Finally, the fourth limit on the right side of (\ref{biggerlimdecomp}) still falls under case \ref{fourthcase}, and we compute it in section \ref{s4} of appendix \ref{asymp}.  (See the ``main result" in that section.)  There, we deform $\Gamma_1''$ and $\Gamma_2''$ in a way that generates terms only falling under cases \ref{firstcase} and \ref{secondcase}.  Only the latter type of term has a non-vanishing limit as $x_{i+1}\rightarrow x_i$, and a factor of $n^{-1}$ accompanies it.  Thus, the last limit on the right side of (\ref{biggerlimdecomp}) is the element of $\mathcal{B}_{N-1}$ with contours $\Gamma_0':=[x_{i-1},x_{i+2}]^+$, $\Gamma_3,\Gamma_4,\ldots,\Gamma_{N-1}$.  After summing all four terms, we find that the right side of (\ref{biggerlimdecomp}) equals the element of $\mathcal{B}_{N-1}$ with contours $\Gamma:=\Gamma_0'+\Gamma_1'+\Gamma_2',$ and $\Gamma_3,$ $\Gamma_4,\ldots,\Gamma_{N-1}$ (figure \ref{Cases}).

\begin{figure}[p]
\centering
\includegraphics[scale=0.27]{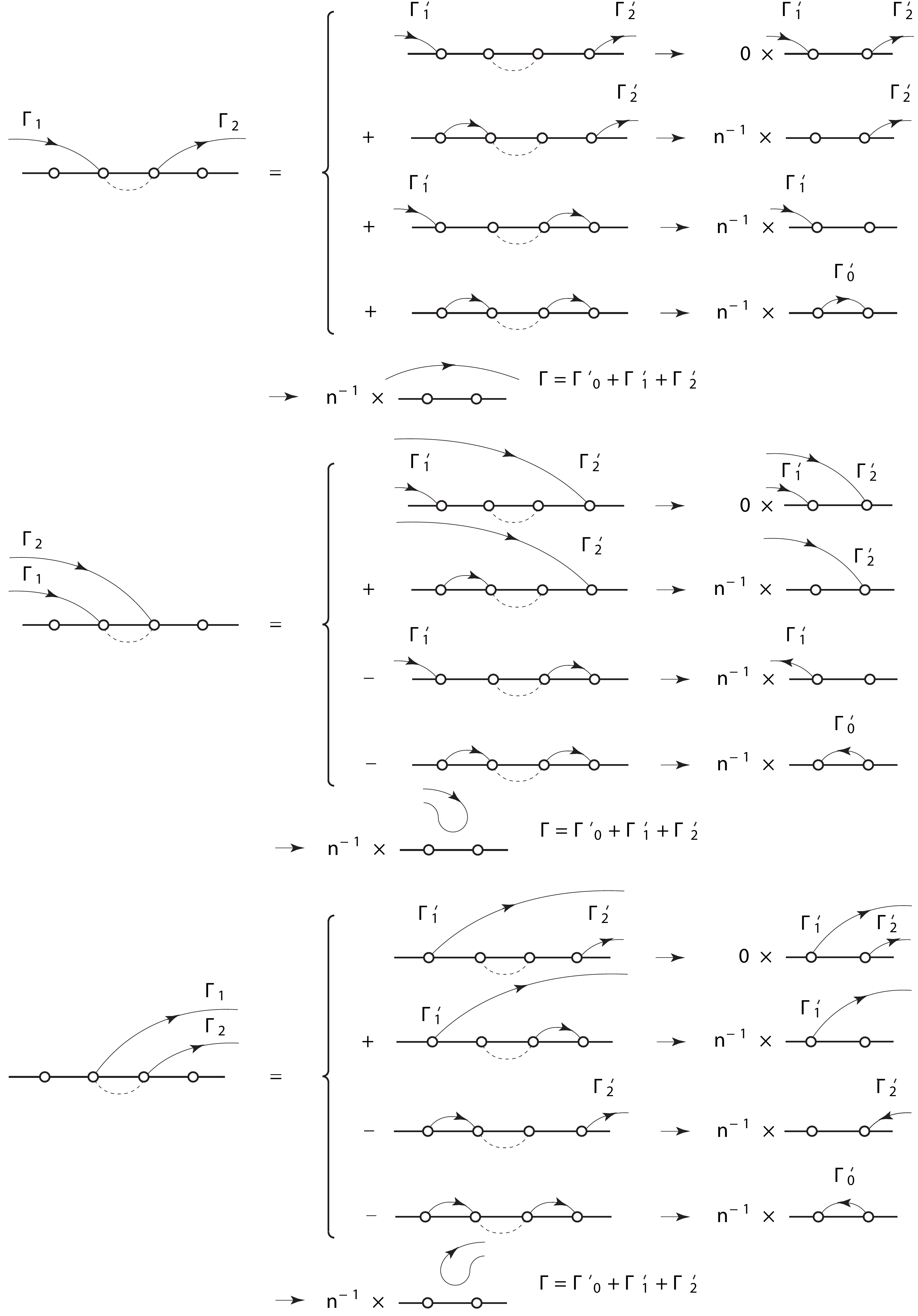}
\caption{The decomposition of case \ref{fourthcase} into cases \ref{firstcase}--\ref{thirdcase} and a simpler version of case \ref{fourthcase}.  The top and middle two terms to the right of each bracket fall under cases \ref{firstcase} and \ref{thirdcase} respectively, and the bottom term to the right of each bracket falls under case \ref{fourthitem}.}
\label{Case4}
\end{figure}

\noindent
\textbf{Main result:} In case \ref{fourthcase}, the limit $\bar{\ell}_1\mathcal{F}_\vartheta$ (\ref{thelim}) equals the element of $\mathcal{B}_{N-1}$ generated from the formula (\ref{firstFexplicit1}) for $\mathcal{F}_\vartheta$ by dropping all factors involving $x_i$, $x_{i+1}$, and $u_1$, dropping the integration along $\Gamma_1$, integrating $u_2$ along the contour $\Gamma$ generated by pulling $x_i$ and $x_{i+1}$ together to join $\Gamma_1$ with $\Gamma_2$, and reducing the prefactor power in (\ref{firstprefactor}) or (\ref{secondprefactor}) by one.

\noindent
\textbf{Further details:} Above, we assume that $i>1$ and $\kappa>4$.  Here, we extend our proof to situations in which they do not hold.
\begin{enumerate}
\item\label{i-1=0}We may have $i=1$.  Here, we informally identify the index $i-1=0$ with $2N$ and treat this scenario identically to scenarios with $i>1$.  More formally, we expand $\Gamma_1$, with its endpoints fixed, into the upper half-plane and deform it into the shape of a semicircle with a very large radius $R$ and with its base flush against the real axis.  As $R\rightarrow\infty$, the integration along the arc of the semicircle vanishes like $R^{-1}$, and the integration path along the base decomposes into three segments: $[-\infty,x_1]^+$, $[x_{2N},\infty]^+$, and the remainder of the base, which we call $\Gamma_1'$.  Because, the integrand of $\mathcal{F}_\vartheta$, as a function of $u_1$, vanishes like $u_1^{-2}$ as $u_1\rightarrow\pm\infty$ (\ref{firstFexplicit1}), the integrations along the first two segments converge, and we join them into one integration along $\Gamma_1''=[x_{2N},x_1]^+$, where the interval $[x_{2N},x_1]$ contains infinity.

\item\label{lastb}To extend our results to $\kappa\in(0,4]$, we use the analytic continuation described above in \ref{kappaextend}.  (We note that the point $x_{2N}$ may be an endpoint of $\Gamma_2'$ and $\Gamma_2''$ but never of $\Gamma_1'$ or $\Gamma_1''$, and this happens only if $i=2N-2$.)
\end{enumerate}
\end{enumerate}
To summarize, we compute the limit $\bar{\ell}_1\mathcal{F}_\vartheta$ (\ref{thelim}) for each of the four possible cases (figure \ref{Cases}) in which an integration contour may surround the points $x_i$ or $x_{i+1}$.  Items \ref{firstcase}--\ref{fourthcase} above respectively give these calculations for these cases.

\begin{figure}[b]
\centering
\includegraphics[scale=0.3]{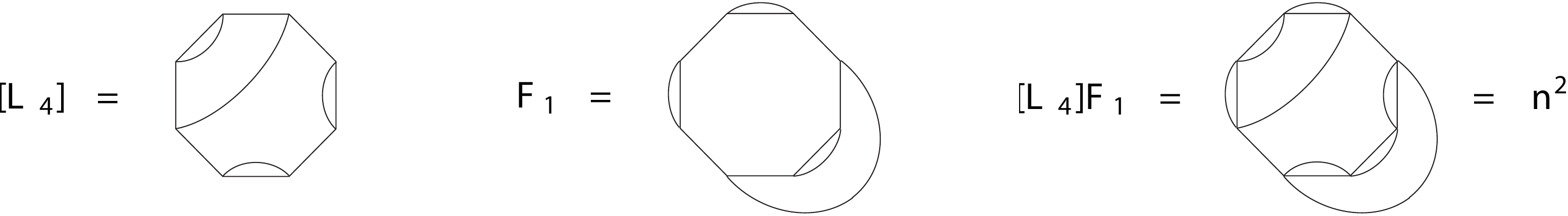}
\caption{The diagram for $[\mathscr{L}_4]\in\mathscr{B}_4^*$, for $\mathcal{F}_1\in\mathcal{B}_4$, and for their product $[\mathscr{L}_4]\mathcal{F}_1\in\mathbb{R}$.  The product diagram contains two loops and therefore evaluates to $n^2$.}
\label{innerproduct}
\end{figure}

We recall that $\bar{\ell}_1$ (\ref{thelim}) is the first of a collection $\{\bar{\ell}_1,\bar{\ell}_2,\ldots,\bar{\ell}_N\}$ of $N$ limits for some element $\mathscr{L}_\varsigma:=\bar{\ell}_N\dotsm\bar{\ell}_2\bar{\ell}_1$ of the equivalence class $[\mathscr{L}_\varsigma]$, and computing $\bar{\ell}_1\mathcal{F}_\vartheta$ is our first step toward our ultimate goal of finding $[\mathscr{L}_\varsigma]\mathcal{F}_\vartheta$.  Now with this first limit determined, computing $\bar{\ell}_2\bar{\ell}_1\mathcal{F}_\vartheta$ is the next step.  But with $\bar{\ell}_1\mathcal{F}_\vartheta\in\mathcal{S}_{N-1}$ thanks to lemma \red{5} of \cite{florkleb}, the calculation of this second limit is identical to the calculation of the first, with $N$ replaced by $N-1$.  The same is true of the other $N-2$ limits that compose $\mathscr{L}_\varsigma$.  Exploiting the similarities of these calculations, we use a diagrammatic method introduced in \cite{js} to facilitate our calculation of $[\mathscr{L}_\varsigma]\mathcal{F}_\vartheta$.  We draw the polygon diagram for $[\mathscr{L}_\varsigma]$ and that for $\mathcal{F}_\vartheta$ on the same polygon (figure \ref{innerproduct}) and call the result the diagram for $[\mathscr{L}_\varsigma]\mathcal{F}_\vartheta$.  The interior and exterior arcs of this diagram respectively represent the limits of $[\mathscr{L}_\varsigma]$ to be taken and the integration contours of $\mathcal{F}_\vartheta$  (except for the exterior arc with an endpoint at $x_{2N}$, which has no associated integration contour). Now, each vertex of the $2N$-sided polygon $\mathcal{P}$ in this diagram is the endpoint of a unique exterior arc and a unique interior arc.  Thus, starting on an arbitrary interior arc, we may follow it in a given (say clockwise) direction, passing onto an exterior arc, and then another interior arc, etc., until we return to our starting point.  All of the arcs thus traversed  join to form a loop that dodges in and out of $\mathcal{P}$ through its vertices.  If an arc in the diagram for $[\mathscr{L}_\varsigma]\mathcal{F}_\vartheta$ is not a part of this loop, then we repeat the process starting with that arc, and we continue this until all arcs are included in a loop.  Thus, all of the arcs in the diagram for $[\mathscr{L}_\varsigma]\mathcal{F}_\vartheta$ join to form $l_{\varsigma,\vartheta}\in\mathbb{Z}^+$ loops. 

We recall from the first paragraph of this proof that $x_i$ and $x_{i+1}$ are endpoints of a common interior arc in the half-plane diagram for $[\mathscr{L}_\varsigma]$.  As such, the corresponding polygon vertices $w_i$ and $w_{i+1}$ are endpoints of a common interior arc in the diagram for $[\mathscr{L}_\varsigma]\mathcal{F}_\vartheta$.  Now one of two possibilities may happen.
\begin{enumerate}[I.]
\item The vertices $w_i$ and $w_{i+1}$ may be endpoints of the same exterior arc in the diagram for $[\mathscr{L}_\varsigma]\mathcal{F}_\vartheta$, joining with their interior arc to form a loop that intersects the polygon $\mathcal{P}$ of that diagram only at those vertices.  We identify this situation with case \ref{secondcase} above.  Collapsing the interval $(x_i,x_{i+1})$ deletes this loop and the side it surrounds from $\mathcal{P}$, fusing its adjacent sides together to create a $(2N-2)$-sided polygon $\mathcal{P}'$.  This modification sends $\mathcal{F}_\vartheta$ to $n$ times the element of $\mathcal{B}_{N-1}$ whose diagram is given by the remaining $N-1$ exterior arcs attached to $\mathcal{P}'$.    
\item The vertices $w_i$ and $w_{i+1}$ may not be endpoints of the same exterior arc in the diagram for $[\mathscr{L}_\varsigma]\mathcal{F}_\vartheta$.  We identify this situation with either case \ref{thirdcase} or \ref{fourthcase} above.  Collapsing the interval $(x_i,x_{i+1})$ deletes the corresponding side and its attached interior arc from $\mathcal{P}$, fusing its adjacent sides together to create a $(2N-2)$-sided polygon $\mathcal{P}'$, and joining the two exterior arcs with an endpoint at $w_i$ or $w_{i+1}$ into one exterior arc.  This modification sends $\mathcal{F}_\vartheta$ to the element of $\mathcal{B}_{N-1}$ whose diagram is given by the remaining $N-1$ exterior arcs attached to $\mathcal{P}'$. 
\end{enumerate}
  
We repeat collapsing the sides of $\mathcal{P}$ this way another $N-1$ more times.  As we do this, we eventually contract away each loop in the diagram for $[\mathscr{L}_\varsigma]\mathcal{F}_\vartheta$ (with the polygon deleted), finding a factor of $n$ in its wake.  Thus (figure \ref{innerproduct}),
\be\label{LkFk}\langle[\mathscr{L}_\varsigma],[\mathscr{L}_{\vartheta}]\rangle:=[\mathscr{L}_\varsigma]\mathcal{F}_\vartheta=n^{l_{\varsigma,\vartheta}}.\ee
J.\ Simmons independently discovered this result (\ref{LkFk}) before us for $N\in\{1,2,3\}$, and he published the case $N=2$ in his article on percolation and logarithmic CFT \cite{js}.  (See figure \red{7} of that article.)  Prior to publication, he shared these results with us, and from that, we anticipated the more general result (\ref{LkFk}) for all $N\in\mathbb{Z}^+$.

So far, we have proven (\ref{LkFk}) only for all $\kappa\in(0,8)$ with $8/\kappa\not\in\mathbb{Z}^+$.  To remove this latter restriction, we first note that for some $\epsilon>0$, each limit in every element of $[\mathscr{L}_\varsigma]$ is uniform over $\mathcal{K}:=(\kappa-\epsilon,\kappa+\epsilon)$ with $8/\kappa\in\mathbb{Z}^+$.  We may prove this claim by sending $\delta\downarrow0$ in (\red{65}) and (\red{60}) of \cite{florkleb} (after taking the supremum of the latter over $\mathcal{K}$).  (See the proof of lemma \red{4} of \cite{florkleb} for context, keeping in mind that $\mathcal{K}$ has different meaning in that proof, although this does not matter here.)  Thus, we may commute the limit $\varkappa\rightarrow\kappa$ with each limit of every element of $[\mathscr{L}_\varsigma]$.  So by sending $\varkappa\rightarrow\kappa$ on both sides of (\ref{LkFk}) and commuting this limit with $[\mathscr{L}_\varsigma]$, we prove (\ref{LkFk}) for $8/\kappa\in\mathbb{Z}^+$ too.

Equation (\ref{LkFk}) defines an inner product on the space spanned by the elements of $\mathscr{B}_N^*$, and this inner product is identical to one for the Temperley-Lieb algebras $TL_N(n)$ \cite{tl} studied in \cite{fgg,difranc,fgut,franc}.  
(In particular, see figure \red{39} of \cite{fgut}.)  The Gram matrix $M_N\circ n$ of this inner product, whose $(\varsigma,\vartheta)$th entry is given by (\ref{LkFk}), is called the \emph{meander matrix} \cite{fgg,difranc,fgut,franc}.  In our application, the vectors of $v(\mathcal{B}_N)$ form the columns of $M_N\circ n$.   We conclude from lemma \red{15} of \cite{florkleb} that $\mathcal{B}_N$ is linearly independent if and only if the determinant of $M_N\circ n$ is not zero.

The determinant of this Gram matrix is the \emph{meander determinant}, and P.\ Di Francesco, O.\ Golinelli, and E.\ Guitter compute it in \cite{fgg} (see also \cite{difranc,fgut,franc}), giving the formula
\bea\det(M_N\circ n)(\kappa)&=&\prod_{q=1}^N [(U_q\circ n)(\kappa)]^{a(N,q)}\\
\label{meanderzeros}&=&\prod_{\mathclap{1\leq q''<q\leq N+1}}^N\,\,[n(\kappa)-n_{q,q''}]^{a(N,q-1)},\quad \text{$n_{q,q''}:=-2\cos\left(\frac{\pi q''}{q}\right)$ with $q,q''\in\mathbb{Z}^+$ and $q''<q$,}\eea
where $U_q(n)$ is the $q$th Chebychev polynomial of the second kind \cite{fgg}, and the power $a(N,q)$ is given by 
\be\label{as}a(N,q)=\binom{2N}{N-q}-2\binom{2N}{N-q-1}+\binom{2N}{N-q-2}.\ee
Because $n_{q,q''}$ (\ref{meanderzeros}) only depends on the ratio $q''/q$, we adopt the convention that the pair $\{q,q''\}$ labeling it is coprime.  (Table \ref{thezeros} shows a list of the first few $n_{q,q''}$ and their correspondence to critical lattice models via the O$(n)$ model.)

Each zero $\kappa$ of the meander determinant (\ref{meanderzeros}) satisfies the equation $n(\kappa)=n_{q,q''}$ for some pair $\{q,q''\}$ of coprime positive integers with $q''<q\leq N+1$, and the solutions of this equation are $\kappa_{q,q'}$ with $q'=2mq\pm q''$ for any $m\in\mathbb{Z}$.  All of these solutions are rational, and almost all rational numbers are solutions, with only those of the form $4/r$ for some $r\in\mathbb{Z}\setminus\{0\}$ odd (resp.\ even) excluded by the condition $q''<q$ (resp.\ $q''\neq0$) and zero excluded by the condition $q\neq0$.  Furthermore, we are only interested in the positive solutions, and all of these have $q'=q''$ or $q'=2mq\pm q''$ for some $m\in\mathbb{Z}^+$.  We note that every such positive solution is an exceptional speed (\ref{exceptional}), and every exceptional speed is one such positive solution.  Thus, because the positive solution $\kappa_{q,q'}$ is a zero of the meander determinant (\ref{meanderzeros}) if and only if $q\leq N+1$, the lemma follows.
\end{proof}

\begin{table}[b]\footnotesize\label{MeanderZeros}
\centering
\begin{tabular}{>{\centering}p{1.5cm}|>{\centering}p{1.5cm}>{\centering}p{1.5cm}>{\centering}p{1.5cm}>{\centering}p{1.5cm}>{\centering}p{1.5cm}>{\centering}p{1.5cm}}
\diagbox{$q''$}{$q$} 
& 1 & 2 & 3 & 4 & 5 & 6 
\tabularnewline
\hline
\rule{0pt}{.5cm} 1 & $\times$ & 0 & $-1$ & $-\sqrt{2}$ & $\displaystyle{-\frac{1+\sqrt{5}}{2}}$ & $-\sqrt{3}$
\tabularnewline
\rule{0pt}{.5cm} 2 & $\times$ & $\times$ & $1$ & $\times$ & $\displaystyle{\frac{1-\sqrt{5}}{2}}$ & $\times$ 
\tabularnewline
\rule{0pt}{.5cm} 3 & $\times$ & $\times$ & $\times$ & $\sqrt{2}$ & $\displaystyle{\frac{-1+\sqrt{5}}{2}}$ & $\times$
\tabularnewline
\rule{0pt}{.5cm} 4 & $\times$& $\times$ & $\times$ & $\times$ & $\displaystyle{\frac{1+\sqrt{5}}{2}}$ & $\times$ 
\tabularnewline
\rule{0pt}{.5cm} 5 & $\times$ & $\times$ & $\times$ & $\times$ & $\times$ & $\sqrt{3}$ 
\end{tabular}
\caption{The first zeros $n_{q,q''}$ of $\det M_N\circ n$.  Descending the superdiagonal, we have the dense-phase O$(n)$-model loop fugacity of the uniform spanning tree, percolation, $Q=2$ FK clusters, the tri-critical Ising model, and $Q=3$ FK clusters.}
\label{thezeros}
\end{table}

The proof of lemma \ref{mainlem} establishes this useful corollary.
\begin{cor}\label{meadnertheorem} Suppose that $\kappa\in(0,8)$.  Then $\text{rank}\,\mathcal{B}_N=\text{rank}\,M_N\circ n$.\end{cor}
\noindent
If $n(\kappa)\neq n_{q,q''}$ for coprime integers $1\leq q''<q\leq N+1$, then the nullity of $(M_N\circ n)(\kappa)$ equals zero.  Otherwise, the nullity equals the multiplicity $d_N(q,q'')$ of the zero $n_{q,q''}$ (\ref{meanderzeros}) of the meander determinant \cite{franc}.  Hence,
\be\label{rankMN}\text{rank}\,\mathcal{B}_N(\kappa)=\text{rank}\,(M_N\circ n)(\kappa)=\left\{\begin{array}{ll} C_N-d_N(q,q''), & \text{$\kappa=\kappa_{q,q'}$ (\ref{exceptional}) with $q\leq N+1$} \\ C_N, & \text{otherwise}\end{array}\right.\ee
thanks to the dimension theorem and corollary \ref{meadnertheorem}.  In (\ref{rankMN}), the precise relationship between $q''$ and $q'$ is irrelevant because the multiplicity $d_N(q,q'')$ is independent of $q''$.  Indeed, we see this from the following formula derived in \cite{fgg}:
\bea\label{rankdeg}d_N(q,q'')\,\,\,\,=\,\,\,\,\sum_{p=1}^{\mathclap{\lfloor(N+1)/q\rfloor}}\,\, a(N,pq-1)\,\,\,\,=\,\,\,\,C_N-\frac{1}{2q}\sum_{p=1}^{q-1}\left(2\sin\frac{\pi p}{q}\right)^2\left(2\cos\frac{\pi p}{q}\right)^{2N}.\eea

Now we use lemma \ref{mainlem} to prove most of the following theorem, which is the main result of this article.
\begin{theorem}\label{maintheorem} Suppose that $\kappa\in(0,8)$.  Then the following are true.
\begin{enumerate}
\item\label{firstitem} $\mathcal{B}_N$ is a basis for $\mathcal{S}_N$ if and only if $\kappa$ is not an exceptional speed (\ref{exceptional}) with $q\leq N+1$.
\item\label{seconditem}  $\dim\mathcal{S}_N=C_N,$ with $C_N$ the $N$th Catalan number (\ref{catalan}).
\item\label{thirditem} $\mathcal{S}_N$ has a basis consisting entirely of real-valued Coulomb gas solutions.  (See definition \ref{CGsolnsdef}.)
\item\label{fourthitem} The map $v:\mathcal{S}_N\rightarrow\mathbb{R}^{C_N}$ with $v(F)_\varsigma:=[\mathscr{L}_\varsigma]F$ is a vector-space isomorphism.
\item\label{fifthitem} $\mathscr{B}_N^*:=\{[\mathscr{L}_1],[\mathscr{L}_2],\ldots,[\mathscr{L}_{C_N}]\}$ is a basis for $\mathcal{S}_N^*$.
\end{enumerate}\end{theorem}
\noindent
(According to definition \ref{CGsolnsdef}, we could equivalently state item \ref{thirditem} of this theorem as ``every element of $\mathcal{S}_N$ is a real-valued Coulomb gas solution.")

\begin{proof} Item \ref{firstitem} follows immediately from lemma \ref{mainlem} after we recall that $|\mathcal{B}_N|=C_N$ and $\dim\mathcal{S}_N\leq C_N$ (according to lemma \red{15} in \cite{florkleb}).  If $\kappa$ is not an exceptional speed (\ref{exceptional}) with $q\leq N+1$, then items \ref{seconditem} and \ref{thirditem} immediately follow from item \ref{firstitem}, but if $\kappa$ is such a speed, then a separate proof of items \ref{seconditem} and \ref{thirditem} is needed, which we provide below.  Because $v$ is linear and injective (according to lemma \red{15} in \cite{florkleb}), its rank equals $\dim\mathcal{S}_N$  by the dimension theorem.  Combining this with item \ref{seconditem}, we have $\text{rank}\,v=\dim\mathcal{S}_N=\dim\mathbb{R}^{C_N}$, the dimension of the codomain of $v$.  Hence, $v$ is surjective too, and this proves item \ref{fourthitem}.  Therefore, all that remains is to prove items \ref{seconditem} and \ref{thirditem} for $\kappa$ an exceptional speed (\ref{exceptional}) with $q\leq N+1$ and to prove item \ref{fifthitem} in general.

To prove items \ref{seconditem} and \ref{thirditem} for $\kappa$ an exceptional speed (\ref{exceptional}) with $q\leq N+1$, we construct a new set of $C_N$ real-valued Coulomb gas solutions that is linearly independent for such $\kappa$ and invoke the fact that $\dim\mathcal{S}_N\leq C_N$ (according to lemma \red{15} in \cite{florkleb}).  If $q''<q\leq N+1$ are positive coprime integers such that $n(\kappa)=n_{q,q''}$, then because $\mathcal{B}_N(\kappa)$ has rank $C_N-d_N(q,q'')$ (\ref{rankMN}), where $d_N(q,q'')$ is the multiplicity of the zero $n_{q,q''}$ (\ref{meanderzeros}) of the meander determinant \cite{fgg}, the solutions in $\mathcal{B}_N(\kappa)$ satisfy exactly $d:=d_N(q,q'')$ different linear dependencies.  We write each as
\be\label{lindep}\sum_{\vartheta=1}^{C_N} a_{\varsigma,\vartheta}\mathcal{F}_\vartheta(\kappa)=0,\quad \varsigma\in\{1,2,\ldots,d\},\ee
where the set $\{\boldsymbol{a}_1,\boldsymbol{a}_2,\ldots,\boldsymbol{a}_d\}$ of vectors $\boldsymbol{a}_\varsigma:=(a_{\varsigma,1},a_{\varsigma,2},\ldots,a_{\varsigma,C_N})\in\mathbb{R}^{C_N}$ is a basis for the kernel of $(M_N\circ n)(\kappa)$.

Next, we construct a new set $\mathcal{B}_N^{\scaleobj{0.75}{\bullet}}(\varkappa)$ of cardinality $C_N$.  With $A$ any $C_N\times C_N$ invertible matrix whose first $d$ columns are $\boldsymbol{a}_1$, $\boldsymbol{a}_2,\ldots,\boldsymbol{a}_d$, we consider the set of solutions
\begin{multline}\label{set}\left\{\sideset{}{_\vartheta}\sum a_{1,\vartheta}\mathcal{F}_\vartheta(\varkappa),\quad\sideset{}{_\vartheta}\sum a_{2,\vartheta}\mathcal{F}_\vartheta(\varkappa),\quad\ldots,\right.\\
\left.\sideset{}{_\vartheta}\sum a_{d,\vartheta}\mathcal{F}_\vartheta(\varkappa),\quad\sideset{}{_\vartheta}\sum a_{d+1,\vartheta}\mathcal{F}_\vartheta(\varkappa),\quad\ldots,\quad\sideset{}{_\vartheta}\sum a_{C_N,\vartheta}\mathcal{F}_\vartheta(\varkappa)\right\}.\end{multline}
If $\varkappa\neq\kappa$, then this new set is also linearly independent because $\det A\neq0$.  But if $\varkappa=\kappa$, then the first $d$ entries are zero while the others collectively form a linearly independent set of full rank $C_N -d$.  Because each $\mathcal{F}_\vartheta(\varkappa)$ is analytic at $\varkappa=\kappa$, for each $\varrho\in\{1,2,\ldots,d\}$, the $\varrho$th element of (\ref{set}) is $a_\varrho(\varkappa-\kappa)^{m_\varrho}+O((\varkappa-\kappa)^{m_\varrho+1})$ as $\varkappa\rightarrow\kappa$ for some $m_\varrho\in\mathbb{Z}^+$ and $a_\varrho\neq0$.  Now, we adjust the set (\ref{set}) so all of its elements remain finite and nonzero as $\varkappa\rightarrow\kappa$ thus:
\begin{multline}\label{BNprime}\mathcal{B}_N^{\scaleobj{0.75}{\bullet}}(\varkappa):=\left\{(\varkappa-\kappa)^{-m_1}\sideset{}{_\vartheta}\sum a_{1,\vartheta}\mathcal{F}_\vartheta(\varkappa),\,(\varkappa-\kappa)^{-m_2}\sideset{}{_\vartheta}\sum a_{2,\vartheta}\mathcal{F}_\vartheta(\varkappa),\,\ldots\right.\\
\left.(\varkappa-\kappa)^{-m_d}\sideset{}{_\vartheta}\sum a_{d,\vartheta}\mathcal{F}_\vartheta(\varkappa),\,\sideset{}{_\vartheta}\sum a_{d+1,\vartheta}\mathcal{F}_\vartheta(\varkappa),\,\ldots,\,\sideset{}{_\vartheta}\sum a_{C_N,\vartheta}\mathcal{F}_\vartheta(\varkappa)\right\}.\end{multline}
This new set $\mathcal{B}_N^{\scaleobj{0.75}{\bullet}}(\varkappa)$ comprises $C_N$ Coulomb gas solutions, and we let $\mathcal{F}_\varrho^{\scaleobj{0.75}{\bullet}}(\varkappa)$ be its $\varrho$th element.  According to item \ref{itemdef2} of definition \ref{CGsolnsdef}, $\mathcal{F}_\varrho^{\scaleobj{0.75}{\bullet}}(\kappa)$ is a Coulomb gas solution, and according to theorem \ref{vertexop}, it is an element of $\mathcal{S}_N(\kappa)$.

Now, to finally prove items \ref{seconditem} and \ref{thirditem} for $\kappa$ an exceptional speed (\ref{exceptional}) with $q\leq N+1$, we show that $\mathcal{B}_N^{\scaleobj{0.75}{\bullet}}(\kappa)$ (\ref{BNprime}) is a basis for $\mathcal{S}_N(\kappa)$.  Pursuing the strategy of lemma \ref{mainlem}, we define $v(\mathcal{B}_N^{\scaleobj{0.75}{\bullet}}):=\{v(\mathcal{F}_1^{\scaleobj{0.75}{\bullet}}),v(\mathcal{F}_2^{\scaleobj{0.75}{\bullet}}),\ldots,v(\mathcal{F}_{C_N}^{\scaleobj{0.75}{\bullet}})\}$, and we let $M_N^{\scaleobj{0.75}{\bullet}}\circ n$ be the matrix whose $\varrho$th column is the $\varrho$th element of this set.  Now for $\varkappa\neq\kappa$ and some $C\neq0$, we have
\bea\label{detMN'}\det (M_N^{\scaleobj{0.75}{\bullet}}\circ n)(\varkappa)&=&(\varkappa-\kappa)^{-m_1-m_2-\dotsm-m_d}\det A\det (M_N\circ n)(\varkappa)\nonumber\\
\label{O}&=&(\varkappa-\kappa)^{d-m_1-m_2-\dotsm-m_d}[C+O(\varkappa-\kappa)],\eea
where $d$ is the multiplicity of the zero $n(\kappa)$ of the meander determinant (\ref{meanderzeros}).  Furthermore, in the discussion following (\ref{LkFk}), we note that the limit $\varkappa\rightarrow\kappa$ commutes with all $[\mathscr{L}_\varsigma]\in\mathscr{B}_N^*$.  Therefore, $(M_N^{\scaleobj{0.75}{\bullet}}\circ n)(\kappa)=\lim_{\varkappa\rightarrow\kappa}(M_N^{\scaleobj{0.75}{\bullet}}\circ n)(\varkappa)$, so we find the determinant of $(M_N^{\scaleobj{0.75}{\bullet}}\circ n)(\kappa)$ by sending $\varkappa\rightarrow\kappa$ in (\ref{detMN'}).  Now, because this determinant is necessarily finite and each power in (\ref{detMN'}) is a positive integer, we must have $m_\varrho=1$ for all $\varrho\in\{1,2,\ldots,d\}$.  Then it follows that $\det(M_N^{\scaleobj{0.75}{\bullet}}\circ n)(\kappa)\neq0$, and $v(\mathcal{B}_N^{\scaleobj{0.75}{\bullet}}(\kappa))$ is thus linearly independent.  Finally, because $v$ is injective and $|\mathcal{B}_N^{\scaleobj{0.75}{\bullet}}|=C_N$, we conclude that $\mathcal{B}_N^{\scaleobj{0.75}{\bullet}}(\kappa)$ is a basis for $\mathcal{S}_N(\kappa)$.  This proves items \ref{seconditem} and \ref{thirditem} for $\kappa$ an exceptional speed (\ref{exceptional}) with $q\leq N+1$.

Item \ref{seconditem} implies that $\dim\mathcal{S}_N^*=C_N$.  To prove item \ref{fifthitem}, we let $\mathcal{M}:=\{[\mathscr{L}_1],[\mathscr{L}_2],\ldots,[\mathscr{L}_M]\}$ be a maximal linearly independent subset of $\mathscr{B}_N^*$, and we prove that $M:=|\mathcal{M}|=C_N$.  To prove that $\mathcal{M}$ is nonempty first, we show that no element of $\mathscr{B}_N^*$ is the zero-functional.  If one element $[\mathscr{L}_\varsigma]\in\mathscr{B}_N^*$ were the zero functional and $\kappa$ is not (resp.\ is) an exceptional speed (\ref{exceptional}) with $q\leq N+1$, then the $\varsigma$th row of the matrix $M_N\circ n$ (resp.\ $M_N^{\scaleobj{0.75}{\bullet}}\circ n$) would have each entry zero.  But then the determinant of this matrix would vanish, a contradiction.  Thus, $\mathcal{M}$ is nonempty.

Now we suppose that $M<C_N$.  Then by item \ref{seconditem}, $\dim\mathcal{S}_N^*=C_N$, and $\mathcal{S}_N^*$ has a finite basis for which $\mathcal{M}$ may serve as a proper subset.  We let
\be\begin{gathered} B_N^*=\{[\mathscr{L}_1],[\mathscr{L}_2],\ldots,[\mathscr{L}_M],f_{M+1},f_{M+2},\ldots,f_{C_N}\},\\
B_N=\{\Pi_1,\Pi_2,\ldots, \Pi_M,\Pi_{M+1},\Pi_{M+2},\ldots, \Pi_{C_N}\}\end{gathered}\ee
be dual bases for $\mathcal{S}_N^*$ and $\mathcal{S}_N$ respectively, so $[\mathscr{L}_\varsigma]\Pi_\vartheta=0$ for all $\vartheta>M$ because $\varsigma\leq M$.  Moreover, the elements $[\mathscr{L}_{M+1}]$, $[\mathscr{L}_{M+2}],\ldots,[\mathscr{L}_{C_N}]$ of $\mathscr{B}_N^*$ that are not in $\mathcal{M}$ must be linear combinations of those in $\mathcal{M}$ because $\mathcal{M}$ is maximal, so they annihilate $\Pi_\vartheta$ for all $\vartheta>M$ too.   Then $v(\Pi_\vartheta)=0$ for all $\vartheta>M$, and because $v$ is injective, $\Pi_\vartheta$ is therefore zero for all $\vartheta>M$.  But this contradicts the fact that each $\Pi_\vartheta$ is an element of a basis.  We therefore conclude that $M=C_N$, proving item \ref{fifthitem}.
\end{proof}

If $N=2$, then we may write $\mathcal{F}_{c,\vartheta}$ in terms of a hypergeometric function by using the contour integral definition of the latter.  (Equations (\red{17}--\red{19}) of \cite{florkleb} indicate that this is always possible.)  After doing so, appropriate hypergeometric identities show that $\mathcal{F}_{c,1}$ is the same function for all $c\in\{1,2,3,4\}$, and similarly for $\mathcal{F}_{c,2}$.  Extending this observation, we might expect that for any $N\geq2$, $\mathcal{F}_{c,\vartheta}$ is independent of our choice of $c\in\{1,2,\ldots,2N\}$ for the following reason.  CFT seems to suggest that a correlation function of $2N$ one-leg boundary operators with only, for example, the identity fusion channel propagating between each pair of operators whose points are joined by an arc in the $\vartheta$th connectivity is unique.  (This $2N$-point function is sometimes called a ``conformal block.")  But the Coulomb gas formalism gives $2N$ ostensibly different formulas $\mathcal{F}_{1,\vartheta},$ $\mathcal{F}_{2,\vartheta},\ldots,\mathcal{F}_{2N,\vartheta}$ that vary by the location of the point $x_c$ bearing the conjugate charge.  To prove that these formulas indeed give the same function directly from the formulas themselves seems to be too difficult.  So instead, we use the methods of this article to prove this fact indirectly.

\begin{cor}\label{moveconjchargecor}Suppose that $\kappa\in(0,8)$.  Then $\mathcal{F}_{c,\vartheta}=\mathcal{F}_{c',\vartheta}$ for all $c,c'\in\{1,2,\ldots,2N\}$ and all $\vartheta\in\{1,2,\ldots,C_N\}$.
\end{cor}

\begin{proof} 
We prove the corollary by showing that $v(\mathcal{F}_{c,\vartheta})=v(\mathcal{F}_{\vartheta})$ for all $c\in\{1,2,\ldots,2N-1\}$ and $\vartheta\in\{1,2,\ldots,C_N\}$.  (We recall from item \ref{item2a} of definition \ref{Fkdefn} that $\mathcal{F}_\vartheta:=\mathcal{F}_{2N,\vartheta}$.)  Then the corollary follows from item \ref{fourthitem} of theorem \ref{maintheorem}.

To prove that $v(\mathcal{F}_{c,\vartheta})=v(\mathcal{F}_{\vartheta})$ for all $c\in\{1,2,\ldots,2N-1\}$ and $\vartheta\in\{1,2,\ldots,C_N\}$, we compute $[\mathscr{L}_\varsigma]\mathcal{F}_{c,\vartheta}$ for each $\varsigma\in\{1,2,\ldots,C_N\}$ and verify that it equals $[\mathscr{L}_\varsigma]\mathcal{F}_\vartheta$, whose value is given by (\ref{LkFk}).  If we imitate the method and reasoning in the proof of lemma \ref{mainlem}, then this computation uses an element of $[\mathscr{L}_\varsigma]$ whose last limit involves the point $x_c$ bearing the conjugate charge.  Now, the many limits $\ell_1$, $\ell_2,\ldots,\ell_N$ that compose this element act on functions in $\mathcal{S}_N$ in one of two ways.  The first way that the first limit $\ell_1$ may act on $\mathcal{F}_{c,\vartheta}$ is via (\ref{lim}, \ref{thelim})
\begin{multline}\label{nextthelim}\bar{\ell}_1\mathcal{F}_{c,\vartheta}(\kappa\,|\,x_1,x_2,\ldots,x_{i-1},x_{i+2},\ldots,x_{2N})\\=\lim_{x_{i+1}\rightarrow x_i}(x_{i+1}-x_i)^{6/\kappa-1}\mathcal{F}_{c,\vartheta}(\kappa\,|\, \boldsymbol{x}),\quad i\in\{1,2,\ldots,2N-1\}\setminus\{c-1,c\},\quad\kappa\in(0,8),\end{multline}
and in this case, we write $\ell_1=\bar{\ell}_1$.  By following the proof of lemma \ref{mainlem} but instead using the first main result of sections \ref{s50}--\ref{s53} in appendix \ref{asymp}, we find that $\bar{\ell}_1\mathcal{F}_{c,\vartheta}$ equals an element of $\mathcal{S}_{N-1}$ of the kind falling under definition \ref{Fkdefn}.  In fact, we generate the half-plane diagram for this element by joining together the two arcs in the half-plane diagram for $\mathcal{F}_{c,\vartheta}$ whose endpoints are brought together by the limit $\bar{\ell}_1$.  Furthermore, this new diagram matches that of $\bar{\ell}_1\mathcal{F}_\vartheta$, a fact useful for this proof.  The second way that the first limit $\ell_1$ may act on $\mathcal{F}_{c,\vartheta}$ is via
\be\label{inflim}\underline{\ell}_1\mathcal{F}_{c,\vartheta}(\kappa\,|\,x_2,x_3,\ldots,x_{2N-1})=\lim_{R\rightarrow\infty}(2R)^{6/\kappa-1}\mathcal{F}_{c,\vartheta}(\kappa\,|\,x_1=-R,x_2,\ldots,x_{2N}=R),\quad c\not\in\{1,2N\},\quad\kappa\in(0,8),\ee
and in this case, we write $\ell_1=\underline{\ell}_1$.  (See lemma \red{5} and definition \red{9} in \cite{florkleb}.)  Little of our analysis in \cite{florkleb} and none in this article so far have involved this second type of limit.

As we observe in the proof of lemma \ref{mainlem}, if $c=2N$ (and also by symmetry, if $c=1$), then there is always an element of $[\mathscr{L}_\varsigma]$ whose last limit involves $x_c$ and with none of its limits $\ell_1$, $\ell_2,\ldots,\ell_N$ of the second type (\ref{inflim}).  Thus, we may complete the proof of lemma \ref{mainlem} without computing limits of this latter kind.  However, if $c\not\in\{1,2N\}$, then there are $\vartheta\in\{1,2,\ldots,C_N\}$ where this is not the case, so to compute $[\mathscr{L}_\varsigma]\mathcal{F}_{c,\vartheta}$ in this situation, we must calculate this second type of limit $\underline{\ell}_1\mathcal{F}_{c,\vartheta}$ (\ref{inflim}).  As previously, there are four cases to consider, and each mirrors a case from the proof of lemma \ref{mainlem}.  In the following, we recall that, according to item \ref{2cit} in definition \ref{Fkdefn}, we orient all integration contours of $\mathcal{F}_{c,\vartheta}$ that pass (resp.\ do not pass) over the point $x_c$ bearing the conjugate charge leftward (resp.\ rightward).
\begin{enumerate}
\item\label{firstcase2} 
\textbf{Configuration:} In case \ref{firstcase2}, neither $x_1$ nor $x_{2N}$ is an endpoint of an integration contour of $\mathcal{F}_{c,\vartheta}$.

\noindent
\textbf{Calculation:}  With $x_{2N}=-x_1=R$, we estimate $\mathcal{F}_{c,\vartheta}(\kappa\,|\,\boldsymbol{x})$ for $R$ very large and then send $R\rightarrow\infty$.  (See the second main result in section \ref{s50} of appendix \ref{asymp}.)

\noindent
\textbf{Main result:} In case \ref{firstcase2}, the limit (\ref{inflim}) is zero.

\item\label{secondcase2} 
\textbf{Configuration:} In case \ref{secondcase2}, both $x_1$ and $x_{2N}$ are endpoints of a single, common integration contour of $\mathcal{F}_{c,\vartheta}$, $\Gamma_1=\mathscr{P}(x_{2N},x_1)$ or $[x_{2N},x_1]^+$, following a semicircular arc in the upper half-plane that passes over the points $x_2,$ $x_3,\ldots,x_{2N-1}$ and all other integration contours of $\mathcal{F}_{c,\vartheta}$.

\noindent
\textbf{Calculation:} The last paragraphs of section \ref{s51} in appendix \ref{asymp} present the calculation of the limit $\underline{\ell}_1\mathcal{F}_{c,\vartheta}$ (\ref{inflim}) in full detail.

\noindent
\textbf{Main result:} In case \ref{secondcase2}, the limit $\underline{\ell}_1\mathcal{F}_{c,\vartheta}$ (\ref{inflim}) equals $n$ (\ref{fugacity}) times an element of $\mathcal{S}_{N-1}$ of the form in definition \ref{Fkdefn}.  We generate it from the formula (\ref{firstFexplicit1}) for $\mathcal{F}_{c,\vartheta}$ by dropping all factors involving $x_1$, $x_{2N}$, and $u_1$, dropping the integration along $\Gamma_1$, and reducing the prefactor power in (\ref{firstprefactor}) or (\ref{secondprefactor}) by one.

\item\label{thirdcase2} 
\textbf{Configuration:} In case \ref{thirdcase2}, either $x_1$ or $x_{2N}$ is an endpoint of a single integration contour $\Gamma_1$ of $\mathcal{F}_{c,\vartheta}$, but the other is not an endpoint of any contour.  Thus, the latter point is an endpoint of the arc terminating at $x_c$ in the half-plane diagram for $\mathcal{F}_{c,\vartheta}$.

\noindent
\textbf{Initial assumptions:} We assume that $\kappa>4$, $x_1$ is the endpoint of $\Gamma_1$, and $\Gamma_1$ has rightward orientation.

\noindent
\textbf{Calculation:} We repeat the calculation for case \ref{thirdcase} in the proof of lemma \ref{mainlem} with $i=2N$ and $x_{i+1}$ identified with $x_1$.  While that previous calculation uses the main result of section \ref{s3} in appendix \ref{asymp}, the present calculation uses the second main result of section \ref{s52} in that appendix.

\noindent
\textbf{Main result:} In case \ref{thirdcase2}, the limit $\underline{\ell}_1\mathcal{F}_{c,\vartheta}$ (\ref{inflim}) equals an element of $\mathcal{S}_{N-1}$ of the form in definition \ref{Fkdefn}.  We generate it from the formula (\ref{firstFexplicit1}) for $\mathcal{F}_{c,\vartheta}$ by dropping all factors involving $x_1$, $x_{2N}$, and $u_1$, dropping the integration along $\Gamma_1$, and reducing the prefactor power in (\ref{firstprefactor}) or (\ref{secondprefactor}) by one.

\noindent
\textbf{Further details:} Above, we assume that $\kappa>4$, $x_1$ is the endpoint of $\Gamma_1$, and $\Gamma_1$ has rightward orientation.  Here, we justify these assumptions, or we extend our proof to situations in which they do not hold.
\begin{enumerate}
\item\label{rightitem} $\Gamma_1$ has rightward orientation.  Indeed, $\Gamma_1$ may not pass over $x_c$, or else its arc crosses the arc with endpoints at $x_c$ and $x_{2N}$ in the half-plane diagram for $\mathcal{F}_{c,\vartheta}$.  The claim then follows from item \ref{2cit} in definition \ref{Fkdefn}.
\item If $x_{2N}$ and not $x_1$ is an endpoint of $\Gamma_1$, then repeating the above analysis yields the same result.  Again, $\Gamma_1$ has rightward orientation for a reason almost identical to that presented in item \ref{rightitem} above.
\item The extension of this result to $\kappa\in(0,4]$ is identical to that presented in item \ref{kappaextend} in the proof of lemma \ref{mainlem}, appropriately adapted to the present situation with $c\not\in\{1,2N\}$.
\end{enumerate}

\item\label{fourthcase2} \textbf{Configuration}: In case \ref{fourthcase2}, the most complicated case, $x_{2N}$ is an endpoint of one contour $\Gamma_1$ of $\mathcal{F}_{c,\vartheta}$, and $x_1$ is an endpoint of a different contour $\Gamma_2$.

\noindent
\textbf{Initial assumptions:} We assume that $\kappa>4$ and both $\Gamma_1$ and $\Gamma_2$ have rightward orientation.

\noindent
\textbf{Calculation:} We repeat the calculation of case \ref{fourthcase} in the proof of lemma \ref{mainlem} with $i=2N$ and $x_{i+1}$ and $x_{i+2}$ identified with $x_1$ and $x_2$ respectively.  While that previous calculation uses the main result of section \ref{s4} in appendix \ref{asymp}, the present calculation uses the second main result of section \ref{s53} in that appendix.

\noindent
\textbf{Main result:} In case \ref{fourthcase2}, the limit $\underline{\ell}_1\mathcal{F}_{c,\vartheta}$ (\ref{inflim}) equals an element of $\mathcal{S}_{N-1}$ of the form in definition \ref{Fkdefn}.  We generate it from the formula (\ref{firstFexplicit1}) for $\mathcal{F}_{c,\vartheta}$ by dropping all factors involving $x_{2N}$, $x_1$, and $u_1$, dropping the integration along $\Gamma_1$, integrating $u_2$ along the contour $\Gamma$ generated by disconnecting $\Gamma_1$ and $\Gamma_2$ from their respective endpoints at $x_{2N}$ and $x_1$ and joining their dangling ends together in the upper half-plane, and reducing the prefactor power in (\ref{firstprefactor}) or (\ref{secondprefactor}) by one.  The new contour $\Gamma$ has leftward orientation (as it should because it passes over $x_c$).

\noindent
\textbf{Further details:} Above, we assume that $\kappa>4$ and both $\Gamma_1$ and $\Gamma_2$ have rightward orientation.  Here, we extend our proof to situations in which this is not true.
\begin{enumerate}
\item According to item \ref{2cit} in definition \ref{Fkdefn}, whether or not an integration contour passes over $x_c$ determines that contour's orientation.  Because $\Gamma_1$ and $\Gamma_2$ may not simultaneously pass over $x_c$ (or else they would intersect), there are two possibilities.  First, neither passes over $x_c$, so both have rightward orientation.  We already considered this case above.  Second, one of them passes over $x_c$, but the other does not.  Thus, the former (resp.\ latter) has leftward (resp.\ rightward) orientation.  If we reverse the former's orientation and compute the limit $\underline{\ell}_1\mathcal{F}_{c,\vartheta}$, then we find the opposite of what is described under ``main result" above.  However, the new contour $\Gamma$ does not pass over $x_c$ yet has leftward orientation.  By reversing the orientation of $\Gamma$ to conform with what item \ref{2cit} in definition \ref{Fkdefn} specifies, we absorb the extra minus sign into the limit.
\item The extension of this result to $\kappa\in(0,4]$ is identical to that presented in item \ref{kappaextend} in the proof of lemma \ref{mainlem}.
\end{enumerate}
\end{enumerate}
In cases \ref{thirdcase2} and \ref{fourthcase2}, we find that the limit $\underline{\ell}_1\mathcal{F}_{c,\vartheta}$ (\ref{inflim}) equals a function in $\mathcal{S}_{N-1}$ that falls under definition \ref{Fkdefn}.  Furthermore, we create the diagram for this function by disconnecting the leftmost and rightmost arcs in the half-plane diagram of $\mathcal{F}_{c,\vartheta}$ from their respective endpoints at $x_1$ and $x_{2N}$ and joining their dangling ends together to form one arc.  In case \ref{secondcase}, we find the same limit, but multiplied by an extra factor of $n$ (\ref{fugacity}).  These findings for the limit $\underline{\ell}_1\mathcal{F}_{c,\vartheta}$ (\ref{inflim}) are identical to those of the similar limit $\bar{\ell}_1\mathcal{F}_{c,\vartheta}$ (\ref{nextthelim}).  Thus, for any $c\in\{1,2,\ldots,2N-1\}$, the part of the proof of lemma \ref{mainlem} that follows item \ref{lastb} also justifies (\ref{LkFk}) with $\mathcal{F}_\vartheta$ replaced by $\mathcal{F}_{c,\vartheta}$, so $v(\mathcal{F}_{c,\vartheta})=v(\mathcal{F}_\vartheta)$.
\end{proof}

\section{Summary}

The purpose of this article and its predecessors \cite{florkleb, florkleb2} is to completely and rigorously determine a solution space $\mathcal{S}_N$ for the system of ``null-state" PDEs (\ref{nullstate}) with the ``conformal Ward identities" (\ref{wardid}) that govern multiple-SLE$_{\kappa}$ partition functions and CFT $2N$-point correlation functions of one-leg boundary operators (\ref{2Npoint}).  Our main result, theorem \ref{maintheorem}, states that the vector space $\mathcal{S}_N$ of all classical solutions for this system that satisfy the growth bound (\ref{powerlaw}) has dimension $C_N$ and is spanned by real-valued Coulomb gas solutions (definition \ref{CGsolnsdef} and (\ref{CGsolns}, \ref{eulerintegralch2})).  Here, $C_N$ is the $N$th Catalan number (\ref{catalan}).  To obtain this result, we first construct a linear injective mapping $v:\mathcal{S}_N\rightarrow\mathbb{R}^{C_N}$ and use it to prove that $\dim\mathcal{S}_N\leq C_N$ in \cite{florkleb, florkleb2}.  Then in sections \ref{CGsolnsSect} and \ref{MeanderMatrix} of this article, we construct a set $\mathcal{B}_N:=\{\mathcal{F}_1,\mathcal{F}_2,\ldots,\mathcal{F}_{C_N}\}\subset\mathcal{S}_N$ of $C_N$ explicit solutions using the Coulomb gas (contour integral) formalism of conformal field theory (CFT).  Next, we prove lemma \ref{mainlem}, which states that $v(\mathcal{B}_N):=\{v(\mathcal{F}_1),v(\mathcal{F}_2),\ldots,v(\mathcal{F}_{C_N})\}$, and therefore $\mathcal{B}_N$, is linearly independent if and only if the Schramm-L\"{o}wner evolution (SLE$_\kappa$) speed $\kappa \in(0,8)$ is not an exceptional speed (\ref{exceptional}) with $q\leq N+1$.  To reach this conclusion, we  identify the matrix formed by the columns of $v(\mathcal{B}_N)$ with the Gram matrix of an inner product on the Temperley-Lieb algebra $TL_N(n)$, called the ``meander matrix," and we use results of \cite{fgg,difranc,fgut,franc} to show that the zeros of its determinant correspond with these exceptional speeds.  If $\kappa$ is one of these exceptional speeds, then we use $\mathcal{B}_N$ to construct an alternative set $\mathcal{B}_N^{\scaleobj{0.75}{\bullet}}\subset\mathcal{S}_N$ of $C_N$ real Coulomb gas solutions that is linearly independent, completing the proof of theorem \ref{maintheorem}.  The calculations that support this proof are complicated, and appendix \ref{asymp} presents all of their details.  In particular, lemma \ref{reallem} in appendix \ref{reallemproof} proves that the elements of $\mathcal{B}_N$ (and therefore of $\mathcal{B}_N^{\scaleobj{0.75}{\bullet}}$, see the proof of theorem \ref{maintheorem}) are real-valued.  Finally, we remind the reader that, although the system (\ref{nullstate}, \ref{wardid}) arises in CFT in a way that is typically non-rigorous, our treatment of this system here and in \cite{florkleb,florkleb2,florkleb4} is completely rigorous.

\section{Acknowledgements}

We thank J.\ J.\ H.\ Simmons for insightful conversations and for sharing some of his unpublished results, in particular, his calculation of $v(\mathcal{B}_N)$ for $N\in\{1,2,3\}$ \cite{js}.  We also thank K.\ Kyt\"ol\"a for helpful conversations concerning, among other things, the proof in appendix \ref{reallemproof}, and  S.\ Fomin for showing us the connection between the arc connectivity diagrams and the Temperley-Lieb algebra, an observation that led us to the calculation of the meander determinant by P.\ Di Francesco, O.\ Golinelli, and E.\ Guitter in \cite{fgg}.  In addition, we are grateful to C.\ Townley Flores for carefully proofreading the manuscript.

During the writing of this article, we learned that K.\ Kyt\"ol\"a and E.\ Peltola recently obtained results very similar to ours by using a completely different approach based on quantum group methods \cite{kype,kype2}.

This work was supported by National Science Foundation Grants Nos.\ PHY-0855335 (SMF) and DMR-0536927 (PK and SMF).

\appendix{}

\section{Asymptotic behavior of Coulomb gas integrals under interval collapse}\label{asymp}

In this appendix, we justify the claims of items \ref{firstcase}--\ref{fourthcase} in the proof of lemma \ref{mainlem} (and later, of corollary \ref{moveconjchargecor}).  The main thrust of the proof is to calculate the limit
(\ref{thelim}) for all $\mathcal{F}_\vartheta\in\mathcal{B}_N$ and $\kappa\in(0,8)$ with $8/\kappa\not\in\mathbb{Z}^+$.  The explicit formula for $\mathcal{F}_\vartheta(\kappa\,|\,\boldsymbol{x})$ is (\ref{firstFexplicit1}) with $c=2N$, or
\begin{multline}\label{Fexplicit1}\mathcal{F}_\vartheta(\kappa\,|\,\boldsymbol{x})=n(\kappa)\left[\frac{n(\kappa)\Gamma(2-8/\kappa)}{4\sin^2(4\pi/\kappa)\Gamma(1-4/\kappa)^2}\right]^{N-1}\Bigg(\prod_{j<k}^{2N-1}(x_k-x_j)^{2/\kappa}\Bigg)\Bigg(\prod_{k=1}^{2N-1}(x_{2N}-x_k)^{1-6/\kappa}\Bigg)\\ 
\underbrace{\oint_{\Gamma_{N-1}}{\rm d}u_{N-1}\dotsm\oint_{\Gamma_2}{\rm d}u_{2}\,\,\oint_{\Gamma_1}{\rm d}u_1\,\,\mathcal{N}\Bigg[\Bigg(\prod_{l=1}^{2N-1}\prod_{m=1}^{N-1}(x_l-u_m)^{-4/\kappa}\Bigg)\Bigg(\prod_{m=1}^{N-1}(x_{2N}-u_m)^{12/\kappa-2}\Bigg)\Bigg(\prod_{p<q}^{N-1}(u_p-u_q)^{8/\kappa}\Bigg)\Bigg],}_{\mathcal{J}^{(N-1,2N)}(\boldsymbol{x})}\end{multline}
where we have indicated the Coulomb gas integral $\mathcal{J}^{(N-1,2N)}(\boldsymbol{x})$ (\ref{eulerintegralch2}) with braces, and where we have the following.
\begin{enumerate}[label=(\alph*),ref=\alph*]
\item Each $\Gamma_m$ is a Pochhammer contour entwining two ``endpoints" among $x_1$, $x_2,\ldots,x_{2N}$.  (See (\ref{PochDecomp}) and figure \ref{BreakDown}.)
\item Item \ref{4thitem} of definition \ref{Fkdefn} explains the symbol $\mathcal{N}[\,\,\ldots\,\,]$.  (If no contour of $\mathcal{F}_\vartheta$ passes over another, then $\mathcal{N}[\,\,\ldots\,\,]$ orders the terms of the differences of the integrand so the enclosed power functions are real-valued.)
\item\label{itemc} If $\text{Re}\,\kappa>4$, then we use identity (\ref{Pochtostraight}) to simplify (\ref{Fexplicit1}).  We replace each Pochhammer contour by a simple contour with the same endpoints and orientation and drop each factor of $4\sin^2(4\pi/\kappa)$ in the denominator of the prefactor.
\end{enumerate}
According to item \ref{itemc}, if the integration contours of $\mathcal{F}_\vartheta$ are Pochhammer (resp.\ simple), then the prefactor for $\mathcal{F}_\vartheta$ is (\ref{firstprefactor}) (resp.\ (\ref{secondprefactor})).  Formula (\ref{Fexplicit1}) only shows the Pochhammer contour version.

To compute the limit (\ref{thelim}), we must determine the asymptotic behavior of the definite integral $\mathcal{J}^{(N-1,2N)}(\boldsymbol{x})$ in (\ref{Fexplicit1}) as $x_{i+1}\rightarrow x_i$.  As we noted in the proof of lemma \ref{mainlem}, there are four cases to consider: 
\begin{enumerate}
\item\label{sc1} Neither $x_i$ nor $x_{i+1}$ are endpoints of an integration contour.
\item\label{sc2} Both $x_i$ and $x_{i+1}$ are endpoints of one common contour, say $\Gamma_1$.
\item\label{sc3} $x_i$ (resp.\ $x_{i+1}$) is an endpoint of one contour, say $\Gamma_1$, and $x_{i+1}$ (resp.\ $x_i$) is not an endpoint of any contour.
\item\label{sc4} $x_i$ is an endpoint of one contour, say $\Gamma_1$, and $x_{i+1}$ is an endpoint of a different contour, say $\Gamma_2$.
\end{enumerate}
We study these four cases in items \ref{firstcase}--\ref{fourthcase} of the proof of lemma \ref{mainlem} respectively.  By deforming the integration contours in that proof, we find that we may replace items \ref{sc3} and \ref{sc4} directly above with these respective scenarios.
\begin{enumerate}[label=\arabic*$''$., ref=\arabic*$''$]
\setcounter{enumi}{2}
\item\label{sc3new} $x_i$ and $x_{i-1}$ (resp.\ $x_{i+1}$ and $x_{i+2}$) are endpoints of one common contour, say $\Gamma_1''$, and $x_{i+1}$ (resp.\ $x_i$) is not an endpoint of any contour.  (According to item \ref{thirdcase} in the proof of lemma \ref{mainlem}, we may assume $i>1$.)
\item\label{sc4new} $x_i$ and $x_{i-1}$ are endpoints of one common contour, say $\Gamma_1''$, and $x_{i+1}$ and $x_{i+2}$ are endpoints of one other common contour, say $\Gamma_2''$.  (If $i=1$, then we identify $i-1$ with $2N$.  
See section \ref{s4} for further details.)
\end{enumerate}
Figure \ref{asympcases} illustrates cases \ref{sc1} and \ref{sc2} above, and the simpler cases \ref{sc3new} and \ref{sc4new} directly above.  We determine the asymptotic behavior of $\mathcal{J}^{(N-1,2N)}(\boldsymbol{x})$ as $x_{i+1}\rightarrow x_i$ for these four cases in sections \ref{s1}--\ref{s4} respectively.  Afterwards in section \ref{s5}, we determine the asymptotic behavior of the similar Coulomb gas integral in (\ref{firstFexplicit1}) for any $c\in\{1,2,\ldots,2N\}$ as $x_{i+1}\rightarrow x_i$ and, in particular, as $-x_1=x_{2N}=R\rightarrow\infty$.  The latter behavior is useful for finding the limit (\ref{inflim}).

\begin{figure}[t]
\centering
\includegraphics[scale=0.27]{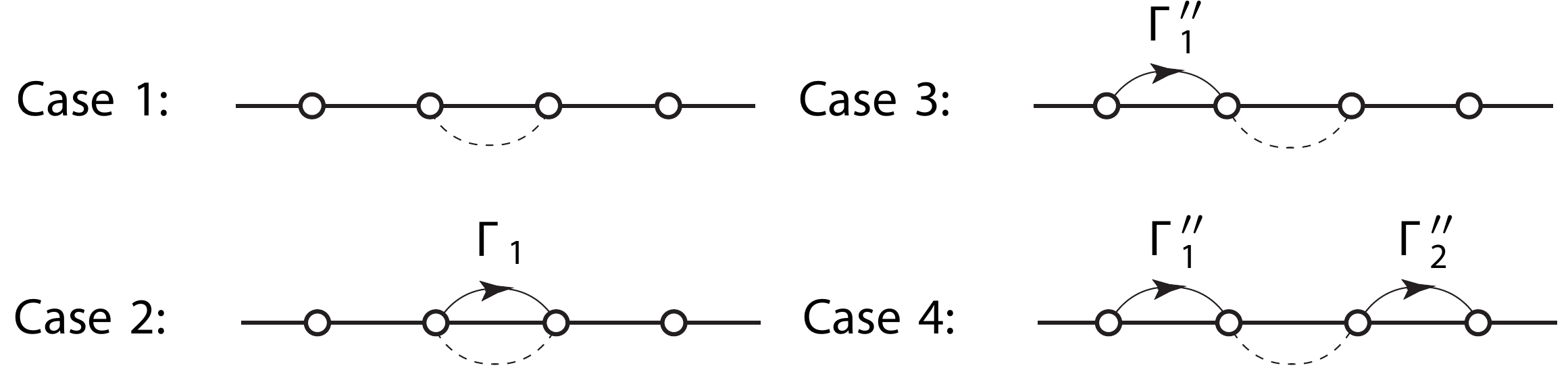}
\caption{The four cases of interval collapse.  The dashed curves connect the endpoints of the interval to be collapsed, and the solid curves are the integration contours.}
\label{asympcases}
\end{figure}

Determining the asymptotic behavior of $\mathcal{J}^{(N-1,2N)}(\boldsymbol{x})$ as $x_{i+1}\rightarrow x_i$ in cases \ref{sc1} and \ref{sc2} is straightforward.  In the former, we simply set $x_{i+1}=x_i$, and in the latter, we use a beta-function identity.  However, determining this behavior in cases \ref{sc3new} and \ref{sc4new} is more involved.  Indeed, if $x_i$ is an endpoint of $\Gamma_1''$, then the integration variable $u_1$ is not bounded away from $x_i$.  Because the difference $x_i-u_1$ is always much smaller than $x_{i+1}-u_1$ for some $u_1\in\Gamma_1''$ even as we send $x_{i+1}\rightarrow x_i$, we may not simply set $x_{i+1}=x_i$ in the integration with respect to $u_1$ to determine the asymptotic behavior.  In order to handle these cases, we cast them into cases \ref{sc1} and \ref{sc2} by deforming any integration contour with an endpoint at $x_i$ or $x_{i+1}$ into one that has neither or both of its endpoints at these locations.

In order to track the phase factors arising from these deformations, we adopt the convention that $-\pi<\arg z\leq\pi$ for all complex $z$, which implies the following identity.  If $x,u,\epsilon,\beta\in\mathbb{R}$ with $x<u$ and $\epsilon>0$, then for $\epsilon\ll |u-x|$,
\be\label{factorout}(x-(u\pm i\epsilon))^\beta=e^{\mp\pi i\beta}(u\pm i\epsilon-x)^\beta.\ee
Figure \ref{phases} shows the relative phase factors accrued by the left side of (\ref{factorout}) as the complex variable $u\pm i\epsilon$ passes over or under or circles around the real branch point $x$, as a consequence of (\ref{factorout}).

\begin{figure}[t]
\centering
\includegraphics[scale=0.27]{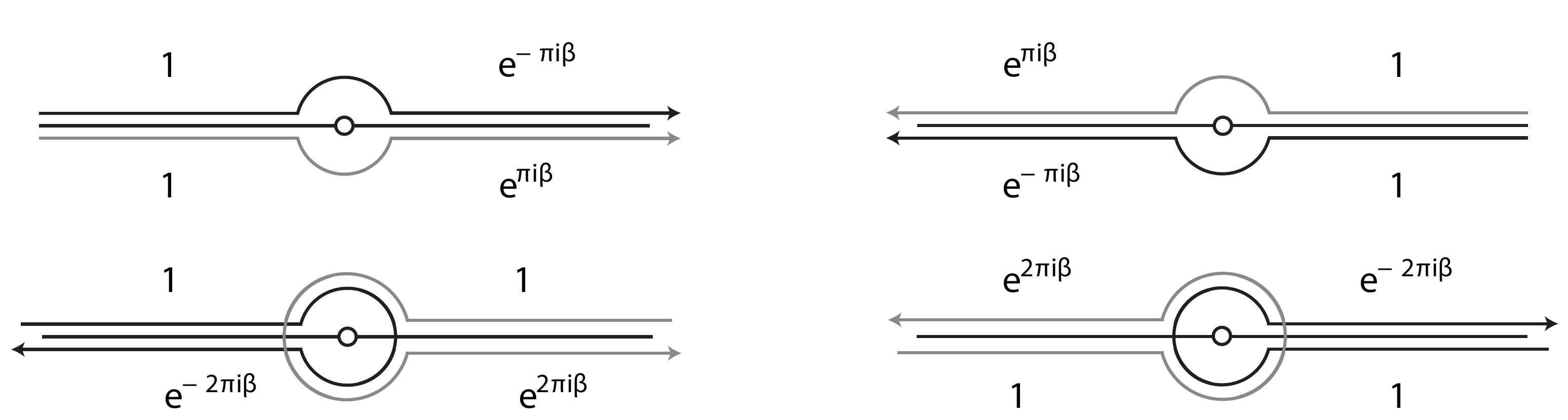}
\caption{Monodromy factors accrued by $(x-u)^\beta$ as $u$ passes over or under (upper illustration) or circles around (lower illustration) the branch point $x$ (open circle on the real axis).}
\label{phases}
\end{figure}

As in the proof of lemma \ref{mainlem}, we let $[x_i,x_{i+1}]^+$ be the simple contour formed by bending $[x_i,x_{i+1}]$ into the upper half-plane while keeping its endpoints fixed.

\subsection{Proof of lemma \ref{mainlem}, item \ref{firstcase}}\label{s1}

The purpose of this section is to complete the argument for item \ref{firstcase} in the proof of lemma \ref{mainlem}.  We do this in two steps.  First, we determine the asymptotic behavior as $x_{i+1}\rightarrow x_i$ of the Coulomb gas integral $\mathcal{J}^{(N-1,2N)}(\boldsymbol{x})$ in (\ref{Fexplicit1}) for all $\kappa\in(0,8)$.  As stated above in case \ref{sc1}, neither $x_i$ nor $x_{i+1}$ is an endpoint of any integration contour.  (Actually, items \ref{third} and \ref{item2a} of definition \ref{Fkdefn} show that only one point among $x_1$, $x_2,\ldots,x_{2N-1}$ is not an endpoint of any contour of $\mathcal{F}_\vartheta$.  Therefore, this case, strictly speaking, does not occur.  However, it does arise from deforming integration contours in cases \ref{sc3new} and \ref{sc4new}, and for that reason, we study it here.)   After we find the asymptotic behavior of $\mathcal{J}^{(N-1,2N)}(\boldsymbol{x})$, we use it to compute the limit $\bar{\ell}_1\mathcal{F}_\vartheta$ (\ref{thelim}), our second step.

To generate this case, we replace all contours, call them $\Gamma_1$ and/or $\Gamma_2$ (there are at most two), with endpoints at $x_i$ or $x_{i+1}$ by new contours $\Gamma_1'$ and/or $\Gamma_2'$ that do not surround or touch these points.  Then the limit as $x_{i+1}\rightarrow x_i$ of each factor of $(x_{i+1}-u_m)^{-4/\kappa}$ with $m\in\{1,2,\ldots,N-1\}$ in the integrand of $\mathcal{J}^{(N-1,2N)}(\boldsymbol{x})$ is uniform over $u_m\in\Gamma_m$.  So by setting $x_{i+1}=x_i$, we find the limit of $\mathcal{J}^{(N-1,2N)}(\boldsymbol{x})$ as $x_{i+1}\rightarrow x_i$.

To finish, we justify item \ref{firstcase} in the proof of lemma \ref{mainlem}.  In that part of the proof, we require the limit $\bar{\ell}_1\mathcal{F}_\vartheta$ (\ref{thelim}) for $\kappa\in(0,8)$, with $\Gamma_1$ and/or $\Gamma_2$ replaced by $\Gamma_1'$ and or $\Gamma_2'$ as described above.  This limit is
\begin{multline}\label{case1lim}\bar{\ell}_1\big(\mathcal{F}_\vartheta\big|_{(\Gamma_1,\Gamma_2)\mapsto(\Gamma_1',\Gamma_2')}\big)\,(\kappa\,|\,\boldsymbol{x})\,\,\,:=\lim_{x_{i+1}\rightarrow x_i}(x_{i+1}-x_i)^{6/\kappa-1}\,\,\times\\
(\mathcal{F}_\vartheta|_{(\Gamma_1,\Gamma_2)\mapsto(\Gamma_1',\Gamma_2')})\,(\kappa\,|\,\boldsymbol{x})\\ \parallel \\
\boxed{\begin{aligned}&n(\kappa)\left[\frac{n(\kappa)\Gamma(2-8/\kappa)}{4\sin^2(4\pi/\kappa)\Gamma(1-4/\kappa)^2}\right]^{N-1}\,\Bigg(\prod_{j<k}^{2N-1}(x_k-x_j)^{2/\kappa}\Bigg)\Bigg(\prod_{k=1}^{2N-1}(x_{2N}-x_k)^{1-6/\kappa}\Bigg)\overbrace{\oint_{\Gamma_{N-1}}{\rm d}u_{N-1}\dotsm}^{\mathcal{J}^{(N-1,2N)}(\boldsymbol{x})}\\ 
&\underbrace{\dotsm\,\oint_{\Gamma_2'}{\rm d}u_{2}\,\,\oint_{\Gamma_1'}{\rm d}u_1\,\,\mathcal{N}\Bigg[\Bigg(\prod_{l=1}^{2N-1}\prod_{m=1}^{N-1}(x_l-u_m)^{-4/\kappa}\Bigg)\Bigg(\prod_{m=1}^{N-1}(x_{2N}-u_m)^{12/\kappa-2}\Bigg)\Bigg(\prod_{p<q}^{N-1}(u_p-u_q)^{8/\kappa}\Bigg)\Bigg].}_{\mathcal{J}^{(N-1,2N)}(\boldsymbol{x})}\end{aligned}}\end{multline}
According to the previous paragraph, we find the limit of $\mathcal{J}^{(N-1,2N)}(\boldsymbol{x})$ as $x_{i+1}\rightarrow x_i$ by setting $x_{i+1}=x_i$.  We similarly find the limit as $x_{i+1}\rightarrow x_i$ of the algebraic factors multiplying this Coulomb gas integral.  In particular, one of these factors is $(x_{i+1}-x_i)^{2/\kappa}$ because $i<2N-1$ (\ref{thelim}), and it multiplies the factor of $(x_{i+1}-x_i)^{6/\kappa-1}$ accompanying $\bar{\ell}_1$ to give an overall factor of $(x_{i+1}-x_i)^{8/\kappa-1}$.  Because $\kappa\in(0,8)$, the limit of this factor as $x_{i+1}\rightarrow x_i$ is zero while the limits of all others in (\ref{case1lim}) are finite. We thus have our main result, for use in item \ref{firstcase} in the proof of lemma \ref{mainlem}.
\begin{quote}\textbf{Main result:} In case \ref{sc1}, where we replace all integration contours with endpoints at $x_i$ or $x_{i+1}$ by contours that do not surround of touch these points in the formula (\ref{Fexplicit1}) for $\mathcal{F}_\vartheta$, the limit $\bar{\ell}_1\mathcal{F}_\vartheta$ (\ref{thelim}) vanishes.
\end{quote}

\subsection{Proof of lemma \ref{mainlem}, item \ref{secondcase}}\label{s2}

The purpose of this section is to complete the argument for item \ref{secondcase} in the proof of lemma \ref{mainlem}.  Again, we do this in two steps.  First, we determine the asymptotic behavior as $x_{i+1}\rightarrow x_i$ of the Coulomb gas integral $\mathcal{J}^{(N-1,2N)}(\boldsymbol{x})$ in (\ref{Fexplicit1}) for all $\kappa\in(0,8)$.  In case \ref{sc2}, $x_i$ and $x_{i+1}$ are endpoints of only one common integration contour:
\be\label{Gamma1}\Gamma_1=\begin{cases}\mathscr{P}(x_i,x_{i+1}), & 0<\kappa\leq4\\
 [x_i,x_{i+1}]^+, & 4<\kappa<8\end{cases}.\ee
After we find the asymptotic behavior of $\mathcal{J}^{(N-1,2N)}(\boldsymbol{x})$, we use it to compute the limit $\bar{\ell}_1\mathcal{F}_\vartheta$ (\ref{thelim}), our second step.

To determine the behavior of $\mathcal{J}^{(N-1,2N)}(\boldsymbol{x})$ as $x_{i+1}\rightarrow x_i$, we follow these steps first.  The subsequent calculation is fairly straightforward.
\begin{enumerate}
\item\label{oneit}We order the integrations of $\mathcal{J}^{(N-1,2N)}(\boldsymbol{x})$ via Fubini's theorem so $u_1$ is integrated first, followed by integration with respect to $u_2,$ $u_3,\ldots,u_{N-1}$.
\item\label{twoit} The limit as $x_{i+1}\rightarrow x_i$ of $(x_{i+1}-u_m)^{-4/\kappa}$ is uniform over $u_m\in\Gamma_m$ only if $m>1$.  Hence, we set $x_{i+1}=x_i$ in every such factor in the integrand of $\mathcal{J}^{(N-1,2N)}(\boldsymbol{x})$ with $m>1$ to find its limit.
\item\label{threeit} Thus, we need only determine the behavior as $x_{i+1}\rightarrow x_i$ of the integration in $\mathcal{J}^{(N-1,2N)}(\boldsymbol{x})$ with respect to $u_1$.  The corresponding definite integral is a function of  $x_1$, $x_2,\ldots,x_{2N}$, $u_2$, $u_3 ,\ldots,u_{N-1}$, and it is given by (\ref{eulerintegralch2}) with $M=1$ and $K=3N-2$ (and with the integrand enclosed by the symbol $\mathcal{N}[\,\,\ldots\,\,]$) (\ref{I1here}).   
\item\label{fourit} If any contour $\Gamma_m$ with $m>1$ and endpoints at, say, $x_j<x_k$, passes over $\Gamma_1$, then we replace it with two contours in the upper half-plane that do not pass over $\Gamma_1$ and with orientation opposite that of $\Gamma_m$.  The first has its endpoints at $x_j$ and minus infinity, and the second has its endpoints at positive infinity and $x_k$. 
\item After step \ref{fourit}, no contour passes over $\Gamma_1$.  Furthermore, each product $(x_i-u_m)^{-4/\kappa}(x_{i+1}-u_m)^{-4/\kappa}\Delta(u_1-u_m)$ with $m>2$ in the integrand of $\mathcal{J}^{(N-1,2N)}(\boldsymbol{x})$ now equals (figure \ref{Orderings}, (\ref{Deltadefn}))
\be\label{firstordering}(x_i-u_m)^{-4/\kappa}(x_{i+1}-u_m)^{-4/\kappa}(u_1-u_m)^{8/\kappa}\ee
in the left new contour with an endpoint at minus infinity, and
\be\label{secondordering}(u_m-x_i)^{-4/\kappa}(u_m-x_{i+1})^{-4/\kappa}(u_m-u_1)^{8/\kappa}\ee
in the right new contour with an endpoint at plus infinity.  The ordering of the terms in the differences in (\ref{firstordering}, \ref{secondordering}) agrees with what the symbol $\mathcal{N}$ prescribes for pairs of un-nested contours (figure \ref{Orderings}, item \ref{4thitem} of definition \ref{Fkdefn}).
\item We note below (\ref{eulerintegralch2}) that the contours of $\mathcal{J}^{(N-1,2N)}(\boldsymbol{x})$ may intersect because $\gamma=8/\kappa>0$.  Thus, we push all integration contours flush against the real axis (except for the circular integrations of figure \ref{BreakDown}).
\item\label{secondlastit} Restricting our attention to the integration in $\mathcal{J}^{(N-1,2N)}(\boldsymbol{x})$ with respect to $u_1$, we freeze the other integration variables $u_2$, $u_3,\ldots,u_{N-1}$ at arbitrary locations within their respective contours.
\item\label{lastit} With all of the variables $x_1$, $x_2,\ldots,x_{2N}$, $u_2$, $u_3,\ldots,u_{N-1}$ real-valued, we re-index them in increasing order as $x_1<x_2<\ldots<x_K$, with $K=3N-2$ to simplify the notation of the integration with respect to $u_1$.  (The order depends on where we freeze each $u_m$ in its contour $\Gamma_m$.  Also, this re-indexing likely changes the value of the index $i$, as $x_i$ remains the left endpoint of the interval we are collapsing.  We note that if $i+1<2N$ before this re-indexing, as it does in (\ref{thelim}), then $i+1<K$ after this re-indexing.)
\end{enumerate}

Now we explain step \ref{fourit} further.  With its endpoints  fixed, we deform $\Gamma_m$ into a semicircle with large radius $R$ and its base flush against the real axis.  According to the Cauchy integral theorem (Thm.\ 2.3 of \cite{kod}), this alteration does not change the value of $\mathcal{J}^{(N-1,2N)}(\boldsymbol{x})$.  Furthermore, the integration along the arc of the semicircle vanishes like $R^{-1}$ as $R\rightarrow\infty$ thanks to (\ref{powers1}--\ref{powers3}), so only the integration along the contours replacing $\Gamma_m$ in step \ref{fourit} remain.  Finally, these improper integrals converge because their integrands vanish like $u_m^{-2}$ as $u_m\rightarrow\infty$.

Steps \ref{oneit}--\ref{threeit} show that we only need to determine the behavior of the integration with respect to $u_1$ in order to find the asymptotic behavior of $\mathcal{J}^{(N-1,2N)}(\boldsymbol{x})$ as $x_{i+1}\rightarrow x_i$.  After steps \ref{fourit}--\ref{lastit}, this integration has the form
\be\label{I1here}\mathcal{J}^{(1,K)}\Big(\{\beta_j\}\,\Big|\,\,\Gamma_1\,\,\Big|\,x_1,x_2,\ldots,x_K\Big)=\sideset{}{_{\Gamma_1}}\int\mathcal{N}\Bigg[\prod^K_{j=1}(u_1-x_j)^{\beta_j}\Bigg]\,{\rm d}u_1,\quad x_j\not\in(x_i,x_{i+1}),\ee
with $K=3N-2$, $\Gamma_1$ given in (\ref{Gamma1}), and the symbol $\mathcal{N}[\,\,\ldots\,\,]$ ordering the terms of the differences in the function it encloses so this function is real-valued over the domain of integration.

In (\ref{I1here}), we have relabeled each power $\gamma$ of the factors $(u_1-u_m)^\gamma$ in $\mathcal{J}^{(N-1,2N)}(\boldsymbol{x})$ with $m>1$ as $\beta_j$ for some index $j$ for convenience.  After identifying the integration with respect to $u_1$ in (\ref{Fexplicit1}) with (\ref{I1here}), we find that
\be\label{betai}\text{$\beta_j\in\{-4/\kappa,8/\kappa,12/\kappa-2\}$ for all $j\in\{1,2,\ldots,K\}$},\quad\beta_i=\beta_{i+1}=-4/\kappa.\ee

Now we determine the asymptotic behavior of (\ref{I1here}) as $x_{i+1}\rightarrow x_i$.  First, if $\kappa>4$, then we have $\Gamma_1=[x_i,x_{i+1}]^+$.  After substituting $u_1(t)=(1-t)x_i+tx_{i+1}$ in (\ref{I1here}) and identifying the beta function in the result, we find
\begin{multline}\label{result2.5}\mathcal{J}^{(1,K)}\Big(\{\beta_j\}\,\Big|\,[x_i,x_{i+1}]^+\,\Big|\,x_1,x_2,\ldots,x_K\Big)\\
\begin{aligned}&\underset{x_{i+1}\rightarrow x_i}{\sim}(x_{i+1}-x_i)^{\beta_i+\beta_{i+1}+1}\mathcal{N}\Bigg[\prod_{j\neq i,i+1}^K(x_{i}-x_j)^{\beta_j}\Bigg]\sideset{}{_0^1}\int t^{\beta_i}(1-t)^{\beta_{i+1}}\,{\rm d}t\hspace{1in}\\
&=\frac{\Gamma(\beta_i+1)\Gamma(\beta_{i+1}+1)}{\Gamma(\beta_i+\beta_{i+1}+2)}(x_{i+1}-x_i)^{\beta_i+\beta_{i+1}+1}\mathcal{N}\Bigg[\prod_{j\neq i,i+1}^K(x_{i}-x_j)^{\beta_j}\Bigg].\end{aligned}\end{multline}
(Here, $\mathcal{N}[\,\,\ldots\,\,]$ orders the terms of the differences between its brackets so the function it encloses is real-valued.)

On the other hand, if $\kappa\in(0,4]$ (in which case $\beta_i,\beta_{i+1}\leq-1$ thanks to (\ref{betai}), so the improper integral (\ref{I1here}) with $\Gamma_1=[x_i,x_{i+1}]^+$ diverges), then we have $\Gamma_1=\mathscr{P}(x_i,x_{i+1})$ instead.  After the same substitution as before, we find 
\begin{multline}\label{result2}\mathcal{J}^{(1,K)}\Big(\{\beta_j\}\,\Big|\,\mathscr{P}(x_i,x_{i+1})\,\Big|\,x_1,x_2,\ldots,x_K\Big)\\\begin{aligned}&\underset{x_{i+1}\rightarrow x_i}{\sim}(x_{i+1}-x_i)^{\beta_i+\beta_{i+1}+1}\mathcal{N}\Bigg[\prod_{j\neq i,i+1}^K(x_{i}-x_j)^{\beta_j}\Bigg]\sideset{}{_{\mathclap{\hspace{1cm}\mathscr{P}(0,1)}}}\oint \hspace{.5cm}t^{\beta_i}(1-t)^{\beta_{i+1}}\,{\rm d}t\\
& =4e^{\pi i(\beta_i-\beta_{i+1})}\sin\pi\beta_i\sin\pi\beta_{i+1}\frac{\Gamma(\beta_i+1)\Gamma(\beta_{i+1}+1)}{\Gamma(\beta_i+\beta_{i+1}+2)}(x_{i+1}-x_i)^{\beta_i+\beta_{i+1}+1}\mathcal{N}\Bigg[\prod_{j\neq i,i+1}^K(x_{i}-x_j)^{\beta_j}\Bigg],\end{aligned}\end{multline}
where we have used an analytic continuation of the beta function \cite{witt} to evaluate the contour integral in (\ref{result2}).

Because (\ref{result2}) is just the analytic continuation of (\ref{result2.5}) (viewing the latter as a complex-analytic function of $\beta_i$ or $\beta_{i+1}$) to $\text{Re}\,\beta_i\leq-1$ or $\text{Re}\,\beta_{i+1}\leq-1$, displaying (\ref{result2.5}) with (\ref{result2}) might seem redundant.  However, in sections \ref{s3} and \ref{s4}, we find that working with simple contours first and then analytically continuing our results with Pochhammer contours simplifies our analysis considerably.  Hence, we need both types of integration contours in those sections, so we display (\ref{result2.5}) in addition to (\ref{result2}).

To finish, we use (\ref{result2.5}, \ref{result2}) to justify item \ref{secondcase} in the proof of lemma \ref{mainlem}.  In that part of the proof, we require the limit $\bar{\ell}_1\mathcal{F}_\vartheta$ (\ref{thelim}) for $\kappa\in(0,8)$, with $\mathcal{F}_\vartheta$ given by (\ref{Fexplicit1}) and $\Gamma_1$ given by (\ref{Gamma1}) (so the sine functions drop from the prefactor (\ref{firstprefactor}) in (\ref{Fexplicit1}) if $\kappa>4$).  For $\kappa\leq4$, this limit is 
\begin{multline}\label{thebiglimit}\bar{\ell}_1\mathcal{F}_\vartheta(\kappa\,|\,x_1,x_2,\ldots,x_{i-1},x_{i+2},\ldots,x_{2N})\,\,=\lim_{x_{i+1}\rightarrow x_i}(x_{i+1}-x_i)^{6/\kappa-1}\,\,\times\\ \mathcal{F}_\vartheta(\kappa\,|\,\boldsymbol{x}) \\ \parallel \\
\boxed{\begin{aligned}&n(\kappa)\left[\frac{n(\kappa)\Gamma(2-8/\kappa)}{4\sin^2(4\pi/\kappa)\Gamma(1-4/\kappa)^2}\right]^{N-1}\Bigg(\prod_{\substack{1\leq j<k\\j,k\neq i,i+1}}^{2N-1}(x_k-x_j)^{2/\kappa}\Bigg)\Bigg(\prod_{\substack{k=1 \\ k\neq i,i+1}}^{2N-1}(x_{2N}-x_k)^{1-6/\kappa}\Bigg)\sideset{}{_{\Gamma_{N-1}}}\oint {\rm d}u_{N-1}\\
&\dotsm\,\sideset{}{_{\Gamma_3}}\oint {\rm d}u_3\,\, \sideset{}{_{\Gamma_2}}\oint {\rm d}u_2\,\,\mathcal{N}\Bigg[\Bigg(\prod_{l\neq i,i+1}^{2N-1}\prod_{m=2}^{N-1}(x_l-u_m)^{-4/\kappa}\Bigg)\Bigg(\prod_{m=2}^{N-1}(x_{2N}-u_m)^{12/\kappa-2}\Bigg)\Bigg(\,\,\prod_{\mathclap{2\leq p<q}}^{N-1}(u_p-u_q)^{8/\kappa}\Bigg)\\
&\Bigg(\prod_{l=i}^{i+1}\prod_{m=2}^{N-1}(x_l-u_m)^{-4/\kappa}\Bigg)\Bigg]\Bigg(\prod_{j\neq i,i+1}^{2N-1}|x_j-x_i|^{2/\kappa}|x_j-x_{i+1}|^{2/\kappa}\Bigg)\Big(x_{2N}-x_i\Big)^{1-6/\kappa}\Big(x_{2N}-x_{i+1}\Big)^{1-6/\kappa}\\
&\hspace{.33cm}\Big(x_{i+1}-x_i\Big)^{2/\kappa}\underbrace{\sideset{}{_{\mathscr{P}(x_i,x_{i+1})}}\oint
{\rm d}u_1\,\,\mathcal{N}\Bigg[\Bigg(\prod_{l=1}^{2N-1}(x_l-u_1)^{-4/\kappa}\Bigg)\Bigg(\prod_{m=2}^{N-1}(u_m-u_1)^{8/\kappa}\Bigg)\Big(x_{2N}-u_1\Big)^{12/\kappa-2}\Bigg].}_{\mathcal{J}^{(1,K)}}\end{aligned}}\end{multline}
(If $\kappa>4$, then we adjust (\ref{thebiglimit}) as per item \ref{itemc} in the introduction of this appendix.)  We have rewritten the formula (\ref{Fexplicit1}) for $\mathcal{F}_\vartheta(\kappa\,|\,\boldsymbol{x})$ slightly to clarify the calculation, and we indicate the contour integral $\mathcal{J}^{(1,K)}$ (\ref{I1here}) with braces.

Now we find the limit (\ref{thebiglimit}).  After doing steps \ref{oneit}--\ref{lastit}, setting $x_{i+1}=x_i$ in all factors of (\ref{Fexplicit1}) without $u_1$, identifying the definite integral with respect to $u_1$ with (\ref{Gamma1}, \ref{I1here}, \ref{betai}), and replacing it by the right side of (\ref{result2}), we find 
\begin{multline}\label{asympinsert}(x_{i+1}-x_i)^{6/\kappa-1}\mathcal{F}_\vartheta(\kappa\,|\,\boldsymbol{x})\underset{x_{i+1}\rightarrow x_i}{\sim}(x_{i+1}-x_i)^{6/\kappa-1}\,\,\times\\
\boxed{\begin{aligned}&n(\kappa)\left[\frac{n(\kappa)\Gamma(2-8/\kappa)}{4\sin^2(4\pi/\kappa)\Gamma(1-4/\kappa)^2}\right]^{N-1}\Bigg(\prod_{\substack{1\leq j<k\\j,k\neq i,i+1}}^{2N-1}(x_k-x_j)^{2/\kappa}\Bigg)\Bigg(\prod_{\substack{k=1 \\ k\neq i,i+1}}^{2N-1}(x_{2N}-x_k)^{1-6/\kappa}\Bigg)\sideset{}{_{\Gamma_{N-1}}}\oint {\rm d}u_{N-1}\\
&\dotsm\sideset{}{_{\Gamma_3}}\oint {\rm d}u_3\,\, \sideset{}{_{\Gamma_2}}\oint {\rm d}u_2\,\,\,\mathcal{N}\Bigg[\Bigg(\prod_{l\neq i,i+1}^{2N-1}\prod_{m=2}^{N-1}(x_l-u_m)^{-4/\kappa}\Bigg)\Bigg(\prod_{m=2}^{N-1}(x_{2N}-u_m)^{12/\kappa-2}\Bigg)\Bigg(\,\,\prod_{\mathclap{2\leq p<q}}^{N-1}(u_p-u_q)^{8/\kappa}\Bigg)\\
&\Bigg(\prod_{m=2}^{N-1}(x_i-u_m)^{-8/\kappa}\Bigg)\Bigg]\Bigg(\prod_{j\neq i,i+1}^{2N-1}|x_j-x_i|^{4/\kappa}\Bigg)\Big(x_{2N}-x_i\Big)^{2-12/\kappa}\Big(x_{i+1}-x_i\Big)^{2/\kappa}\overbrace{4\sin^2\left(\frac{4\pi}{\kappa}\right)}^{\text{right side of (\ref{result2})}}\\
&\hspace{.25cm}\underbrace{\frac{\Gamma(1-4/\kappa)^2}{\Gamma(2-8/\kappa)}\Big(x_{i+1}-x_i\Big)^{1-8/\kappa}\mathcal{N}\Bigg[\Bigg(\prod_{j\neq i,i+1}^{2N-1}(x_j-x_i)^{-4/\kappa}\Bigg)\Big(x_{2N}-x_i\Big)^{12/\kappa-2}\Bigg(\prod_{m=2}^{N-1}(x_i-u_m)^{8/\kappa}\Bigg)\Bigg].}_{\text{right side of (\ref{result2})}}\end{aligned}}\end{multline}
(If $\kappa>4$ and (\ref{thebiglimit}) is adjusted as per item \ref{itemc} in the introduction of this appendix, then we replace the definite integral with respect to $u_1$ with the right side of (\ref{result2.5}) instead, finding (\ref{asympinsert}) again, but with all factors of $4\sin^2(4\pi/\kappa)$ dropped and all contours simple.)  After some simplification, we finally send $x_{i+1}\rightarrow x_i$ in (\ref{asympinsert}) to find
\begin{multline}\label{firstlimprime}\bar{\ell}_1\mathcal{F}_\vartheta(\kappa\,|\,x_1,x_2,\ldots,x_{i-1},x_{i+2},\ldots,x_{2N})\,\,=\,\,n(\kappa)\,\,\times\\
\left\{\begin{aligned}&n(\kappa)\left[\frac{n(\kappa)\Gamma(2-8/\kappa)}{4\sin^2(4\pi/\kappa)\Gamma(1-4/\kappa)^2}\right]^{N-2}\Bigg(\prod_{\substack{1\leq j<k \\ j,k\neq i,i+1}}^{2N-1}(x_k-x_j)^{2/\kappa}\Bigg)\Bigg(\prod_{\substack{k=1 \\ k\neq i,i+1}}^{2N-1}(x_{2N}-x_k)^{1-6/\kappa}\Bigg)\oint_{\Gamma_{N-1}} {\rm d}u_{N-1}\\ 
&\dotsm\,\oint_{\Gamma_3}{\rm d}u_3\,\,\oint_{\Gamma_2}{\rm d}u_2\,\,\mathcal{N}\Bigg[\Bigg(\prod_{l\neq i,i+1}^{2N-1}\prod_{m=2}^{N-1}(x_l-u_m)^{-4/\kappa}\Bigg)\Bigg(\prod_{m=2}^{N-1}(x_{2N}-u_m)^{12/\kappa-2}\Bigg)\Bigg(\prod_{2\leq p<q}^{N-1}(u_p-u_q)^{8/\kappa}\Bigg)\Bigg]\end{aligned}\right\}.\end{multline}
(Again, if $\kappa>4$, then we find the same result, but with all factors of $4\sin^2(4\pi/\kappa)$ dropped and with all Pochhammer contours replaced by simple contours with the same endpoints and orientation.)

Now, the quantity of (\ref{firstlimprime}) in brackets is the element of $\mathcal{B}_{N-1}$ described in the following conclusion.  We thus have our main result, for use in item \ref{secondcase} in the proof of lemma \ref{mainlem}.
\begin{quote}\textbf{Main result:} In case \ref{sc2}, where $\Gamma_1$ is given by (\ref{Gamma1}), the limit $\bar{\ell}_1\mathcal{F}_\vartheta$ (\ref{thelim}) equals $n$ (\ref{fugacity}) times the element of $\mathcal{B}_{N-1}$ generated from the formula (\ref{Fexplicit1}) for $\mathcal{F}_\vartheta$ by dropping all factors involving $x_i$, $x_{i+1}$, and $u_1$, dropping the integration along $\Gamma_1$, and reducing the power of the prefactor (\ref{firstprefactor}) or (\ref{secondprefactor}) by one.
\end{quote} 
We extend this result to the value $i=2N$ in section \ref{s51} below, for use in item \ref{secondcase2} in the proof of corollary \ref{moveconjchargecor}.

\subsection{Proof of lemma \ref{mainlem}, item \ref{thirdcase}}\label{s3}

The purpose of this section is to complete the argument for item \ref{thirdcase} in the proof of lemma \ref{mainlem}.  We do this in two steps.  First, we determine the asymptotic behavior as $x_{i+1}\rightarrow x_i$ of the Coulomb gas integral $\mathcal{J}^{(N-1,2N)}(\boldsymbol{x})$ in (\ref{Fexplicit1}) for all $\kappa\in(0,8)$ with $8/\kappa\not\in\mathbb{Z}^+$.  In case \ref{sc3new}, $x_i$ and $x_{i-1}$ (resp.\ $x_{i+1}$ and $x_{i+2}$) are endpoints of one contour $\Gamma_1''$ (replacing $\Gamma_1$ in item \ref{sc3} of this appendix), and $x_{i+1}$ (resp.\ $x_i$) is not an endpoint of any contour.   Among these two possibilities, we choose without loss of generality
\be\label{Gamma1''}\Gamma_1''=\begin{cases}\mathscr{P}(x_{i-1},x_i), & 0<\kappa\leq4\\
[x_{i-1},x_i]^+, & 4<\kappa<8\end{cases}.\ee 
(According to item \ref{ineq1} in the proof of lemma \ref{mainlem}, we have $i>1$.)  After we find the asymptotic behavior of $\mathcal{J}^{(N-1,2N)}(\boldsymbol{x})$, we use it to compute the limit $\bar{\ell}_1\mathcal{F}_\vartheta$ (\ref{thelim}), our second step.

Determining the asymptotic behavior of $\mathcal{J}^{(N-1,2N)}(\boldsymbol{x})$ in this case is more delicate than determining this behavior in cases \ref{sc1} and \ref{sc2}.  With $x_i$ an endpoint of $\Gamma_1''$, the integration variable $u_1$ is not bounded away from $x_i$ as $x_{i+1}\rightarrow x_i$.  Then because the difference $x_i-u_1$ is always much smaller than $x_{i+1}-u_1$ for some $u_1\in\Gamma_1''$ even as we send $x_{i+1}\rightarrow x_i$, we may not simply set $x_{i+1}=x_i$ in the integration with respect to $u_1$ to determine this asymptotic behavior.  Instead, we deform the contour $\Gamma_1''$ into contours that fall under cases \ref{sc1} and \ref{sc2} (figure \ref{case3}), and we apply the results of sections \ref{s1} and \ref{s2} to determine the asymptotic behavior of the new terms that thus appear.

To determine the asymptotic behavior of $\mathcal{J}^{(N-1,2N)}(\boldsymbol{x})$ as $x_{i+1}\rightarrow x_i$, we execute steps \ref{oneit}--\ref{lastit} of section \ref{s2} first.  Actually, steps \ref{oneit}--\ref{threeit} show that we only need to determine the behavior of the integration in $\mathcal{J}^{(N-1,2N)}(\boldsymbol{x})$ with respect to $u_1$ as $x_{i+1}\rightarrow x_i$. After executing steps \ref{fourit}--\ref{lastit} in section \ref{s2}, we find that this integration is
\be\label{prescenario3}\mathcal{J}^{(1,K)}\Big(\{\beta_j\}\,\Big|\,\,\Gamma_1''\,\,\Big|\,x_1,x_2,\ldots,x_K\Big)=\sideset{}{_{\Gamma_1''}}\int\mathcal{N}\Bigg[\prod_{j=1}^K(u_1-x_j)^{\beta_j}\Bigg]\,{\rm d}u_1,\quad x_j\not\in(x_{i-1},x_{i+1}),\ee
with $K=3N-2$, $\Gamma_1''$ given in (\ref{Gamma1''}), and the symbol $\mathcal{N}[\,\,\ldots\,\,]$ ordering the terms of the differences in the function it encloses so this function is real-valued over the domain of integration.

In (\ref{prescenario3}), we have relabeled each power $\gamma$ of the factors $(u_1-u_m)^\gamma$ in $\mathcal{J}^{(N-1,2N)}(\boldsymbol{x})$ with $m>1$ as $\beta_j$ for some index $j$ for convenience.  After identifying the integration with respect to $u_1$ in (\ref{Fexplicit1}) with (\ref{prescenario3}), we find that
\be\begin{gathered}\label{applicationsc3} 
s:=\sideset{}{_{j=1}^K}\sum\beta_j=-2,\qquad\text{$\beta_j\in\{-4/\kappa,8/\kappa,12/\kappa-2\}$ for all $j\in\{1,2,\ldots,K\},$}\\
\quad\beta_{i-1}=\beta_i=\beta_{i+1}=-4/\kappa,\qquad\beta_{i+2}\in\{-4/\kappa,12/\kappa-2\}.
\end{gathered}\ee
Initially, we assume $\kappa\in(4,8)$ in our calculations, and in this range, we automatically have $8/\kappa\not\in\mathbb{Z}^+$.  Combined with the former restriction, (\ref{applicationsc3}) implies the weaker conditions
\gdef\thesubequation{\theequation \textit{a,b}}
\begin{subeqnarray}
\label{extracond0}&&s:=\sideset{}{_{j=1}^K}\sum\beta_j\in\mathbb{Z}^-\setminus\{-1\},\qquad\text{$\beta_j>-1$ for all $j\in\{1,2,\ldots,K\},$}\\
\refstepcounter{equation}
\gdef\thesubequation{\theequation \textit{a,b}}
\label{extracond}
&&\hspace{2.5cm}\beta_i+\beta_{i+1}\not\in\mathbb{Z},\qquad\beta_i+\beta_{i+1}<-1,
\end{subeqnarray}
which we use exclusively in the rest of this section.  We note that the branch cuts of the integrand anchored to $x_1$, $x_2,\ldots,x_{K-1}$ collectively end at $x_K$ because $s\in\mathbb{Z}$.

Thus, our present goal is to determine the asymptotic behavior of $\mathcal{J}^{(1,K)}(\boldsymbol{x})$ (\ref{prescenario3}, \ref{Gamma1''}) as $x_{i+1}\rightarrow x_i$, first for $\kappa\in(4,8)$ with conditions (\ref{extracond0}, \ref{extracond}).  For this purpose, we define for each $k\in\{1,2,\ldots,K\}$ the definite integral
\be\label{Ikintegrals}
I_k(x_1,x_2,\ldots,x_K):=\sideset{}{_{x_k}^{x_{k+1}}}\int\mathcal{N}\Bigg[\prod_{j=1}^K(u_1-x_j)^{\beta_j}\Bigg]\,{\rm d}u_1,\quad\text{with}\,\,\sideset{}{_{x_K}^{x_{K+1}}}\int\,\,:=\,\,\sideset{}{_{x_K}^\infty}\int\,\,+\,\,\sideset{}{_{-\infty}^{x_1}}\int\,\,\text{ if $k=K$.}\ee
(Conditions (\ref{extracond0}) imply that any improper integral among (\ref{Ikintegrals}) converges.)  Here, $I_{i-1}(\boldsymbol{x})$ is precisely the definite integral (\ref{prescenario3}) with $\Gamma_1''=[x_{i-1},x_i]^+$ whose asymptotic behavior as $x_{i+1}\rightarrow x_i$ we wish to determine under conditions (\ref{extracond0}, \ref{extracond}).  Also, the definite integrals $I_{i\pm1}$ fall under case \ref{sc3new}, $I_i$ falls under case \ref{sc2}, and the others fall under case \ref{sc1}.

\begin{figure}[b]
\centering
\includegraphics[scale=0.27]{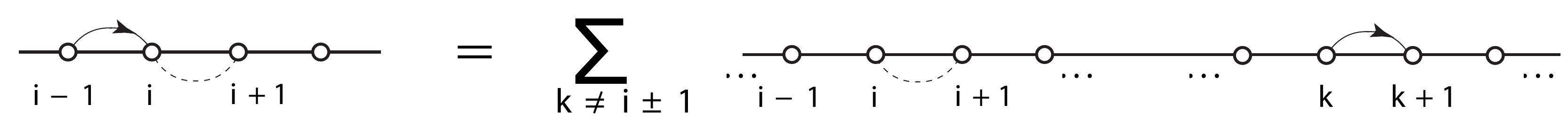}
\caption{The third case.  The dashed curve connects the endpoints of the interval to be collapsed.  We push the integration contour from $[x_{i-1},x_i]$ onto all intervals except $[x_{i+1},x_{i+2}]$.}
\label{case3}
\end{figure}

Now, to determine the asymptotic behavior of $I_{i-1}(\boldsymbol{x})$ (\ref{prescenario3}) as $x_{i+1}\rightarrow x_i$ under conditions (\ref{extracond0}, \ref{extracond}), we express this quantity as a linear combination of definite integrals $I_k$ (\ref{Ikintegrals}) with $k\neq i\pm1$ (figure \ref{case3}).  These latter definite integrals fall under cases \ref{sc1} and \ref{sc2}, which we study in sections \ref{s1} and \ref{s2} respectively.  To this end, we replace the integration contour in (\ref{Ikintegrals}) by a large semicircle of radius $R$, with counterclockwise (resp.\ clockwise) orientation in the upper (resp.\ lower) half-plane, and with its base on the real axis.  By the Cauchy integral theorem, the integration around this semicircle gives zero.  Furthermore, the integration along the arc of the semicircle behaves like $R^{s+1}$ and thus vanishes as $R\rightarrow\infty$, thanks to (\ref{extracond0}\red{a}).  Thus, we find the two equations (with the $-$ (resp.\ $+$) sign corresponding with the upper (resp.\ lower) half-plane setting)
\be\label{case3int}\sum_{k=1}^Ke^{\pm\pi i\sum_{l=1}^k\beta_l}I_k=0.\ee
The phase factors in (\ref{case3int}) arise from passing over or under the branch points $x_j$ of the function enclosed within the brackets of the integrand in (\ref{Ikintegrals}) (figure \ref{phases}).Using the two equations (\ref{case3int}) to solve for $I_{i-1}$ in terms of $I_k$ with $k\neq i\pm1$ (figure \ref{case3}), we find (recalling that $i+1<K$ thanks to step \ref{lastit} of section \ref{s2})
\be\label{result}I_{i-1}=-\sum_{k=1}^{i-2}\frac{\sin\pi\sum_{l=k+1}^{i+1}\beta_l}{\sin\pi(\beta_i+\beta_{i+1})}I_k+\sum_{k=i+2}^K\frac{\sin\pi\sum_{l=i+2}^k\beta_l}{\sin\pi(\beta_i+\beta_{i+1})}I_k-\frac{\sin\pi\beta_{i+1}}{\sin\pi(\beta_i+\beta_{i+1})}I_i.\ee
(We revise this equation to (\ref{resulti2}) if $i=2$ in section \ref{s52} below, without affecting any of the following results.)  On the right side of (\ref{result}), the definite integral $I_k$ with $k\neq i$ (resp.\ $k=i$) falls under case \ref{sc1} (resp.\ \ref{sc2}).  Thus, after inserting $x_{i+1}=x_i$ (resp.\ (\ref{result2.5})) in $I_k$ with $k\neq i$ (resp.\ $k=i$), we find ((\ref{Ikintegrals}) defines integration from $x_K$ to $x_{K+1}$)
\be\begin{aligned}\label{preresult}I_{i-1}(x_1,x_2,\ldots,x_K)\underset{x_{i+1}\rightarrow x_i}{\sim}&-\sum_{k=1}^{i-2}\frac{\sin\pi\sum_{l=k+1}^{i+1}\beta_l}{\sin\pi(\beta_i+\beta_{i+1})}\sideset{}{_{x_k}^{x_{k+1}}}\int\mathcal{N}\Bigg[(u-x_i)^{\beta_i+\beta_{i+1}}\prod_{j\neq i,i+1}^K(u_1-x_j)^{\beta_j}\Bigg]\,{\rm d}u_1\\
&+\sum_{k=i+2}^K\frac{\sin\pi\sum_{l=i+2}^k\beta_l}{\sin\pi(\beta_i+\beta_{i+1})}\sideset{}{_{x_k}^{x_{k+1}}}\int\mathcal{N}\Bigg[(u-x_i)^{\beta_i+\beta_{i+1}}\prod_{j\neq i,i+1}^K(u_1-x_j)^{\beta_j}\Bigg]\,{\rm d}u_1&\\
&-\frac{\sin\pi\beta_{i+1}\,\Gamma(\beta_i+1)\Gamma(\beta_{i+1}+1)}{\sin\pi(\beta_i+\beta_{i+1})\Gamma(\beta_i+\beta_{i+1}+2)}(x_{i+1}-x_i)^{\beta_i+\beta_{i+1}+1}\mathcal{N}\Bigg[\prod_{j\neq i,i+1}^K(x_i-x_j)^{\beta_j}\Bigg],&\end{aligned}\ee
assuming conditions (\ref{extracond0}, \ref{extracond}\red{a}).  (The denominators of (\ref{preresult}) are not zero, thanks to (\ref{extracond}\red{a}).)  Condition (\ref{extracond}\red{b}) shows that only the last term of (\ref{preresult}) blows up as $x_{i+1}\rightarrow x_i$.  Thus, after changing notation from $I_{i-1}$ to (\ref{prescenario3}) with $\Gamma_1''=[x_{i-1},x_i]^+$, we have
\begin{multline}\label{result3}\mathcal{J}^{(1,K)}\Big(\{\beta_j\}\,\Big|\,[x_{i-1},x_i]^+\,\Big|\,x_1,x_2,\ldots,x_K\Big)\underset{x_{i+1}\rightarrow x_i}{\sim}\\
-\frac{\sin\pi\beta_{i+1}\,\Gamma(\beta_i+1)\Gamma(\beta_{i+1}+1)}{\sin\pi(\beta_i+\beta_{i+1})\Gamma(\beta_i+\beta_{i+1}+2)}(x_{i+1}-x_i)^{\beta_i+\beta_{i+1}+1}\mathcal{N}\Bigg[\prod_{j\neq i,i+1}^K(x_i-x_j)^{\beta_j}\Bigg],\end{multline}
assuming conditions (\ref{extracond0}, \ref{extracond}).  We note that (\ref{result3}) is identical to (\ref{result2.5}), except that the integration contours differ and a ratio of sine functions and factor of negative one multiplies the right side of the former.  Thanks to (\ref{applicationsc3}), the product of these extra factors equals $n(\kappa)^{-1}$ (\ref{fugacity}), and it justifies the factors of $n(\kappa)^{-1}$ appearing in case \ref{thirdcase} of figure \ref{Cases} and in the middle two lines of each bracketed collection in figure \ref{Case4}.

By solving the two equations (\ref{case3int}) for $I_{i+1}$ in terms of $I_k$ with $k\neq i\pm1$, we find that the right side of (\ref{result3}) with $\beta_i$ and $\beta_{i+1}$ switched gives the asymptotic behavior of (\ref{prescenario3}) in the other situation with $\Gamma_1''=[x_{i+1},x_{i+2}]^+$.

Now we extend our findings from $\kappa\in(4,8)$ to $\kappa\in(0,8)$ with $8/\kappa\not\in\mathbb{Z}^+$.  If $\kappa\in(0,4]$, then (\ref{applicationsc3}) implies that $\beta_k\leq-1$ for several $k\in\{1,2,\ldots,K\}$, so the improper integral $I_k$ (\ref{Ikintegrals}) diverges.  If we extend $I_k$ to an analytic function of $\text{Re}\,\beta_k>-1$, then inserting the replacement, inspired by (\ref{Pochtostraight}),
\be\label{analyticcont}\int_{x_k}^{x_{k+1}}\quad\longmapsto\quad\frac{1}{4e^{\pi i(\beta_k-\beta_{k+1})}\sin\pi\beta_k\sin\pi\beta_{k+1}}\sideset{}{_{\mathscr{P}(x_k,x_{k+1})}}\oint\ee
into $I_k$ (\ref{Ikintegrals}) analytically continues this function to all complex $\beta_k\not\in\mathbb{Z}^-$, where it has simple poles.  This replacement also analytically continues $I_k$ from $\text{Re}\,\beta_{k+1}>-1$ to all complex $\beta_{k+1}\not\in\mathbb{Z}^-$ (with $\beta_{K+1}:=\beta_1$).  To avoid the poles at the negative integers, we weaken condition (\ref{extracond0}\red{b}) to
\be\label{extracondmod}\text{$\beta_j\not\in\mathbb{Z}^-$ for all $j\in\{1,2,\ldots,K\}$},\qquad(\beta_{K+1}:=\beta_1).\ee
Indeed, if $\kappa\in(0,8)$ and $8/\kappa\not\in\mathbb{Z}^+$, then (\ref{applicationsc3}) implies this new condition (\ref{extracondmod}).  Equation (\ref{result}) holds after we alter each $I_k$ as in (\ref{analyticcont}) because both of its sides are analytic on $\text{Re}\,\beta_j>-1$.  In particular, if $\kappa\in(0,4]$ so $\beta_i\leq-1$ (\ref{applicationsc3}), then after solving (this analytic continuation of) (\ref{result}) for $I_{i-1}$ (which is also (\ref{prescenario3}) with $\Gamma_1''=\mathscr{P}(x_{i-1},x_i)$), we find 
\begin{multline}\label{result3.5}\mathcal{J}^{(1,K)}\Big(\{\beta_j\}\,\Big|\,\mathscr{P}(x_{i-1},x_i)\,\Big|\,x_1,x_2,\ldots,x_K\Big)\underset{x_{i+1}\rightarrow x_i}{\sim}4e^{\pi i(\beta_{i-1}-\beta_i)}\sin\pi\beta_{i-1}\sin\pi\beta_i\\
\times(-1)\frac{\sin\pi\beta_{i+1}\,\Gamma(\beta_i+1)\Gamma(\beta_{i+1}+1)}{\sin\pi(\beta_i+\beta_{i+1})\Gamma(\beta_i+\beta_{i+1}+2)}(x_{i+1}-x_i)^{\beta_i+\beta_{i+1}+1}\mathcal{N}\Bigg[\prod_{j\neq i,i+1}^K(x_i-x_j)^{\beta_j}\Bigg],\end{multline}
assuming (\ref{extracond0}\red{a}, \ref{extracond}, \ref{extracondmod}).  Thanks to (\ref{Pochtostraight}), this new result (\ref{result3.5}) includes our previous result (\ref{result3}) for $\beta_{i-1},\beta_i>-1$.  Or if, instead of (\ref{Gamma1''}), $x_{i+1}$ and $x_{i+2}$ are endpoints of $\Gamma_1''$ and $\kappa\in(0,4]$ so $\beta_{i+1}\leq-1$ (\ref{applicationsc3}), then after solving (this analytic continuation of) (\ref{result}) for $I_{i+1}$ (which is also (\ref{prescenario3}) with $\Gamma_1''=\mathscr{P}(x_{i+1},x_{i+2})$), we find
\begin{multline}\label{result3.75}\mathcal{J}^{(1,K)}\Big(\{\beta_j\}\,\Big|\,\mathscr{P}(x_{i+1},x_{i+2})\,\Big|\,x_1,x_2,\ldots,x_K\Big)\underset{x_{i+1}\rightarrow x_i}{\sim}4e^{\pi i(\beta_{i+1}-\beta_{i+2})}\sin\pi\beta_{i+1}\sin\pi\beta_{i+2}\\
\times(-1)\frac{\sin\pi\beta_i\,\Gamma(\beta_i+1)\Gamma(\beta_{i+1}+1)}{\sin\pi(\beta_i+\beta_{i+1})\Gamma(\beta_i+\beta_{i+1}+2)}(x_{i+1}-x_i)^{\beta_i+\beta_{i+1}+1}\mathcal{N}\Bigg[\prod_{j\neq i,i+1}^K(x_i-x_j)^{\beta_j}\Bigg],\end{multline}
again, assuming (\ref{extracond0}\red{a}, \ref{extracond}, \ref{extracondmod}).

To finish, we use (\ref{result3}, \ref{result3.5}) to justify item \ref{thirdcase} in the proof of lemma \ref{mainlem}.  In that part of the proof, we require the limit $\bar{\ell}_1\mathcal{F}_\vartheta$ (\ref{thelim}) for $\kappa\in(0,8)$ and $8/\kappa\not\in\mathbb{Z}^+$, with $\mathcal{F}_\vartheta$ given by (\ref{Fexplicit1}) and $\Gamma_1$ replaced by $\Gamma_1''$ (\ref{Gamma1''}) (so the sine functions drop from the prefactor (\ref{firstprefactor}) in (\ref{Fexplicit1}) if $\kappa>4$). This calculation is almost identical to the calculation spanning (\ref{thebiglimit}--\ref{firstlimprime}) in section \ref{s2}, with one difference.  After we insert the asymptotic result (\ref{result3.5}) (with each $\beta_j$ assigned its appropriate value among those in (\ref{applicationsc3})) into the formula (\ref{thebiglimit}) for $\mathcal{F}_\vartheta$, we again find (\ref{asympinsert}), but with an extra factor of $n(\kappa)^{-1}$ (\ref{fugacity}) arising from the factor of negative one and the ratio of sine functions on the right side of (\ref{result3.5}).  (The condition $8/\kappa\not\in\mathbb{Z}^+$ implies that $n(\kappa)\neq0$.)   Absent from the right side of our previous case \ref{sc2} result (\ref{result2}), this extra factor removes the factor of $n(\kappa)$ that appears outside the braces on the right side of (\ref{firstlimprime}).  Thus, we arrive with our main result, for use in item \ref{thirdcase} of the proof of lemma \ref{mainlem}.  (Again, we assume here that $8/\kappa\not\in\mathbb{Z}^+$, but the result remains true even if this condition is not met.  See the last paragraph of section \ref{s4}.)
\begin{quote}\textbf{Main result:} In case \ref{sc3new}, where $\Gamma_1''$ (\ref{Gamma1''}) replaces $\Gamma_1$ in the formula (\ref{Fexplicit1}) for $\mathcal{F}_\vartheta$, the limit $\bar{\ell}_1\mathcal{F}_\vartheta$ (\ref{thelim}) equals the element of $\mathcal{B}_{N-1}$ generated from the formula for $\mathcal{F}_\vartheta$ by dropping all factors involving $x_i$, $x_{i+1}$, and $u_1$, dropping the integration along $\Gamma_1$, and reducing the power of the prefactor (\ref{firstprefactor}) or (\ref{secondprefactor}) by one.
\end{quote}
As mentioned, if $i=2$, then we modify (\ref{result}) and derive this main result with the modified equation in section \ref{s52}.

On the other hand, if $x_{i+1}$ and $x_{i+2}$ are the endpoints of $\Gamma_1''$, then we repeat the analysis of the above paragraph, using (\ref{result3.75}) in place of (\ref{result3.5}), and ultimately find the same result.  There is one exceptional detail worth mentioning though.  If $\beta_{i+2}=12/\kappa-2$ (\ref{applicationsc3}) and $\Gamma_1''=\mathscr{P}(x_{i+1},x_{i+2})$, then according to (\ref{nextadjustment}, \ref{analyticcont}), as we replace $\Gamma_1$ with $\Gamma_1''$, we must also replace one factor of $4\sin^2(4\pi/\kappa)$ in (\ref{thebiglimit}) with $4e^{-16\pi i/\kappa}\sin(-4\pi/\kappa)\sin(12\pi/\kappa)$.  With $\beta_{i+1}=-4/\kappa$ too  (\ref{applicationsc3}), this factor exactly cancels the factor of $4e^{\pi i(\beta_{i+1}-\beta_{i+2})}\sin\pi\beta_{i+1}\sin\pi\beta_{i+2}$ appearing in (\ref{result3.75}), just as what happens if $\beta_{i+2}=-4/\kappa$ (\ref{applicationsc3}) instead.

\subsection{Proof of lemma \ref{mainlem}, item \ref{fourthcase}}\label{s4}

The purpose of this section is to complete the argument for item \ref{fourthcase} in the proof of lemma \ref{mainlem}.  We do this in two steps.  First, we determine the asymptotic behavior as $x_{i+1}\rightarrow x_i$ of the Coulomb gas integral $\mathcal{J}^{(N-1,2N)}(\boldsymbol{x})$ in (\ref{Fexplicit1}) for all $\kappa\in(0,8)$ with $8/\kappa\not\in\mathbb{Z}^+$.  In case \ref{sc4new}, $x_i$ and $x_{i-1}$ are endpoints of one contour $\Gamma_1''$ (replacing $\Gamma_1$ in item \ref{sc4} of this appendix), and $x_{i+1}$ and $x_{i+2}$ are endpoints of a different contour $\Gamma_2''$ (replacing $\Gamma_2$ in item \ref{sc4} of this appendix).  Thus,
\be\label{Gamma12''}\Gamma_1''=\begin{cases}\mathscr{P}(x_{i-1},x_i), & 0<\kappa\leq4\\
[x_{i-1},x_i]^+, & 4<\kappa<8\end{cases},\qquad\Gamma_2''=\begin{cases}\mathscr{P}(x_{i+1},x_{i+2}), & 0<\kappa\leq4\\
[x_{i+1},x_{i+2}]^+, & 4<\kappa<8\end{cases}.\ee 
(According to item \ref{i-1=0} in the proof of lemma \ref{mainlem}, we identify the index $i-1=0$ with $2N$.)  After we find the asymptotic behavior of $\mathcal{J}^{(N-1,2N)}(\boldsymbol{x})$, we use it to compute the limit $\bar{\ell}_1\mathcal{F}_\vartheta$ (\ref{thelim}), our second step.

Determining the behavior of $\mathcal{J}^{(N-1,2N)}(\boldsymbol{x})$ in this case is more delicate than determining this behavior in cases \ref{sc1} and \ref{sc2} for the reason given at the beginning of section \ref{s3}.  Just as in section \ref{s3}, we deform the contours $\Gamma_1''$ and $\Gamma_2''$ into contours that fall under cases \ref{sc1} and \ref{sc2} (figure \ref{case4alt}), and we apply the results of sections \ref{s1} and \ref{s2} to determine the asymptotic behavior of all of the new terms that thus appear.

To determine the asymptotic behavior of $\mathcal{J}^{(N-1,2N)}(\boldsymbol{x})$ as $x_{i+1}\rightarrow x_i$, we follow steps very similar to steps \ref{oneit}--\ref{lastit} of section \ref{s2}:
\begin{enumerate}
\item\label{oneit2}We order the integrations of $\mathcal{J}^{(N-1,2N)}(\boldsymbol{x})$ (\ref{Fexplicit1}) via Fubini's theorem so $u_1$ is integrated first, followed by $u_2$, followed by integration with respect to $u_3,$ $u_4,\ldots,u_{N-1}$.
\item\label{twoit2}The limit as $x_{i+1}\rightarrow x_i$ of $(x_{i+1}-u_m)^{-4/\kappa}$ is uniform over $u_m\in\Gamma_m$ only if $m>2$.  Hence, we set $x_{i+1}=x_i$ in every such factor in the integrand of $\mathcal{J}^{(N-1,2N)}(\boldsymbol{x})$ with $m>2$ to find its limit.
\item\label{threeit2} Thus, we only need to determine the behavior as $x_{i+1}\rightarrow x_i$ of the integration in $\mathcal{J}^{(N-1,2N)}(\boldsymbol{x})$ (\ref{Fexplicit1}) with respect to $u_1$ and $u_2$.  The corresponding definite integral is a function of  $x_1$, $x_2,\ldots,x_{2N}$, $u_3$, $u_4,\ldots,u_{N-1}$, and it is given by (\ref{eulerintegralch2}) with $M=2$ and $K=3N-3$ (with the integrand enclosed by the symbol $\mathcal{N}[\,\,\ldots\,\,]$) (\ref{scenario4}).
\item\label{fourit2} If any contour $\Gamma_m$ with $m>2$ and endpoints at $x_j<x_k$ passes over $\Gamma_1$ and $\Gamma_2$, then we replace it with two contours in the upper half-plane that do not pass over $\Gamma_1$ or $\Gamma_2$ and with orientation opposite that of $\Gamma_m$.  The first has its endpoints at $x_j$ and minus infinity, and the second has its endpoints at positive infinity and $x_k$.
\item After step \ref{fourit2}, no contour passes over $\Gamma_1$ or $\Gamma_2$.  Furthermore, each product $(x_i-u_q)^{-4/\kappa}(x_{i+1}-u_q)^{-4/\kappa}\Delta(u_p-u_q)$ with $p\in\{1,2\}$ and $q>2$ in the integrand of $\mathcal{J}^{(N-1,2N)}(\boldsymbol{x})$ now equals (figure \ref{Orderings}, (\ref{Deltadefn}))
\be\label{firstordering2}(x_i-u_q)^{-4/\kappa}(x_{i+1}-u_q)^{-4/\kappa}(u_p-u_q)^{8/\kappa}\ee
in the left new contour with an endpoint at minus infinity, and 
\be\label{secondordering2}(u_q-x_i)^{-4/\kappa}(u_q-x_{i+1})^{-4/\kappa}(u_q-u_p)^{8/\kappa}\ee
in the right new contour with an endpoint at plus infinity.  The ordering of the terms in the differences in (\ref{firstordering2}, \ref{secondordering2}) agrees with what the symbol $\mathcal{N}$ prescribes for pairs of un-nested contours (figure \ref{Orderings}, item \ref{4thitem} of definition \ref{Fkdefn}).
\item We note below (\ref{eulerintegralch2}) that the contours of $\mathcal{J}^{(N-1,2N)}(\boldsymbol{x})$ may intersect because $\gamma=8/\kappa>0$.  Thus, we push all integration contours flush against the real axis (except for the circular integrations of figure \ref{BreakDown}).
\item\label{secondlastit2} Restricting our attention to the integrations in $\mathcal{J}^{(N-1,2N)}(\boldsymbol{x})$ with respect to $u_1$ and $u_2$, we freeze the other integration variables $u_3$, $u_4,\ldots,u_{N-1}$ at arbitrary locations within their respective contours.
\item\label{lastit2} With all of the variables $x_1$, $x_2,\ldots,x_{2N}$, $u_3$, $u_4,\ldots,u_{N-1}$ real-valued, we re-index them in increasing order as $x_1<x_2<\ldots<x_K$, with $K=3N-3$ to simplify the notation of the integration with respect to $u_1$ and $u_2$.  (The order depends on where we freeze $u_m$ in its contour $\Gamma_m$.  Also, this re-indexing likely changes the value of the index $i$, as $x_i$ remains the left endpoint of the interval we are collapsing.  We note that if $i+1<2N$ before this re-indexing, as it does in (\ref{thelim}), then $i+1<K$ after this re-indexing.)
\end{enumerate}
Steps \ref{oneit2}--\ref{threeit2} show that we only need to determine the behavior of the integration with respect to $u_1$ and $u_2$ in order to find the asymptotic behavior of $\mathcal{J}^{(N-1,2N)}(\boldsymbol{x})$ (\ref{Fexplicit1}) as $x_{i+1}\rightarrow x_i$.  After steps \ref{fourit}--\ref{lastit}, this integration has the form
\begin{multline}\label{scenario4} 
\mathcal{J}^{(2,K)}\Big(\{\beta_j\};\gamma\,\Big|\,\,\Gamma_1'',\,\Gamma_2''\,\,\Big|\,x_1,x_2,\ldots,x_K\Big)\\
=\sideset{}{_{\Gamma_1''}}\int\sideset{}{_{\Gamma_2''}}\int\mathcal{N}\Bigg[\prod_{j=1}^K(u_1-x_j)^{\beta_j}(u_2-x_j)^{\beta_j}(u_2-u_1)^\gamma\Bigg]\,{\rm d}u_2\,{\rm d}u_1,\quad x_j\not\in(x_{i-1},x_{i+2}),
\end{multline}
with $K=3N-3$, $\Gamma_1$ and $\Gamma_2$ given in (\ref{Gamma12''}), and the symbol $\mathcal{N}[\,\,\ldots\,\,]$ ordering the terms of the differences in the function it encloses so this function is real-valued over the domain of integration.

In (\ref{scenario4}), we have relabeled each power $\gamma$ of the factors $(u_p-u_q)^\gamma$ in $\mathcal{J}^{(N-1,2N)}(\boldsymbol{x})$ with $p\in\{1,2\}$ and $q>2$ (\ref{eulerintegralch2}) as $\beta_j$ for some index $j$ for convenience.  After identifying the integration with respect to $u_1$ and $u_2$ in (\ref{Fexplicit1}) with (\ref{scenario4}), we find that
\be\label{application}\begin{gathered}s:=\sideset{}{_{j=1}^K}\sum\beta_j+\gamma=-2,\qquad\text{$\beta_j\in\{-4/\kappa,8/\kappa,12/\kappa-2\}$ for all $j\in\{1,2,\ldots,K\},$}\\
\beta_{i-1}=\beta_i=\beta_{i+1}=-4/\kappa,\qquad\beta_{i+2}\in\{-4/\kappa,12/\kappa-2\},\qquad\gamma=8/\kappa.\end{gathered}\ee
Initially, we assume $\kappa\in(4,8)$ in our calculations, and in this range, we automatically have $8/\kappa\not\in\mathbb{Z}^+$.  Combined with this restriction, (\ref{application}) implies the weaker conditions
\gdef\thesubequation{\theequation \textit{a,b,c}}
\begin{subeqnarray}
\label{extracond1}
&&s:=\sideset{}{_{j=1}^K}\sum\beta_j+\gamma\in\mathbb{Z}^-\setminus\{-1\},\qquad\beta_i+\beta_{i+1}+\gamma/2\not\in\mathbb{Z},\qquad \text{$\beta_j>-1$ for all $j\in\{1,2,\ldots,K\}$},\hspace{1cm}\\
\gdef\thesubequation{\theequation \textit{a,b,c,d}}
\refstepcounter{equation}
\label{extracond2}
&&\hspace{2.5cm}\beta_i+\beta_{i+1}<-1,\qquad\beta_i+\gamma/2=0,\qquad\beta_i+\beta_{i+1}\not\in\mathbb{Z},\qquad\beta_i=\beta_{i+1},\hspace{2cm}
\end{subeqnarray}
which we use exclusively in the rest of this section.  We note that the branch cuts of the integrand anchored to $x_1$, $x_2,\ldots,x_{K-1}$ collectively end at $x_K$ because $s\in\mathbb{Z}$.

Thus, our present goal is to determine the asymptotic behavior of $\mathcal{J}^{(2,K)}(\boldsymbol{x})$ (\ref{Gamma12''}, \ref{scenario4}) as $x_{i+1}\rightarrow x_i$, first for $\kappa\in(4,8)$ with conditions (\ref{extracond1}, \ref{extracond2}).  For this purpose, we define for each $m,k\in\{1,2,\ldots,K\}$ with $m\neq k$ the definite integral
\begin{multline}\label{Ijkintegrals}
I_{m,k}(x_1,x_2,\ldots,x_K):=\sideset{}{_{x_m}^{x_{m+1}}}\int\sideset{}{_{x_k}^{x_{k+1}}}\int\mathcal{N}\Bigg[\prod_{j=1}^K(u_1-x_j)^{\beta_j}(u_2-x_j)^{\beta_j}(u_2-u_1)^\gamma\Bigg]\,{\rm d}u_2\,{\rm d}u_1,\\
\text{with}\,\,\sideset{}{_{x_K}^{x_{K+1}}}\int\,\,:=\,\,\sideset{}{_{x_K}^\infty}\int\,\,+\,\,\sideset{}{_{-\infty}^{x_1}}\int\,\,\text{ if $m=K$ or $k=K$.}\end{multline}
(Conditions (\ref{extracond1}\red{a,b}) imply that any improper integral among (\ref{Ijkintegrals}) converges.)  Here, $I_{i-1,i+1}(\boldsymbol{x})$ is precisely the definite integral in (\ref{scenario4}) with $\Gamma_1''=[x_{i-1},x_i]^+$ and $\Gamma_2''=[x_{i+1},x_{i+2}]^+$ whose asymptotic behavior as $x_{i+1}\rightarrow x_i$ we wish to determine under conditions (\ref{extracond1}, \ref{extracond2}).  (If $i=1$, then we identify the index $m,k=i-1=0$ with $K$.  We consider this case separately in item \ref{itemi=1} in section \ref{s53} below.)  For $m=k$, we define
\begin{multline}\label{Ijjintegrals}I_{k,k}(x_1,x_2,\ldots,x_K):=\sideset{}{_{x_{k-1}}^{x_k}}\int\sideset{}{_{x_{k-1}}^{u_1}}\int\mathcal{N}\Bigg[\prod_{j=1}^K(u_1-x_j)^{\beta_j}(u_2-x_j)^{\beta_j}(u_1-u_2)^\gamma\Bigg]\,{\rm d}u_2\,{\rm d}u_1\\
\text{with}\,\,\sideset{}{_{x_K}^{x_{K+1}}}\int\sideset{}{_{x_K}^{u_1}}\int\,\,:=\,\,\lim_{R\rightarrow\infty}\Bigg[\sideset{}{_{x_K}^R}\int\sideset{}{_{x_K}^{u_1}}\int\,\,+\,\,\sideset{}{_{-R}^{x_1}}\int\Bigg(\sideset{}{_{x_K}^R}\int\,\,+\,\,\sideset{}{_{-R}^{u_1}}\int\Bigg)\Bigg]\,\,\text{ if $k=K$.}\end{multline}
Thanks to the ordering symbol $\mathcal{N}[\,\,\ldots\,\,]$, the integrand $f(u_1,u_2)$ of (\ref{Ijjintegrals}), as a function of $u_1$ and $u_2$, has the symmetry property $f(u_1,u_2)=f(u_2,u_1)$ for all $(u_1,u_2)\in[x_k,x_{k+1}]^2$.  Thus, we also have
\begin{multline}\label{Ijjintegrals2}I_{k,k}(x_1,x_2,\ldots,x_K)=\sideset{}{_{x_{k-1}}^{x_k}}\int\sideset{}{_{u_1}^{x_k}}\int\mathcal{N}\Bigg[\prod_{j=1}^K(u_1-x_j)^{\beta_j}(u_2-x_j)^{\beta_j}(u_1-u_2)^\gamma\Bigg]\,{\rm d}u_2\,{\rm d}u_1\\
\text{with}\,\,\sideset{}{_{x_K}^{x_{K+1}}}\int\sideset{}{_{u_1}^{x_{K+1}}}\int\,\,:=\,\,\lim_{R\rightarrow\infty}\Bigg[\sideset{}{_{x_K}^R}\int\Bigg(\sideset{}{_{u_1}^R}\int\,\,+\,\,\sideset{}{_{-R}^{x_1}}\int\Bigg)\,\,+\,\,\sideset{}{_{-R}^{x_1}}\int\sideset{}{_{u_1}^{x_1}}\int\Bigg]\,\,\text{ if $k=K$.}\end{multline}

Now, to determine the asymptotic behavior of $I_{i-1,i+1}(\boldsymbol{x})$ (\ref{scenario4}) as $x_{i+1}\rightarrow x_i$ under conditions (\ref{extracond1}, \ref{extracond2}), we follow the strategy used in section \ref{s3} and express this quantity as a linear combination of definite integrals that fall under cases \ref{sc1} and \ref{sc2} of sections \ref{s1} and \ref{s2} respectively.  To this end, we replace the integration contour $[x_{k+1},x_{k+2}]$ in (\ref{Ijkintegrals}) with $m=i-1$ by a large semicircle of radius $R$, with counterclockwise (resp.\ clockwise) orientation in the upper (resp.\ lower) half-plane, and with its base on the real axis.  By the Cauchy integral theorem, the integration around this semicircle gives zero.  Furthermore, the integration along the arc of the semicircle behaves like $R^{s+1}$ and thus vanishes as $R\rightarrow\infty$, thanks to (\ref{extracond1}\red{a}).  Similar to (\ref{case3int}), we find the two equations (with the $-$ (resp.\ $+$) sign corresponding with the upper (resp.\ lower) half-plane setting)
\be\label{overunder}\sum_{k=1}^{i-2}e^{\pm\pi i\sum_{l=1}^k\beta_l}I_{i-1,k}+e^{\pm\pi i\sum_{l=1}^{i-1}\beta_l}(1+e^{\pm\pi i\gamma})I_{i-1,i-1}+\sum_{k=i}^Ke^{\pm\pi i(\sum_{l=1}^k\beta_l+\gamma)}I_{i-1,k}=0.\ee
The phase factors in (\ref{overunder}) arise from passing over or under the branch points $x_j$ of the function enclosed within the brackets of the integrand in (\ref{Ijkintegrals}) (figure \ref{phases}).  Using the two equations (\ref{overunder}) to solve for $I_{i-1,i+1}$ in terms of $I_{i-1,k}$ with $k\neq i\pm1$ (figure \ref{case4alt}), we find (recalling that $i+1<K$, thanks to step \ref{lastit2}, and $\beta_i+\beta_{i+1}+\gamma/2\not\in\mathbb{Z}$ (\ref{extracond1}\red{b}))
\begin{multline}\label{result1}I_{i-1,i+1}=\sum_{k=1}^{i-2}\frac{\sin\pi(\sum_{l=k+1}^{i-1}\beta_l+\gamma/2)}{\sin\pi(\beta_i+\beta_{i+1}+\gamma/2)}I_{i-1,k}\\
-\sum_{k=i+2}^K\frac{\sin\pi(\sum_{l=i}^k\beta_l+\gamma/2)}{\sin\pi(\beta_i+\beta_{i+1}+\gamma/2)}I_{i-1,k}-\frac{\sin\pi(\beta_i+\gamma/2)}{\sin\pi(\beta_i+\beta_{i+1}+\gamma/2)}I_{i-1,i}.\end{multline}
(We replace (\ref{result1}), and later (\ref{result3.55}--\ref{result4}), with new equations in section \ref{s53} below for $i\in\{1,2\}$.  This does not affect any of the following results.)

\begin{figure}[t]
\centering
\includegraphics[scale=0.27]{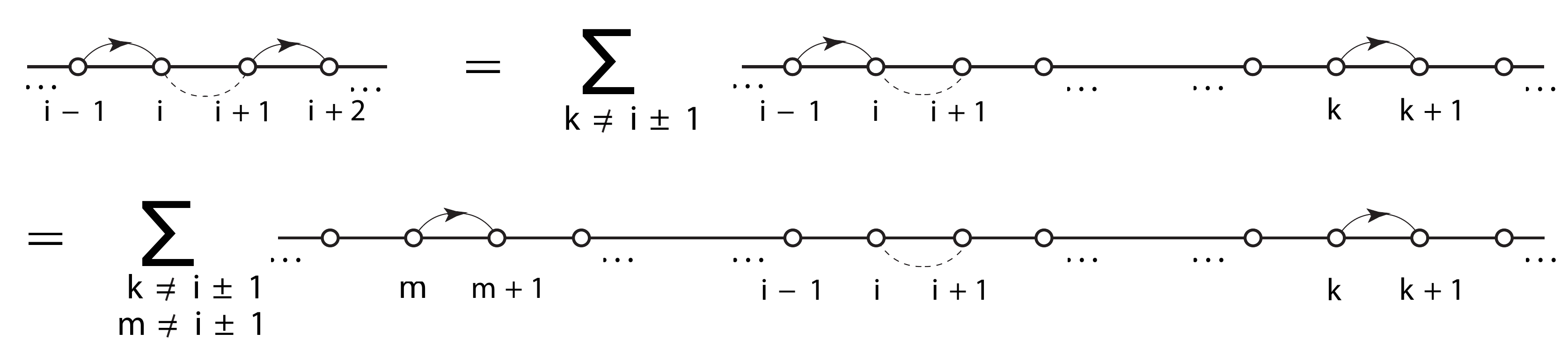}
\caption{The fourth case.  The dashed curve connects the endpoints of the interval to be collapsed.  We push the left (resp.\ right) integration contour from $[x_{i-1},x_i]$ (resp.\ $[x_{i+1},x_{i+2}]$) onto all intervals except $[x_{i+1},x_{i+2}]$ (resp.\ $[x_{i-1},x_i]$).}
\label{case4alt}
\end{figure}

Our goal is still to determine the behavior of the definite integral $I_{i-1,i+1}(\boldsymbol{x})$ as $x_{i+1}\rightarrow x_i$ under conditions (\ref{extracond1}, \ref{extracond2}).  We achieve this goal after expressing each $I_{i-1,k}$ with $k\neq i\pm1$ on the right side of (\ref{result1}) in terms of all $I_{m,k}$ with $m,k\neq i\pm1$.  Indeed, because these latter terms fall under cases \ref{sc1} and \ref{sc2}, we may use the results of sections \ref{s1} and \ref{s2} to determine their behavior as $x_{i+1}\rightarrow x_i$.  To this end, we replace the integration contour $[x_{i-1},x_i]^+$ of each term on the right side of (\ref{result1})  by a large semicircle of radius $R$, with counterclockwise (resp.\ clockwise) orientation in the upper (resp.\ lower) half-plane, and with its base on the real axis.  Again, the integration along the arc of the semicircle behaves like $R^{s+1}$ and thus vanishes as $R\rightarrow\infty$, thanks to (\ref{extracond1}\red{a}).  Therefore, after sending $R\rightarrow\infty$ and using the Cauchy integral theorem, we find the two equations
\be\label{result123}\sum_{m=1}^{k-1}e^{\pm\pi i\sum_{n=1}^m\beta_n}I_{m,k}+e^{\pm\pi i\sum_{n=1}^k\beta_n}(1+e^{\pm\pi i\gamma})I_{k,k}+\sum_{m=k+1}^Ke^{\pm\pi i(\sum_{n=1}^m\beta_n+\gamma)}I_{m,k}=0.\ee
Using the two equations (\ref{result123}), we solve for each $I_{i-1,k}$ with $k\neq i\pm1$ appearing on the right side of (\ref{result1}) in terms of $I_{m,k}$ with $m,k\neq i\pm1$ and substitute the result into (\ref{result1}) as desired (figure \ref{case4alt}).  With $k\neq i\pm1$, two cases arise.
\begin{enumerate}[I.] 
\item $k=i$: We find that $I_{i-1,i}$ equals a sum of many terms, each with a contour falling under case \ref{sc2}.  Because we have $\beta_i+\beta_{i+1}<-1$ (\ref{extracond2}\red{a}), all of these terms diverge as $x_{i+1}\rightarrow x_i$ thanks to (\ref{result2.5}), so simultaneously tracking all of them is difficult.  However, condition (\ref{extracond2}\red{b}) implies that the coefficient of $I_{i-1,i}$ in (\ref{result1}) vanishes, fortunately.
\item $k\not\in\{i,i\pm1\}$: Following the calculation in section \ref{s3} of appendix \ref{asymp}, we find that $I_{i-1,k}$ equals a sum of terms with $I_{m,k}$ and $m,k\not\in\{i,i\pm1\}$, falling under case \ref{sc1}, and a term with $I_{i,k}$ and $k\not\in\{i,i\pm1\}$, falling under case \ref{sc2}.  The appearance of this expression for $I_{i-1,k}$ is very similar to that for $I_{i-1}$ in (\ref{result}).  According to section \ref{s1}, the case \ref{sc1} terms are bounded in the limit $x_{i+1}\rightarrow x_i$, and according to result (\ref{result2.5}) of section \ref{s2}, the case \ref{sc2} term diverges in this limit because $\beta_i+\beta_{i+1}<1$ (\ref{extracond2}\red{a}).  Furthermore, this case \ref{sc2} term is identical to the last term on the right side of (\ref{result}) after we replace $I_i$ with $I_{i,k}$.  Therefore, we have
\be\label{condasymp}I_{i-1,k}(x_1,x_2,\ldots,x_K)\underset{x_{i+1}\rightarrow x_i}{\sim}-\frac{\sin\pi\beta_{i+1}}{\sin\pi(\beta_i+\beta_{i+1})}I_{i,k}(x_1,x_2,\ldots,x_K),\quad k\not\in\{i,i\pm1\}.\ee
Because $\beta_i+\beta_{i+1}\not\in\mathbb{Z}$ (\ref{extracond2}\red{c}), the denominator on the right side of (\ref{condasymp}) is not zero.
\end{enumerate}
The definite integral from $x_i$ to $x_{i+1}$ within $I_{i,k}(\boldsymbol{x})$ of (\ref{condasymp}) falls under case \ref{sc2}, so (\ref{result2.5}) gives its asymptotic behavior as $x_{i+1}\rightarrow x_i$.  After inserting the right side of (\ref{result2.5}) into (\ref{condasymp}), then inserting what results into the summations on the right side of (\ref{result1}) for each $k\not\in\{i,i\pm1\}$, and finally recalling that the term in (\ref{result1}) with $I_{i-1,i}$ vanishes, we find
\begin{multline}\label{result3.55}I_{i-1,i+1}(x_1,x_2,\ldots,x_K)\underset{x_{i+1}\rightarrow x_i}{\sim}\\
\begin{aligned}&-\frac{\sin\pi\beta_{i+1}\Gamma(\beta_i+1)\Gamma(\beta_{i+1}+1)}{\sin\pi(\beta_i+\beta_{i+1})\Gamma(\beta_i+\beta_{i+1}+2)}(x_{i+1}-x_i)^{\beta_i+\beta_{i+1}+1}\mathcal{N}\Bigg[\prod_{j\neq i,i+1}^K(x_i-x_j)^{\beta_j}\Bigg]\\
&\times\Bigg[\sum_{k=1}^{i-2}\frac{\sin\pi(\sum_{l=k+1}^{i-1}\beta_l+\gamma/2)}{\sin\pi(\beta_i+\beta_{i+1}+\gamma/2)}-\sum_{k=i+2}^K\frac{\sin\pi(\sum_{l=i}^k\beta_l+\gamma/2)}{\sin\pi(\beta_i+\beta_{i+1}+\gamma/2)}\Bigg]\\
&\times\,\,\sideset{}{_{x_k}^{x_{k+1}}}\int\mathcal{N}\Bigg[(u_2-x_i)^{\beta_i+\beta_{i+1}+\gamma}\prod_{j\neq i,i+1}^K(u_2-x_j)^{\beta_j}\Bigg]\,{\rm d}u_2,\end{aligned}\end{multline}
where so far we have assumed conditions (\ref{extracond1}\red{a}--\red{c}, \ref{extracond2}\red{a}--\red{c}) but not yet condition (\ref{extracond2}\red{d}).  (We revise (\ref{result3.55}) for $i\in\{1,2\}$ in section \ref{s53} below.)

Condition (\ref{extracond2}\red{d}) allows us to simplify our result (\ref{result3.55}) considerably.  First, after joining the intervals $[x_{i-1},x_i],$ $[x_i,x_{i+1}]$, and $[x_{i+1},x_{i+2}]$ together into one interval $[x_{i-1},x_{i+2}]$, we consider the collection of definite integrals
\begin{gather}\label{Ii+2'}I_{i-1}'(x_1,x_2,\ldots,x_{i-1},x_{i+2},\ldots,x_K):=\sideset{}{_{x_{i-1}}^{x_{i+2}}}\int\mathcal{N}\Bigg[\prod_{j\neq i,i+1}^K(u_2-x_j)^{\beta_j}\Bigg]\,{\rm d}u_2,\\
\label{Ik'}\begin{gathered}I_k'(x_1,x_2,\ldots,x_{i-1},x_{i+2},\ldots,x_K):=\sideset{}{_{x_k}^{x_{k+1}}}\int\mathcal{N}\Bigg[\prod_{j\neq i,i+1}^K(u_2-x_j)^{\beta_j}\Bigg]\,{\rm d}u_2,\\
 k\in\{1,2,\ldots,i-2,i+2,\ldots,K\}.\end{gathered}\end{gather}
(Again, we revise this equation and the following for $i\in\{1,2\}$ in section \ref{s53} below.)  The prime signifies that $I_{i-1}'$ and $I_k'$ are functions of only $x_1$, $x_2,\ldots,x_{i-1},$ $x_{i+2},\ldots,x_K$.  We note that the branch cuts of the integrand anchored to $x_1$, $x_2,\ldots,x_{K-1}$ collectively end at $x_K$ because
\be\label{sprime} s':=\sideset{}{_{j={i,i+1}}^K}\sum\beta_j\quad\Longrightarrow\quad s'=s-\beta_i-\beta_{i+1}-\gamma=s\in\mathbb{Z}^-\setminus\{-1\},\ee
thanks to conditions (\ref{extracond1}\red{a}, \ref{extracond2}\red{b,d}).  Also, the bottom line of (\ref{Ijkintegrals}) defines the integration for $I_K$.

Now we find a useful expression for the definite integral $I_{i-1}'$ in terms of the $I_k'$ with $k\neq i-1$.  Following the derivation of (\ref{case3int}), we replace the integration contour in (\ref{Ii+2'}, \ref{Ik'}) by a large semicircle of radius $R$, with counterclockwise (resp.\ clockwise) orientation in the upper (resp.\ lower) half-plane, and with its base on the real axis.  The integration along the arc of the semicircle behaves like $R^{s'+1}$ and thus vanishes as $R\rightarrow\infty$, thanks to (\ref{sprime}).  Therefore, after sending $R\rightarrow\infty$ and using the Cauchy integral theorem, we find the two equations
\be\label{intabovebelow} A^\pm:=\sideset{}{'}\sum_{k=1}^K e^{\pm\pi i\sum_{l=1}^{'k}\beta_l}I_{k}'=0,\ee
where the prime indicates summation over all indices $k,l\not\in\{i,i+1\}$.  Now after isolating the definite integral $I_{i-1}'$ from the linear combination
\be\label{Alincmb} e^{-\pi i\sum_{l=1}^{i-1}\beta_l}e^{\pi i(\beta_i+\beta_{i+1}+\gamma/2)}A^+-e^{\pi i\sum_{l=1}^{i-1}\beta_l}e^{-\pi i(\beta_i+\beta_{i+1}+\gamma/2)}A^-=0,\ee
we find (here, the bottom line of (\ref{Ijkintegrals}) defines integration from $x_K$ to $x_{K+1}$ and condition (\ref{extracond1}\red{b}) says that $\beta_i+\beta_{i+1}+\gamma/2\not\in\mathbb{Z}$)
\begin{multline}\label{primed}I_{i-1}'(x_1,x_2,\ldots,x_{i-1},x_{i+2},\ldots,x_K)=\\
\Bigg[\sum_{k=1}^{i-2}\frac{\sin\pi(\sum_{l=k+1}^{i-1}\beta_l-\beta_i-\beta_{i+1}-\gamma/2)}{\sin\pi(\beta_i+\beta_{i+1}+\gamma/2)}-\sum_{k=i+2}^K\frac{\sin\pi(\sum_{l=i}^k\beta_l+\gamma/2)}{\sin\pi(\beta_i+\beta_{i+1}+\gamma/2)}\Bigg]\\
\times\,\sideset{}{_{x_k}^{x_{k+1}}}\int\mathcal{N}\Bigg[\prod_{j\neq i,i+1}^K(u_2-x_j)^{\beta_j}\Bigg]\,{\rm d}u_2.\end{multline}
With $\beta_i=\beta_{i+1}=-\gamma/2$ thanks to (\ref{extracond2}\red{b,d}), the numerators of the fractions in the first summations of (\ref{primed}) and (\ref{result3.55}) agree.  Thus, the entire right side of (\ref{primed}) equals the difference of summations on the right side of (\ref{result3.55}).  After replacing the latter with the former, (\ref{result3.55}) becomes
\begin{multline}\label{result4}I_{i-1,i+1}(x_1,x_2,\ldots,x_K)\underset{x_{i+1}\rightarrow x_i}{\sim}-\frac{\sin\pi\beta_{i+1}\Gamma(\beta_i+1)\Gamma(\beta_{i+1}+1)}{\sin\pi(\beta_i+\beta_{i+1})\Gamma(\beta_i+\beta_{i+1}+2)}(x_{i+1}-x_i)^{\beta_i+\beta_{i+1}+1}\\
\times\,\mathcal{N}\Bigg[\prod_{j\neq i,i+1}^K(x_i-x_j)^{\beta_j}\Bigg]\sideset{}{_{x_{i-1}}^{x_{i+2}}}\int\mathcal{N}\Bigg[\prod_{j\neq i,i+1}^K(u_2-x_j)^{\beta_j}\Bigg]\,{\rm d}u_2,\quad \beta_i=\beta_{i+1},\end{multline}
where we have now assumed all conditions among (\ref{extracond1}\red{a}--\red{c}, \ref{extracond2}\red{a}--\red{d}).  We note that the prefactor of sine and Gamma functions on the right side of (\ref{result4}) equals that of (\ref{result3}).  After substituting $\beta_{i+1}=\beta_i$ in (\ref{result4}) and changing notation from $I_{i-1,i+1}$ to the notation $\mathcal{J}^{(2,K)}$ (\ref{eulerintegralch2}, \ref{scenario4}), we finally have (assuming conditions (\ref{extracond1}, \ref{extracond2}))
\begin{multline}\label{result4.5}\mathcal{J}^{(2,K)}\Big(\{\beta_j\};\gamma\,\Big|\,[x_{i-1},x_i]^+,[x_{i+1},x_{i+2}]^+\,\Big|\,x_1,x_2\ldots,x_K\Big)\\
\begin{aligned}&\underset{x_{i+1}\rightarrow x_i}{\sim}\frac{\Gamma(\beta_i+1)^2}{-2\cos\pi\beta_i\,\Gamma(2\beta_i+2)}(x_{i+1}-x_i)^{2\beta_i+1}\\
&\times\,\mathcal{N}\Bigg[\prod_{j\neq i,i+1}^K(x_i-x_j)^{\beta_j}\Bigg]\sideset{}{_{x_{i-1}}^{x_{i+2}}}\int\mathcal{N}\Bigg[\prod_{j\neq i,i+1}^K(u_2-x_j)^{\beta_j}\Bigg]\,{\rm d}u_2.\end{aligned}\end{multline}
Remarkably, sending $x_{i+1}\rightarrow x_i$ joins the contours $\Gamma_1''=[x_{i-1},x_i]^+$ and $\Gamma_2''=[x_{i+1},x_{i+2}]^+$ of (\ref{scenario4}) into the single contour $\Gamma_0':=[x_{i-1},x_{i+2}]^+$ of (\ref{result4.5}).  The points $x_i$ and $x_{i+1}$ do not participate in the definite integral that remains in (\ref{result4.5}).  Also, condition (\ref{application}) implies that the factor of $-2\cos\pi\beta_i$ in the denominator of (\ref{result4.5}) again equals $n(\kappa)$ (\ref{fugacity}), and its presence justifies the factors of $n(\kappa)^{-1}$ that appear in the bottom line of figure \ref{Cases} and on the bottom line of each bracketed collection in figure \ref{Case4}.

\begin{figure}[t!]
\centering
\includegraphics[scale=0.27]{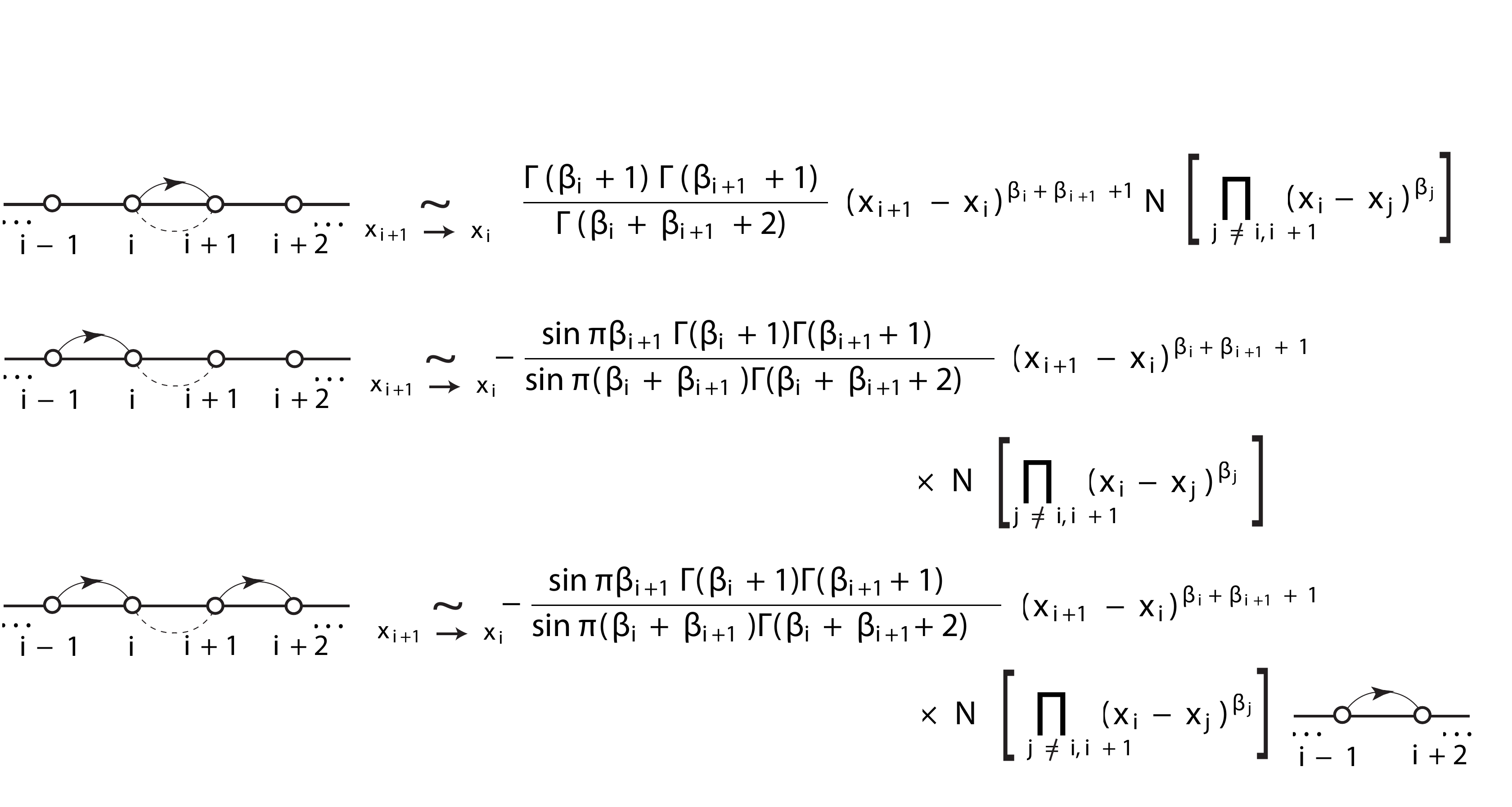}
\caption{Summary of cases \ref{sc2}, \ref{sc3new}, and \ref{sc4new}.  Case \ref{sc2} (top illustration) assumes $\beta_1+\beta_2+\dotsm+\beta_K\in\mathbb{Z}^-\setminus\{-1\}$.  Case \ref{sc3new} (middle illustration) assumes (\ref{extracond0}, \ref{extracond}).  Case \ref{sc4new} (bottom illustration) assumes (\ref{extracond1}, \ref{extracond2}).}
\label{Summary}
\end{figure}

Figure \ref{Summary} summarizes the asymptotic behaviors of the definite integrals studied in cases \ref{sc2}, \ref{sc3new}, and \ref{sc4new}.

Now we extend our findings from $\kappa\in(4,8)$ to $\kappa\in(0,8)$ with $8/\kappa\not\in\mathbb{Z}^+$.  If $\kappa\in(0,4]$, then (\ref{application}) implies that $\beta_k\leq-1$ for several $k\in\{1,2,\ldots,K\}$, so the improper integral $I_{m,k}$ (\ref{Ikintegrals}) diverges.  The end of section \ref{s3} presents the remedy.  To summarize, we replace condition (\ref{extracond1}\red{c}) with (\ref{extracondmod}) first.  Then in (\ref{result1}), we alter both integrations of each $I_{m,k}$ as (\ref{analyticcont}) shows.  In particular, if $\kappa\in(0,4]$ so $\beta_i,\beta_{i+1}\leq-1$ (\ref{application}), then after solving (this analytic continuation of) (\ref{result1}) for $I_{i-1,i+1}$ (which is also (\ref{scenario4}) with $\Gamma''_1=\mathscr{P}(x_{i-1},x_i)$ and $\Gamma_2''=\mathscr{P}(x_{i+1},x_{i+2})$), we find
\begin{multline}\label{result4.75}\mathcal{J}^{(2,K)}\Big(\{\beta_j\};\gamma\,\Big|\,\mathscr{P}(x_{i-1},x_i),\mathscr{P}(x_{i+1},x_{i+2})\,\Big|\,x_1,x_2\ldots,x_K\Big)\underset{x_{i+1}\rightarrow x_i}{\sim}4\sin^2\pi\beta_i\,\frac{\Gamma(\beta_i+1)^2}{-2\cos\pi\beta_i\,\Gamma(2\beta_i+2)}\\
\times\,(x_{i+1}-x_i)^{2\beta_i+1}\mathcal{N}\Bigg[\prod_{j\neq i,i+1}^K(x_i-x_j)^{\beta_j}\Bigg]\sideset{}{_{\mathscr{P}(x_{i-1},x_{i+2})}}\oint\mathcal{N}\Bigg[\prod_{j\neq i,i+1}^K(u_2-x_j)^{\beta_j}\Bigg]\,{\rm d}u_2,\end{multline}
assuming conditions (\ref{extracond1}, \ref{extracond2}).

To finish, we use (\ref{result4.5}, \ref{result4.75}) to justify item \ref{fourthcase} in the proof of lemma \ref{mainlem}.  In that part of the proof, we require the limit $\bar{\ell}_1\mathcal{F}_\vartheta$ (\ref{thelim}) for $\kappa\in(0,8)$ and $8/\kappa\not\in\mathbb{Z}^+$, with $\mathcal{F}_\vartheta$ given by (\ref{Fexplicit1}) and $\Gamma_1$ (resp.\ $\Gamma_2$) replaced by $\Gamma_1''$ (resp.\ $\Gamma_2''$) which is given in (\ref{Gamma12''}) (so the sine functions drop from the prefactor (\ref{firstprefactor}) in (\ref{Fexplicit1}) if $\kappa>4$).  For $\kappa\leq4$, this limit is 
\begin{multline}\label{thebiglimit2}\bar{\ell}_1\big(\mathcal{F}_\vartheta\big|_{(\Gamma_1,\Gamma_2)\mapsto(\Gamma_1'',\Gamma_2'')}\big)\,(\kappa\,|\,x_1,x_2,\ldots,x_{i-1},x_{i+2},\ldots,x_{2N})\,\,=\lim_{x_{i+1}\rightarrow x_i}(x_{i+1}-x_i)^{6/\kappa-1}\,\,\times\\ 
(\mathcal{F}_\vartheta|_{(\Gamma_1,\Gamma_2)\mapsto(\Gamma_1'',\Gamma_2'')})\,(\kappa\,|\,\boldsymbol{x}) \\ \parallel \\
\boxed{\begin{aligned}&n(\kappa)\left[\frac{n(\kappa)\Gamma(2-8/\kappa)}{4\sin^2(4\pi/\kappa)\Gamma(1-4/\kappa)^2}\right]^{N-1}\Bigg(\prod_{\substack{1\leq j<k\\j,k\neq i,i+1}}^{2N-1}(x_k-x_j)^{2/\kappa}\Bigg)\Bigg(\prod_{\substack{k=1 \\ k\neq i,i+1}}^{2N-1}(x_{2N}-x_k)^{1-6/\kappa}\Bigg)\sideset{}{_{\Gamma_{N-1}}}\oint {\rm d}u_{N-1}\\
&\dotsm\,\sideset{}{_{\Gamma_4}}\oint {\rm d}u_4\,\, \sideset{}{_{\Gamma_3}}\oint {\rm d}u_3\,\,\mathcal{N}\Bigg[\Bigg(\prod_{l\neq i,i+1}^{2N-1}\prod_{m=3}^{N-1}(x_l-u_m)^{-4/\kappa}\Bigg)\Bigg(\prod_{m=3}^{N-1}(x_{2N}-u_m)^{12/\kappa-2}\Bigg)\Bigg(\prod_{3\leq p<q}^{N-1}(u_p-u_q)^{8/\kappa}\Bigg)\\
&\Bigg(\prod_{m=3}^{N-1}(x_i-u_m)^{-4/\kappa}(x_{i+1}-u_m)^{-4/\kappa}\Bigg)\Bigg]
\Bigg(\prod_{j\neq i,i+1}^{2N-1}|x_j-x_i|^{2/\kappa}|x_j-x_{i+1}|^{2/\kappa}\Bigg)\Big(x_{i+1}-x_i\Big)^{2/\kappa}\\
&\begin{aligned}&\Big(x_{2N}-x_i\Big)^{1-6/\kappa} \\ &\,\,\,\Big(x_{2N}-x_{i+1}\Big)^{1-6/\kappa}\end{aligned}\underbrace{\left\{\begin{aligned}
&\sideset{}{_{\mathscr{P}(x_{i+1},x_{i+2})}}\oint {\rm d}u_2\,\,\sideset{}{_{\mathscr{P}(x_{i-1},x_i)}}\oint {\rm d}u_1
\,\,\mathcal{N}\Bigg[\Bigg(\prod_{l=1}^{2N-1}(x_l-u_1)^{-4/\kappa}(x_l-u_2)^{-4/\kappa}\Bigg)\Big(x_{2N}\\
&-u_1\Big)^{12/\kappa-2}\Big(x_{2N}-u_2\Big)^{12/\kappa-2}\Big(u_2-u_1\Big)^{8/\kappa}\Bigg(\prod_{m=3}^{N-1}(u_m-u_1)^{8/\kappa}(u_m-u_2)^{8/\kappa}\Bigg)\Bigg]\end{aligned}\right\}.}_{\mathcal{J}^{(2,K)}}\end{aligned}}\end{multline}
(If $\kappa>4$, then we adjust (\ref{thebiglimit2}) as per item \ref{itemc} in the introduction of this appendix.)  We have rewritten the formula (\ref{Fexplicit1}) for $\mathcal{F}_\vartheta(\kappa\,|\,\boldsymbol{x})$ slightly to clarify the calculation, and we indicate the contour integral $\mathcal{J}^{(2,K)}$ (\ref{scenario4}) with braces.

Now we find the limit (\ref{thebiglimit2}).  After doing steps \ref{oneit2}--\ref{lastit2}, setting $x_{i+1}=x_i$ in all factors of (\ref{Fexplicit1}) without $u_1$ or $u_2$, identifying the double integral with respect to $u_1$ and $u_2$ with (\ref{Gamma12''}, \ref{scenario4}, \ref{application}), and replacing it by the right side of (\ref{result4.75}), we find 
\begin{multline}\label{asympinsert2}(x_{i+1}-x_i)^{6/\kappa-1}\big(\mathcal{F}_\vartheta\big|_{(\Gamma_1,\Gamma_2)\mapsto(\Gamma_1'',\Gamma_2'')}\big)\,(\kappa\,|\,\boldsymbol{x})\underset{x_{i+1}\rightarrow x_i}{\sim}(x_{i+1}-x_i)^{6/\kappa-1}\,\,\times\\ 
\boxed{\begin{aligned}&n(\kappa)\left[\frac{n(\kappa)\Gamma(2-8/\kappa)}{4\sin^2(4\pi/\kappa)\Gamma(1-4/\kappa)^2}\right]^{N-1}\Bigg(\prod_{\substack{1\leq j<k\\j,k\neq i,i+1}}^{2N-1}(x_k-x_j)^{2/\kappa}\Bigg)\Bigg(\prod_{\substack{k=1 \\ k\neq i,i+1}}^{2N-1}(x_{2N}-x_k)^{1-6/\kappa}\Bigg)\sideset{}{_{\Gamma_{N-1}}}\oint {\rm d}u_{N-1}\\
&\dotsm\,\sideset{}{_{\Gamma_4}}\oint {\rm d}u_4\,\, \sideset{}{_{\Gamma_3}}\oint {\rm d}u_3\,\,\mathcal{N}\Bigg[\Bigg(\prod_{l\neq i,i+1}^{2N-1}\prod_{m=3}^{N-1}(x_l-u_m)^{-4/\kappa}\Bigg)\Bigg(\prod_{m=3}^{N-1}(x_{2N}-u_m)^{12/\kappa-2}\Bigg)\Bigg(\prod_{3\leq p<q}^{N-1}(u_p-u_q)^{8/\kappa}\Bigg)\\
&\Bigg(\prod_{m=3}^{N-1}(x_i-u_m)^{-8/\kappa}\Bigg)\Bigg]\Bigg(\prod_{j\neq i,i+1}^{2N-1}|x_j-x_i|^{4/\kappa}\Bigg)\Big(x_{i+1}-x_i\Big)^{2/\kappa}\,\,\,\overbrace{\frac{4\sin^2(4\pi/\kappa)\Gamma(1-4/\kappa)^2}{n(\kappa)\Gamma(2-8/\kappa)}}^{\text{right side of (\ref{result4.75})}}\\
&\Big(x_{2N}-x_i\Big)^{2-12/\kappa}\underbrace{\left\{\begin{aligned}\Big(x_{i+1}-x_i\Big)^{1-8/\kappa}&\mathcal{N}\Bigg[\Bigg(\prod_{j\neq i,i+1}^{2N-1}(x_j-x_i)^{-4/\kappa}\Bigg)\Big(x_{2N}-x_i\Big)^{12/\kappa-2}\Bigg(\prod_{m=3}^{N-1}(x_i-u_m)^{8/\kappa}\Bigg)\Bigg]\\
\sideset{}{_{\mathscr{P}(x_{i-1},x_{i+2})}}\oint {\rm d}u_2\,\,
&\mathcal{N}\Bigg[\Big(x_{2N}-u_2\Big)^{12/\kappa-2}\Bigg(\prod_{l=1}^{2N-1}(x_l-u_2)^{-4/\kappa}\Bigg)\Bigg(\prod_{m=3}^{N-1}(u_m-u_2)^{8/\kappa}\Bigg)\hspace{.15cm}\Bigg]\end{aligned}\right\}}_{\text{right side of (\ref{result4.75})}}.\end{aligned}}\end{multline}
(If $\kappa>4$ and we adjust (\ref{thebiglimit2}) as per item \ref{itemc} in the introduction of this appendix, then we replace the double integral with respect to $u_1$ and $u_2$ in (\ref{thebiglimit2}) with the right side of (\ref{result4.5}) instead, finding (\ref{asympinsert2}) again, but with factors of $4\sin^2(4\pi/\kappa)$ dropped and all contours simple.)  After some simplification, we finally send $x_{i+1}\rightarrow x_i$ in (\ref{asympinsert2}) to find
\begin{multline}\label{firstlimprime2}\bar{\ell}_1\big(\mathcal{F}_\vartheta\big|_{(\Gamma_1,\Gamma_2)\mapsto(\Gamma_1'',\Gamma_2'')}\big)\,(\kappa\,|\,x_1,x_2,\ldots,x_{i-1},x_{i+2},\ldots,x_{2N})\,\,=\\
\left\{\begin{aligned}&n(\kappa)\left[\frac{n(\kappa)\Gamma(2-8/\kappa)}{4\sin^2(4\pi/\kappa)\Gamma(1-4/\kappa)^2}\right]^{N-2}\,\,\Bigg(\prod_{\substack{1\leq j<k \\ j,k\neq i,i+1}}^{2N-1}(x_k-x_j)^{2/\kappa}\Bigg)\Bigg(\prod_{\substack{k=1 \\ k\neq i,i+1}}^{2N-1}(x_{2N}-x_k)^{1-6/\kappa}\Bigg)\oint_{\Gamma_{N-1}} {\rm d}u_{N-1}\dotsm\\ 
&\oint_{\Gamma_3}{\rm d}u_3\,\,\sideset{}{_{\hspace{1cm}\mathclap{\mathscr{P}(x_{i-1},x_{i+2})}}}\oint\hspace{.5cm}{\rm d}u_2\,\,\mathcal{N}\Bigg[\Bigg(\prod_{l\neq i,i+1}^{2N-1}\prod_{m=2}^{N-1}(x_l-u_m)^{-4/\kappa}\Bigg)\Bigg(\prod_{m=2}^{N-1}(x_{2N}-u_m)^{12/\kappa-2}\Bigg)\Bigg(\prod_{2\leq p<q}^{N-1}(u_p-u_q)^{8/\kappa}\Bigg)\Bigg]\end{aligned}\right\}.\end{multline}
(Again, if $\kappa>4$, then we find the same result, but with the factors of $4\sin^2(4\pi/\kappa)$ dropped and with the Pochhammer contours replaced by simple contours with the same endpoints.)  

Now, the quantity of (\ref{firstlimprime2}) in brackets is the element of $\mathcal{B}_{N-1}$ described in the following conclusion.  We thus have our main result, for use in item \ref{fourthcase} in the proof of lemma \ref{mainlem}.  (Again, we assume here that $8/\kappa\in\mathbb{Z}^+$, but the result remains true even if this condition is not met.  See the last paragraph of this section.)
\begin{quote}\textbf{Main result:} In case \ref{sc4new}, where $\Gamma_1''$ and $\Gamma_2''$ (\ref{Gamma12''}) respectively replace $\Gamma_1$ and $\Gamma_2$ in the formula (\ref{Fexplicit1}) for $\mathcal{F}_\vartheta$, the limit $\bar{\ell}_1\mathcal{F}_\vartheta$ (\ref{thelim}) equals the element of $\mathcal{S}_{N-1}$ generated from the formula for $\mathcal{F}_\vartheta$ by dropping all factors involving $x_i$, $x_{i+1}$, and $u_1$, dropping the integration along $\Gamma_1$, replacing $\Gamma_2$ by $\mathscr{P}(x_{i-1},x_{i+2})$ (or $[x_{i-1},x_{i+2}]^+$ if $\kappa>4$, as per item \ref{itemc} in the introduction of this appendix), and reducing the power of the prefactor (\ref{firstprefactor}) or (\ref{secondprefactor}) by one.
\end{quote}
As previously noted, (\ref{result1}, \ref{result3.55}--\ref{result4.5}) require revision if $i\in\{1,2\}$, so this main result remains to be proved if $i=1$ or $i=2$.  We provide this proof in items \ref{itemi=1} and \ref{itemi=2} of section \ref{s53} below.

So far, the analysis of sections \ref{s3} and \ref{s4} assumes that $8/\kappa\not\in\mathbb{Z}^+$, but as we show below (\ref{LkFk}), the main results of sections \ref{s1}--\ref{s4} (and \ref{s5} below) remain true even if this condition is not met.  This situation sorts into two groups, those with $4/\kappa\in\mathbb{Z}^+$, and those with $4/\kappa\not\in\mathbb{Z}^+$.  If $4/\kappa\in\mathbb{Z}^+$, then there is no reason to adapt the analysis of this appendix because we have algebraic formulas (\ref{explicit}) for the elements of $\mathcal{B}_N$. The required limits then follow straightforwardly.  If $4/\kappa\not\in\mathbb{Z}^+$, then we may adapt the analysis of sections \ref{s2} and \ref{s3} to this situation.  We do this in appendix \red{A} of \cite{florkleb4} as part of the proof of a theorem regarding Frobenius series expansions for elements of $\mathcal{S}_N$.

\subsection{Proof of corollary \ref{moveconjchargecor}, items \ref{firstcase2}--\ref{fourthcase2}}\label{s5}

The purpose of this section is to complete the argument for items \ref{firstcase2}--\ref{fourthcase2} in the proof of lemma \ref{mainlem} and extend some of the analysis in sections \ref{s3} and \ref{s4} to $i\in\{1,2,2N-1\}$.  As usual, we do this in two steps. First, we determine the asymptotic behavior of the Coulomb gas integral $\mathcal{J}^{(N-1,2N)}(\boldsymbol{x})$ in $\mathcal{F}_{c,\vartheta}(\kappa\,|\,\boldsymbol{x})$ (\ref{firstFexplicit1}) as $x_{i+1}\rightarrow x_i$ for $i\in\{1,2,\ldots,2N-1\}\setminus\{c-1,c\}$ or as $-x_1=x_{2N}=R\rightarrow\infty$ for $c\in\{1,2,\ldots,2N\}$, with $\kappa\in(0,8)$ and $8/\kappa\not\in\mathbb{Z}^+$.  After we find this behavior, we use it to compute the limits $\bar{\ell}_1\mathcal{F}_{c,\vartheta}$ (\ref{nextthelim}) and $\underline{\ell}_1\mathcal{F}_{c,\vartheta}$ (\ref{inflim}), our second step.  We need these limits for the proof of corollary \ref{moveconjchargecor} (and we need the limit $\bar{\ell}_1\mathcal{F}_{c,\vartheta}$ for item \ref{thirdcase} (resp.\ item \ref{fourthcase}) in the proof of lemma \ref{mainlem} if $c=2N$ and $i=2$ (resp.\ $i\in\{1,2\}$)).  Throughout this section, we assign new values to any indices outside the range $1$, $2,\ldots,2N$ using modular $2N$ arithmetic.  For example, if $i=1$ (resp.\ $i=2N$), then we assign $i-1$ the value $2N$ (resp.\ $i+1$ the value one).  In particular, if $i=2N$, then sending $x_{i+1}\rightarrow x_i$ as in (\ref{nextthelim}) becomes $-x_1=x_{2N}=R\rightarrow\infty$ as in (\ref{inflim}).

For the proof of corollary \ref{moveconjchargecor}, we determine the limit $\bar{\ell}_1\mathcal{F}_{c,\vartheta}$ (\ref{nextthelim}) in cases \ref{sc1}, \ref{sc2}, \ref{sc3new}, and \ref{sc4new} listed in the introduction of this appendix.  By repeating the analysis of sections \ref{s1} and \ref{s2}, we find the same main results for this limit in cases \ref{sc1} and \ref{sc2} respectively.  And by repeating the analysis of sections \ref{s3}  and \ref{s4} with $i\in\{3,4,\ldots,2N-2\}\setminus\{c-1,c\}$, we find the same main results for this limit in cases \ref{sc3new} and \ref{sc4new} respectively.  However, if $i\in\{1,2,2N-1\}$, then some equations in the analysis of cases \ref{sc3new} and \ref{sc4new} are different.  Although slight, these differences are too significant to overlook altogether and too involved to include in the previous sections.  Hence, we have postponed their consideration up to now.  This includes the case with $c=2N$ and $i=2$ (resp.\ $i\in\{1,2\}$) required for the main result of section \ref{s3} (resp.\ section \ref{s4}).

Also for the proof of corollary \ref{moveconjchargecor}, we determine the limit $\underline{\ell}_1\mathcal{F}_{c,\vartheta}$ (\ref{nextthelim}) in cases \ref{sc1}, \ref{sc2}, \ref{sc3new}, and \ref{sc4new} listed in the introduction of this appendix.  None of the previous sections in this appendix consider this limit, and although its calculation is similar to that of $\bar{\ell}_1\mathcal{F}_{c,\vartheta}$, it is different enough that we show it here.

\subsubsection{Proof of corollary \ref{moveconjchargecor}, item \ref{firstcase2}}\label{s50}

As we mentioned in the introduction of this section, the analysis of section \ref{s1} with $i\in\{1,2,\ldots,2N-1\}\setminus\{c-1,c\}$ immediately gives this result, for use in the proof of corollary \ref{moveconjchargecor}.
\begin{quote}\textbf{Main result 1:} In case \ref{sc1} with $i\in\{1,2,\ldots,2N-1\}\setminus\{c-1,c\}$, where we replace all integration contours with endpoints at $x_i$ or $x_{i+1}$ by contours that do not surround or touch these points in the formula (\ref{firstFexplicit1}) for $\mathcal{F}_{c,\vartheta}$, the limit $\bar{\ell}_1\mathcal{F}_{c,\vartheta}$ (\ref{nextthelim}) vanishes.
\end{quote}
This result includes the main result of section \ref{s2} as the special case $c=2N$, and it is similarly proven.  On the other hand, we may set $-x_1=x_{2N}=R$ in the formula (\ref{firstFexplicit1}) for $\mathcal{F}_{c,\vartheta}(\kappa\,|\,\boldsymbol{x})$ to find that this quantity times $(2R)^{6/\kappa-1}$ is $O(R^{1-8/\kappa})$ as $R\rightarrow\infty$.   Thus, we have the following result.
\begin{quote}\textbf{Main result 2:} In case \ref{sc1}, where we replace all integration contours with endpoints at $x_{2N}$ or $x_1$ by contours that do not surround or touch these points in the formula (\ref{firstFexplicit1}) for $\mathcal{F}_{c,\vartheta}$, the limit $\underline{\ell}_1\mathcal{F}_{c,\vartheta}$ (\ref{nextthelim}) vanishes.
\end{quote}
We use this second main result to repeat the analysis of sections \ref{s2}--\ref{s4} and to find the second main results of sections \ref{s51}--\ref{s53}.

\subsubsection{Proof of corollary \ref{moveconjchargecor}, item \ref{secondcase2}}\label{s51}

As we mentioned in the introduction of this section, the analysis of section \ref{s2} with $i\in\{1,2,\ldots,2N-1\}\setminus\{c-1,c\}$ immediately gives this result, for use in the proof of corollary \ref{moveconjchargecor}.
\begin{quote}\textbf{Main result 1:} In case \ref{sc2} with $i\in\{1,2,\ldots,2N-1\}\setminus\{c-1,c\}$, where $\Gamma_1$ is given by (\ref{Gamma1}), the limit $\bar{\ell}_1\mathcal{F}_{c,\vartheta}$ (\ref{nextthelim}) equals $n$ (\ref{fugacity}) times the element of $\mathcal{S}_{N-1}$ generated from the formula (\ref{firstFexplicit1}) for $\mathcal{F}_{c,\vartheta}$ by dropping all factors involving $x_i$, $x_{i+1}$, and $u_1$, dropping the integration along $\Gamma_1$, and reducing the power of the prefactor (\ref{firstprefactor}) or (\ref{secondprefactor}) by one.
\end{quote}
This result includes the main result of section \ref{s2} as the special case $c=2N$, and it is similarly proven.

We proceed as follows.  First, we determine the asymptotic behavior as $-x_1=x_{2N}=R\rightarrow\infty$ of the Coulomb gas integral $\mathcal{J}^{(N-1,2N)}(\boldsymbol{x})$ in $\mathcal{F}_{c,\vartheta}(\kappa\,|\,\boldsymbol{x})$ (\ref{firstFexplicit1}) for all $\kappa\in(0,8)$, with $x_1$ and $x_{2N}$ as endpoints of only one common integration contour $\Gamma_1$.  And after we find this behavior, we use it to compute the limit $\underline{\ell}_1\mathcal{F}_{c,\vartheta}$ (\ref{inflim}).  We may interpret the forthcoming analysis as an extension of that in section \ref{s2} to the value $i=2N$.

To begin, we execute steps \ref{oneit}--\ref{lastit} of section \ref{s2}.  (After re-indexing the coordinates as in step \ref{lastit}, $x_K$ becomes the rightmost coordinate of $\boldsymbol{x}\in\Omega_0$, previously called $x_{2N}$.)  This leaves us with finding the asymptotic behavior of $\mathcal{J}^{(1,K)}(\boldsymbol{x})$ (\ref{I1here}) with $K=3N-2$ as $-x_1=x_K=R\rightarrow\infty$.    Thus, the integration contour $\Gamma_1$ is 
\be\label{Gamma1i=2N}\Gamma_1=\begin{cases}\mathscr{P}(R,-R), & 0<\kappa\leq4\\
 [R,-R]^+, & 4<\kappa<8\end{cases},\ee  
and we bend it into the shape of a semicircular arc in the upper half-plane with endpoints at $-R$ and $R$.  Because it passes over $x_c$, $\Gamma_1$ has leftward orientation (item \ref{2cit} of definition \ref{Fkdefn}).  After identifying the contour integral with respect to $u_1$ in (\ref{firstFexplicit1}) with (\ref{I1here}), we find that (in the equation number, we assign the letters to the conditions by going from left to right, then down one line, and then from left to right again)
\gdef\thesubequation{\theequation \textit{a,b,c}}
\begin{subeqnarray}
\label{betai2}
&&\begin{gathered}\text{$\beta_j\in\{-4/\kappa,8/\kappa,12/\kappa-2\}$ for all $j\in\{1,2,\ldots,K\}$},\\
s:=\sideset{}{_{j=1}^K}\sum\beta_j=-2,\qquad\beta_K=\beta_1=-4/\kappa,
\end{gathered}\end{subeqnarray}
which is identical to (\ref{betai}), except that we have identified ``$\beta_{K+1}$" with $\beta_1$ and included a condition on the sum of the powers $\beta_j$.

Now we find the asymptotic behavior of $\mathcal{J}^{(1,K)}(\boldsymbol{x})$ (\ref{I1here}) as $-x_1=x_K=R\rightarrow\infty$.  After inserting the substitution $u_1(\theta)=Re^{i\theta}$, this becomes (item \ref{4thitem} of definition \ref{Fkdefn} defines the symbol $\mathcal{N}[\,\,\ldots\,\,]$ enclosing the integrand of (\ref{I1here}))
\newpage
\begin{multline}\label{Jint}\mathcal{J}^{(1,K)}\Big(\{\beta_j\}\,\Big|\,[R,-R]^+\,\Big|\,x_1=-R,x_2,\ldots,x_K=R\Big)\\
=\sideset{}{_0^\pi}\int(-R-Re^{i\theta})^{\beta_1}(R-Re^{i\theta})^{\beta_K}\prod_{j=2}^{K-1}(x_j-Re^{i\theta})^{\beta_j}Rie^{i\theta}\,{\rm d}\theta.\end{multline}
With $R$ large, we may drop $x_j$ from each difference in the product in (\ref{Jint}).  After doing this, simplifying that product, inserting $\beta_2+\beta_3+\dotsm+\beta_{K-1}=-2-\beta_1-\beta_K$ from condition (\ref{betai2}\red{b}) for the power, and simplifying further, we find 
\be\label{semicircJ}\mathcal{J}^{(1,K)}\Big(\{\beta_j\}\,\Big|\,[R,-R]^+\,\Big|\,x_1=-R,x_2,\ldots,x_K=R\Big)\underset{R\rightarrow\infty}{\sim}R^{-1}\sideset{}{_0^\pi}\int(1+e^{-i\theta})^{\beta_1}(1-e^{-i\theta})^{\beta_K}ie^{-i\theta}\,{\rm d}\theta.\ee
Finally, the substitution $t=(1-e^{-i\theta})/2$ converts the definite integral in (\ref{semicircJ}) into the beta-function integral in (\ref{result2.5}) with $\beta_i$ and $\beta_{i+1}$ replaced by $\beta_K$ and $\beta_1$ respectively.  Thus, we find
\be\label{semicircJfin}\mathcal{J}^{(1,K)}\Big(\{\beta_j\}\,\Big|\,[R,-R]^+\,\Big|\,x_1=-R,x_2,\ldots,x_K=R\Big)\underset{R\rightarrow\infty}{\sim}2^{\beta_K+\beta_1+1}R^{-1}\frac{\Gamma(\beta_K+1)\Gamma(\beta_1+1)}{\Gamma(\beta_K+\beta_1+2)}.\ee
On the other hand, if $\kappa\in(0,4]$ (in which case $\beta_K,\beta_1\leq-1$ thanks to (\ref{betai2}\red{c}), so the improper integral (\ref{Jint})  diverges), then we have $\Gamma_1=\mathscr{P}(R,-R)$ instead.  Using the analytic continuation (\ref{betafunc}) of the beta function \cite{witt}, we find
\begin{multline}\label{semicircJPoch}\mathcal{J}^{(1,K)}\Big(\{\beta_j\}\,\Big|\,\mathscr{P}(R,-R)\,\Big|\,x_1=-R,x_2,\ldots,x_K=R\Big)\\
\underset{R\rightarrow\infty}{\sim}4e^{\pi i(\beta_K-\beta_1)}\sin\pi\beta_K\sin\beta_1\,\,2^{\beta_K+\beta_1+1}R^{-1}\frac{\Gamma(\beta_K+1)\Gamma(\beta_1+1)}{\Gamma(\beta_K+\beta_1+2)}.\end{multline}

To finish, we use (\ref{semicircJfin}, \ref{semicircJPoch}) to justify item \ref{secondcase2} in the proof of corollary \ref{moveconjchargecor}.  In that part of the proof, we require the limit $\underline{\ell}_1\mathcal{F}_{c,\vartheta}$ (\ref{inflim}) for $\kappa\in(0,8)$, with $\mathcal{F}_{c,\vartheta}$ given by (\ref{firstFexplicit1}) and $\Gamma_1$ given by (\ref{Gamma1}) (so the sine functions drop from the prefactor (\ref{firstprefactor}) in (\ref{firstFexplicit1}) if $\kappa>4$).  For $\kappa\leq4$, this limit is 
\begin{multline}\label{thebiglimiti=2N}\underline{\ell}_1\mathcal{F}_{c,\vartheta}(\kappa\,|\,x_2,x_3,\ldots,x_{2N-1})\,\,=\lim_{R\rightarrow\infty}(2R)^{6/\kappa-1}\,\,\times\\ 
\mathcal{F}_{c,\vartheta}(\kappa\,|\,x_1=-R,x_2,\ldots,x_{2N}=R) \\ \parallel \\
\boxed{\begin{aligned}&n(\kappa)\left[\frac{n(\kappa)\Gamma(2-8/\kappa)}{4\sin^2(4\pi/\kappa)\Gamma(1-4/\kappa)^2}\right]^{N-1}\Bigg(\prod_{\substack{2\leq j<k\\j,k\neq c}}^{2N-1}(x_k-x_j)^{2/\kappa}\Bigg)\Bigg(\prod_{\substack{k=2 \\ k\neq c}}^{2N-1}|x_c-x_k|^{1-6/\kappa}\Bigg)\Big(2R\Big)^{2/\kappa}\sideset{}{_{\Gamma_{N-1}}}\oint {\rm d}u_{N-1}\\
&\dotsm\sideset{}{_{\Gamma_3}}\oint {\rm d}u_3\,\,\sideset{}{_{\Gamma_2}}\oint {\rm d}u_2\,\,\mathcal{N}\Bigg[\Bigg(\prod_{\substack{l=2 \\ l\neq c}}^{2N-1}\prod_{m=2}^{N-1}(x_l-u_m)^{-4/\kappa}\Bigg)\Bigg(\prod_{m=2}^{N-1}(x_c-u_m)^{12/\kappa-2}\Bigg)\Bigg(\prod_{2\leq p<q}^{N-1}(u_p-u_q)^{8/\kappa}\Bigg)\\
&\Bigg(\prod_{\substack{j=2\\ j\neq c}}^{2N-1}(x_j+R)^{2/\kappa}(R-x_j)^{2/\kappa}\Bigg)\Big(x_c+R\Big)^{1-6/\kappa}\Big(R-x_c\Big)^{1-6/\kappa}\Bigg(\prod_{m=2}^{N-1}(u_m+R)^{-4/\kappa}(R-u_m)^{-4/\kappa}\Bigg)\Bigg]\\
&\underbrace{\sideset{}{_{\mathscr{P}(R,-R)}}\oint {\rm d}u_1
\,\,\mathcal{N}\Bigg[(u_1+R)^{-4/\kappa}(R-u_1)^{-4/\kappa}\Bigg(\prod_{\substack{l=2 \\ l\neq c}}^{2N-1}(x_l-u_1)^{-4/\kappa}\Bigg)\Bigg(\prod_{m=2}^{N-1}(u_m-u_1)^{8/\kappa}\Bigg)\Big(x_c-u_1\Big)^{12/\kappa-2}\Bigg].}_{\mathcal{J}^{(1,K)}}\end{aligned}}\end{multline}
(If $\kappa>4$, then we adjust (\ref{thebiglimiti=2N}) as per item \ref{itemc} in the introduction of this appendix.)  We have rewritten the formula (\ref{firstFexplicit1}) for $\mathcal{F}_{c,\vartheta}(\kappa\,|\,\boldsymbol{x})$ slightly to clarify the calculation, and we indicate the contour integral $\mathcal{J}^{(1,K)}$ (\ref{I1here}) with braces.

Now we find the limit (\ref{thebiglimiti=2N}).  After doing steps \ref{oneit}--\ref{lastit} of section \ref{s2}, identifying the definite integral with respect to $u_1$ with (\ref{I1here}, \ref{Gamma1i=2N}, \ref{betai2}), and replacing it with the right side of (\ref{semicircJPoch}), we find 
\begin{multline}\label{asympinserti=2N}(2R)^{6/\kappa-1}\mathcal{F}_{c,\vartheta}(\kappa\,|\,x_1=-R,x_2,\ldots,x_{2N}=R)\,\,\underset{R\rightarrow\infty}{\sim}(2R)^{6/\kappa-1}\,\,\times\\ 
\boxed{\begin{aligned}&n(\kappa)\left[\frac{n(\kappa)\Gamma(2-8/\kappa)}{4\sin^2(4\pi/\kappa)\Gamma(1-4/\kappa)^2}\right]^{N-1}\Bigg(\prod_{\mathclap{\substack{2\leq j<k\\j,k\neq c}}}^{2N-1}(x_k-x_j)^{2/\kappa}\Bigg)\Bigg(\prod_{\substack{k=2 \\ k\neq c}}^{2N-1}|x_c-x_k|^{1-6/\kappa}\Bigg)\Big(2R\Big)^{2/\kappa}\sideset{}{_{\Gamma_{N-1}}}\oint {\rm d}u_{N-1}\\
&\dotsm\sideset{}{_{\Gamma_3}}\oint {\rm d}u_3\,\, \sideset{}{_{\Gamma_2}}\oint {\rm d}u_2\,\,\mathcal{N}\Bigg[\Bigg(\prod_{\substack{l=2 \\ l\neq c}}^{2N-1}\prod_{m=2}^{N-1}(x_l-u_m)^{-4/\kappa}\Bigg)\Bigg(\prod_{m=2}^{N-1}(x_c-u_m)^{12/\kappa-2}\Bigg)\Bigg(\prod_{2\leq p<q}^{N-1}(u_p-u_q)^{8/\kappa}\Bigg)\Bigg]R^{2-8/\kappa}\\
&\hspace{10cm}\underbrace{4\sin^2\left(\frac{4\pi}{\kappa}\right)\,\,2^{1-8/\kappa}R^{-1}\frac{\Gamma(1-4/\kappa)^2}{\Gamma(2-8/\kappa)}.}_{\text{right side of (\ref{semicircJPoch})}}\end{aligned}}\end{multline}
(If $\kappa>4$ and (\ref{thebiglimiti=2N}) is adjusted as per item \ref{itemc} in the introduction of this appendix, then we replace the definite integral with respect to $u_1$ with the right side of (\ref{semicircJfin}) instead, finding the same result but with all factors of $4\sin^2(4\pi/\kappa)$ dropped and all contours simple.)  After some simplification, we finally send $R\rightarrow\infty$ in (\ref{asympinserti=2N}) to find
\begin{multline}\label{firstlimprimei=2N}\underline{\ell}_1\mathcal{F}_{c,\vartheta}(\kappa\,|\,x_2,x_3,\ldots,x_{2N-1})=n(\kappa)\,\,\times\\
\left\{\begin{aligned}&n(\kappa)\left[\frac{n(\kappa)\Gamma(2-8/\kappa)}{4\sin^2(4\pi/\kappa)\Gamma(1-4/\kappa)^2}\right]^{N-2}\Bigg(\prod_{\substack{2\leq j<k \\ j,k\neq c}}^{2N-1}(x_k-x_j)^{2/\kappa}\Bigg)\Bigg(\prod_{\substack{k=2 \\ k\neq c}}^{2N-1}|x_c-x_k|^{1-6/\kappa}\Bigg)\oint_{\Gamma_{N-1}} {\rm d}u_{N-1}\\ 
&\dotsm\,\oint_{\Gamma_3}{\rm d}u_3\oint_{\Gamma_2}{\rm d}u_2\,\,\mathcal{N}\Bigg[\Bigg(\prod_{\substack{l=2 \\ l\neq c}}^{2N-1}\prod_{m=2}^{N-1}(x_l-u_m)^{-4/\kappa}\Bigg)\Bigg(\prod_{m=2}^{N-1}(x_c-u_m)^{12/\kappa-2}\Bigg)\Bigg(\prod_{2\leq p<q}^{N-1}(u_p-u_q)^{8/\kappa}\Bigg)\Bigg]\end{aligned}\right\}.\end{multline}
(Again, if $\kappa>4$, then we find the same result, but with all factors of $4\sin^2(4\pi/\kappa)$ dropped and with all Pochhammer contours replaced by simple contours with the same endpoints and orientation.)

Now, the quantity of (\ref{firstlimprimei=2N}) in brackets is the element of $\mathcal{S}_{N-1}$ described in the following conclusion.  We thus have our main result, for use in item \ref{secondcase2} in the proof of corollary \ref{moveconjchargecor}.
\begin{quote}\textbf{Main result 2:} In case \ref{sc2} with $i=2N$, where $\Gamma_1$ is given by (\ref{Gamma1i=2N}), the limit $\underline{\ell}_1\mathcal{F}_{c,\vartheta}$ (\ref{inflim}) equals $n$ (\ref{fugacity}) times the element of $\mathcal{S}_{N-1}$ generated from the formula (\ref{firstFexplicit1}) for $\mathcal{F}_{c,\vartheta}$ by dropping all factors involving $x_{2N}$, $x_1$, and $u_1$, dropping the integration along $\Gamma_1$, and reducing the power of the prefactor (\ref{firstprefactor}, \ref{secondprefactor}) by one.
\end{quote}
We note that this second main result is identical to the first if, with $i=2N$, we identify $x_{i+1}$ with $x_1$.

\subsubsection{Proof of corollary \ref{moveconjchargecor}, item \ref{thirdcase2}}\label{s52}

In case \ref{sc3new}, we assume without loss of generality that $x_i$ and $x_{i-1}$ are endpoints of one contour $\Gamma_1''$ (\ref{Gamma1''}) (replacing $\Gamma_1$ of case \ref{sc3}), and $x_{i+1}$ is not an endpoint of any contour.  To determine $\bar{\ell}_1\mathcal{F}_{c,\vartheta}$ (\ref{nextthelim}) for $i\in\{1,2,\ldots,2N-1\}\setminus\{c-1,c\}$, we repeat the analysis of section \ref{s3}, and to determine $\underline{\ell}_1\mathcal{F}_{c,\vartheta}$ (\ref{inflim}) (which we identify with the value $i=2N$), we follow a similar route, using the results of section \ref{s51} instead of \ref{s2}.  For either, we pick up the analysis of section \ref{s3} at (\ref{case3int}).  If $i\not\in\{1,2,2N-1,2N\}$, then we may solve the system (\ref{case3int}) of two equations to find (\ref{result}) again.   However, if $i=1$ (resp.\ 2, resp.\ $2N-1$, resp.\ $2N$), then we find that $i=1$ (resp.\ 2, resp.\ $K-1$, resp.\ $K$) in (\ref{result}) after re-indexing the points according to step \ref{lastit} of section \ref{s2}, and vice versa.  Now, if $i\in\{1,2,K-1,K\}$, then (\ref{result}) has an index outside of its allowed range or a lower summation index greater than an accompanying upper summation index.  Therefore, we must revise (\ref{case3int}) as follows.
\begin{enumerate}
\item\label{i1}If $i=1$, then we use the two equations (\ref{case3int}) to solve for $I_K$ (\ref{Ikintegrals}) in terms of $I_k$ with $k\not\in\{2,K\}$, finding
\be\label{resulti1}I_K=\sum_{k=3}^{K-1}\frac{\sin\pi\sum_{l=3}^k\beta_l}{\sin\pi(\beta_1+\beta_2)}I_k-\frac{\sin\pi\beta_2}{\sin\pi(\beta_1+\beta_2)}I_1.\ee
\item\label{i2} If $i=2$, then we use the two equations (\ref{case3int}) to solve for $I_1$ (\ref{Ikintegrals}) in terms of $I_k$ with $k\not\in\{1,3\}$, finding
\be\label{resulti2}I_1=\sum_{k=4}^K\frac{\sin\pi\sum_{l=4}^k\beta_l}{\sin\pi(\beta_2+\beta_3)}I_k-\frac{\sin\pi\beta_3}{\sin\pi(\beta_2+\beta_3)}I_2.\ee
\item\label{i3} If $i=K-1$, then we use the two equations (\ref{case3int}) to solve for $I_{K-2}$ (\ref{Ikintegrals}) in terms of $I_k$ with $k\not\in\{K-2,K\}$, finding 
\be\label{resultiK-1}I_{K-2}=-\sum_{k=1}^{K-3}\frac{\sin\pi\sum_{l=k+1}^K\beta_l}{\sin\pi(\beta_{K-1}+\beta_K)}I_k-\frac{\sin\pi\beta_K}{\sin\pi(\beta_{K-1}+\beta_K)}I_{K-1}.\ee
\item\label{i4} If $i=K$, then we use the two equations (\ref{case3int}) to solve for $I_{K-1}$ (\ref{Ikintegrals}) in terms of $I_k$ with $k\not\in\{1,K-1\}$, finding (here, it helps to recall from (\ref{applicationsc3}) that $\beta_1+\beta_2+\dotsm+\beta_K=-2$)
\be\label{resultiK}I_{K-1}=\sum_{k=2}^{K-2}\frac{\sin\pi\sum_{l=2}^k\beta_l}{\sin\pi(\beta_K+\beta_1)}I_k-\frac{\sin\pi\beta_1}{\sin\pi(\beta_K+\beta_1)}I_K.\ee
\end{enumerate}
Similar to (\ref{result}), each term in the summation on the right side of (\ref{resulti1}--\ref{resultiK}) falls under case \ref{sc1}, but the last term falls under case \ref{sc2}.  Thanks to condition (\ref{extracond}\red{b}), the latter term is asymptotically dominant over the former terms as $x_{i+1}\rightarrow x_i$.  By starting from items \ref{i1}--\ref{i3} immediately above and following the exact reasoning that takes us from (\ref{result}) to the ``main result" of section \ref{s3}, we find the following result, for use in the proof of corollary \ref{moveconjchargecor}. 
\begin{quote}\textbf{Main result 1:} In case \ref{sc3new} with $i\in\{1,2,\ldots,2N-1\}$, where $\Gamma_1''$ (\ref{Gamma1''}) (with $x_{i-1}:=x_{2N}$ if $i=1$) replaces $\Gamma_1$ in the formula (\ref{firstFexplicit1}) for $\mathcal{F}_{c,\vartheta}$, the limit $\bar{\ell}_1\mathcal{F}_{c,\vartheta}$ (\ref{nextthelim}) equals the element of $\mathcal{S}_{N-1}$ generated from the formula for $\mathcal{F}_{c,\vartheta}$ by dropping all factors involving $x_i$, $x_{i+1}$, and $u_1$, dropping the integration along $\Gamma_1$, and reducing the power of the prefactor (\ref{firstprefactor}) or (\ref{secondprefactor}) by one.
\end{quote}
This result includes the main result of section \ref{s3} as the special case $c=2N$.

The last item \ref{i4} pertains to the limit $\underline{\ell}_1\mathcal{F}_{c,\vartheta}$ (\ref{inflim}), so we are sending $-x_1=x_K\rightarrow\infty$ in (\ref{resultiK}).  In the last term of (\ref{resultiK}), the integration contour of $I_K$ is (\ref{Gamma1i=2N}), and (\ref{semicircJfin}, \ref{semicircJPoch}) determines the asymptotic behavior of $I_K$ as $R\rightarrow\infty$.  Moreover, (\ref{applicationsc3}) implies that $\beta_K=\beta_1=-4/\kappa$, so the ratio of sine functions and factor of negative one multiplying $I_K$ in (\ref{resultiK}) is again $n(\kappa)^{-1}$ (\ref{fugacity}).  After using these facts and the second main result of section \ref{s51}, we have the following result, for use in item \ref{thirdcase2} in the proof of corollary \ref{moveconjchargecor}.
\begin{quote}\textbf{Main result 2:} In case \ref{sc3new} with $i=2N$, where $\Gamma_1''$ (\ref{Gamma1''}) replaces $\Gamma_1$ in the formula (\ref{firstFexplicit1}) for $\mathcal{F}_{c,\vartheta}$, the limit $\underline{\ell}_1\mathcal{F}_{c,\vartheta}$ (\ref{nextthelim}) equals the element of $\mathcal{S}_{N-1}$ generated from the formula for $\mathcal{F}_{c,\vartheta}$ by dropping all factors involving $x_i$, $x_{i+1}$, and $u_1$, dropping the integration along $\Gamma_1$, and reducing the power of the prefactor (\ref{firstprefactor}) or (\ref{secondprefactor}) by one.
\end{quote}
(Again, we assume here that $8/\kappa\not\in\mathbb{Z}^+$, but these results remain true even if this condition is not met.  See the last paragraph of section \ref{s4}.)  We note that this second main result is identical to the first if, with $i=2N$, we identify $x_{i+1}$ with $x_1$.  Also, the analysis for the situation with $x_{i+1}$ and $x_{i+2}$ as endpoints of $\Gamma_1''$ is similar and leads to the same results.

\subsubsection{Proof of corollary \ref{moveconjchargecor}, item \ref{fourthcase2}}\label{s53}

In case \ref{sc4new}, the integration contours $\Gamma_1''$ and $\Gamma_2''$ (\ref{Gamma12''}) respectively replace the integration contours $\Gamma_1$ and $\Gamma_2$ of case \ref{sc4} in the formula (\ref{firstFexplicit1}) for $\mathcal{F}_{c,\vartheta}$.  After assigning any indices outside the allowed range $1$, $2,\ldots,2N$ values within this range using modular $2N$ arithmetic, we have
\begin{align}\label{contouri=1} &i=1:&\Gamma_1''=&\begin{cases}\mathscr{P}(x_{2N},x_1), & 0<\kappa\leq4\\
[x_{2N},x_1]^+, & 4<\kappa<8\end{cases},& &\Gamma_2''=\begin{cases}\mathscr{P}(x_2,x_3), & 0<\kappa\leq4\\
[x_2,x_3]^+, & 4<\kappa<8\end{cases},&\\
\label{contouri=2}&i\in\{2,3,\ldots,2N-2\}: &\Gamma_1''=&\begin{cases}\mathscr{P}(x_{i-1},x_i), & 0<\kappa\leq4\\
[x_{i-1},x_i]^+, & 4<\kappa<8\end{cases},& &\Gamma_2''=\begin{cases}\mathscr{P}(x_{i+1},x_{i+2}), & 0<\kappa\leq4\\
[x_{i+1},x_{i+2}]^+, & 4<\kappa<8\end{cases},&\\
\label{contouri=2N-1}&i=2N-1:&\Gamma_1''=&\begin{cases}\mathscr{P}(x_{2N-2},x_{2N-1}), & 0<\kappa\leq4\\
[x_{2N-2},x_{2N-1}]^+, & 4<\kappa<8\end{cases},& &\Gamma_2''=\begin{cases}\mathscr{P}(x_{2N},x_1), & 0<\kappa\leq4\\
[x_{2N},x_1]^+, & 4<\kappa<8\end{cases},&\\
\label{contouri=2N}&i=2N:&\Gamma_1''=&\begin{cases}\mathscr{P}(x_{2N-1},x_{2N}), & 0<\kappa\leq4\\
[x_{2N-1},x_{2N}]^+, & 4<\kappa<8\end{cases},& &\Gamma_2''=\begin{cases}\mathscr{P}(x_1,x_2), & 0<\kappa\leq4\\
[x_1,x_2]^+, & 4<\kappa<8\end{cases}.&\end{align}
To determine $\bar{\ell}_1\mathcal{F}_{c,\vartheta}$ (\ref{nextthelim}) for $i\in\{1,2,\ldots,2N-1\}\setminus\{c-1,c\}$, we repeat the analysis of section \ref{s4}, and to determine $\underline{\ell}_1\mathcal{F}_{c,\vartheta}$ (\ref{inflim}) (which we identify with $i=2N$), we follow a similar route, using the results of section \ref{s51} instead of \ref{s2}.  For either, we pick up the analysis of section \ref{s4} at (\ref{overunder}).  If $i\not\in\{1,2,2N-1,2N\}$, then this analysis proceeds as before.   Otherwise, many equations in this analysis require revision.  Although these revisions are similar for different values of $i$, they are different enough to warrant showing all of them for completeness.  Thus, we explicitly re-derive the main result of section \ref{s4} for $i\in\{1,2,2N-1,2N\}$.
\begin{enumerate}
\item\label{itemi=1} If $i=1$, then we solve for $I_{K,2}$ in terms of $I_{K,k}$ with $k\not\in\{2,K\}$.  In place of (\ref{result1}), we thus find
\be\label{IK2} I_{K,2}=-\sum_{k=3}^{K-1}\frac{\sin\pi(\sum_{l=1}^k\beta_l+\gamma/2)}{\sin\pi(\beta_1+\beta_2+\gamma/2)}I_{K,k}-\frac{\sin\pi(\beta_1+\gamma/2)}{\sin\pi(\beta_1+\beta_2+\gamma/2)}I_{K,1}.\ee
Again, (\ref{extracond2}\red{b}) implies that the second term on the right side of (\ref{IK2}) vanishes.  Meanwhile,  (\ref{condasymp}) with the index $i-1$ replaced by $K$ gives the asymptotic behavior of $I_{K,k}$ with $k\not\in\{1,2,K\}$ in terms of $I_{1,k}$ as $x_2\rightarrow x_1$ under condition (\ref{extracond2}\red{a}), and (\ref{result2.5}) with $i=1$ gives the behavior of the integration in $I_{1,k}$ (\ref{Ijkintegrals}) from $x_1$ to $x_2$ in this limit.  After inserting these into (\ref{IK2}), we find (assuming conditions (\ref{extracond1}\red{a}--\red{c}, \ref{extracond2}\red{a}--\red{c}) but not yet (\ref{extracond2}\red{d}))
\begin{multline}\label{result3.55IK2}I_{K,2}(x_1,x_2,\ldots,x_K)\underset{x_2\rightarrow x_1}{\sim}-\frac{\sin\pi\beta_2\,\Gamma(\beta_1+1)\Gamma(\beta_2+1)}{\sin\pi(\beta_1+\beta_2)\Gamma(\beta_1+\beta_2+2)}(x_2-x_1)^{\beta_1+\beta_2+1}\mathcal{N}\Bigg[\prod_{j=3}^K(x_1-x_j)^{\beta_j}\Bigg]\\
\times(-1)\sum_{k=3}^{K-1}\frac{\sin\pi(\sum_{l=1}^k\beta_l+\gamma/2)}{\sin\pi(\beta_1+\beta_2+\gamma/2)}\sideset{}{_{x_k}^{x_{k+1}}}\int\mathcal{N}\Bigg[(u_2-x_1)^{\beta_1+\beta_2+\gamma}\prod_{j=3}^K(u_2-x_j)^{\beta_j}\Bigg]\,{\rm d}u_2.\end{multline}
This equation replaces (\ref{result3.55}).  To simplify (\ref{result3.55IK2}), we include condition (\ref{extracond2}\red{d}) with the other conditions (\ref{extracond1}\red{a}--\red{c}, \ref{extracond2}\red{a}--\red{c}) already assumed and repeat the analysis from (\ref{Ii+2'}) to (\ref{result4}).  Thus, we replace (\ref{Ii+2'}) with
\be\label{Ii+2'i=1}I_K'(x_3,x_4,\ldots,x_K):=\Bigg(\sideset{}{_{x_K}^{\infty}}\int+\sideset{}{_{-\infty}^{x_3}}\int\Bigg)\mathcal{N}\Bigg[\prod_{j=3}^K(u_2-x_j)^{\beta_j}\Bigg]\,{\rm d}u_2,\ee
and we define $I_k'$ as in (\ref{Ik'}), but with the index $k\in\{3,4,\ldots,K-1\}$. With $A^\pm$ defined in (\ref{intabovebelow}), we isolate $I_K'$ from the linear combination (replacing (\ref{Alincmb}))
\be\label{Alincmbi=1} e^{-\pi i\sum_{l=3}^K\beta_l}e^{\pi i(\beta_1+\beta_2+\gamma/2)}A^+-e^{\pi i\sum_{l=3}^K\beta_l}e^{-\pi i(\beta_1+\beta_2+\gamma/2)}A^-=0\ee
to find the following formula for $I_K'$ (which replaces (\ref{primed})):
\be\label{primedi=1}I_K'(x_3,x_4,\ldots,x_K)=\sum_{k=3}^{K-1}\frac{\sin\pi(\sum_{l=k+1}^K\beta_l-\beta_1-\beta_2-\gamma/2)}{\sin\pi(\beta_1+\beta_2+\gamma/2)}\sideset{}{_{x_k}^{x_{k+1}}}\int\mathcal{N}\Bigg[\prod_{j=3}^K(u_2-x_j)^{\beta_j}\Bigg]\,{\rm d}u_2.\ee
After using conditions (\ref{extracond1}\red{a}, \ref{extracond2}\red{b,d}) to rewrite the numerator of the fraction in (\ref{primedi=1}), we find that the right side of (\ref{primedi=1}) matches the entire bottom line of (\ref{result3.55IK2}).  After replacing the latter expression by the former, (\ref{result3.55IK2}) becomes 
\begin{multline}\label{result4i=1}I_{K,2}(x_1,x_2,\ldots,x_K)\underset{x_2\rightarrow x_1}{\sim}-\frac{\sin\pi\beta_2\Gamma(\beta_1+1)\Gamma(\beta_2+1)}{\sin\pi(\beta_1+\beta_2)\Gamma(\beta_1+\beta_2+2)}(x_2-x_1)^{\beta_1+\beta_2+1}\\
\times\mathcal{N}\Bigg[\prod_{j=3}^K(x_1-x_j)^{\beta_j}\Bigg]\Bigg(\sideset{}{_{x_K}^{\infty}}\int+\sideset{}{_{-\infty}^{x_3}}\int\Bigg)\mathcal{N}\Bigg[\prod_{j=3}^K(u_2-x_j)^{\beta_j}\Bigg]\,{\rm d}u_2,\quad \beta_1=\beta_2,\end{multline}
which replaces (\ref{result4}).  The remaining analysis that leads from here to the main result of section \ref{s4} is identical to the analysis from (\ref{result4.5}) to (\ref{firstlimprime2}).
\item\label{itemi=2} If $i=2$, then we solve for $I_{1,3}$ in terms of $I_{1,k}$ with $k\not\in\{1,3\}$.  In place of (\ref{result1}), we thus find
\be\label{I13} I_{1,3}=-\sum_{k=4}^K\frac{\sin\pi(\sum_{l=2}^k\beta_l+\gamma/2)}{\sin\pi(\beta_2+\beta_3+\gamma/2)}I_{1,k}-\frac{\sin\pi(\beta_2+\gamma/2)}{\sin\pi(\beta_2+\beta_3+\gamma/2)}I_{1,2}.\ee
Again, (\ref{extracond2}\red{b}) implies that the second term on the right side of (\ref{I13}) vanishes.  Meanwhile,  (\ref{condasymp}) with $i=2$ gives the asymptotic behavior of $I_{1,k}$ with $k\not\in\{1,2,3\}$ in terms of $I_{2,k}$ as $x_3\rightarrow x_2$ under condition (\ref{extracond2}\red{a}), and (\ref{result2.5}) with $i=2$ gives the behavior of the integration in $I_{2,k}$ (\ref{Ijkintegrals}) from $x_2$ to $x_3$ in this limit.  After inserting these into (\ref{I13}), we find (assuming conditions (\ref{extracond1}\red{a}--\red{c}, \ref{extracond2}\red{a}--\red{c}) but not yet (\ref{extracond2}\red{d}))
\begin{multline}\label{result3.55I13}I_{1,3}(x_1,x_2,\ldots,x_K)\underset{x_3\rightarrow x_2}{\sim}-\frac{\sin\pi\beta_3\,\Gamma(\beta_2+1)\Gamma(\beta_3+1)}{\sin\pi(\beta_2+\beta_3)\Gamma(\beta_2+\beta_3+2)}(x_3-x_2)^{\beta_2+\beta_3+1}\mathcal{N}\Bigg[\prod_{j\neq2,3}^K(x_2-x_j)^{\beta_j}\Bigg]\\
\times(-1)\sum_{k=4}^K\frac{\sin\pi(\sum_{l=2}^k\beta_l+\gamma/2)}{\sin\pi(\beta_2+\beta_3+\gamma/2)}\sideset{}{_{x_k}^{x_{k+1}}}\int\mathcal{N}\Bigg[(u_2-x_2)^{\beta_2+\beta_3+\gamma}\prod_{j\neq2,3}^K(u_2-x_j)^{\beta_j}\Bigg]\,{\rm d}u_2.\end{multline}
This equation replaces (\ref{result3.55}).  To simplify (\ref{result3.55I13}), we repeat the analysis from (\ref{Ii+2'}) to (\ref{result4}).  Thus, we replace (\ref{Ii+2'}) with
\be\label{Ii+2'i=2}I_1'(x_1,x_4,x_5,\ldots,x_K):=\sideset{}{_{x_1}^{x_4}}\int\mathcal{N}\Bigg[\prod_{j\neq2,3}^K(u_2-x_j)^{\beta_j}\Bigg]\,{\rm d}u_2,\ee
and we define $I_k'$ as in (\ref{Ik'}), but with the index $k\in\{4,5\ldots,K\}$. With $A^\pm$ defined in (\ref{intabovebelow}), we isolate $I_1'$ from the linear combination (replacing (\ref{Alincmb}))
\be\label{Alincmbi=2} e^{-\pi i\beta_1}e^{\pi i(\beta_2+\beta_3+\gamma/2)}A^+-e^{\pi i\beta_1}e^{-\pi i(\beta_2+\beta_3+\gamma/2)}A^-=0\ee
to find the following formula for $I_K'$ (which replaces (\ref{primed})):
\be\label{primedi=2}I_K'(x_1,x_4,x_5,\ldots,x_K)=-\sum_{k=4}^K\frac{\sin\pi(\sum_{l=2}^k\beta_l+\gamma/2)}{\sin\pi(\beta_2+\beta_3+\gamma/2)}\sideset{}{_{x_k}^{x_{k+1}}}\int\mathcal{N}\Bigg[\prod_{j\neq2,3}^K(u_2-x_j)^{\beta_j}\Bigg]\,{\rm d}u_2.\ee
We notice that the right side of (\ref{primedi=2}) matches the entire bottom line of (\ref{result3.55I13}).  After replacing the latter expression by the former, (\ref{result3.55I13}) becomes 
\begin{multline}\label{result4i=2}I_{1,3}(x_1,x_2,\ldots,x_K)\underset{x_2\rightarrow x_1}{\sim}-\frac{\sin\pi\beta_3\Gamma(\beta_2+1)\Gamma(\beta_3+1)}{\sin\pi(\beta_2+\beta_3)\Gamma(\beta_2+\beta_3+2)}(x_3-x_2)^{\beta_2+\beta_3+1}\\
\times\mathcal{N}\Bigg[\prod_{j\neq2,3}^K(x_2-x_j)^{\beta_j}\Bigg]\sideset{}{_{x_1}^{x_4}}\int\mathcal{N}\Bigg[\prod_{j\neq2,3}^K(u_2-x_j)^{\beta_j}\Bigg]\,{\rm d}u_2,\end{multline}
which replaces (\ref{result4}).  (We notice that we did not need to assume (\ref{extracond2}\red{d}) to obtain this result.)  The remaining analysis that leads from here to the main result of section \ref{s4} is identical to the analysis from (\ref{result4.5}) to (\ref{firstlimprime2}).
\item\label{itemi=K-1} If $i=K-1$, then we solve for $I_{K-2,K}$ in terms of $I_{K,k}$ with $k\not\in\{K-2,K\}$.  In place of (\ref{result1}), we thus find
\be\label{IK-2K} I_{K-2,K}=\sum_{k=1}^{K-3}\frac{\sin\pi(\sum_{l=k+1}^{K-2}\beta_l+\gamma/2)}{\sin\pi(\beta_{K-1}+\beta_K+\gamma/2)}I_{K-2,k}-\frac{\sin\pi(\beta_{K-1}+\gamma/2)}{\sin\pi(\beta_{K-1}+\beta_K+\gamma/2)}I_{K-2,K-1}.\ee
Again, (\ref{extracond2}\red{b}) implies that the second term on the right side of (\ref{IK-2K}) vanishes.  Meanwhile,  (\ref{condasymp}) with $i=K-1$ gives the asymptotic behavior of $I_{K-2,k}$ with $k\not\in\{1,2,K\}$ in terms of $I_{K-1,k}$ as $x_K\rightarrow x_{K-1}$ under condition (\ref{extracond2}\red{a}), and (\ref{result2.5}) with $i=K-1$ gives the behavior of the integration in $I_{K-1,k}$ (\ref{Ijkintegrals}) from $x_{K-1}$ to $x_K$ in this limit.  After inserting these into (\ref{IK-2K}), we find (assuming conditions (\ref{extracond1}\red{a}--\red{c}, \ref{extracond2}\red{a}--\red{c}) but not yet (\ref{extracond2}\red{d}))
\begin{multline}\label{result3.55IK-2K}I_{K-2,K}(x_1,x_2,\ldots,x_K)\underset{x_K\rightarrow x_{K-1}}{\sim}\\
\begin{aligned}&-\frac{\sin\pi\beta_K\,\Gamma(\beta_{K-1}+1)\Gamma(\beta_K+1)}{\sin\pi(\beta_{K-1}+\beta_K)\Gamma(\beta_{K-1}+\beta_K+2)}(x_K-x_{K-1})^{\beta_{K-1}+\beta_K+1}\mathcal{N}\Bigg[\prod_{j=1}^{K-2}(x_{K-1}-x_j)^{\beta_j}\Bigg]\\
&\times\sum_{k=1}^{K-3}\frac{\sin\pi(\sum_{l=k+1}^{K-2}\beta_l+\gamma/2)}{\sin\pi(\beta_{K-1}+\beta_K+\gamma/2)}\sideset{}{_{x_k}^{x_{k+1}}}\int\mathcal{N}\Bigg[(u_2-x_{K-1})^{\beta_{K-1}+\beta_K+\gamma}\prod_{j=1}^{K-2}(u_2-x_j)^{\beta_j}\Bigg]\,{\rm d}u_2.\end{aligned}\end{multline}
This equation replaces (\ref{result3.55}).  To simplify (\ref{result3.55IK-2K}), we include condition (\ref{extracond2}\red{d}) with the other conditions (\ref{extracond1}\red{a}--\red{c}, \ref{extracond2}\red{a}--\red{c}) already assumed and repeat the analysis from (\ref{Ii+2'}) to (\ref{result4}).  Thus, we replace (\ref{Ii+2'}) with
\be\label{Ii+2'i=K-1}I_{K-2}'(x_1,x_2,\ldots,x_{K-2}):=\Bigg(\sideset{}{_{x_{K-2}}^{\infty}}\int+\sideset{}{_{-\infty}^{x_1}}\int\Bigg)\mathcal{N}\Bigg[\prod_{j=1}^{K-2}(u_2-x_j)^{\beta_j}\Bigg]\,{\rm d}u_2,\ee
and we define $I_k'$ as in (\ref{Ik'}), but with the index $k\in\{1,2,\ldots,K-3\}$. With $A^\pm$ defined in (\ref{intabovebelow}), we isolate $I_{K-2}'$ from the linear combination (replacing (\ref{Alincmb}))
\be\label{Alincmbi=K-1} e^{-\pi i\sum_{l=1}^{K-2}\beta_l}e^{\pi i(\beta_{K-1}+\beta_K+\gamma/2)}A^+-e^{\pi i\sum_{l=1}^{K-2}\beta_l}e^{-\pi i(\beta_{K-1}+\beta_K+\gamma/2)}A^-=0\ee
to find the following formula for $I_{K-2}'$ (which replaces (\ref{primed})):
\be\label{primedi=K-1}I_{K-2}'(x_1,x_2,\ldots,x_{K-2})=\sum_{k=1}^{K-3}\frac{\sin\pi(\sum_{l=k+1}^{K-2}\beta_l-\beta_{K-1}-\beta_K-\gamma/2)}{\sin\pi(\beta_{K-1}+\beta_K+\gamma/2)}\sideset{}{_{x_k}^{x_{k+1}}}\int\mathcal{N}\Bigg[\prod_{j=1}^{K-2}(u_2-x_j)^{\beta_j}\Bigg]\,{\rm d}u_2.\ee
After using conditions (\ref{extracond1}\red{a}, \ref{extracond2}\red{b,d}) to rewrite the numerator of the fraction in (\ref{primedi=K-1}), we find that the right side of (\ref{primedi=K-1}) matches the entire bottom line of (\ref{result3.55IK-2K}).  After replacing the latter expression by the former, (\ref{result3.55IK-2K}) becomes 
\begin{multline}\label{result4i=K-1}I_{K-2,K}(x_1,x_2,\ldots,x_K)\underset{x_2\rightarrow x_1}{\sim}-\frac{\sin\pi\beta_K\Gamma(\beta_{K-1}+1)\Gamma(\beta_K+1)}{\sin\pi(\beta_{K-1}+\beta_{K-2})\Gamma(\beta_{K-1}+\beta_K+2)}(x_K-x_{K-1})^{\beta_{K-1}+\beta_K+1}\\
\times\mathcal{N}\Bigg[\prod_{j=1}^{K-2}(x_{K-1}-x_j)^{\beta_j}\Bigg]\Bigg(\sideset{}{_{x_{K-2}}^{\infty}}\int+\sideset{}{_{-\infty}^{x_1}}\int\Bigg)\mathcal{N}\Bigg[\prod_{j=1}^{K-2}(u_2-x_j)^{\beta_j}\Bigg]\,{\rm d}u_2,\quad \beta_{K-1}=\beta_K,\end{multline}
which replaces (\ref{result4}).  The remaining analysis that leads from here to the main result of section \ref{s4} is identical to the analysis from (\ref{result4.5}) to (\ref{firstlimprime2}).
\item\label{itemi=K} If $i=K$, then we solve for $I_{K-1,1}$ in terms of $I_{K,k}$ with $k\not\in\{1,K-1\}$.  In place of (\ref{result1}), we thus find
\be\label{IK-11} I_{K-1,1}=\sum_{k=2}^{K-2}\frac{\sin\pi(\sum_{l=k+1}^{K-1}\beta_l+\gamma/2)}{\sin\pi(\beta_K+\beta_1+\gamma/2)}I_{K-1,k}-\frac{\sin\pi(\beta_K+\gamma/2)}{\sin\pi(\beta_K+\beta_1+\gamma/2)}I_{K-1,K}.\ee
Again, (\ref{extracond2}\red{b}) implies that the second term on the right side of (\ref{IK-11}) vanishes.  Meanwhile,  (\ref{condasymp}) with $i=K$ gives the asymptotic behavior of $I_{K-1,k}$ with $k\not\in\{K-1,K,1\}$ in terms of $I_{K,k}$ as $-x_1=x_K=R\rightarrow\infty$ under condition (\ref{extracond2}\red{a}) (with $\beta_{K+1}=\beta_1$), and (\ref{semicircJfin}) gives the behavior of the integration in $I_{K,k}$ (\ref{Ijkintegrals}) from $x_K$ to infinity and then from minus infinity to $x_1$ in this limit, also under condition (\ref{betai2}\red{b}).  After inserting these into (\ref{IK-11}), we find (assuming conditions (\ref{extracond1}\red{a}--\red{c}, \ref{extracond2}\red{a}--\red{c}) with $\beta_{K+1}=\beta_1$ but not yet (\ref{extracond2}\red{d}))
\begin{multline}\label{result3.55IK-11}I_{K-1,1}(x_1=-R,x_2,\ldots,x_K=R)\underset{R\rightarrow\infty}{\sim}-\frac{\sin\pi\beta_1\,\Gamma(\beta_K+1)\Gamma(\beta_1+1)}{\sin\pi(\beta_K+\beta_1)\Gamma(\beta_K+\beta_1+2)}2^{\beta_1+\beta_K+1}R^{\beta_K+\beta_1-1}\\
\times\sum_{k=2}^{K-2}\frac{\sin\pi(\sum_{l=k+1}^{K-1}\beta_l+\gamma/2)}{\sin\pi(\beta_K+\beta_1+\gamma/2)}\sideset{}{_{x_k}^{x_{k+1}}}\int\mathcal{N}\Bigg[\prod_{j=2}^{K-1}(u_2-x_j)^{\beta_j}\Bigg]\,{\rm d}u_2.\end{multline}
This equation replaces (\ref{result3.55}).  To simplify (\ref{result3.55IK-11}), we include condition (\ref{extracond2}\red{d}) with the other conditions (\ref{extracond1}\red{a}--\red{c}, \ref{extracond2}\red{a}--\red{c}) (with $\beta_{K+1}=\beta_1$) already assumed and repeat the analysis from (\ref{Ii+2'}) to (\ref{result4}).  Thus, we replace (\ref{Ii+2'}) with
\be\label{Ii+2'i=K}I_{K-1}'(x_2,x_3\ldots,x_{K-1}):=\Bigg(\sideset{}{_{x_{K-1}}^{\infty}}\int+\sideset{}{_{-\infty}^{x_2}}\int\Bigg)\mathcal{N}\Bigg[\prod_{j=2}^{K-1}(u_2-x_j)^{\beta_j}\Bigg]\,{\rm d}u_2,\ee
and we define $I_k'$ as in (\ref{Ik'}), but with the index $k\in\{2,3,\ldots,K-2\}$. With $A^\pm$ defined in (\ref{intabovebelow}) (where now, the indices $i=1$ and $i=K$ drop from the sum), we isolate $I_{K-1}'$ from the linear combination (replacing (\ref{Alincmb}))
\be\label{Alincmbi=K} e^{-\pi i\sum_{l=2}^{K-1}\beta_l}e^{\pi i(\beta_K+\beta_1+\gamma/2)}A^+-e^{\pi i\sum_{l=2}^{K-1}\beta_l}e^{-\pi i(\beta_K+\beta_1+\gamma/2)}A^-=0\ee
to find the following formula for $I_{K-1}'$ (which replaces (\ref{primed})):
\be\label{primedi=K}I_{K-1}'(x_2,x_3,\ldots,x_{K-1})=\sum_{k=2}^{K-2}\frac{\sin\pi(\sum_{l=k+1}^{K-1}\beta_l-\beta_K-\beta_1-\gamma/2)}{\sin\pi(\beta_K+\beta_1+\gamma/2)}\sideset{}{_{x_k}^{x_{k+1}}}\int\mathcal{N}\Bigg[\prod_{j=2}^{K-1}(u_2-x_j)^{\beta_j}\Bigg]\,{\rm d}u_2.\ee
After using conditions (\ref{extracond1}\red{a}, \ref{extracond2}\red{b,d}) (with $\beta_{K+1}=\beta_1$) to rewrite the numerator of the fraction (\ref{primedi=K}), we find that the right side of (\ref{primedi=K}) matches the entire bottom line of (\ref{result3.55IK-11}).  After replacing the latter expression by the former, (\ref{result3.55IK-11}) becomes 
\begin{multline}\label{result4i=K}I_{K-1,1}(x_1=-R,x_2,\ldots,x_K=R)\underset{R\rightarrow\infty}{\sim}-\frac{\sin\pi\beta_1\,\Gamma(\beta_K+1)\Gamma(\beta_1+1)}{\sin\pi(\beta_K+\beta_1)\Gamma(\beta_K+\beta_1+2)}2^{\beta_K+\beta_1+1}R^{\beta_K+\beta_1-1}\\
\times\Bigg(\sideset{}{_{x_{K-1}}^{\infty}}\int+\sideset{}{_{-\infty}^{x_2}}\int\Bigg)\mathcal{N}\Bigg[\prod_{j=2}^{K-1}(u_2-x_j)^{\beta_j}\Bigg]\,{\rm d}u_2,\quad \beta_K=\beta_1,\end{multline}
which replaces (\ref{result4}).  
\end{enumerate}

The results (\ref{result4i=1}, \ref{result4i=2}, \ref{result4i=K-1}) found above in items \ref{itemi=1}--\ref{itemi=K-1} respectively are essentially identical to (\ref{result4}).  Thus, by finishing the remaining analysis of section \ref{s4} from (\ref{result4.5}) to (\ref{firstlimprime2}), we find the following result for the limit $\bar{\ell}_1\mathcal{F}_{c,\vartheta}$, for use in the proof of corollary \ref{moveconjchargecor}.
\begin{quote}\textbf{Main result 1:} In case \ref{sc4new} with $i\in\{1,2,\ldots,2N-1\}$, where $\Gamma_1''$ and $\Gamma_2''$ (\ref{contouri=1}--\ref{contouri=2N-1}) respectively replace $\Gamma_1$ and $\Gamma_2$ in the formula (\ref{firstFexplicit1}) for $\mathcal{F}_{c,\vartheta}$, the limit $\bar{\ell}_1\mathcal{F}_{c,\vartheta}$ (\ref{nextthelim}) equals the element of $\mathcal{S}_{N-1}$ generated from the formula for $\mathcal{F}_{c,\vartheta}$ by dropping all factors involving $x_i$, $x_{i+1}$, and $u_1$, dropping the integration along $\Gamma_1$, replacing $\Gamma_2$ by a Pochhammer contour $\Gamma_0'$ (or simple contour with the same endpoints if $\kappa>4$, as per item \ref{itemc} in the introduction of this appendix) described below, and reducing the power of the prefactor (\ref{firstprefactor}) or (\ref{secondprefactor}) by one.  If $x_j\neq x_i$ and $x_k\neq x_{i+1}$ are the other endpoints of $\Gamma_1''$ and $\Gamma_2''$ respectively, then $\Gamma_0'=\mathscr{P}(x_j,x_k)$.
\end{quote}

The result (\ref{result4i=K}) found in item \ref{itemi=K} of this section is very similar to (\ref{result4}).  However, if $i=2N$, then we require the limit (\ref{thelim}) for $\kappa\in(0,8)$ and $8/\kappa\not\in\mathbb{Z}^+$.  For $\kappa\leq4$, this limit is 
\begin{multline}\label{thebiglimit2i=2N}\underline{\ell}_1\big(\mathcal{F}_{c,\vartheta}\big|_{(\Gamma_1,\Gamma_2)\mapsto(\Gamma_1'',\Gamma_2'')}\big)\,(\kappa\,|x_2,x_3,\ldots,x_{2N-1})\,\,=\lim_{R\rightarrow\infty}(2R)^{6/\kappa-1}\,\,\times\\ 
(\mathcal{F}_{c,\vartheta}|_{(\Gamma_1,\Gamma_2)\mapsto(\Gamma_1'',\Gamma_2'')})\,(\kappa\,|\,x_1=-R,x_2,\ldots,x_{2N}=R) \\ \parallel \\ 
\boxed{\begin{aligned}&n(\kappa)\left[\frac{n(\kappa)\Gamma(2-8/\kappa)}{4\sin^2(4\pi/\kappa)\Gamma(1-4/\kappa)^2}\right]^{N-1}\Bigg(\prod_{\substack{2\leq j<k\\j,k\neq c}}^{2N-1}(x_k-x_j)^{2/\kappa}\Bigg)\Bigg(\prod_{\substack{k=2 \\ k\neq c}}^{2N-1}|x_c-x_k|^{1-6/\kappa}\Bigg)\Big(2R\Big)^{2/\kappa}\sideset{}{_{\Gamma_{N-1}}}\oint {\rm d}u_{N-1}\\
&\dotsm\sideset{}{_{\Gamma_4}}\oint {\rm d}u_4\,\,\sideset{}{_{\Gamma_3}}\oint {\rm d}u_3\,\,\mathcal{N}\Bigg[\Bigg(\prod_{\substack{l=2 \\ l\neq c}}^{2N-1}\prod_{m=3}^{N-1}(x_l-u_m)^{-4/\kappa}\Bigg)\Bigg(\prod_{m=3}^{N-1}(x_c-u_m)^{12/\kappa-2}\Bigg)\Bigg(\prod_{3\leq p<q}^{N-1}(u_p-u_q)^{8/\kappa}\Bigg)\\
&\Bigg(\prod_{m=3}^{N-1}(u_m+R)^{-4/\kappa}(R-u_m)^{-4/\kappa}\Bigg)\Bigg(\prod_{\substack{j=2\\ j\neq c}}^{2N-1}(x_j+R)^{2/\kappa}(R-x_j)^{2/\kappa}\Bigg)\Big(x_c+R\Big)^{1-6/\kappa}\Big(R-x_c\Big)^{1-6/\kappa}\Bigg]\\
&\underbrace{\left\{\begin{aligned}
\sideset{}{_{\mathclap{\hspace{1.2cm}\mathscr{P}(x_1,x_2)}}}\oint\hspace{1cm}{\rm d}u_2\hspace{.1cm}\sideset{}{_{\mathclap{\hspace{2cm}\mathscr{P}(x_{2N-1},x_{2N})}}}\oint\hspace{1.8cm}{\rm d}u_1\,\,\mathcal{N}\Bigg[\Big(u_1+R\Big)^{-4/\kappa}\Big(R-u_1\Big)^{-4/\kappa}\Big(u_2+R\Big)^{-4/\kappa}\Big(R-u_2\Big)^{-4/\kappa}\Big(u_2-u_1\Big)^{8/\kappa}\\
\Big(x_c-u_1\Big)^{12/\kappa-2}\Big(x_c-u_2\Big)^{12/\kappa-2}\Bigg(\prod_{\substack{l=2 \\ l\neq c}}^{2N-1}(x_l-u_1)^{-4/\kappa}(x_l-u_2)^{-4/\kappa}\Bigg)\Bigg(\prod_{m=3}^{N-1}(u_m-u_1)^{8/\kappa}(u_m-u_2)^{8/\kappa}\Bigg)\Bigg]\end{aligned}\right\}.}_{\mathcal{J}^{(2,K)}}\end{aligned}}\end{multline}
(If $\kappa>4$, then we adjust (\ref{thebiglimit2i=2N}) as per item \ref{itemc} in the introduction of this appendix.)  We have rewritten the formula (\ref{firstFexplicit1}) for $\mathcal{F}_{c,\vartheta}(\kappa\,|\,\boldsymbol{x})$ slightly to clarify the calculation, and we indicate the contour integral $\mathcal{J}^{(2,K)}$ (\ref{scenario4}) with braces.

Now we find the limit (\ref{thebiglimit2i=2N}).  After doing steps \ref{oneit2}--\ref{lastit2} of section \ref{s4}, identifying the double integral with respect to $u_1$ and $u_2$ with (\ref{scenario4}, \ref{application}, \ref{contouri=2N}) (with $\beta_{K+1}=\beta_1$ and $\beta_{K+2}=\beta_2$), and replacing it by the right side of (\ref{result4i=K}), we find 
\begin{multline}\label{asympinsert2i=2N}(2R)^{6/\kappa-1}\big(\mathcal{F}_{c,\vartheta}\big|_{(\Gamma_1,\Gamma_2)\mapsto(\Gamma_1'',\Gamma_2'')}\big)\,(\kappa\,|\,x_1=-R,x_2,\ldots,x_{2N}=R)\underset{R\rightarrow\infty}{\sim}(2R)^{6/\kappa-1}\,\,\times\\ 
\boxed{\begin{aligned}n(\kappa)\left[\frac{n(\kappa)\Gamma(2-8/\kappa)}{4\sin^2(4\pi/\kappa)\Gamma(1-4/\kappa)^2}\right]^{N-1}\Bigg(\prod_{\substack{2\leq j<k\\j,k\neq c}}^{2N-1}(x_k-x_j)^{2/\kappa}\Bigg)\Bigg(\prod_{\substack{k=2 \\ k\neq c}}^{2N-1}|x_c-x_k|^{1-6/\kappa}\Bigg)\Big(2R\Big)^{2/\kappa}\sideset{}{_{\Gamma_{N-1}}}\oint {\rm d}u_{N-1}\\
\dotsm\,\sideset{}{_{\Gamma_4}}\oint {\rm d}u_4\,\, \sideset{}{_{\Gamma_3}}\oint {\rm d}u_3\,\,\mathcal{N}\Bigg[\Bigg(\prod_{\substack{l=2 \\ l\neq c}}^{2N-1}\prod_{m=3}^{N-1}(x_l-u_m)^{-4/\kappa}\Bigg)\Bigg(\prod_{m=3}^{N-1}(x_c-u_m)^{12/\kappa-2}\Bigg)\Bigg(\prod_{3\leq p<q}^{N-1}(u_p-u_q)^{8/\kappa}\Bigg)\Bigg]R^2\\
\underbrace{\left\{\begin{aligned}\frac{4\sin^2(4\pi/\kappa)\Gamma(1-4/\kappa)^2}{n(\kappa)\Gamma(2-8/\kappa)}\,\,2^{1-8/\kappa}\,\,R^{-8/\kappa-1}\sideset{}{_{\mathscr{P}(x_{2N-1},x_2)}}\oint {\rm d}u_2\,\,\mathcal{N}\Bigg[\Big(x_c-u_2\Big)^{12/\kappa-2}\hspace{2cm}\\
\Bigg(\prod_{\substack{l=2 \\ l\neq c}}^{2N-1}(x_l-u_2)^{-4/\kappa}\Bigg)\Bigg(\prod_{m=3}^{N-1}(u_m-u_2)^{8/\kappa}\Bigg)\Bigg]\end{aligned}\right\}.}_{\text{right side of (\ref{result4i=K})}}\end{aligned}}\end{multline}
(Because the integration contours in present use are Pochhammer contours, we must include the factor $4\sin^2\beta_i=4\sin^2(4\pi/\kappa)$ on the right side of (\ref{result4.75}) with this substitution.)  After some simplification, we finally send $R\rightarrow\infty$ in (\ref{asympinsert2i=2N}) to find
\begin{multline}\label{firstlimprime2i=2N}\underline{\ell}_1\big(\mathcal{F}_{c,\vartheta}\big|_{(\Gamma_1,\Gamma_2)\mapsto(\Gamma_1'',\Gamma_2'')}\big)\,(\kappa\,|\,x_2,x_3,\ldots,x_{2N-1})\,\,=\\
\left\{\begin{aligned}&n(\kappa)\left[\frac{n(\kappa)\Gamma(2-8/\kappa)}{4\sin^2(4\pi/\kappa)\Gamma(1-4/\kappa)^2}\right]^{N-2}\Bigg(\prod_{\substack{2\leq j<k \\ j,k\neq c}}^{2N-1}(x_k-x_j)^{2/\kappa}\Bigg)\Bigg(\prod_{\substack{k=2 \\ k\neq c}}^{2N-1}|x_c-x_k|^{1-6/\kappa}\Bigg)\oint_{\Gamma_{N-1}} {\rm d}u_{N-1}\dotsm\\ 
&\oint_{\Gamma_3}{\rm d}u_3\,\,\sideset{}{_{\mathclap{\hspace{2cm}\mathscr{P}(x_{2N-1},x_2)}}}\oint \hspace{1.2cm}{\rm d}u_2\,\,\mathcal{N}\Bigg[\Bigg(\prod_{\substack{l=2 \\ l\neq c}}^{2N-1}\prod_{m=2}^{N-1}(x_l-u_m)^{-4/\kappa}\Bigg)\Bigg(\prod_{m=2}^{N-1}(x_c-u_m)^{12/\kappa-2}\Bigg)\Bigg(\prod_{2\leq p<q}^{N-1}(u_p-u_q)^{8/\kappa}\Bigg)\Bigg]\end{aligned}\right\}.\end{multline}
(Again, if $\kappa>4$, then we find the same result, but with all factors of $4\sin^2(4\pi/\kappa)$ dropped and with all Pochhammer contours replaced by simple contours with the same endpoints and orientation.)

Now, the quantity of (\ref{firstlimprime2i=2N}) in brackets is the element of $\mathcal{S}_{N-1}$ described in the following conclusion.  We thus have our main result, for use in item \ref{fourthcase2} in the proof of corollary \ref{moveconjchargecor}.
\begin{quote}\textbf{Main result 2:} In case \ref{sc4new} with $i=2N$, where $\Gamma_1''$ and $\Gamma_2''$ (\ref{contouri=2N}) respectively replace $\Gamma_1$ and $\Gamma_2$ in the formula (\ref{firstFexplicit1}) for $\mathcal{F}_{c,\vartheta}$, the limit $\underline{\ell}_1\mathcal{F}_{c,\vartheta}$ (\ref{nextthelim}) equals the element of $\mathcal{S}_{N-1}$ generated from the formula for $\mathcal{F}_{c,\vartheta}$ by dropping all factors involving $x_{2N}$, $x_1$, and $u_1$, replacing $\Gamma_2$ by $\mathscr{P}(x_{2N-1},x_2)$ (or $[x_{2N-1},x_2]^+$ if $\kappa>4$, as per item \ref{itemc} in the introduction of this appendix), and reducing the power of the prefactor (\ref{firstprefactor}) or (\ref{secondprefactor}) by one.
\end{quote}
\noindent
(Again, we assume here that $8/\kappa\not\in\mathbb{Z}^+$, but these results remain true even if this condition is not met.  See the last paragraph of section \ref{s4}.)

\section{A proof that $\mathcal{F}_{c,\vartheta}$ is real-valued}\label{reallemproof}

The purpose of this section is to prove the following lemma.
\begin{lem}\label{reallem}  If $\kappa>0$, then for each $c\in\{1,2,\ldots,2N\}$ and $\vartheta\in\{1,2,\ldots,C_N\}$, the function $\mathcal{F}_{c,\vartheta}$ of definition \ref{Fkdefn} (with part of its domain naturally extended to $\kappa\geq8$) is real-valued.
\end{lem}

\begin{figure}[t]
\centering
\includegraphics[scale=0.27]{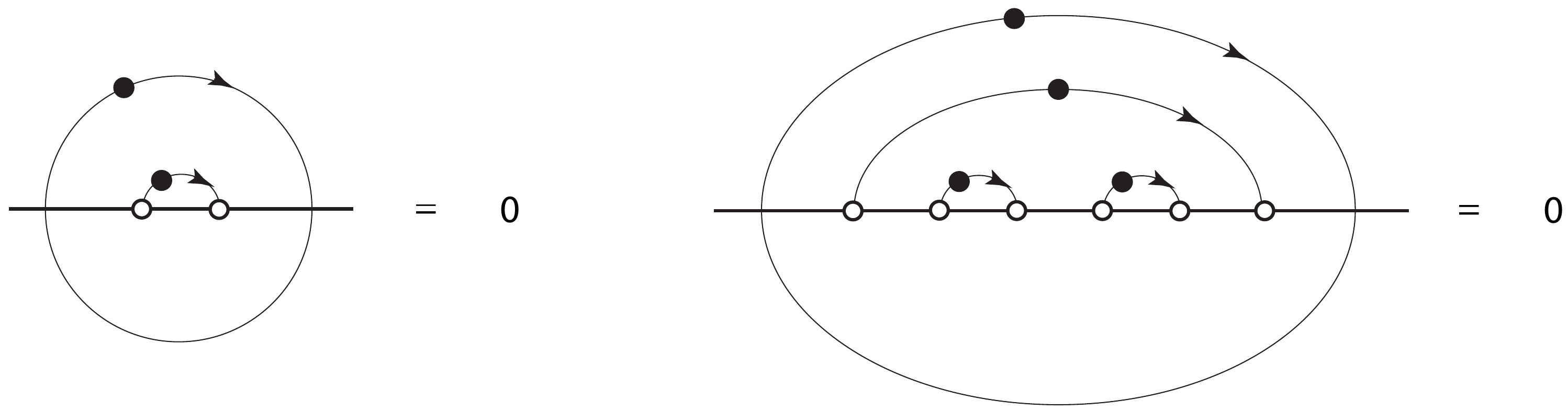}
\caption{If an integration contour surrounds a collection of other integration contours of $\mathcal{F}_{c,\vartheta}$, then the integration along all of these contours gives zero.}
\label{Oint}
\end{figure}

\begin{proof} K.\ Kyt\"ol\"a suggested to us the idea for this proof.  To begin, we suppose that $\kappa>4$ so the integration contours of $\mathcal{F}_{c,\vartheta}$ are simple.  In order to prove that $\mathcal{F}_{c,\vartheta}$ is real-valued, we deform each of its integration contours $\Gamma_m$, all of which are in the upper half-plane, into the mirror image $\bar{\Gamma}_m:=\{\bar{z}\in\mathbb{C}\,|\,z\in\Gamma_m\}$.  Then we show that the combined effect of these deformations is to send $\mathcal{F}_{c,\vartheta}(\kappa\,|\boldsymbol{x})$ to its complex conjugate and to leave it unchanged for all $(\kappa,\boldsymbol{x})\in\mathbb{R}\times\Omega_0$ at the same time.  Thus, we conclude that $\mathcal{F}_{c,\vartheta}(\kappa)$ is real-valued for all $\kappa>0$.

First, we continuously deform each contour $\Gamma_m$ of $\mathcal{F}_{c,\vartheta}$ into its mirror image $\bar{\Gamma}_m$ without changing the function $\mathcal{F}_{c,\vartheta}$.  This process invokes the Cauchy integral theorem and the following identity.  If $\kappa>4$, $\Gamma_1$ is a simple contour with endpoints at $x_i<x_j$ (we recall that $i\neq j\neq c$), and $\Gamma_2$ is a simple loop surrounding $\Gamma_1$, then (figure \ref{Oint})
\be\label{oint}\int_{\Gamma_1}\sideset{}{_{\Gamma_2}}\oint
[\,\ldots\,\text{the integrand of $\mathcal{F}_{c,\vartheta}$ (\ref{firstFexplicit1}) \,\ldots}]\,{\rm d}u_2\,{\rm d}u_1=0.\ee
To verify (\ref{oint}), we wrap $\Gamma_2$ tightly around $\Gamma_1$ and parameterize these contours as $u_1:[0,2]\rightarrow\mathbb{C}$ and $u_2:[0,1]\rightarrow\mathbb{C}$ so $u_1(t) \approx u_2(t) \approx u_2(2-t)$ for $t\in[0,1]$.  After choosing an arbitrary $s\in[0,1]$, we note that the four segments
\be\label{segments}u_2[0,s],\quad u_2[s,1],\quad u_2[1,2-s],\quad u_2[2-s,2]\ee
of $\Gamma_2$ are immediately above $u_1[0,s]$ and $u_1[s,1]$ and below $u_1[s,1]$ and $u_1[0,s]$ respectively (figure \ref{Loop}), and their near agreement with these segments of $\Gamma_1$ only improves as we wrap $\Gamma_2$ more tightly around $\Gamma_1$.  

Now we decompose the integration with respect to $u_2$ in (\ref{oint}) into integrations along the four segments of (\ref{segments}).  If $0\leq t<s$, then $u_1(s)-u_2(t)$ is in the first or fourth quadrant of the complex plane, so we have $\Delta(u_1(s)-u_2(t))=(u_1(s)-u_2(t))^{8/\kappa}$ (\ref{Deltadefn}).  After integrating $t$ over this range, we are left with 
\be\label{I1ex}I_1(s):=\sideset{}{_0^s}\int\left[\dotsm\begin{array}{l}\text{the integrand of $\mathcal{F}_{c,\vartheta}$ (\ref{firstFexplicit1}) with $u_1=u_1(s)$ and $u_2=$}\\
\text{$u_2(t)$ and $\Delta(u_1-u_2)$ replaced by $(u_1(s)-u_2(t))^{8/\kappa}$}\end{array}\dotsm\right]\,u_1'(s)\,u_2'(t)\,{\rm d}t.\\
\ee
Now as $t$ increases beyond $s$, the difference $u_1(s)-u_2(t)$ passes from the fourth quadrant into the third quadrant.  If $u_2-u_1$ is in either of these quadrants, then we may write (recalling that $-\pi<\arg z\leq\pi$ for all complex $z$)
\be\label{redo}\Delta(u_1-u_2)=e^{-8\pi i/\kappa}(u_2-u_1)^{8/\kappa}.\ee
After replacing $\Delta(u_1-u_2)$ with the right side of (\ref{redo}), but for the moment leaving off the phase factor of $\exp[-8\pi i/\kappa]$, in the integrand of $\mathcal{F}_{c,\vartheta}$ (\ref{firstFexplicit1}) and integrating the result over $s\leq t\leq1$, we are left with 
\be\label{I2ex}I_2(s):=\sideset{}{_s^1}\int\left[\dotsm\begin{array}{l}\text{the integrand of $\mathcal{F}_{c,\vartheta}$ (\ref{firstFexplicit1}) with $u_1=u_1(s)$ and $u_2=$}\\
\text{$u_2(t)$ and $\Delta(u_1-u_2)$ replaced by $(u_2(t)-u_1(s))^{8/\kappa}$}\end{array}\dotsm\right]\,u_1'(s)\,u_2'(t)\,{\rm d}t.\\
\ee
Now, by further contracting the top $u_2[0,1]$ and bottom $u_2[1,2]$ portions of $\Gamma_2$ so they perfectly agree with $\Gamma_1$, the integration of (\ref{oint}) with respect to only $u_2$ gives (figure \ref{Loop})
\be\label{ointdecomp}I_1(s)+e^{-8\pi i/\kappa}I_2(s)-I_2(s)-e^{-8\pi i/\kappa}I_1(s),\ee
with $u_2(t)=u_1(t)$ for all $t\in[0,1]$ and $u_2(t)=u_1(2-t)$ for all $t\in[1,2]$.  The first and second terms of (\ref{ointdecomp}) arise from integrating along $u_2[0,s]$ and $u_2[s,1]$ respectively.  The third arises from integrating along $u_2[1,2-s]$ after acquiring an additional phase factor of $\exp[-2\pi i(-4/\kappa)]$ from cycling clockwise around $x_j$.  The fourth arises from integrating along $u_2[2-s,2]$.  Its phase factor equals the product of $\exp[-2\pi i(-4/\kappa)]$, acquired from cycling clockwise around $x_j$, with $\exp[-2\pi i(8/\kappa)]$, acquired from crossing the branch cut of $\Delta(u_1-u_2)$ as $u_1-u_2$ passes from the second quadrant into the first. (This branch cut anchors to $u_1$ and points downward as we view $\Delta(u_1-u_2)$ as a function of $u_2$.)

\begin{figure}[t]
\centering
\includegraphics[scale=0.27]{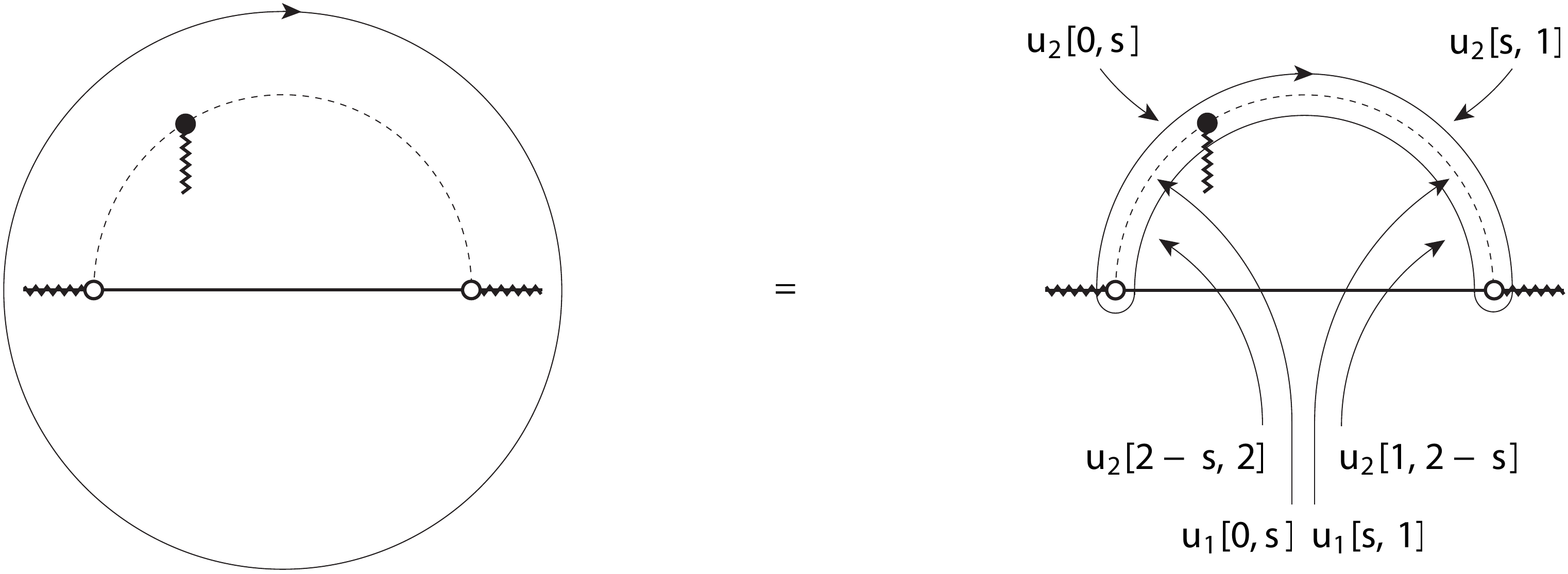}
\caption{Illustration of the proof of identity (\ref{oint}).}
\label{Loop}
\end{figure}

Finally, we note that the integrands of $I_1(s)$ (\ref{I1ex}) and $I_2(s)$ (\ref{I2ex}) exchange under the switch from $(s,t)\in[0,1]^2$ to $(t,s)\in[0,1]^2$.  Therefore, the definite integral of $I_1(s)$ over $0\leq s\leq1$ equals that of $I_2(s)$, so 
\be\label{intointdecomp}\int_{\Gamma_1}\sideset{}{_{\Gamma_2}}\oint[\,\ldots\,\text{the integrand of $\mathcal{F}_{c,\vartheta}$ (\ref{firstFexplicit1}) \ldots}\,]\,{\rm d}u_2\,{\rm d}u_1=(1-e^{-8\pi i/\kappa})\int_0^1[I_1(s)-I_2(s)]\,{\rm d}s=0,\ee
proving the identity (\ref{oint}) for $\kappa>4$.

In fact, if we suppose that $\kappa$ is complex with $\text{Re}\,\kappa>0$, then the identity (\ref{oint}), with $\Gamma_1$ replaced by the Pochhammer contour $\mathscr{P}(x_i,x_j)$, remains true.  (Though we do not use this fact in this proof, it is useful in other applications, so we mention it now.)  Indeed, thanks to identity (\ref{Pochtostraight}) with $[x_i,x_j]=\Gamma_1$, if $\text{Re}\,\kappa>4$, then
\begin{multline}\label{analyticcont2}\oint_{\mathscr{P}(x_i,x_j)}\sideset{}{_{\Gamma_2}}\oint[\,\ldots\,\text{the integrand of $\mathcal{F}_{c,\vartheta}$ (\ref{firstFexplicit1}) \ldots}\,]\,{\rm d}u_2\,{\rm d}u_1\\
=4\sin^2(4\pi/\kappa)\int_{\Gamma_1}\sideset{}{_{\Gamma_2}}\oint[\,\ldots\,\text{the integrand of $\mathcal{F}_{c,\vartheta}$ (\ref{firstFexplicit1}) \ldots}\,]\,{\rm d}u_2\,{\rm d}u_1.\end{multline}
The left side of (\ref{analyticcont2}) gives the analytic continuation of the right side to the half-plane $\{\kappa\in\mathbb{C}\,|\,\text{Re}(\kappa)>0\}$.  Because the right side of (\ref{analyticcont2}) vanishes for all (real) $\kappa>4$, it follows that both sides of (\ref{analyticcont2}) vanish on this entire half-plane.  Thus, if $\text{Re}\,\kappa>0$, then identity (\ref{oint}) with $\Gamma_1$ replaced by the Pochhammer contour $\mathscr{P}(x_i,x_j)$ remains true.

Furthermore, we may extend identity (\ref{oint}) as follows.  If $\Gamma_m$ is a simple loop surrounding the integration contours $\Gamma_1$, $\Gamma_2,\ldots,\Gamma_{m-1}$ of $\mathcal{F}_{c,\vartheta}$ with $m\leq N-1$, then (figure \ref{Oint})
\be\label{ointext}\sideset{}{_{\Gamma_1}}\int\sideset{}{_{\Gamma_2}}\int\dotsm\sideset{}{_{\Gamma_m}}\oint
[\,\ldots\,\text{the integrand of $\mathcal{F}_{c,\vartheta}$ (\ref{firstFexplicit1}) \ldots}\,]\,{\rm d}u_m\dotsm {\rm d}u_2\,{\rm d}u_1=0.\ee
We prove (\ref{ointext}) by decomposing $\Gamma_m$ into loops that each surround a single integration contour, using Fubini's theorem to re-order the integrations in each resulting term, and invoking identity (\ref{oint}) to show that each term is zero.  Again, by replacing simple contours with Pochhammer contours, we extend this identity to the half-plane $\{\kappa\in\mathbb{C}\,|\,\text{Re}(\kappa)>0\}$.

\begin{figure}[b]
\centering
\includegraphics[scale=0.27]{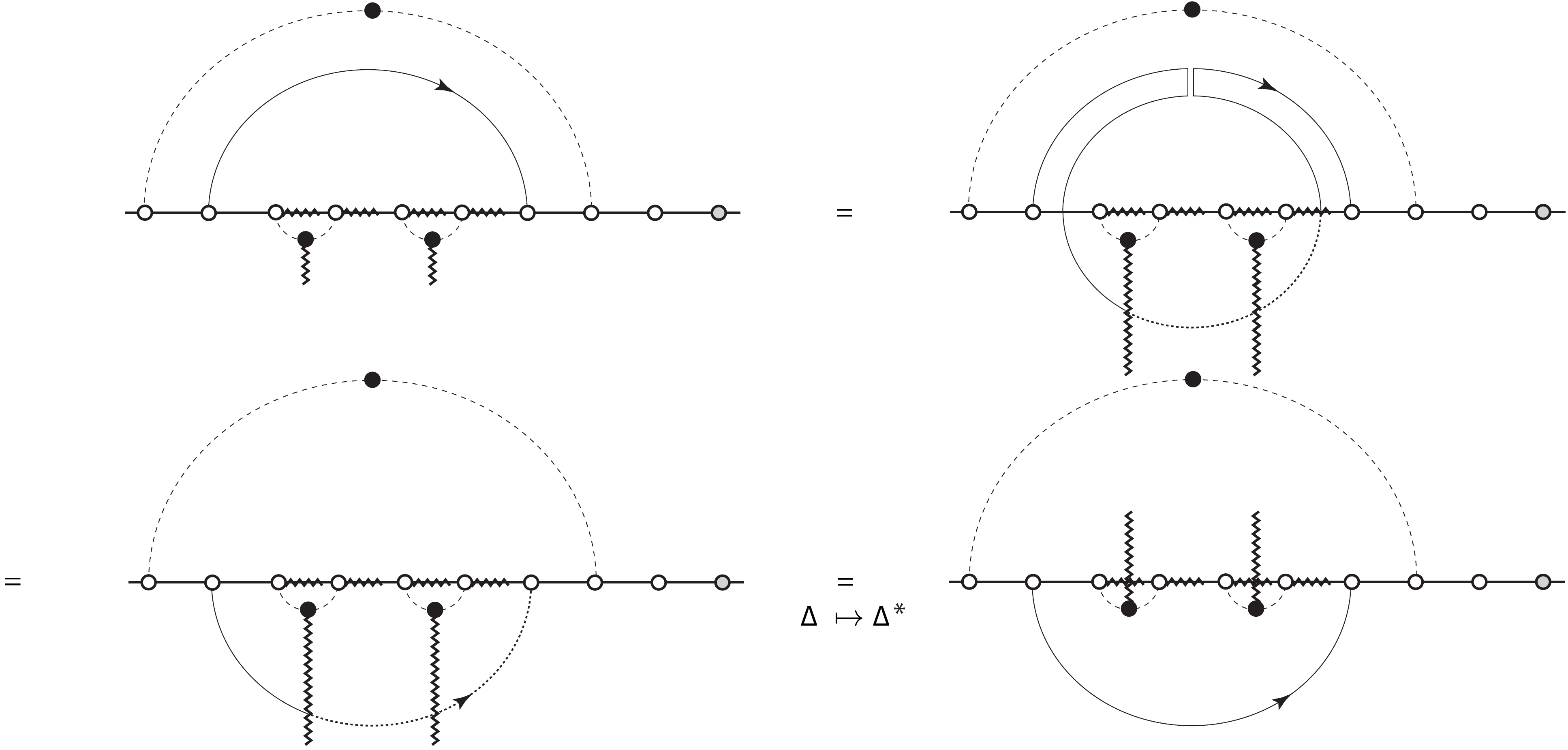}
\caption{Deforming an integration contour (solid curve) into the lower half-plane.  The coarse-dashed curves are other contours, the fine-dashed curve is the solid curve after it passes onto a different Riemann sheet, and the zig-zags indicate branch cuts.}
\label{Deform1}
\end{figure}

Now we deform every integration contour $\Gamma_m$ of $\mathcal{F}_{c,\vartheta}$ into its mirror image $\bar{\Gamma}_m$ in the lower half-plane.  (For simplicity, we again assume that $\kappa>4$ for now, so all contours are simple.)  Each deformation is a homotopy that fixes the endpoints of its contour.  We first deform all of those contours that do not pass over another integration contour.  Then, we deform all contours that pass over only the contours deformed in the previous step.  Then, we deform all contours that pass over only the contours deformed in the previous two steps, and so on.  We repeat until we have deformed all contours into their mirror images in the lower half-plane.

In order to deform one of the contours $\Gamma_m$ into its mirror image, the homotopy may need to push $\Gamma_m$ through some of the other contours of $\mathcal{F}_{c,\vartheta}$.  By using identities (\ref{oint}, \ref{ointext}), we do this without changing the function $\mathcal{F}_{c,\vartheta}$.  However, the homotopy must not collide the deformed contour with either $x_{\iota(2N-1)}$ or $x_{\iota(2N)}$ (one of these is $x_c$) because these points are not endpoints of any other contour.  Therefore, if $x_c$ does not lie between the endpoints of $\Gamma_m$, then its homotopy must pass between those endpoints, and if otherwise, then its homotopy must pass through the upper half-plane, through infinity, and into the lower half-plane.

In the first case (with $x_c$  not between the endpoints of $\Gamma_m$), we must deform a contour $\Gamma_q$ through the space between its endpoints.  If $\Gamma_q$ does not pass over another contour, then this homotopy is straightforward.  But if otherwise, then the nested contours seem to block the homotopy.  In order to maneuver around them, we insert into $\Gamma_q$ a large loop that surrounds all of these blocking contours, we pull this manipulated contour apart into left and right segments at the insertion point, and we continuously deform these two segments until their combination gives $\bar{\Gamma}_q$.  As a consequence of identities (\ref{oint}, \ref{ointext}) and the (strong form of the) Cauchy integral theorem (Thm.\ 2.3 of \cite{kod}), neither inserting this loop nor the subsequent homotopy alters $\mathcal{F}_{c,\vartheta}$.

The integrand (\ref{eulerintegralch2}) of $\mathcal{F}_{c,\vartheta}$ has many branch cuts as a function of the integration variable $u_q$.  Indeed, if the contour $\Gamma_q$ passes over another contour $\Gamma_p$, then a branch cut anchors to $u_p\in\bar{\Gamma}_p$  and points downward to its other endpoint at infinity.  (We recall that we deformed $\Gamma_p$ into its mirror image before deforming $\Gamma_q$ because the latter passes over the former.)  Furthermore, a branch cut anchors to each endpoint of $\Gamma_p$ and points rightwards to its other endpoint at infinity.  As we deform $\Gamma_q$, its deformed version, at every moment of the homotopy and followed in the direction of its orientation from left to right, crosses the branch cut anchored to $u_p$ first, then crosses the branch cuts anchored to the endpoints of $\Gamma_p$ next, and finally terminates at its right endpoint (figure \ref{Deform1}).  Furthermore, this deformed contour crosses only these mentioned branch cuts and in this order during its homotopy.

Still viewing the integrand of $\mathcal{F}_{c,\vartheta}$ as a function of $u_q$, we study what happens to it as the deformed version of $\Gamma_q$ crosses one of the branch cuts.   Starting at its left endpoint and following it rightward, this deformed contour initially resides on the same Riemann sheet of the factor $\Delta(u_p-u_q)$ (\ref{Deltadefn}) as that of the original contour $\Gamma_q$.  But as it crosses the branch cut anchored to $u_p$, it passes onto a new sheet, and the integrand acquires a phase factor of $\exp[2\pi i(8/\kappa)]$.  To show this acquisition explicitly, we insert the replacement 
\be\label{bardelta}\Delta(u_p-u_q)\mapsto\Delta^*(u_p-u_q),\qquad\Delta^*(z):=z^{8/\kappa}\times\begin{cases}1, & -\pi/2<\arg z\leq\pi \\ e^{2\pi i(8/\kappa)}, & -\pi<\arg z\leq-\pi/2\end{cases}\ee
in the integrand of $\mathcal{F}_{c,\vartheta}$.  We notice that the branch cut for $\Delta^*$ lies on the negative imaginary axis.  As such, this replacement (\ref{bardelta}) rotates the direction of the branch cut anchored to $u_p$ so it points upward (figure \ref{Deform1}).

\begin{figure}[t]
\centering
\includegraphics[scale=0.27]{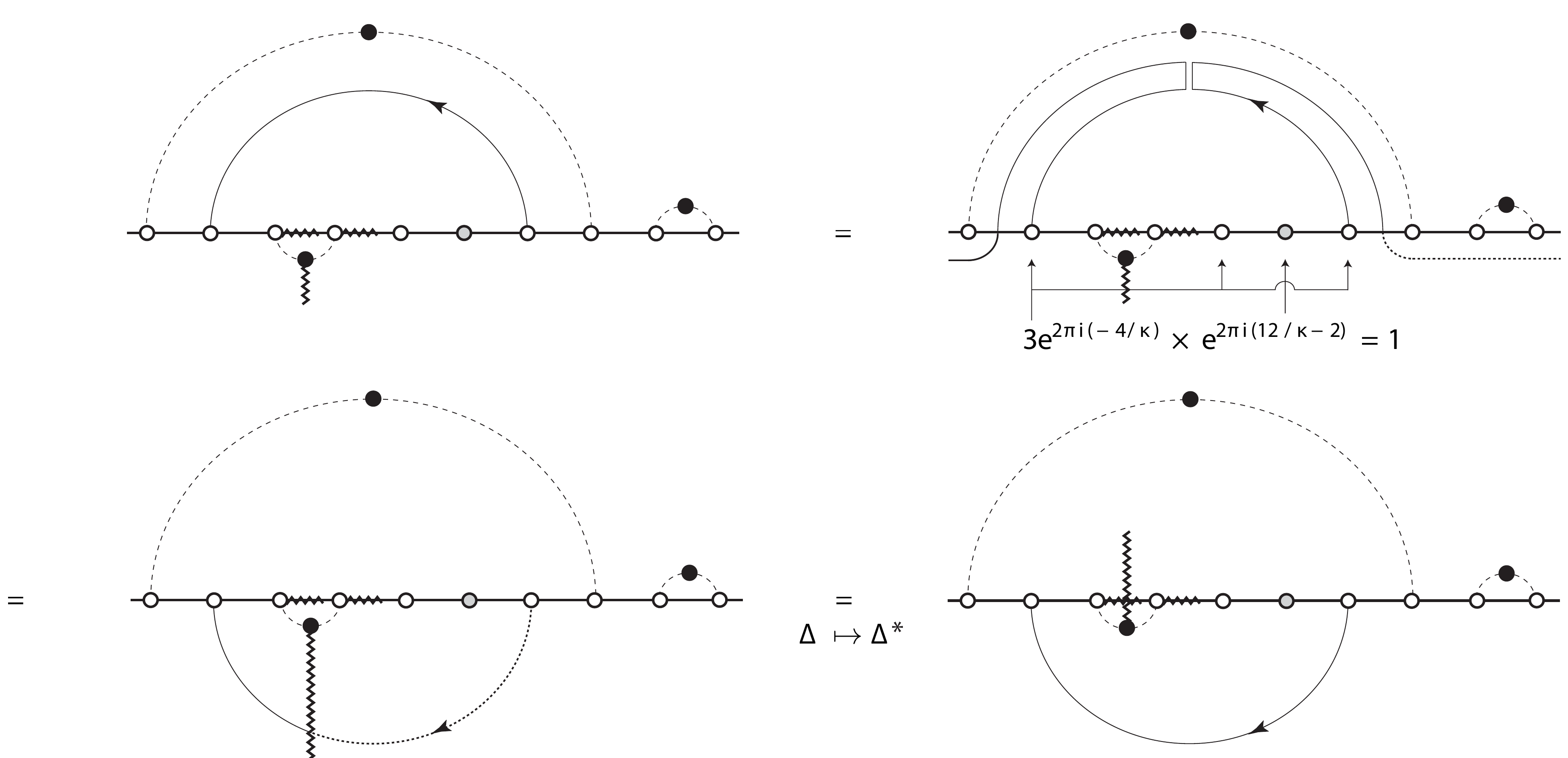}
\caption{Deforming an integration contour into the lower half-plane.  The key difference between this figure and figure \ref{Deform1} is that here, the contour is deformed outward through the upper half-plane, through infinity, and into the lower half-plane.}
\label{Deform2}
\end{figure}

In the second case (with $x_c$  between the endpoints of $\Gamma_m$), we must deform a contour $\Gamma_q$ through the upper half-plane and into the lower half-plane.  If $\Gamma_q$ passes over all of the other contours, then this homotopy is straightforward.  But if otherwise, then the un-nested contours seem to block the homotopy.  In order to maneuver around them, we insert into $\Gamma_q$ a large loop that surrounds all of these blocking contours, we pull this manipulated contour apart into left and right segments at the insertion point, giving it the approximate shape of a semicircle with its base just beneath the real axis and its arc in the upper half-plane, and we uniformly expand this contour outward through the upper half-plane (figure \ref{Deform2}).

If $R$ is the radius of the semicircle, then the integrand of $\mathcal{F}_{c,\vartheta}$ with $u_q$ restricted to the arc of that semicircle vanishes like $R^{-1}$ as $R\rightarrow\infty$.  (See section \ref{s3} of appendix \ref{asymp}.)  Therefore, we restrict our attention to the integration along the base.  If $R$ is very large, then the base essentially comprises two halves, one to the left of all other contours of $\mathcal{F}_{c,\vartheta}$ and one to the right.  Now, the left half resides on the same Riemann sheet as the original contour $\Gamma_q$, but the sheet on which the right half resides depends on whether $\Gamma_q$ passes over any other contours.  If it does, then the right half crosses all branch cuts anchored to the endpoints of those nested contours immediately upon leaving its left endpoint (figure \ref{Deform2}), and the integrand of $\mathcal{F}_{c,\vartheta}$ accumulates some phase factors explicitly described in the next paragraph.  The right half also crosses the branch cuts anchored to the endpoints of $\Gamma_q$ and the points $x_{\iota(2N-1)}$ and $x_{\iota(2N)}$ between them.  After this happens, the integrand accumulates three phase factors of $\exp[-2\pi i(-4/\kappa)]$ and one phase factor of $\exp[-2\pi i(12/\kappa-2)]$ from these crossings. However the product of these factors is one, so these four crossings do not cause the right half to pass onto yet another different Riemann sheet.

The limit $R\rightarrow\infty$ joins the distant endpoints of the semicircle's base at infinity, and after they meet, we retract them away from infinity, join them together with a semicircular arc in the lower half-plane, contract this semicircle inward, and deform the result into $\bar{\Gamma}_q$.  According to item \ref{third} of definition \ref{Fkdefn}, the orientation of $\Gamma_q$, and therefore of $\bar{\Gamma}_q$, points leftward because the point $x_c$ lies between its endpoints.  In light of this fact, the previous paragraph implies that if $\Gamma_q$ passes over another contour $\Gamma_p$, then its mirror image $\bar{\Gamma}_q$ crosses both of the branch cuts anchored to the endpoints of $\Gamma_p$ immediately after leaving its right endpoint, and the integrand of $\mathcal{F}_{c,\vartheta}$ acquires a phase factor of $\exp[-2\pi i(-4/\kappa)]$ from both crossings.  Subsequently, $\bar{\Gamma}_q$ crosses the branch cut anchored to $u_p$, giving the integrand of $\mathcal{F}_{c,\vartheta}$ another phase factor of $\exp[-2\pi i(8/\kappa)]$ that multiplies the phase factors of the two previous crossings.  Because this product equals one, the mirror-image contour $\bar{\Gamma}_q$ returns to the Riemann sheet on which the original contour $\Gamma_q$ resides.  Finally, $\bar{\Gamma}_q$ terminates at its left endpoint.  To show these acquired phase factors explicitly, we replace the factor of $\Delta(u_p-u_q)$ in the integrand of $\mathcal{F}_{c,\vartheta}$ with $\Delta^*(u_p-u_q)$ (\ref{bardelta}) for the entire integration along $\bar{\Gamma}_q$ (figure \ref{Deform2}).

Finally, we explain the combined effect of all of these contour homotopies on $\mathcal{F}_{c,\vartheta}$.  We have already showed that none change this function, thanks to identities (\ref{oint}, \ref{ointext}) and the (strong form of the) Cauchy integral theorem (Thm.\ 2.3 of \cite{kod}).  Now we show that their combined effect also sends this function to its complex conjugate.  Among the factors in the integrand (\ref{eulerintegralch2}) of $\mathcal{F}_{c,\vartheta}$, the homotopy of $\Gamma_m$ sends (we recall that $-\pi<\arg z\leq\pi$ for all complex $z$)
\be\label{switchit}(u_m-x_l)^{\beta_l}\xrightarrow[\Gamma_m\longrightarrow\bar{\Gamma}_m]{} (\bar{u}_m-x_l)^{\beta_l}=\overline{(u_m-x_l)^{\beta_l}},\quad u_m\in\Gamma_m,\quad x_l\in\mathbb{R},\ee
for each $m\in\{1,2,\ldots,N-1\}$ and $l\in\{1,2,\ldots,2N\}$ (and similarly if the terms of the difference in (\ref{switchit}) is switched).  Furthermore, if one contour $\Gamma_p$ lies completely to the right of another $\Gamma_q$, then their combined homotopies send
\be (u_p-u_q)^{8/\kappa}\xrightarrow[\Gamma_p\times\Gamma_q\longrightarrow\bar{\Gamma}_p\times\bar{\Gamma}_q]{}(\bar{u}_p-\bar{u}_q)^{8/\kappa}=\overline{(u_p-u_q)^{8/\kappa}},\quad (u_p,u_q)\in\Gamma_p\times\Gamma_q.\ee
Finally, if one contour $\Gamma_q$ passes over another $\Gamma_p$, then $\Delta(u_p-u_q)$ (\ref{Deltadefn}) replaces the factor $(u_p-u_q)^{8/\kappa}$ in the integrand of $\mathcal{F}_{c,\vartheta}$.  In light of the above discussion, we see that the homotopy of $\Gamma_p$ followed by the homotopy of $\Gamma_q$ sends
\be\label{deltadeltastar} \Delta(u_p-u_q)\xrightarrow[\Gamma_p\times\Gamma_q\longrightarrow\bar{\Gamma}_p\times\bar{\Gamma}_q]{}\Delta^*(\bar{u}_p-\bar{u}_q)=\overline{\Delta(u_p-u_q)},\quad (u_p,u_q)\in\Gamma_p\times\Gamma_q,\ee
with the equality in (\ref{deltadeltastar}) following from (\ref{Deltadefn}, \ref{bardelta}).  Thus, the homotopies sending the integration contours of $\mathcal{F}_{c,\vartheta}$ to their mirror images altogether send every factor in the integrand of $\mathcal{F}_{c,\vartheta}$, every  integration measure of $\mathcal{F}_{c,\vartheta}$, and thus $\mathcal{F}_{c,\vartheta}$ itself, to its complex conjugate.  But because none of these homotopies change the function itself, we conclude that $\mathcal{F}_{c,\vartheta}$ is real-valued.  This concludes the proof for the case $\kappa>4$.

According to definition \ref{Fkdefn}, if $\kappa\leq4$, then the contours of $\mathcal{F}_{c,\vartheta}(\kappa)$ are not simple, but are Pochhammer contours.  In spite of this change, we may use the above reasoning together with (\ref{PochDecomp}) to show that $\mathcal{F}_{c,\vartheta}(\kappa)$ is real-valued if $\kappa\in(0,4]$ too.  However, we find it is easier to use this fact instead: if $f(z)$ is analytic in the strip $-b<\text{Re}\,z<b$ with $b>0$ and real-valued for all (real) $z\in(0,b)$, then $f(z)$ is real-valued for all (real) $z\in(-b,b)$.  To prove this claim, we use the Taylor series
\be\label{Taylor}f(z)=\sum_{n=0}^\infty a_n z^n,\quad |z|<b\ee
for $f(z)$ centered on zero.  Here, the first coefficient $a_0$ is real because $a_0=\lim_{z\downarrow0}f(z)$ is real.  Thus, the function $g$, defined as
\be g(z):=\left\{\begin{array}{ll} \dfrac{f(z)-a_0}{z}, & z\neq0 \\ a_1, & z=0\end{array}\right\}\quad\Longrightarrow\quad g(z)=\sum_{n=0}^\infty a_{n+1} z^n,\quad |z|<b,\ee
satisfies the same conditions as those of $f(z)$, so the first coefficient $a_1$ of its Taylor series is real too.  Continuing, we see that each coefficient $a_n$ of (\ref{Taylor}) is real, so $f(z)$ is real-valued for all (real) $z\in(-b,b)$.  After setting $b=4$, $z=\kappa-4$, and $f(z)=\mathcal{F}_{c,\vartheta}(\kappa)$ (which the discussion surrounding (\ref{Freplace}, \ref{explicit}) shows to be analytic on the strip $\{\kappa\in\mathbb{C}\,|\,0<\text{Re}(\kappa)<8\}$), it then follows that $\mathcal{F}_{c,\vartheta}(\kappa)$ is real-valued if $\kappa\in(0,4]$ too.
\end{proof}

\section{A proof of theorem \ref{vertexop}}\label{proofappendix}

In this appendix, we complete the proof of theorem \ref{vertexop} in section \ref{CGsolnsSect} by showing that any Coulomb gas solution of item \ref{itemdef1} in definition \ref{CGsolnsdef} satisfies the system (\ref{nullstate}, \ref{wardid}) for $\kappa>0$.  The proof is due to J.\ Dub\'edat \cite{dub}.  Here, we present his proof in a way that directly indicates the relationship between the conformal Ward identities (\ref{wardid}) and the Coulomb gas neutrality condition (\ref{vertexcorrformula}).  Throughout, we let $x_{2N+k}:=u_k$, and without loss of generality, we use Coulomb gas notation consistent with the dense phase ($\kappa>4$).  Thus, from (\ref{alphars}, \ref{densedilute}), we have
\begin{gather}\label{thecharges1}\alpha^+=\sqrt{\kappa}/2,\quad\alpha^-=-2/\sqrt{\kappa},\quad 2\alpha_0=\alpha^++\alpha^-,\\
\label{thecharges2}\alpha_{1,2}^+=-\alpha^-/2=1/\sqrt{\kappa},\quad\alpha_{1,2}^-=\alpha^++3\alpha^-/2=(\kappa-6)/2\sqrt{\kappa}.\end{gather}
In the dilute phase ($\kappa\leq4$), we switch $\alpha^\pm\mapsto\alpha^\mp$ and $\alpha_{1,2}^\pm\mapsto\alpha^\mp_{2,1}$.  This change does not affect the powers (\ref{powers1}--\ref{powers3}) that appear in (\ref{eulerintegralch2}).

We begin with a different construction of the Coulomb gas solution (\ref{CGsolns}) that more directly suggests why these solutions satisfy the null-state PDEs (\ref{nullstate}).  Working with real numbers, or ``charges," $\alpha_1,\alpha_2,\ldots,\alpha_{2N+M}$, we define
\be\label{Phi}\Phi(x_1,x_2,\ldots,x_{2N+M}):=\prod_{j<k}^{\mathclap{2N+M}}(x_k-x_j)^{2\alpha_j\alpha_k}.\ee
In the CFT Coulomb gas formalism, $\alpha_j$ is the charge associated with a chiral operator located at the point $x_j$, and (\ref{Phi}) is the formula (\ref{vertexcorrformula}) for the correlation function of this collection of operators.  Our strategy is to choose the $\alpha_j$ and $M$ such that for all $j\in\{1,2\ldots,2N\}$, we have
\be\label{actonthis}\Bigg[\frac{\kappa}{4}\partial_j^2+\sum_{k\neq j}^{2N}\left(\frac{\partial_k}{x_k-x_j}-\frac{(6-\kappa)/2\kappa}{(x_k-x_j)^2}\right)\Bigg]\Phi(x_1,x_2,\ldots,x_{2N+M})\,\,=\,\,\sum_{\mathclap{k=2N+1}}^{\mathclap{2N+M}}\,\,\partial_k(\,\,\ldots\,\,),\ee
where the ellipses on the right side stand for some analytic function of $x_1$, $x_2,\ldots,x_{2N+M}$.  Once we have done this, we integrate the coordinates $x_{2N+1}$, $x_{2N+2},\ldots,x_{2N+M}$ on both sides of (\ref{actonthis}) around closed, nonintersecting contours $\Gamma_1,$ $\Gamma_2,\ldots,\Gamma_M$ (such as nonintersecting Pochhammer contours).   Because either side of (\ref{actonthis}) is absolutely integrable on each path, we may perform these integrations in any order according to Fubini's theorem.  Integrating the right side of (\ref{actonthis}) therefore gives zero.  Finally, because the contours do not intersect, we have sufficient continuity to use the Leibniz rule of integration to exchange the order of differentiation and integration on the left side of (\ref{actonthis}).  (If $\Gamma_p$ intersects $\Gamma_q$ but $2\alpha_p\alpha_q>0$, then the contour integral $\oint\Phi$ is not improper.  In this event, we may still use the Leibniz rule to perform this last step as long as we may continuously deform these contours so they do not intersect.)  We therefore find that $F:=\oint\Phi$ satisfies the null-state PDEs (\ref{nullstate}).  We note that $M$ counts the number of screening charges to be used in the Coulomb gas construction (\ref{chiralrep}).  This is the plan for the proof, which we now begin.

With some algebra, we find that for any positive integer $M$, any collection of real ``conformal weights" $h_1,$ $h_2,\ldots,h_{2N+M}$ and ``charges" $\alpha_1$, $\alpha_2,\ldots,\alpha_{2N+M}$, and for each $j\in\{1,2,\ldots,2N+M\}$, we have
\begin{multline}\label{algebrach2}\Bigg[\frac{\kappa}{4}\partial_j^2\,\,+\,\,\sum_{k\neq j}^{\mathclap{2N+M}}\,\,\left(\frac{\partial_k}{x_k-x_j}-\frac{h_k}{(x_k-x_j)^2}\right)\Bigg]\Phi(x_1,x_2,\ldots,x_{2N+M})\\
=\Bigg[\sum_{\substack{k,l\neq j\\ k\neq l}}^{2N+M}\frac{\alpha_k\alpha_l(\kappa\alpha_j^2-1)}{(x_k-x_j)(x_l-x_j)}\,\,+\,\,\sum_{k\neq j}^{\mathclap{2N+M}}\frac{\alpha_j\alpha_k(\kappa\alpha_j\alpha_k-\kappa/2+2)-h_k}{(x_k-x_j)^2}\Bigg]\Phi(x_1,x_2,\ldots,x_{2N+M}).\end{multline}
We choose $h_k=(6-\kappa)/2\kappa$ for $k\in\{1,2,\ldots,2N\}$ and $h_k=1$ for $k\in\{2N+1,2N+2,\ldots,2N+M\}$. (These are the conformal weight of a one-leg boundary operator and a chiral operator with charge $\alpha^\pm$ respectively.)  With this choice, we may write (\ref{algebrach2}) as
\begin{multline}\label{totalderiv+ch2}\Bigg[\frac{\kappa}{4}\partial_j^2+\sum_{k\neq j}^{2N}\left(\frac{\partial_k}{x_k-x_j}-\frac{(6-\kappa)/2\kappa}{(x_k-x_j)^2}\right)\Bigg]\Phi(x_1,x_2,\ldots,x_{2N+M})\,\,=\,\,\sum_{\mathclap{k=2N+1}}^{\mathclap{2N+M}}\,\,\partial_k\left(-\frac{\Phi(x_1,x_2,\ldots,x_{2N+M})}{x_k-x_j}\right)\\
+\Bigg[\sum_{\substack{k,l\neq j\\ k\neq l}}^{2N+M}\frac{\alpha_k\alpha_l(\kappa\alpha_j^2-1)}{(x_k-x_j)(x_l-x_j)}\,\,+\,\,\sum_{k\neq j}^{\mathclap{2N+M}}\,\,\frac{\alpha_j\alpha_k(\kappa\alpha_j\alpha_k-\kappa/2+2)-h_k}{(x_k-x_j)^2}\Bigg]\Phi(x_1,x_2,\ldots,x_{2N+M})\end{multline}
for all $j\in\{1,2,\ldots,2N\}$. We recognize the differential operator of the $j$th null-state PDE (\ref{nullstate}) on the left side of (\ref{totalderiv+ch2}).  Now we choose a particular $j\neq2N$.  Next, if we choose the values for $\alpha_j$ and the elements of $\{\alpha_k\}_{k\neq j}$ as
\be\label{branchesch2}\alpha_j=\alpha_{1,2}^+=1/\sqrt{\kappa},\quad\alpha_k^{\pm}=\alpha_0\pm\sqrt{\alpha_0^2+h_k},\quad k\neq j,\ee
then the term in brackets on the right side of (\ref{totalderiv+ch2}) vanishes (for either choice of sign for $\alpha_k^{\pm}$), casting (\ref{totalderiv+ch2}) in the desired form (\ref{actonthis}) for this particular $j$.  

Next, we search for a choice of $\pm$ signs for the charges $\alpha_1^\pm$, $\alpha_2^\pm,\ldots,\alpha_{2N+M}^\pm$ in (\ref{branchesch2}) such that we achieve the form (\ref{actonthis}) not just for the one selected $j\in\{1,2,\ldots,2N\}$ that appears, but for all indices in this set.  We note that for $k\in\{1,2,\ldots,2N\}$, the choice $h_k=(6-\kappa)/2\kappa$ and (\ref{branchesch2}) imply $\alpha_k^{\pm}=\alpha_{1,2}^\pm$, and for $k>2N$, the choice $h_k=1$ and (\ref{branchesch2}) imply $\alpha_k^{\pm}=\alpha^\pm$.  This opens the possibility of achieving the desired form (\ref{actonthis}) for all $k\in\{1,2,\ldots,2N\}$.  We highlight two possible choices.
\begin{enumerate}
\item\label{thefirstitem} If we choose the $+$ sign for all $\alpha_j$ with $j\in\{1,2,\ldots,2N\}$, then the bracketed term on the right side of (\ref{totalderiv+ch2}) vanishes, and we have
\be\label{totalderivreg}\Bigg[\frac{\kappa}{4}\partial_j^2+\sum_{k\neq j}^{2N}\left(\frac{\partial_k}{x_k-x_j}-\frac{(6-\kappa)/2\kappa}{(x_k-x_j)^2}\right)\Bigg]\Phi(x_1,x_2,\ldots,x_{2N+M})\,\,=\,\,\sum_{\mathclap{k=2N+1}}^{\mathclap{2N+M}}\,\,\partial_k\left(-\frac{\Phi(x_1,x_2,\ldots,x_{2N+M})}{x_k-x_j}\right)\ee
for all $j\in\{1,2,\ldots,2N\}$.  Thus, we attain the desired form (\ref{actonthis}) for all $j\in\{1,2,\ldots,2N\}$.  Presently, $M$ and the signs for the $\alpha_k^\pm$ with $k\in\{2N+1,2N+2,\ldots,2N+M\}$ are still unspecified.
\item\label{theseconditem} If $M=N-1$ and we choose the $+$ sign for all $\alpha_j^\pm$ with $j\in\{1,2,\ldots,2N-1\}$, the $-$ sign for $\alpha_{2N}^\pm$, and the $-$ sign for all $\alpha_k^\pm$ with $k\in\{2N+1,2N+2,\ldots,3N-1\}$, then we have (\ref{totalderivreg}) for $j\in\{1,2,\ldots,2N-1\}$.  Thus, we attain the desired form (\ref{actonthis}) for all indices $j$ in this range.  Furthermore, J.\ Dub\'edat proved that \cite{dub}  
\begin{multline}\label{totalderivconj}\Bigg[\frac{\kappa}{4}\partial_{2N}^2+\sum_{k=1}^{2N-1}\left(\frac{\partial_k}{x_k-x_{2N}}-\frac{(6-\kappa)/2\kappa}{(x_k-x_{2N})^2}\right)\Bigg]\Phi(x_1,x_2,\ldots,x_{2N+M})\,\,=\,\,\sum_{\mathclap{k=2N+1}}^{3N-1}\partial_k\left(-\frac{\Phi(x_1,x_2,\ldots,x_{2N+M})}{x_k-x_{2N}}\right)\\
+\frac{1}{2}\sum_{k=2N+1}^{3N-1}\partial_k\Bigg[\frac{8-\kappa}{x_k-x_{2N}}\Bigg(\prod_{l=1}^{2N-1}\frac{x_k-x_l}{x_{2N}-x_l}\prod_{\substack{m=2N+1\\m\neq k}}^{3N-1}\left(\frac{x_{2N}-x_m}{x_k-x_m}\right)^2\Bigg) \Phi(x_1,x_2,\ldots,x_{2N+M})\Bigg].\end{multline}
Because the right side of (\ref{totalderivconj}) equals a sum of derivatives with respect to $x_k$ with $k\in\{2N+1,2N+2,\ldots,3N-1\}$, we attain the desired form (\ref{actonthis}) with $j=2N$ too.
\end{enumerate}
As previously discussed, the function $F:=\oint\Phi$ is annihilated by the differential operator on the left side of (\ref{totalderivreg}, \ref{totalderivconj}) for all $j\in\{1,2,\ldots,2N\}$ provided that none of the $M$ integration contours intersect, thus giving a solution of all of the null-state PDEs  (\ref{nullstate}) in either case.

In addition to satisfying the null-state PDEs (\ref{nullstate}), $F$ must also satisfy the conformal Ward identities (\ref{wardid}) too.  These identities imply that the function (with the integration contours described beneath (\ref{actonthis}))
\begin{multline}\label{Psi} G(x_1,x_2,\ldots,x_{2N})\,\,:=\,\,\prod_{\mathclap{j=1,\,\,\text{odd}}}^{2N}\,\,(x_{j+1}-x_j)^{6/\kappa-1}F(x_1,x_2,\ldots,x_{2N})\\
=\,\,\prod_{\mathclap{j=1,\,\,\text{odd}}}^{2N}\,\,(x_{j+1}-x_j)^{6/\kappa-1}\sideset{}{_{\Gamma_M}}\oint\dotsm\oint_{\Gamma_2}\sideset{}{_{\Gamma_1}}\oint\Phi(x_1,x_2,\ldots,x_{2N+M})\,{\rm d}x_{2N+1}\,{\rm d}x_{2N+2}\dotsm\,{\rm d}x_{2N+M}\qquad\end{multline}
is invariant under M\"{o}bius transformations, or equivalently, depends only on a set of $2N-3$ independent cross-ratios that we may form from $x_1,$ $x_2,\ldots,x_{2N}$ \cite{bpz, fms, henkel, florkleb}.  We choose these cross-ratios to be
\be\label{f}\lambda_i=f(x_i)\quad\text{with}\quad f(x):=\frac{(x-x_1)(x_{2N}-x_{2N-1})}{(x_{2N-1}-x_1)(x_{2N}-x)},\ee
so $\lambda_1=0<\lambda_2<\lambda_3<\ldots<\lambda_{2N-2}<\lambda_{2N-1}=1<\lambda_{2N}=\infty$.  Then this condition is equivalent to $G$ satisfying
\be\label{PsiRequirement}G(x_1,x_2,x_3,\ldots,x_{2N-2},x_{2N-1},x_{2N})=G(0,\lambda_2,\lambda_3,\ldots,\lambda_{2N-2},1,\infty).\ee

Next, we motivate a choice of $\pm$ signs for the charges $\alpha_1^\pm$, $\alpha_2^\pm,\ldots,\alpha_{2N+M}^\pm$ in (\ref{branchesch2}) for $G$ to fulfill the identity (\ref{PsiRequirement}), and then we verify that it indeed is satisfied.  We anticipate that the possible choices we find agree with those suggested in items \ref{thefirstitem} and \ref{theseconditem} above. Because the right side of (\ref{PsiRequirement}) is necessarily finite, we ignore any infinite factors  that result from setting $x_{2N}=\infty$ for now.   From (\ref{Phi}) and (\ref{Psi}), we see that for all $m\in\{1,2,\ldots,M\}$, the $m$th integral on the right side of (\ref{PsiRequirement})  has the form  
\be\label{firstint}\int\lambda_l^{\beta_1}(1-\lambda_l)^{\beta_{2N-1}}\prod_{j=2}^{2N-2}(\lambda_j-\lambda_l)^{\beta_j}\prod_{\substack{k=2N+1 \\ k\neq l}}^{2N+M}(\lambda_k-\lambda_l)^{\beta_k}\,{\rm d}\lambda_l,\quad l:=2N+m,\ee
with $\beta_k := 2 \alpha_k \alpha_l$, and the $m$th integral on the left side of (\ref{PsiRequirement}) has the form
\be\label{secondint}\int\prod_{j=1}^{2N}(x_j-x_l)^{\beta_j}\prod_{\substack{k=2N+1 \\ k\neq l}}^{2N+M}(x_k-x_l)^{\beta_k}\,{\rm d}x_l,\quad l:=2N+m.\ee
We note that the integrand of (\ref{secondint}) contains an extra factor that was dropped in (\ref{firstint}) after we sent $x_{2N}$ to infinity.  The simplest condition that is ostensibly consistent with (\ref{PsiRequirement}) is for the integrals (\ref{firstint}) and (\ref{secondint}) to be the same up to algebraic prefactors.  After the change of variables $\lambda_j=f(x_j)$, (\ref{firstint}) transforms into 
\be\label{transint}\mathcal{P}(x_1,x_2,\ldots,x_{2N})\int\prod_{j=1}^{2N-1}\left(\frac{x_j-x_l}{x_{2N}-x_l}\right)^{\beta_j}\prod_{\substack{k=2N+1 \\ k\neq l}}^{2N+M}\left(\frac{x_k-x_l}{x_{2N}-x_l}\right)^{\beta_k}\frac{{\rm d}x_l}{(x_{2N}-x_l)^2},\quad l:=2N+m,\ee
where $\mathcal{P}(x_1,x_2,\ldots,x_{2N})$ is an algebraic prefactor. To match the integral in (\ref{transint}) with that in (\ref{secondint}), we must have
\be\label{betacond}\beta_{2N}\,\,=\,\,-\sum_{\mathclap{k\neq2N, 2N+m}}\,\,\beta_k-2.\ee
That is, the sum $\sigma_m$ of the powers in (\ref{secondint}), 
\be\label{preneut}\sigma_m\,\,:=\,\,\sum_{\mathclap{k\neq 2N+m}}\,\,\beta_k\,\,=\,\,\sum_{\mathclap{k\neq 2N+m}}\,\,2\alpha_k\alpha_{2N+m}=2\alpha_{2N+m}\left(\sum_k\alpha_k-\alpha_{2N+m}\right)\ee
must equal negative two.  We recall that $\alpha_{2N+m}=\alpha^\pm$ for some choice of $\pm$ because $2N+m>2N$.  Thus, using the identities $\alpha^++\alpha^-=2\alpha_0$ and $\alpha^+\alpha^-=-1$ following from (\ref{alphapm}) and (\ref{screeningcharges}), we find that the Coulomb gas neutrality condition discussed in section \ref{CGsolnsSect} is satisfied if and only if $\sigma_m=-2$ for some $m\in\{1,2,\ldots,M\}$.
\be \label{neut}\sigma_m=2\alpha^\pm\left(\sum_k\alpha_k-\alpha^\pm\right)=-2\qquad\Longleftrightarrow\qquad\sum_k\alpha_k=2\alpha_0.\ee
This in turn implies that if $\sigma_m=-2$ for some $m\in\{1,2,\ldots,M\}$, then $\sigma_m=-2$ for all $m$ in this range.
  
Now we search for sign choices for (\ref{branchesch2}) and a value for $M$ such that the neutrality condition (\ref{neut}) is satisfied.  Without loss of generality, we write
\be\label{signchoice}\alpha_k=\begin{cases}\alpha_{1,2}^+,&1\leq k\leq p\\
\alpha_{1,2}^-,& p+1\leq k\leq 2N\\
\alpha^-,&2N+1\leq k\leq 2N+q\\
\alpha^+,& 2N+q+1\leq k\leq 2N+M\end{cases}\ee
for some $p\in\{0,1,\ldots,2N\}$ and $q\in\{0,1,\ldots,M\}$.  Letting $p':=2N-p$ and $q':=M-q$, (\ref{neut}) with (\ref{thecharges1}, \ref{thecharges2}) gives
\bea\label{sigmafirst}\sigma_m&=&\begin{cases}2\alpha^-[p\alpha_{1,2}^++p'\alpha_{1,2}^-+(q-1)\alpha^-+q'\alpha^+],&1\leq m\leq q\\ 2\alpha^+[p\alpha_{1,2}^++p'\alpha_{1,2}^-+q\alpha^-+(q'-1)\alpha^+],&q+1\leq m\leq M\end{cases}\\
\label{sigmasecond}&=&\left.\begin{cases}4\kappa^{-1}[-p+3p'+2(q-1)]-2(p'+q'),&1\leq m\leq q\\
\kappa(p'+q'-1)/2-(-p+3p'+2q),&q+1\leq m\leq M\end{cases}\right\}=-2\eea
for all $m\in\{1,2,\ldots,M\}$ and $\kappa>0$.

First, we suppose that $q=M$, so $q'=0$ and the bottom line of (\ref{sigmasecond}) gives $p'=1$ and $p=2M+1$.  Then the top line of (\ref{sigmasecond}) is also satisfied, and we see that $M=N-1$.  That is, we use the $\alpha_{1,2}^+$ charge for the points $x_1$, $x_2,\ldots,x_{2N-1}$, we use the $\alpha_{1,2}^-$ ``conjugate charge" for $x_{2N}$, we use the $\alpha^-$ screening charge for all $N-1$ integration variables, and we use no $\alpha^+$ screening charges for any integration variable.  This situation falls under item \ref{theseconditem} above.  So far, we have simply predicted a choice of $p$ and $q$ in (\ref{signchoice}) such that $\oint\Phi$ should satisfy the Ward identities (\ref{wardid}).  To prove that $\oint\Phi$ does indeed solve them if we use this choice, we show that $G$, defined in (\ref{Psi}), satisfies condition (\ref{PsiRequirement}).  We may do this by changing integration variables on the right side of (\ref{PsiRequirement}) from $\lambda_j$ to $x_j$ via $f$ in (\ref{f}) as described above, and doing some straightforward but lengthy algebra.  We omit the details.  This proves that linear combinations of the functions (\ref{CGsolns}) with $c=2N$ satisfy the system (\ref{nullstate}, \ref{wardid}).  Because the system is invariant under permutation of the points $x_1,$ $x_2,\ldots,x_{2N}$, we see that (\ref{CGsolns}), with $c$ equaling any index among $1$, $2,\ldots,2N-1$, satisfies this system too.  This proves theorem \ref{vertexop}.

Now we suppose that $q<M$ so $q'>0$.  Then the bottom line of (\ref{sigmasecond}) implies that $p'+q'-1=0$ and $-p+3p'+2q=2$. The first of these equations implies that $p'=0$ and $q'=1$, or $p=2N$ and $q=M-1$, and with these conditions, the second implies that $M=N+2$.  That is, we use the $\alpha_{1,2}^+$ charge for the points $x_1$, $x_2,\ldots,x_{2N}$, we do not use the $\alpha_{1,2}^-$ ``conjugate charge" for any of these points, we use the $\alpha^-$ screening charges for $N+1$ of the $N+2$ integration variables, and we use the $\alpha^+$ screening charges for the remaining integration variable.  This situation falls under item \ref{thefirstitem} above.  Again, it is possible to prove that $\oint\Phi$, with this choice of $p$ and $q$ for (\ref{signchoice}), satisfies the Ward identities (\ref{wardid}).

We did not pay attention to these $q<M$ solutions in this article because theorem \ref{maintheorem} renders them extraneous.  For example, if $N=1$ so $M=3$, then (writing $u_1=x_5$, $u_2=x_4$, and $u_3=x_3$)
\begin{multline}F(x_1,x_2)=(x_2-x_1)^{2/\kappa}\oint_{\Gamma_3}\oint_{\Gamma_2}\oint_{\Gamma_1} (u_3-x_1)^{-4/\kappa}(u_2-x_1)^{-4/\kappa}(u_1-x_1)(x_2-u_3)^{-4/\kappa}\\
\times(x_2-u_2)^{-4/\kappa}(x_2-u_1)(u_3-u_2)^{8/\kappa}(u_3-u_1)^{-2}(u_2-u_1)^{-2}\,{\rm d}u_1\,{\rm d}u_2\,{\rm d}u_3\end{multline}
is an element of $\mathcal{S}_1$.  By substituting $u_k(t_k)=(1-t_k)x_1+t_kx_2$ for $k=1,2,$ and 3, we may factor the dependence of the triple contour integral on $x_2-x_1$ out of the integrand to find that $F(x_1,x_2)$ is proportional to $(x_2-x_1)^{1-6/\kappa}$. Hence, $F$ is indeed an element of $\mathcal{S}_1$.  (See (\red{16}) of \cite{florkleb}.)

In this article, we restrict our attention to Coulomb gas solutions that obey the Coulomb gas neutrality condition, discussed above, because such functions manifestly satisfy the conformal Ward identities (\ref{wardid}).  Interestingly, K.\ Kyt\"ol\"a and E.\ Peltola have discovered Coulomb gas solutions that do not satisfy the neutrality condition but do satisfy the conformal Ward identities (\ref{wardid}) anyway (in addition to the usual null-state PDEs (\ref{nullstate})) \cite{kype2}.  According to theorem \ref{maintheorem}, these solutions are necessarily linear combinations of Coulomb gas solutions that do satisfy the neutrality condition.  Nonetheless, their existence is interesting.  In our nomenclature for this appendix, the solutions of K.\ Kyt\"ol\"a and E.\ Peltola fall under item \ref{thefirstitem} above, with $M=N$ and the $-$ sign chosen for each $\alpha_k^\pm$ with $k\in\{2N+1,2N+2,\ldots,2N+M\}$.

\end{document}